\documentclass{article}
\usepackage[top=30truemm,bottom=30truemm,left=25truemm,right=25truemm]{geometry}
\usepackage{lscape} 
\usepackage{comment} 

\usepackage[dvipdfmx]{graphicx}
\usepackage{xcolor}
\usepackage{multirow}
\usepackage{graphicx}
\usepackage{bm}
\usepackage{float}
\usepackage{subcaption}
\usepackage{amsmath}
\usepackage{amssymb}
\usepackage{amsthm}
\newtheorem{prop}{Proposition}
\newtheorem{thm}{Theorem}

\usepackage{algorithm,algorithmic}
\usepackage{amsmath}
\usepackage{amsfonts}
\usepackage[]{graphicx}
\usepackage{booktabs}
\chardef\bslash=`\\ 

\begin{document}

\title{Extension of W-method and A-learner \\ for multiple binary outcomes}

\author{Shintaro Yuki\thanks{Graduate School of Culture and Information Science, Doshisha University, Tataramiyakodani 1-3, Kyotanabe City, Kyoto, Japan.}\and
		Kensuke Tanioka\thanks{Department of Life and Medical Sciences, Doshisha University, Tataramiyakodani 1-3, Kyotanabe City, Kyoto, Japan.}  \and
			Hiroshi Yadohisa\thanks{Department of Culture and Information Science, Doshisha University, Tataramiyakodani 1-3, Kyotanabe City, Kyoto, Japan} 
	}

\date{}
\maketitle


\begin{abstract}
In this study, we compared two groups, in which subjects were assigned to either the treatment or the control group. In such trials, if the efficacy of the treatment cannot be demonstrated in a population that meets the eligibility criteria, identifying the subgroups for which the treatment is effective is desirable. Such subgroups can be identified by estimating heterogeneous treatment effects (HTE). In recent years, methods for estimating HTE have increasingly relied on complex models. Although these models improve the estimation accuracy, they often sacrifice interpretability. Despite significant advancements in the methods for continuous or univariate binary outcomes, methods for multiple binary outcomes are less prevalent, and existing interpretable methods, such as the W-method and A-learner, while capable of estimating HTE for a single binary outcome, still fail to capture the correlation structure when applied to multiple binary outcomes. We thus propose two methods for estimating HTE for multiple binary outcomes: one based on the W-method and the other based on the A-learner. We also demonstrate that the conventional A-learner introduces bias in the estimation of the treatment effect. The proposed method employs a framework based on reduced-rank regression to capture the correlation structure among multiple binary outcomes. We correct for the bias inherent in the A-learner estimates and investigate the impact of this bias through numerical simulations. Finally, we demonstrate the effectiveness of the proposed method using a real data application.
\end{abstract}


\section{Introduction}
Clinical trials are currently being conducted to evaluate the efficacy of treatment as an outcome. We specifically compared the two groups in these trials. The subjects were allocated to either the treatment or the control group. In such comparative studies, even if the efficacy of a treatment cannot be demonstrated in the entire eligible population, the treatment may be effective for specific subgroups with particular characteristics. The exploratory identification of such subgroups can be meaningful (e.g., \cite{rothwell2005subgroup}, \cite{strobl2009introduction}, \cite{foster2011subgroup}, and \cite{shen2015inference}).
In recent years, clinical statistics research has aimed at providing optimal treatments based on the characteristics of each subject, known as personalized medicine \cite{kosorok2015adaptive}. Per the studies, estimating the treatment effects tailored to individual subjects, known as heterogeneous treatment effects (HTE), can help us develop more effective treatment strategies from the perspective of personalized medicine.
\par
Statistical methods for estimating HTE in comparisons between two groups are continuously advancing. 
By simply transforming data based on propensity scores and treatment assignments, these methods can estimate HTE in both randomized controlled trials (RCTs) and observational studies, making them easy to implement.
These methods were designed to estimate HTE for a single outcome. However, clinical trials have multiple outcomes, including primary and secondary endpoints that tend to exhibit high correlations. Methods developed for a single outcome are inadequate for capturing complex correlation structures among multiple outcomes. To address these correlation structures, scholars have proposed methods for estimating the HTE and identifying subgroups using linear functions for multiple continuous outcomes \cite{siriwardhana2020personalized} \cite{kulasekera2022quantiles}. Despite the existence of clinical statistics research that addresses multiple binary outcomes (e.g., \cite{williams1996design}; \cite{inan2017joint}; \cite{dunson2000bayesian}), methods for estimating HTE and identifying subgroups for such outcomes are still under development. Therefore, we estimated HTE and identified subgroups for multiple binary outcomes. In one such study, Yuki et al. (2023) propose a method that extends the W-method to handle multiple binary outcomes, referred to here as the multiple binary W-method (MBWM) \cite{yuki2023estimation}. In the MBWM, based on the principal component regression framework \cite{massy1965principal}, latent variables are introduced into a multivariate linear regression model with multiple binary outcomes as response variables. The subgroup characteristics are interpreted using the loadings of these latent variables. 
However, the existing methods have several issues, which are described as follows:
\begin{enumerate}
\item In MBWM, since the loading matrix used for interpretation does not satisfy the orthogonality constraint, it is possible that completely independent loadings have not been obtained.
In such cases, the presence of correlations between factors makes it difficult to interpret the subgroups.
    \item 
To our knowledge, the A-learner for multiple binary outcomes has not been proposed, and no study uniformly addresses both the W-method and A-learner for multiple binary outcomes.
    \item In the conventional A-learner for a univariate binary outcome, the estimation of HTE contains bias.
As a result, treatments selected based on these estimated effects may not be optimal for the subjects.
\end{enumerate}
For the first issue, we adopt a model based on the framework of reduced rank regression \cite{izenman1975reduced}, specifically logistic reduced rank regression \cite{yee2003reduced}, to ensure the orthogonality of the loadings used for interpretation. This approach facilitates the interpretation of subgroups. 
For the second issue, similar to the approach for the second issue, we extend A-learner to be applicable to multivariate binary outcomes based on the framework of reduced-rank regression. By comparing these results with those of the extended W-method described in the second issue, we provide a unified discussion of their effectiveness.
For the third issue, we formalized the bias inherent in the treatment effect estimated by A-learner and eliminated it to accurately estimate the treatment effect. This approach is expected to enable the development of personalized medical treatments.
\par
The remaining paper is organized as follows. In section 2, we describe the notation and definition of HTE as well as the multiple logistic loss employed in the proposed method. We then present the objective function of the proposed method based on the W-method and discuss its properties. We also provide a theorem regarding the bias in the conventional A-learner and describe the objective function of the proposed method based on A-learner, along with its properties.
Section 3 details the parameter estimation methods for the two approaches proposed in section 2 based on logistic reduced-rank regression.
In section 4, we present a numerical study that compares the proposed method with other methods. Section 5 describes the application of the proposed method to real data obtained from a two-arm comparison trial to interpret the resulting subgroups. Section 6 concludes the paper. Proofs of the theorems and propositions as well as figures related to the results of the numerical study can be found in the supplementary material.

\section{Estimating treatment effects for multiple binary outcomes}
In this section, we first define the notation used to describe the proposed method, the HTE for multiple binary outcomes, and the multiple logistic loss for the proposed method. We then formulate the objective function of the proposed method based on the W method using multiple logistic losses, and demonstrate its properties. We then derive the bias inherent in the treatment effect estimation of the conventional A-learner. Based on this, we show that the proposed method using the multiple logistic loss based on the A-learner also introduces bias, formulates it, and discusses the properties of the objective function, similar to the W-method.
\par
Let $T_i$ be the random variable.
\begin{align}\nonumber
    T_i= \left\{
\begin{array}{ll}
1 & (\mathrm{If\ subject\ \textit{i}\ is\ allocated\ to\ test\ therapy}) \\
-1 & (\mathrm{If\ subject\ \textit{i}\ is\ allocated\ to\ control\ therapy})
\end{array}
\right. 
\end{align}
and we assume that the strongly ignorable assumption \cite{rosenbaum1983central} has been satisfied for $T_i$. We denote $\bm{T}=\mathrm{diag}(t_{1},t_{2},\cdots,t_{n})$ as a binary treatment indicator matrix and $t_i$ as the observed value of $T_i$. Let $\boldsymbol{X}^{(\mathrm{r.v.})}=(X_1,X_2,\cdots,X_n)$ be the matrix of random variable corresponding to the explanatory variables with $n$ subjects and $p$ variables, where $X_i=(X_{i1},X_{i2},\cdots,X_{ip})^\prime\in\mathbb{R}^{p}$ is a vector of random variables corresponding to the explanatory variables of subject $i$, $\boldsymbol{X}=(\boldsymbol{x}_{1},\boldsymbol{x}_{2},\cdots,\boldsymbol{x}_{n})^{\prime}=(\boldsymbol{x}_{(1)},\boldsymbol{x}_{(2)},\cdots,\boldsymbol{x}_{(p)})\in\mathbb{R}^{n\times p}$ is an observed value of $\boldsymbol{X}^{(\mathrm{r.v.})}$, and $\boldsymbol{x}_{i}=(x_{i1},x_{i2},\cdots,x_{ip})^{\prime}$ for the subject $i$. $\bm{X}$ is centered and the intercept term is not considered.
The propensity score is defined as $\pi(\bm{x}_i)=P(T_i=1|\bm{x}_i)$ \cite{rosenbaum1983central}. In a randomized controlled trial, it is set as a constant, independent of the explanatory variables; for example, $\forall i; P(T_i=1)=P(T_i=-1)=1/2$. In observational studies, it is typically estimated using regression modeling techniques. Additionally, we assume that it is satisfied with positivity; that is, $\forall i;0<\pi(\bm{x}_i)<1$.
$\boldsymbol{Y}^{(\mathrm{r.v.})}=(Y_1,Y_2,\cdots,Y_n)^\prime\in\{0,1\}^{n\times m}$ is the matrix of random variables corresponding to the multiple binary outcomes with $n$ subjects, where $Y_{i}=(Y_{i1},Y_{i2},\cdots,Y_{im})^{\prime}\in\{0,1\}^{m}$ is a vector of random variables corresponding to the outcomes of subject $i$, $\bm{Y}=(\bm{y}_{1},\bm{y}_{2},\cdots,\bm{y}_{n})^\prime$ is an observed value of $\bm{Y}$, and $\bm{y}_i=(y_{i1},y_{i2},\cdots,y_{im})^\prime$ for the subject $i$.
The potential outcomes are defined as $Y_i^{(1)}$ and $Y_i^{(-1)}$, representing the outcomes for subject $i$ when assigned to the test therapy and control therapy, respectively \cite{rosenbaum1983central}. However, because subject $i$ can only be allocated to one of the two groups, the outcomes can be expressed as $Y_i=I(T_i=1)Y_i^{(1)}+I(T_i=-1)Y_i^{(-1)}$, where $I(\cdot)$ is the indicator function. 

\subsection{Definition of HTE and Multiple logistic loss}
In this subsection, we define HTE for binary outcomes and introduce the multiple logistic losses used in our proposed method. Similar to the settings in W-method and A-learner, the HTE for binary outcomes is defined as the log odds ratio and can be expressed as a linear function to ensure interpretability. Specifically, for the $j$th outcome of $i$th subject,
\begin{align}
\label{HTE}
    \log \frac{E[Y_{ij}|T_i=1,X_i=\bm{x}_i]}{E[Y_{ij}|T_i=-1,X_i=\bm{x}_i]}
    =\bm{\gamma}_j^{*\prime}\bm{x}_i\quad(i=1,2,\cdots,n;\;j=1,2,\cdots,m)
    ,
\end{align}
where $\bm{\gamma}_j^*\in\mathbb{R}^p$ is the true regression coefficient vector for the treatment effects. In matrix notation, the treatment effects for multiple binary outcomes are expressed as $\bm{X\varGamma}^*$, where $\bm{\varGamma}^*=(\bm{\gamma}^*_1,\bm{\gamma}^*_2,\cdots,\bm{\gamma}^*_m)\in\mathbb{R}^{p\times m}$. Here, assuming that there is a correlation between the binary outcomes, we impose a low-rank structure on $\bm{\varGamma}^*$, that is, $\mathrm{rank}(\bm{\varGamma}^*)=r\leq \min(p,m)$, allowing us to consider the relationships among the outcomes. Then, $\bm{\varGamma}^*$ can be expressed as $\bm{\varGamma}^*=\bm{W}^*\bm{V}^{*\prime}$, where $\bm{W}^*\in\mathbb{R}^{p\times r}$ and $\bm{V}^*=(\bm{v}_{(1)}^{*\prime},\bm{v}_{(2)}^{*\prime},\cdots,\bm{v}_{(m)}^{*\prime})\in\mathbb{R}^{m\times r}$ are of full ranks \cite{reinsel2022multivariate}.
From these things, Eq. (\ref{HTE}) can be rewritten as follows: 
\begin{align}
\label{HTE_lowrank}
    \log \frac{E[Y_{ij}|T_i=1,X_i=\bm{x}_i]}{E[Y_{ij}|T_i=-1,X_i=\bm{x}_i]}
    =\bm{v}_{(j)}^{*\prime}\bm{W}^{*\prime}\bm{x}_i.
\end{align}
\par
Here, we consider situations in which W-method and learner A are applied to each outcome $Y_{ij}$. In the conventional W-method and A-learner, we use the logistic loss function $M(Y_{ij},\bm{\gamma}_j^\prime\bm{x}_i)=Y_{ij}\log \{1+\exp (-\bm{\gamma}_j^\prime\bm{x}_i)\}$ to estimate the treatment effects on $Y_{ij}$. In both the W-method and A-learner, using $\hat{\bm{\varGamma}}=(\hat{\bm{\gamma}}_1,\hat{\bm{\gamma}}_2,\cdots,\hat{\bm{\gamma}}_m)$ minimizes the optimization problem related to $M(Y_{ij},\bm{\gamma}_j^\prime\bm{x}_i)$, and $\bm{X}\hat{\bm{\varGamma}}$ represents the treatment effects (note that section \ref{bias_sec} demonstrates that the A-learner incurs bias).
However, this model does not account for the correlation structure of the true treatment effects, as shown in Eq. (\ref{HTE_lowrank}). Therefore, we impose a rank constraint, $\bm{\varGamma}=\bm{WV}^\prime$, where $\bm{W}\in\mathbb{R}^{p\times r}$ and $\bm{V}=(\bm{v}_{(1)}^\prime,\bm{v}_{(2)}^\prime,\cdots,\bm{v}_{(m)}^\prime)\in\mathbb{R}^{m\times r}$. Thus, we define multiple logistic loss for estimating HTE with correlated multiple binary outcomes as follows:
\begin{align}
\label{logi}
    M(Y_i,\bm{VW}^\prime\bm{x}_i)
    =\sum_{j=1}^m M(Y_{ij},\bm{v}_{(j)}^\prime\bm{W}^\prime \bm{x}_i)=\sum_{j=1}^m Y_{ij}\log \{1+\exp (-\bm{v}_{(j)}^\prime\bm{W}^\prime \bm{x}_i)\}.
\end{align}
\subsection{Objective function using multiple logistic loss}
\label{obj_sec}
In this subsection, we extend the W-method and A learner methods to handle multiple binary outcomes. Traditionally, these methods estimate HTE for a single binary outcome by weighting and minimizing a convex objective function with respect to regression coefficients using propensity scores. However, our proposed objective function uses multiple logistic losses to capture the correlation structure among multiple binary outcomes, making it biconvex with respect to the parameters $\bm{V}$ and $\bm{W}$.
\par
\noindent \textbf{W-method:}

In the W-method, the HTE for multiple binary outcomes can be estimated by weighting the convex multiple logistic loss with propensity scores and treatment indicators, and then minimizing the expected value. We define the loss function $\ell_W(\bm{V},\bm{W},\bm{x}_i)$ based on the W method using multiple logistic loss as follows:
\begin{align*}
    \ell_W(\bm{V},\bm{W},\bm{x}_i)
    =E\left[\frac{M(Y_i,T_i\bm{VW}^\prime\bm{x}_i)}{T_i\pi(\bm{x}_i)+(1-T_i)/2}\middle|X_i=\bm{x}_i\right].
\end{align*}
We also formulate the objective function as $L_W(\bm{V},\bm{W},\bm{X})$, which is an empirical distribution representation of $\ell_W(\bm{V},\bm{W},\bm{X})$.
Specifically, $L_{W}(\bm{V},\bm{W},\bm{X})$ can be expressed as 
\begin{align*}
    L_W(\bm{V},\bm{W},\bm{X})
    =\sum_{i=1}^n \sum_{j=1}^m \frac{M\left(y_{ij},t_i\bm{v}_{(j)}^\prime\bm{W}^\prime\bm{x}_i\right)}{t_i\pi(\bm{x}_i)+(1-t_i)/2}
    =\sum_{i=1}^n \sum_{j=1}^m \frac{y_{ij}\log \left\{1+\exp\left(-t_i\bm{v}_{(j)}^\prime\bm{W}^\prime\bm{x}_i\right)\right\}}{t_i\pi(\bm{x}_i)+(1-t_i)/2}
\end{align*}
Thus, we solve the following optimization problem:
\begin{align*}
    \underset{\bm{V},\bm{W}}{\min}\;L_W(\bm{V},\bm{W},\bm{X})\;\mathrm{subject\;to\;}\bm{V}^\prime\bm{V}=\bm{I}_r,
\end{align*}
where $\bm{I}_r\in\mathbb{R}^{r\times r}$ is the identity matrix.
An orthogonality constraint is imposed on $\bm{V}$ to determine $\bm{V}$ and $\bm{W}$ uniquely \cite{reinsel2022multivariate}. Once the estimates $\hat{\bm{W}}$ and $\hat{\bm{V}}$ for $\bm{W}$ and $\bm{V}$ are obtained, they can be used to represent the treatment effects of the multiple binary outcomes. We discuss the properties of this objective function and parameter estimation algorithm in sections \ref{bias_sec} and \ref{est_w}, respectively.
\par
\noindent \textbf{A-learner:}

In contrast to the W-method, the A-learner estimates the HTE for multiple binary outcomes by directly weighting the parameters of the multiple logistic loss with propensity scores and treatment indicators and then minimizing its expected value. Using multiple logistic losses in Eq. (\ref{logi}), we define the loss function $\ell_A(\bm{V},\bm{W},\bm{x}_i)$ as 
\begin{align*}
    \ell_A(\bm{V},\bm{W},\bm{x}_i)
    =E\left[M(Y_i,\{(T_i+1)/2-\pi(\bm{x}_i)\}\times \bm{VW}^\prime\bm{x}_i)|X_i=\bm{x}_i\right].
\end{align*}
We formulate the objective function as $L_A(\bm{V},\bm{W},\bm{X})$, an empirical distribution representation of $\ell_A(\bm{V},\bm{W},\bm{X})$; that is,
\begin{align*}
    L_A(\bm{V},\bm{W},\bm{X})
    &=\sum_{i=1}^n \sum_{j=1}^m M(y_{ij},\{(t_i+1)/2-\pi(\bm{x}_i)\}\times \bm{v}_{(j)}^\prime\bm{W}^\prime\bm{x}_i)
    \\
    &=\sum_{i=1}^n \sum_{j=1}^m y_{ij}\log \left\{1+\exp\left(-\{(t_i+1)/2-\pi(\bm{x}_i)\}\times \bm{v}_{(j)}^\prime\bm{W}^\prime\bm{x}_i\right)\right\}
\end{align*}
Thus, we solve the following optimization problem:
\begin{align*}
    \underset{\bm{V},\bm{W}}{\min}\;L_A(\bm{V},\bm{W},\bm{X})\;\mathrm{subject\;to\;}\bm{V}^\prime\bm{V}=\bm{I}_r.
\end{align*}
We discuss the properties of this objective function and parameter estimation algorithm in sections \ref{bias_sec} and \ref{est_A}, respectively.
\subsection{Properties of objective functions}
\label{bias_sec}
In this subsection, we first discuss the properties of the objective function $\ell_W(\bm{V},\bm{W},\bm{x}_i)$ based on the W-method proposed in section \ref{obj_sec}.
Next, we demonstrate that the conventional A-learner for a univariate binary outcome introduces bias in the estimation of HTE. Finally, we describe the properties of the objective function $\ell_A(\bm{V},\bm{W},\bm{x}_i)$ based on A-learner proposed in section \ref{obj_sec}. Because this objective function is based on the A-learner, it allows for the formalization of the bias introduced in the estimation of HTE. By correcting this bias, an accurate estimation of HTE becomes possible.
\par
First, we describe the properties of the objective function $\ell_W(\bm{V},\bm{W},\bm{x}_i)$ using the W method. This objective function is biconvex with respect to parameters $\bm{V}$ and $\bm{W}$. Therefore, by fixing one parameter, an optimal solution for the other parameters can be obtained. In particular, if the HTE has a low-rank structure as described in Eq. $(\ref{HTE_lowrank})$, the optimal solution can be obtained by fixing one parameter as a column full-rank matrix. Through the following theorems, we demonstrate that this optimal solution represents the HTE with respect to each parameter.
\begin{thm}
Suppose that the true ranks of $\bm{V}^*$ and $\bm{W}^*$ representing HTE in Eq. (2) are both $r$.
When $\bm{W}$ is fixed to the matrix $\bm{W}^\dagger\in\mathcal{W}_r$, the following holds using $\hat{\bm{v}}_{(j)}$, which minimizes the $\ell_W(\bm{V},\bm{W}^\dagger,\bm{x}_i)$:
    \begin{align*}
        \log \frac{E[Y_{ij}|T_i=1,X_i=\bm{x}_i]}{E[Y_{ij}|T_i=-1,X_i=\bm{x}_i]}=\hat{\bm{v}}_{(j)}^\prime\bm{W}^{\dagger\prime}\bm{x}_i,
    \end{align*}
    where $\mathcal{W}_r=\{\bm{W}\in\mathbb{R}^{p\times r}|\mathrm{rank}(\bm{W})=r\}$.
\end{thm}
\begin{proof}
    First, $\ell_W(\bm{V},\bm{W}^\dagger,\bm{x}_i)$ is expanded as 
    \begin{align*}
        \ell_W(\bm{V},\bm{W}^\dagger,\bm{x}_i)
        &=E\left[\frac{M\left(Y_i,T_i\bm{VW}^{\dagger\prime}\bm{x}_i\right)}{T_i\pi(\bm{x}_i)+(1-T_i)/2}\middle|X_i=\bm{x}_i\right]
    \\
    &=\sum_{j=1}^m\left\{E\left[M\left(Y_{ij},\bm{v}_{(j)}^\prime\bm{W}^{\dagger\prime} \bm{x}_i\right)|T_i=1,X_i=\bm{x}_i\right]
    \right.
    \\
    &\left.
    +E\left[M\left(Y_{ij},-\bm{v}_{(j)}^\prime\bm{W}^{\dagger\prime} \bm{x}_i\}\right)|T_i=-1,X_i=\bm{x}_i\right]\right\}
    \\
    &=\sum_{j=1}^m\left\{E[Y_{ij}|T_i=1,X_i=\bm{x}_i]\log \{1+\exp(-\bm{v}_{(j)}^\prime\bm{W}^{\dagger\prime} \bm{x}_i)\}
    \right.
    \\
    &\left.
    +E[Y_{ij}|T_i=-1,X_i=\bm{x}_i]\log \{1+\exp(\bm{v}_{(j)}^\prime\bm{W}^{\dagger\prime} \bm{x}_i)\}\right\}.
    \end{align*}
By differentiating this with respect to $\bm{v}_{(j)}$ and substituting the optimal value $\hat{\bm{v}}_{(j)}$ with the first-order condition, it holds that
\begin{align*}
E[Y_{ij}|T_i=1,X_i=\bm{x}_i]\frac{\exp(-\hat{\bm{v}}_{(j)}^\prime\bm{W}^{\dagger\prime} \bm{x}_i)}{1+\exp(-\hat{\bm{v}}_{(j)}^\prime\bm{W}^{\dagger\prime} \bm{x}_i)}\bm{W}^{\dagger\prime}\bm{x}_i
&=E[Y_{ij}|T_i=-1,X_i=\bm{x}_i]\frac{\exp(\hat{\bm{v}}_{(j)}^\prime\bm{W}^{\dagger\prime} \bm{x}_i)}{1+\exp(\hat{\bm{v}}_{(j)}^\prime\bm{W}^{\dagger\prime} \bm{x}_i)}\bm{W}^{\dagger\prime}\bm{x}_i.
\end{align*}
By transforming this, we derive
\begin{align*}
\log \frac{E[Y_{ij}|T_i=1,X_i=\bm{x}_i]}{E[Y_{ij}|T_i=-1,X_i=\bm{x}_i]}&=\hat{\bm{v}}_{(j)}^\prime\bm{W}^{\dagger\prime}\bm{x}_i,
    \end{align*}
    where $\bm{0}_r$ is an $r$-dimensional vector with zero elements.
\end{proof}
Given that the estimated $\hat{\bm{V}}$ and the provided $\bm{W}^\dagger$ are column full rank, $\bm{V}$ and $\bm{W}$ can be determined such that $\bm{V}^\prime \bm{V} = \bm{I}_r$ from $\bm{W}^\dagger\hat{\bm{V}}^\prime$.
\begin{thm}
Suppose that the true ranks of $\bm{V}^*$ and $\bm{W}^*$ representing HTE in Equation (2) are both $r$.
    When $\bm{V}$ is fixed to the matrix $\bm{V}^\dagger\in\mathcal{V}_r$, the following holds when using $\hat{\bm{W}}$, which minimizes $\ell_W(\bm{V}^\dagger,\bm{W},\bm{x}_i)$:
    \begin{align*}
        \log \frac{E[Y_{ij}|T_i=1,X_i=\bm{x}_i]}{E[Y_{ij}|T_i=-1,X_i=\bm{x}_i]}=\bm{v}_{(j)}^{\dagger\prime}\hat{\bm{W}}^\prime\bm{x}_i,
    \end{align*}
    where $\mathcal{V}_r=\{\bm{V}\in\mathbb{R}^{m\times r}|\mathrm{rank}(\bm{V})=r\}$.
\end{thm}
\begin{proof}
    First, $\ell_W(\bm{V}^\dagger,\bm{W},\bm{x}_i)$ is expanded as 
    \begin{align*}
        \ell_W(\bm{V}^\dagger,\bm{W},\bm{x}_i)
    &=E\left[\frac{M\left(Y_i,T_i\bm{V}^\dagger\bm{W}^\prime\bm{x}_i\right)}{T_i\pi(\bm{x}_i)+(1-T_i)/2}\middle|X_i=\bm{x}_i\right]
    \\
    &=\sum_{j=1}^m\left\{E\left[M\left(Y_{ij},\bm{v}_{(j)}^{\dagger\prime}\bm{W}^\prime \bm{x}_i\right)|T_i=1,X_i=\bm{x}_i\right]
    \right.
    \\
    &\left.
    +E\left[M\left(Y_{ij},-\bm{v}_{(j)}^{\dagger\prime}\bm{W}^\prime \bm{x}_i\}\right)|T_i=-1,X_i=\bm{x}_i\right]\right\}
    \\
    &=\sum_{j=1}^m\left\{E[Y_{ij}|T_i=1,X_i=\bm{x}_i]\log \{1+\exp(-\bm{v}_{(j)}^{\dagger\prime}\bm{W}^\prime \bm{x}_i)\}
    \right.
    \\
    &\left.
    +E[Y_{ij}|T_i=-1,X_i=\bm{x}_i]\log \{1+\exp(\bm{v}_{(j)}^{\dagger\prime}\bm{W}^\prime \bm{x}_i)\}\right\}.
    \end{align*}
By differentiating this with respect to $\bm{W}$ and substituting the optimal value $\hat{\bm{W}}$, it holds from $\bm{v}_{(j)}^{\dagger\prime}\bm{W}^\prime\bm{x}_i=\mathrm{tr}(\bm{W}^\prime\bm{x}_i\bm{v}_{(j)}^{\dagger\prime})$ implies $\partial (\bm{v}_{(j)}^{\dagger\prime}\bm{W}^\prime\bm{x}_i)/\partial \bm{W}=\bm{x}_i\bm{v}_{(j)}^{\dagger\prime}$, where $\mathrm{tr}(\cdot)$ is trace operator for a matrix, that
    \begin{align*}
        E[Y_{ij}|T_i=1,X_i=\bm{x}_i]\frac{\exp(-\bm{v}_{(j)}^{\dagger\prime}\hat{\bm{W}}^\prime \bm{x}_i)}{1+\exp(-\bm{v}_{(j)}^{\dagger\prime}\hat{\bm{W}}^\prime \bm{x}_i)}\bm{x}_i\bm{v}_{(j)}^{\dagger\prime}=E[Y_{ij}|T_i=-1,X_i=\bm{x}_i]\frac{\exp(\bm{v}_{(j)}^{\dagger\prime}\hat{\bm{W}}^\prime \bm{x}_i)}{1+\exp(\hat{\bm{v}}_{(j)}^\prime\bm{W}^{\dagger\prime} \bm{x}_i)}\bm{x}_i\bm{v}_{(j)}^{\dagger\prime}.
\end{align*}
By transforming this, the following can be derived:
\begin{align*}
\log \frac{E[Y_{ij}|T_i=1,X_i=\bm{x}_i]}{E[Y_{ij}|T_i=-1,X_i=\bm{x}_i]}=\bm{v}_{(j)}^{\dagger\prime}\hat{\bm{W}}^\prime \bm{x}_i,
\end{align*}
\end{proof}
These theorems demonstrate that the treatment effects can be represented using the estimated $\bm{V}$ and $\bm{W}$ provided that the rank is correctly specified.
\par
Subsequently, we demonstrate that using A-learner to estimate the treatment effects for a univariate binary outcome results in bias. For this purpose, we consider the following minimization problem in the A-learner for a univariate binary outcome:
\begin{gather}
\label{obj_uniA}
\underset{\bm{\gamma}_j}{\min}\;\ell^*_A(\bm{\gamma}_j,\bm{x}_i),
\\
\nonumber
\mathrm{where}\;\ell^*_A(\bm{\gamma}_j,\bm{x}_i)=E[M(Y_{ij},\{(T_i+1)/2-\pi(\bm{x}_i)\}\times \bm{\gamma}_j^\prime\bm{x}_i)|X_i=\bm{x}_i]
.
\end{gather}
Chen et al. (2017) state that with $\hat{\bm{\gamma}}_j$, which minimizes Eq. (\ref{obj_uniA}), HTE can be expressed as
\begin{align*}
    \log \frac{E[Y_{ij}|T_i=1,X_i=\bm{x}_i]}{E[Y_{ij}|T_i=-1,X_i=\bm{x}_i]}
    =\hat{\bm{\gamma}}_j^\prime\bm{x}_i\quad(i=1,2,\cdots,n;\;j=1,2,\cdots,m).
\end{align*}
However, it can be shown that the estimated HTE is biased, as demonstrated in the following theorem.
\begin{thm}
\label{thm_bias}
Suppose the univariate loss function for a univariate binary outcome $M(Y_{ij},\bm{\gamma}_j^\prime\bm{x}_i)=Y_{ij}\log \{1+\exp (-\bm{\gamma}_j^\prime\bm{x}_i)\}$.
With $\hat{\bm{\gamma}}_j$, which minimizes $\ell_A^*(\bm{\gamma}_j,\bm{x}_i)$ in Eq. (\ref{obj_uniA}), the treatment effect can be expressed as 
\begin{align}
\label{bias_mod}
    \log \frac{E[Y_{ij}|T_i=1,X_i=\bm{x}_i]}{E[Y_{ij}|T_i=-1,X_i=\bm{x}_i]}
    =\hat{\bm{\gamma}}_j^\prime\bm{x}_i
    +\log \frac{1+e^{-\{1-\pi(\bm{x}_i)\}\hat{\bm{\gamma}}_j^\prime\bm{x}_i}}{1+e^{\pi(\bm{x}_i)\hat{\bm{\gamma}}_j^\prime\bm{x}_i}}.
\end{align}
\end{thm}
\begin{proof}
First, Eq. ($\ref{obj_uniA}$) can be written as 
    \begin{align*}
\ell_A^*(\bm{\gamma}_j,\bm{x}_i)&=E[M(Y_{ij},\{(T_i+1)/2-\pi(\bm{x}_i)\}\times \bm{\gamma}_j^\prime\bm{x}_i)|X_i=\bm{x}_i]
\\
&=E[I(T_i=1)M(Y_{ij},\{1-\pi(\bm{x}_i)\}\bm{\gamma}_j^\prime\bm{x}_i)|X_i=\bm{x}_i]
    \\
    &+E[I(T=-1)M(Y_{ij},-\pi(\bm{x}_i)\bm{\gamma}_j^\prime\bm{x}_i)|X_i=\bm{x}_i]
    \\
    &=\pi(\bm{x}_i)E[M(Y_{ij},\{1-\pi(\bm{x}_i)\}\bm{\gamma}_j^\prime\bm{x}_i)|T_i=1,X_i=\bm{x}_i]
    \\
    &+\{1-\pi(\bm{x}_i)\}E[M(Y_{ij},-\pi(\bm{x}_i)\bm{\gamma}_j^\prime\bm{x}_i)|T_i=-1,X_i=\bm{x}_i]
    \\
&=\pi(\bm{x}_i)E[Y_{ij}\log \{1+\exp (-\{1-\pi(\bm{x}_i)\}\bm{\gamma}_j^\prime\bm{x}_i)\}|T_i=1,X_i=\bm{x}_i]
    \\
    &+
    \{1-\pi(\bm{x}_i)\}E[Y_{ij}\log \{1+\exp (\pi(\bm{x}_i)\bm{\gamma}_j^\prime\bm{x}_i)\}|T_i=-1,X_i=\bm{x}_i]
\end{align*}
Differentiating this with respect to $\bm{\gamma}_j$ and substituting the optimal value $\hat{\bm{\gamma}}_j$ from the first-order condition holds that 
\begin{align*}
    &\pi(\bm{x}_i)\{1-\pi(\bm{x}_i)\} E[Y_{ij}|T_i=1,X_i=\bm{x}_i]\frac{e^{-\{1-\pi(\bm{x}_i)\}\hat{\bm{\gamma}}_j^\prime\bm{x}_i}}{1+e^{-\{1-\pi(\bm{x}_i)\}\hat{\bm{\gamma}}_j^\prime\bm{x}_i}}\bm{x}_i
    \\
    &=\pi(\bm{x}_i)\{1-\pi(\bm{x}_i)\}E[Y_{ij}|T_i=-1,\bm{X}_i]\frac{e^{\pi(\bm{x}_i)\hat{\bm{\gamma}}_j^\prime\bm{x}_i}}{1+e^{\pi(\bm{x}_i)\hat{\bm{\gamma}}_j^\prime\bm{x}_i}}\bm{x}_i
    \\
    \Rightarrow &\frac{E[Y_{ij}|T_i=1,X_i=\bm{x}_i]}{E[Y_{ij}|T_i=-1,X_i=\bm{x}_i]}
    =\frac{e^{\pi(\bm{x}_i)\hat{\bm{\gamma}}_j^\prime\bm{x}_i}}{1+e^{\pi(\bm{x}_i)\hat{\bm{\gamma}}_j^\prime\bm{x}_i}}\times \frac{1+e^{-\{1-\pi(\bm{x}_i)\}\hat{\bm{\gamma}}_j^\prime\bm{x}_i}}{e^{-\{1-\pi(\bm{x}_i)\}\hat{\bm{\gamma}}_j^\prime\bm{x}_i}}
    \\
    \Rightarrow &\frac{E[Y_{ij}|T_i=1,X_i=\bm{x}_i]}{E[Y_{ij}|T_i=-1,X_i=\bm{x}_i]}
    =e^{\hat{\bm{\gamma}}_j^\prime\bm{x}_i}
    \times \frac{1+e^{-\{1-\pi(\bm{x}_i)\}\hat{\bm{\gamma}}_j^\prime\bm{x}_i}}{1+e^{\pi(\bm{x}_i)\hat{\bm{\gamma}}_j^\prime\bm{x}_i}}
    \\
    \Rightarrow &\log \frac{E[Y_{ij}|T_i=1,X_i=\bm{x}_i]}{E[Y_{ij}|T_i=-1,X_i=\bm{x}_i]}
    =\hat{\bm{\gamma}}_j^\prime\bm{x}_i
    +\log \frac{1+e^{-\{1-\pi(\bm{x}_i)\}\hat{\bm{\gamma}}_j^\prime\bm{x}_i}}{1+e^{\pi(\bm{x}_i)\hat{\bm{\gamma}}_j^\prime\bm{x}_i}}
\end{align*}
\end{proof}
Theorem \ref{thm_bias} implies that the treatment effect estimated by A-learner using linear functions includes a bias (the second term on the right-hand side of Eq. (\ref{bias_mod})), depending on the propensity score.
In the proof by Chen et al. (2017), bias occurs because the objective function is not differentiated with respect to the parameters when solving the minimization problem in Eq. (\ref{obj_uniA}). Therefore, the bias term must be removed when calculating the treatment effect using the estimated $\bm{\gamma}_j$.
\par
Finally, we describe the properties of the objective function $\ell_A(\bm{V},\bm{W},\bm{x}_i)$ based on A-learner. Similar to $\ell_W(\bm{V},\bm{W},\bm{x}_i)$, this function is biconvex with respect to the parameters. Therefore, we can prove that the optimal value for one parameter can be reached by fixing the other. However, similar to the conventional A-learner for univariate binary outcomes, bias occurs in the estimation of HTE. Therefore, we discuss how to correct the bias and the optimal values of the parameters to accurately represent the HTE.
\begin{thm}
\label{bias_RRRA1}
Suppose that the true ranks of $\bm{V}^*$ and $\bm{W}^*$ representing HTE in Eq. (2) are both $r$. When $\bm{W}$ is fixed to the matrix $\bm{W}^\dagger\in\mathcal{W}_r$, the following holds using $\hat{\bm{v}}_{(j)}$, which minimizes $\ell_A(\bm{V},\bm{W}^\dagger,\bm{x}_i)$:
    \begin{align*}
       \log \frac{E[Y_{ij}|T_i=1,X_i=\bm{x}_i]}{E[Y_{ij}|T_i=-1,X_i=\bm{x}_i]}
    =\hat{\bm{v}}_{(j)}^\prime\bm{W}^{\dagger\prime}\bm{x}_i+
    \log \frac{1+e^{-\{1-\pi(\bm{x}_i)\}\hat{\bm{v}}_{(j)}^\prime\bm{W}^{\dagger\prime}\bm{x}_i}}{1+e^{\pi(\bm{x}_i)\hat{\bm{v}}_{(j)}^\prime\bm{W}^{\dagger\prime}\bm{x}_i}}
    \end{align*}
\end{thm}
\begin{proof}
    $\ell_A(\bm{V},\bm{W}^\dagger,\bm{x}_i)$ is expanded in the same manner as in Theorem \ref{thm_bias}:
    \begin{align*}
\ell_A(\bm{V},\bm{W}^\dagger,\bm{x}_i)
        &=E[M(Y_{i},\{(T_i+1)/2-\pi(\bm{x_i})\}\times \bm{VW}^{\dagger\prime}\bm{x}_i)|X_i=\bm{x}_i]
    \\
    &=\sum_{j=1}^m \left\{E[I(T_i=1)M(Y_{ij},\{1-\pi(\bm{x}_i)\}\bm{v}_{(j)}^\prime\bm{W}^{\dagger\prime}\bm{x}_i)|X_i=\bm{x}_i]
    \right.
    \\
    &\left.
    +
    E[I(T_i=-1)M(Y_{ij},-\pi(\bm{x}_i)\bm{v}_{(j)}^\prime\bm{W}^{\dagger\prime}\bm{x}_i)|X_i=\bm{x}_i]
    \right\}
    \\
    &=\sum_{j=1}^m \left\{
    \pi(\bm{x}_i)E[M(Y_{ij},\{1-\pi(\bm{x}_i)\}\bm{v}_{(j)}^\prime\bm{W}^{\dagger\prime}\bm{x}_i)|T_i=1,X_i=\bm{x}_i]
    \right.
    \\
    &\left.
    +
    \{1-\pi(\bm{x}_i)\}E[M(Y_{ij},-\pi(\bm{x}_i)\bm{v}_{(j)}^\prime\bm{W}^{\dagger\prime}\bm{x}_i)|T_i=-1,X_i=\bm{x}_i]
    \right\}
    \\
    &=\sum_{j=1}^m \left\{
    \pi(\bm{x}_i)E[Y_{ij}\log \{1+\exp (-\{1-\pi(\bm{x}_i)\}\bm{v}_{(j)}^\prime\bm{W}^{\dagger\prime}\bm{x}_i))\}|T_i=1,X_i=\bm{x}_i]
    \right.
    \\
    &\left.
    +
    \{1-\pi(\bm{x}_i)\}E[Y_{ij}\log \{1+\exp (\pi(\bm{x}_i)\bm{v}_{(j)}^\prime\bm{W}^{\dagger\prime}\bm{x}_i)\}|T_i=-1,X_i=\bm{x}_i]
    \right\}
    \end{align*}
    By differentiating this with respect to $\bm{v}_{(j)}$ and substituting the optimal value $\hat{\bm{v}}_{(j)}$ with the first-order condition, it holds that
    \begin{align*}
        &\pi(\bm{x}_i)\{1-\pi(\bm{x}_i)\} E[Y_{ij}|T_i=1,X_i=\bm{x}_i]\frac{e^{-\{1-\pi(\bm{x}_i)\}\hat{\bm{v}}_{(j)}^\prime\bm{W}^{\dagger\prime}\bm{x}_i}}{1+e^{-\{1-\pi(\bm{x}_i)\}\hat{\bm{v}}_{(j)}^\prime\bm{W}^{\dagger\prime}\bm{x}_i}}\bm{W}^{\dagger\prime}\bm{x}_i
       \\=&\pi(\bm{x}_i)\{1-\pi(\bm{x}_i)\}E[Y_{ij}|T_i=-1,X_i=\bm{x}_i]\frac{e^{\pi(\bm{x})\hat{\bm{v}}_{(j)}^\prime\bm{W}^{\dagger\prime}\bm{x}_i}}{1+e^{\pi(\bm{x})\hat{\bm{v}}_{(j)}^\prime\bm{W}^{\dagger\prime}\bm{x}_i}}\bm{W}^{\dagger\prime}\bm{x}_i
    \\
    \Rightarrow &
    \frac{E[Y_{ij}|T_i=1,X_i=\bm{x}_i]}{E[Y_{ij}|T_i=-1,X_i=\bm{x}_i]}
    =\frac{e^{\pi(\bm{x})\hat{\bm{v}}_{(j)}^\prime\bm{W}^{\dagger\prime}\bm{x}_i}}{1+e^{\pi(\bm{x})\hat{\bm{v}}_{(j)}^\prime\bm{W}^{\dagger\prime}\bm{x}_i}}\times
    \frac{1+e^{-\{1-\pi(\bm{x}_i)\}\hat{\bm{v}}_{(j)}^\prime\bm{W}^{\dagger\prime}\bm{x}_i}}{e^{-\{1-\pi(\bm{x}_i)\}\hat{\bm{v}}_{(j)}^\prime\bm{W}^{\dagger\prime}\bm{x}_i}}
    \\
    \Rightarrow &
    \frac{E[Y_{ij}|T_i=1,X_i=\bm{x}_i]}{E[Y_{ij}|T_i=-1,X_i=\bm{x}_i]}
    =e^{\hat{\bm{v}}_{(j)}^\prime\bm{W}^{\dagger\prime}\bm{x}_i}\times
    \frac{1+e^{-\{1-\pi(\bm{x}_i)\}\hat{\bm{v}}_{(j)}^\prime\bm{W}^{\dagger\prime}\bm{x}_i}}{1+e^{\pi(\bm{x}_i)\hat{\bm{v}}_{(j)}^\prime\bm{W}^{\dagger\prime}\bm{x}_i}}
    \\
    \Rightarrow
    &\log \frac{E[Y_{ij}|T_i=1,X_i=\bm{x}_i]}{E[Y_{ij}|T_i=-1,X_i=\bm{x}_i]}
    =\hat{\bm{v}}_{(j)}^\prime\bm{W}^{\dagger\prime}\bm{x}_i+
    \log \frac{1+e^{-\{1-\pi(\bm{x}_i)\}\hat{\bm{v}}_{(j)}^\prime\bm{W}^{\dagger\prime}\bm{x}_i}}{1+e^{\pi(\bm{x}_i)\hat{\bm{v}}_{(j)}^\prime\bm{W}^{\dagger\prime}\bm{x}_i}}
    \end{align*}
\end{proof}
\begin{thm}
    \label{bias_RRRA2}
   Suppose that the true ranks of $\bm{V}^*$ and $\bm{W}^*$ representing the HTE in Eq. (2) are both $r$. When $\bm{V}$ is fixed to the matrix $\bm{V}^\dagger\in\mathcal{V}_r$, the following holds when using $\hat{\bm{W}}$, which minimizes $\ell_A(\bm{V}^\dagger,\bm{W},\bm{x}_i)$:
    \begin{align*}
       \log \frac{E[Y_{ij}|T_i=1,X_i=\bm{x}_i]}{E[Y_{ij}|T_i=-1,X_i=\bm{x}_i]}
    =\bm{v}_{(j)}^{\dagger\prime}\hat{\bm{W}}^\prime\bm{x}_i+
    \log \frac{1+e^{-\{1-\pi(\bm{x}_i)\}\bm{v}_{(j)}^{\dagger\prime}\hat{\bm{W}}^\prime\bm{x}_i}}{1+e^{\pi(\bm{x}_i)\bm{v}_{(j)}^{\dagger\prime}\hat{\bm{W}}^\prime\bm{x}_i}}
    \end{align*}
\end{thm}
\begin{proof}
   $\ell_A(\bm{V}^\dagger,\bm{W},\bm{x}_i)$ is expanded in the same manner as in the proof of Theorem \ref{thm_bias}, that is,
    \begin{align*}
\ell_A(\bm{V}^\dagger,\bm{W},\bm{x}_i)
        &=E[M(Y_{i},\{(T_i+1)/2-\pi(\bm{x_i})\}\times \bm{V}^\dagger\bm{W}^\prime\bm{x}_i)|X_i=\bm{x}_i]
    \\
    &=\sum_{j=1}^m \left\{E[I(T_i=1)M(Y_{ij},\{1-\pi(\bm{x}_i)\}\bm{v}_{(j)}^{\dagger\prime}\bm{W}^\prime\bm{x}_i)|X_i=\bm{x}_i]
    \right.
    \\
    &\left.
    +
    E[I(T_i=-1)M(Y_{ij},-\pi(\bm{x}_i)\bm{v}_{(j)}^{\dagger\prime}\bm{W}^\prime\bm{x}_i)|X_i=\bm{x}_i]
    \right\}
    \\
    &=\sum_{j=1}^m \left\{
    \pi(\bm{x}_i)E[M(Y_{ij},\{1-\pi(\bm{x}_i)\}\bm{v}_{(j)}^{\dagger\prime}\bm{W}^\prime\bm{x}_i)|T_i=1,X_i=\bm{x}_i]
    \right.
    \\
    &\left.
    +
    \{1-\pi(\bm{x}_i)\}E[M(Y_{ij},-\pi(\bm{x}_i)\bm{v}_{(j)}^{\dagger\prime}\bm{W}^\prime\bm{x}_i)|T_i=-1,X_i=\bm{x}_i]
    \right\}
    \\
    &=\sum_{j=1}^m \left\{
    \pi(\bm{x}_i)E[Y_{ij}\log \{1+\exp (-\{1-\pi(\bm{x}_i)\}\bm{v}_{(j)}^{\dagger\prime}\bm{W}^\prime\bm{x}_i))\}|T_i=1,X_i=\bm{x}_i]
    \right.
    \\
    &\left.
    +
    \{1-\pi(\bm{x}_i)\}E[Y_{ij}\log \{1+\exp (\pi(\bm{x}_i)\bm{v}_{(j)}^{\dagger\prime}\bm{W}^\prime\bm{x}_i))\}|T_i=-1,X_i=\bm{x}_i]
    \right\}
    \end{align*}
    By differentiating this with respect to $\bm{W}$ and substituting the optimal value $\hat{\bm{W}}$, it holds from $\bm{v}_{(j)}^{\dagger\prime}\bm{W}^\prime\bm{x}_i=\mathrm{tr}(\bm{W}^\prime\bm{x}_i\bm{v}_{(j)}^{\dagger\prime})$, which implies $\partial (\bm{v}_{(j)}^{\dagger\prime}\bm{W}^\prime\bm{x}_i)/\partial \bm{W}=\bm{x}_i\bm{v}_{(j)}^{\dagger\prime}$, where $\mathrm{tr}(\cdot)$ is trace operator for a matrix, that
    \begin{align*}
        &\pi(\bm{x}_i)\{1-\pi(\bm{x}_i)\} E[Y_{ij}|T_i=1,X_i=\bm{x}_i]\frac{e^{-\{1-\pi(\bm{x}_i)\}\bm{v}_{(j)}^{\dagger\prime}\hat{\bm{W}}^\prime\bm{x}_i}}{1+e^{-\{1-\pi(\bm{x}_i)\}\bm{v}_{(j)}^{\dagger\prime}\hat{\bm{W}}^\prime\bm{x}_i}}\bm{x}_i\bm{v}_{(j)}^{\dagger\prime}
        \\
        =&\pi(\bm{x}_i)\{1-\pi(\bm{x}_i)\}E[Y_{ij}|T_i=-1,X_i=\bm{x}_i]\frac{e^{\pi(\bm{x}_i)\bm{v}_{(j)}^{\dagger\prime}\hat{\bm{W}}^\prime\bm{x}_i}}{1+e^{\pi(\bm{x}_i)\bm{v}_{(j)}^{\dagger\prime}\hat{\bm{W}}^\prime\bm{x}_i}}\bm{x}_i\bm{v}_{(j)}^{\dagger\prime}
    \\
    \Rightarrow &
    \frac{E[Y_{ij}|T_i=1,X_i=\bm{x}_i]}{E[Y_{ij}|T_i=-1,X_i=\bm{x}_i]}
    =\frac{e^{\pi(\bm{x}_i)\bm{v}_{(j)}^{\dagger\prime}\hat{\bm{W}}^\prime\bm{x}_i}}{1+e^{\pi(\bm{x}_i)\bm{v}_{(j)}^{\dagger\prime}\hat{\bm{W}}^\prime\bm{x}_i}}\times
    \frac{1+e^{-\{1-\pi(\bm{x}_i)\}\bm{v}_{(j)}^{\dagger\prime}\hat{\bm{W}}^\prime\bm{x}_i}}{e^{-\{1-\pi(\bm{x}_i)\}\bm{v}_{(j)}^{\dagger\prime}\hat{\bm{W}}^\prime\bm{x}_i}}
    \\
    \Rightarrow &
    \frac{E[Y_{ij}|T_i=1,X_i=\bm{x}_i]}{E[Y_{ij}|T_i=-1,X_i=\bm{x}_i]}
    =e^{\bm{v}_{(j)}^{\dagger\prime}\hat{\bm{W}}^\prime\bm{x}_i}\times
    \frac{1+e^{-\{1-\pi(\bm{x}_i)\}\bm{v}_{(j)}^{\dagger\prime}\hat{\bm{W}}^\prime\bm{x}_i}}{1+e^{\pi(\bm{x}_i)\bm{v}_{(j)}^{\dagger\prime}\hat{\bm{W}}^\prime\bm{x}_i}}
    \\
    \Rightarrow
    &\log \frac{E[Y_{ij}|T_i=1,X_i=\bm{x}_i]}{E[Y_{ij}|T_i=-1,X_i=\bm{x}_i]}
    =\bm{v}_{(j)}^{\dagger\prime}\hat{\bm{W}}^\prime\bm{x}_i+
    \log \frac{1+e^{-\{1-\pi(\bm{x}_i)\}\bm{v}_{(j)}^{\dagger\prime}\hat{\bm{W}}^\prime\bm{x}_i}}{1+e^{\pi(\bm{x}_i)\bm{v}_{(j)}^{\dagger\prime}\hat{\bm{W}}^\prime\bm{x}_i}}
    \end{align*}
\end{proof}
These theorems determine that the estimation by the A-learner introduces a bias in the treatment effect, even when using multiple logistic loss. Therefore, the treatment effect can be expressed by correcting the bias term using the estimated $\bm{V}$ and $\bm{W}$.
\section{Estimation algorithm}
In this section, we describe the estimation algorithms used in the proposed method. The proposed method, whether based on the W-method or A-learner, estimates the parameters by using the majorization–minimization (MM) algorithm \cite{hunter2004tutorial}. The MM algorithm helps us approximate the minimization problem of the proposed method using a simple quadratic function, allowing optimization based on singular value decomposition.
\subsection{Estimation algorithm in the framework of the W-method}
\label{est_w}
In this subsection, we describe the estimation algorithm used to solve the optimization problem within the framework of the W method. We specifically utilized the MM algorithm.
The fundamental concept of the MM algorithm is to iteratively approximate the original complex function $\mathcal{L}(\theta)$ using an auxiliary function $\mathcal{M}(\theta,\vartheta)$ (referred to as the majorization function), where $\theta$ and $\vartheta$ are the parameters and supporting points, respectively.
To qualify $\mathcal{M}(\theta,\vartheta)$ as the majorization function of $\mathcal{L}(\theta)$, the following conditions must be satisfied.
\begin{enumerate}
    \item The original function is always less than or equal to the auxiliary function, that is, $\forall \theta;\mathcal{L}(\theta)\leq \mathcal{M}(\theta,\vartheta)$.
    \item The original function is tangent to the auxiliary function at the supporting point, that is,  $\mathcal{L}(\vartheta)=\mathcal{M}(\vartheta,\vartheta)$.
\end{enumerate}
Let $\theta^*$ denote the minimum $\mathcal{M}(\theta,\vartheta)$ at the supporting point $\vartheta$. From the conditions above, we derive the following inequality:
\begin{align*}
    \mathcal{L}(\theta^*)\leq \mathcal{M}(\theta^*,\vartheta)\leq \mathcal{M}(\vartheta,\vartheta)=\mathcal{L}(\vartheta).
\end{align*}
Here, we derive a majorization function approximating $L_W(\bm{V},\bm{W},\bm{X})$ following the approach of de Rooij (2024) \cite{de2024new}.
\begin{prop}
\label{prop_w}
The majorization function of $L_W(\bm{V},\bm{W},\bm{X})$ is derived as 
\begin{align*}
    \mathcal{M}_W(\bm{V},\bm{W})=\|\bm{A}^{1/2}\bm{Z}-\bm{A}^{1/2}\bm{TXWV}^\prime\|_F^2,
\end{align*}
where $\bm{A}=\mathrm{diag}(a_1,a_2,\cdots,a_n),\;a_i=\left(t_i\pi(\bm{x}_i)+(1-t_i)/2\right)^{-1}$, $\bm{Z}=\bm{TXW}^{(t)}\bm{V}^{(t)\prime}+4\bm{Y}\odot\bm{\Phi}^{(t)}\in\mathbb{R}^{n\times m}$ using the $\bm{V}^{(t)}$ and $\bm{W}^{(t)}$ from the $t$th iteration, and $\odot$ represents the Hadamard product.
$\bm{A}^{1/2}$ is the square root of the matrix $\bm{A}$, $\bm{\Phi}^{(t)}\in [0,1]^{n\times m}$ is a matrix in which each element is equal to $\phi_{ij}^{(t)}$, and
\begin{align*}
    \phi_{ij}^{(t)}
    =\frac{\exp(-t_i\bm{v}_{(j)}^{(t)\prime}\bm{W}^{(t)\prime}\bm{x}_i)}{1+\exp(-t_i\bm{v}_{(j)}^{(t)\prime}\bm{W}^{(t)\prime}\bm{x}_i)}\;(i=1,2,\cdots,n;\;j=1,2,\cdots,m).
\end{align*}
\end{prop}
\begin{proof}
$L_W(\bm{V},\bm{W},\bm{X})$ can be written as
\begin{align}
\label{MM}
    L_W(\bm{V},\bm{W},\bm{X})
    =\sum_{i=1}^n \sum_{j=1}^m \frac{y_{ij}\log \left\{1+\exp\left(-t_i\bm{v}_{(j)}^\prime\bm{W}^\prime\bm{x}_i\right)\right\}}{t_i\pi(\bm{x}_i)+(1-t_i)/2}
    =\sum_{i=1}^n a_i\sum_{j=1}^m y_{ij}\log \left\{1+\exp\left(-\theta_{ij}\right)\right\},
\end{align}
where $\theta_{ij}=t_i\bm{v}_{(j)}^\prime\bm{W}^\prime\bm{x}_i$.
From the Eq. (\ref{MM}), we isolate the terms involving parameter $\theta_{ij}$ and define them as $\mathcal{L}(\theta_{ij})$, that is,
\begin{align*}
    \mathcal{L}(\theta_{ij})
    =y_{ij}\log \left\{1+\exp\left(-\theta_{ij}\right)\right\}.
\end{align*}
To derive the majorization function of $\mathcal{L}(\theta_{ij})$, we calculate the first and second derivates of $\mathcal{L}(\theta_{ij})$ with respect to the parameter.
The first derivative of $\mathcal{L}(\theta_{ij})$ is
\begin{align*}
    \xi(\theta_{ij})\equiv \frac{\partial \mathcal{L}(\theta_{ij})}{\partial \theta_{ij}}
    =-y_{ij}\frac{\exp(-\theta_{ij})}{1+\exp(-\theta_{ij})}.
\end{align*}
The second derivative of $\mathcal{L}(\theta_{ij})$ is
\begin{align*}
    \frac{\partial^2 \mathcal{L}(\theta_{ij})}{\partial \theta_{ij}^2}
    =y_{ij}\frac{\exp(-\theta_{ij})}{1+\exp(-\theta_{ij})}\frac{1}{1+\exp(-\theta_{ij})}\leq \frac{1}{4}.
\end{align*}
Thus, we majorize this function using the least squares function, that is,
\begin{align*}
    \mathcal{L}(\theta_{ij})
    &\leq \mathcal{L}(\vartheta_{ij})
    +\xi (\vartheta_{ij})(\theta_{ij}-\vartheta_{ij})
    +\frac{1}{8}(\theta_{ij}-\vartheta_{ij})^2
    \\
    &=\mathcal{L}(\vartheta_{ij})
    +\xi (\vartheta_{ij})\theta_{ij}-\xi (\vartheta_{ij})\vartheta_{ij}
    +\frac{1}{8}(\theta_{ij}^2+\vartheta_{ij}^2-2\theta_{ij}\vartheta_{ij})
    \\
    &=\mathcal{L}(\vartheta_{ij})+\frac{1}{8}\theta_{ij}^2 +\xi (\vartheta_{ij})\theta_{ij} -2\frac{1}{8}\theta_{ij}\vartheta_{ij} -\xi (\vartheta_{ij})\vartheta_{ij}+\frac{1}{8}\vartheta_{ij}^2
    \\
    &=\mathcal{L}(\vartheta_{ij})+\frac{1}{8}\theta_{ij}^2 -2\frac{1}{8}\theta_{ij} (\vartheta_{ij}-4\xi (\vartheta_{ij})) -\xi (\vartheta_{ij})\vartheta_{ij}+\frac{1}{8}\vartheta_{ij}^2
    \\
    &=\mathcal{L}(\vartheta_{ij})+\frac{1}{8}\theta_{ij}^2 -2\frac{1}{8}\theta_{ij} z_{ij} + \left(\frac{1}{8}z_{ij}^2 - \frac{1}{8}z_{ij}^2\right) -\xi (\vartheta_{ij})\vartheta_{ij}+\frac{1}{8}\vartheta_{ij}^2
    \\
    &=\mathcal{L}(\vartheta_{ij})+\frac{1}{8}(\theta_{ij}-z_{ij})^2 - \frac{1}{8}z_{ij}^2 -\xi (\vartheta_{ij})\vartheta_{ij}+\frac{1}{8}\vartheta_{ij}^2
    \\
    &=\frac{1}{8}(z_{ij}-\theta_{ij})^2+\mathrm{const.}\;,
\end{align*}
where $\vartheta_{ij}=t_i\bm{v}^{(t)\prime}_{(j)}\bm{W}^{(t)\prime}\bm{x}_i$ in the $t$th iteration, $\mathrm{const}$ is a term that is independent of the parameter, and $z_{ij}=\vartheta_{ij}-4\xi(\vartheta_{ij})$.
Therefore, by the summation property and $\forall i;a_i\geq 0$, the majorization function of $L_W(\bm{V},\bm{W},\bm{X})$ is 
\begin{align*}
    \frac{1}{8}\sum_{i=1}^n a_i\sum_{j=1}^m (z_{ij}-\theta_{ij})^2 + \mathrm{const.}\;
\end{align*}
Furthermore, in the matrix notation, we obtain
\begin{align*}
    \frac{1}{8}\sum_{i=1}^n a_i\sum_{j=1}^m (z_{ij}-\theta_{ij})^2 + \mathrm{const.}
    =\frac{1}{8}\|\bm{A}^{1/2}\bm{Z}-\bm{A}^{1/2}\bm{\Theta}\| + \mathrm{const.},
\end{align*}
where $\bm{\Theta}$ and $\bm{Z}$ are the matrices with elements $\theta_{ij}$ and $\bm{z}_{ij}$, respectively.
Because each $\theta_{ij}$ is expressed as $t_i\bm{v}_{(j)}^\prime\bm{W}^\prime\bm{x}_i$ and $\bm{\Theta}=\bm{TXWV}^\prime$, it follows that 
\begin{align*}
   \frac{1}{8}\|\bm{A}^{1/2}\bm{Z}-\bm{A}^{1/2}\bm{TXWV}^\prime\|_F^2 + \mathrm{const.}.
\end{align*}
Here, since $z_{ij}=t_i\bm{v}^{(t)\prime}_{(j)}\bm{W}^{(t)\prime}\bm{x}_i +4y_{ij}\frac{\exp(-t_i\bm{v}^{(t)\prime}_{(j)}\bm{W}^{(t)\prime}\bm{x}_i)}{1+\exp(-t_i\bm{v}^{(t)\prime}_{(j)}\bm{W}^{(t)\prime}\bm{x}_i)}$ in the $t$th iteration, $\bm{Z}=\bm{TXW}^{(t)}\bm{V}^{(t)\prime}+4\bm{Y}\odot\bm{\Phi}^{(t)}\in\mathbb{R}^{n\times m}$.
By excluding this constant factor, we can derive $\mathcal{M}_W(\bm{V},\bm{W})$.
\end{proof}
Therefore, in every iteration of the MM algorithm, the following minimization problem should be considered:
\begin{gather}
\label{obj}
    \underset{\bm{V},\bm{W}}{\min}\;\mathcal{M}_W(\bm{V},\bm{W})\;\mathrm{subject\;to\;}\bm{V}^\prime\bm{V}=\bm{I}_r.
\end{gather}
We present the following two propositions for estimating the parameters $\bm{V}$ and $\bm{W}$.
\begin{prop}
Given $\bm{W}$, $\bm{V}$ is updated to $\hat{\bm{V}}$ to solve the minimization problem Eq. (\ref{obj}).
\begin{align*}
\hat{\bm{V}}=\bm{KL}^\prime
\end{align*}
where $2\bm{Z}^\prime\bm{A}\bm{TXW}=\bm{K\Lambda L}^\prime$ by the singular value decomposition.
$\bm{K}\in\mathbb{R}^{m\times r}$ and $\bm{L}\in\mathbb{R}^{r\times r}$ are left and right singular vectors, respectively.
$\bm{\Lambda}=\mathrm{diag}(\lambda_1,\lambda_2,\cdots,\lambda_r)$ denotes a square-diagonal matrix.
\end{prop}
\begin{proof}
$\mathcal{M}_W(\bm{V},\bm{W})$ in Eq. ~ (\ref{obj}) can be written as 
\begin{align*}
    \|\bm{A}^{1/2}\bm{Z}-\bm{A}^{1/2}\bm{TXWV}^\prime\|_F^2
    =\mathrm{tr}\left(\bm{Z}^\prime\bm{A}\bm{Z}-2\bm{VW}^\prime\bm{X}^\prime\bm{TA}\bm{Z}+\bm{W}^\prime\bm{X}^\prime\bm{A}\bm{XW}\right).
\end{align*}
Considering only the term related to parameter $\bm{V}$,
\begin{align*}
    \underset{\bm{V}}{\mathrm{argmax}}\;2\mathrm{tr}\left((\bm{Z}^\prime\bm{A}\bm{TXW})^\prime\bm{V}\right)\;\mathrm{s.t.}\;\bm{V}^\prime\bm{V}=\bm{I}_r.
\end{align*}
Referring to Theorem A.4.2 in Adachi (2016) \cite{adachi2016matrix}, this maximization problem can be solved using singular value decomposition, thereby proving this proposition.
\end{proof}
\begin{prop}
    Suppose that $\bm{X}^\prime\bm{X}$ is a regular matrix.
    Given $\bm{V}$, $\bm{W}$ is updated to $\hat{\bm{W}}$ to solve the minimization problem in Eq. (\ref{obj}).
\begin{align*}
\hat{\bm{W}}=(\bm{X}^\prime\bm{AX})^{-1}\bm{X}^\prime\bm{TAZV}
\end{align*}
\end{prop}
\begin{proof}
    $\mathcal{M}_W(\bm{V},\bm{W})$ in Eqs. (\ref{obj}) can be written as 
\begin{align*}
    \|\bm{A}^{1/2}\bm{Z}-\bm{A}^{1/2}\bm{TXWV}^\prime\|_F^2
    =\mathrm{tr}\left(\bm{Z}^\prime\bm{A}\bm{Z}-2\bm{VW}^\prime\bm{X}^\prime\bm{TA}\bm{Z}+\bm{W}^\prime\bm{X}^\prime\bm{A}\bm{XW}\right).
\end{align*}
Considering only the terms related to parameter $\bm{W}$,
\begin{align*}
    \underset{\bm{W}}{\mathrm{argmin}}\;-2\mathrm{tr}\left(\bm{W}^\prime\bm{X}^\prime\bm{TA}\bm{ZV}\right)+\mathrm{tr}\left(\bm{W}^\prime\bm{X}^\prime\bm{A}\bm{XW}\right).
\end{align*}
Differentiating with respect to $\bm{W}$ to determine $\hat{\bm{W}}$, it holds that
\begin{align*}
    -2\bm{X}^\prime\bm{TA}\bm{ZV}
    +2\bm{X}^\prime\bm{A}\bm{XW}=\bm{O}_{p\times r}
    \Rightarrow \hat{\bm{W}}=\left(\bm{X}^\prime\bm{A}\bm{X}\right)^{-1}\bm{X}^\prime\bm{TA}\bm{ZV}.
\end{align*}
\end{proof}
The pseudocode for the estimation algorithm is presented in Algorithm \ref{alg1}.
\begin{figure}[t]
\begin{algorithm}[H]
    \caption{Logistic reduced rank regression based on the W-method}
    \label{alg1}
    \begin{algorithmic}[1]    
    \REQUIRE$\bm{X}$,$\bm{Y}$,$\bm{T}$,$\bm{A}$
    \ENSURE $\bm{W},\bm{V}$
    \STATE $\varepsilon>0$
    \STATE $t\leftarrow 1$
    \STATE Set initial value for $\bm{W}^{(0)}$,$\bm{V}^{(0)}$ 
    \WHILE{$\mathcal{M}_W(\bm{V}^{(t)},\bm{W}^{(t)})-\mathcal{M}_W(\bm{V}^{(t-1)},\bm{W}^{(t-1)})<\varepsilon$}
    \STATE $\bm{\Phi}^{(t)}=(\phi_{ij}^{(t)})\leftarrow \frac{\exp(-t_i\bm{v}_{(j)}^{(t-1)\prime}\bm{W}^{(t-1)\prime}\bm{x}_i)}{1+\exp(-t_i\bm{v}_{(j)}^{(t-1)\prime}\bm{W}^{(t-1)\prime}\bm{x}_i)}$
    \STATE $\bm{Z}\leftarrow \bm{TX}\bm{W}^{(t-1)}\bm{V}^{(t-1)\prime}+4\bm{Y}\odot\bm{\Phi}^{(t)}$
    \STATE  $\bm{K}\bm{\Lambda}\bm{L}^{\prime}$ $\leftarrow$ singular value decomposition of $2\bm{Z}^\prime\bm{A}\bm{TXW}^{(t-1)}$
     \STATE $\bm{V}^{(t)}\leftarrow\bm{K}\bm{L}^{\prime}$
     \STATE $\bm{W}^{(t)}\leftarrow (\bm{X}^\prime\bm{AX})^{-1}\bm{X}^\prime\bm{TAZV}^{(t)}$
     \STATE $t \leftarrow t+1$
    \ENDWHILE
      \STATE$\bm{W} \leftarrow \bm{W}^{(t)}$
      \STATE$\bm{V} \leftarrow \bm{V}^{(t)}$
    \RETURN $\bm{W},\bm{V}$
    \end{algorithmic}
\end{algorithm}
\end{figure}
\subsection{Estimation algorithm in the framework of the A-learner}
\label{est_A}
In this subsection, we describe the estimation algorithm for solving the optimization problem within Learner A’s framework. Similar to the W method, we use the MM algorithm. The majorization function is derived as follows:
\begin{prop}
    The majorization function of $L_A(\bm{V},\bm{W},\bm{X})$ is derived as 
    \begin{align*}
        \mathcal{M}_A(\bm{V},\bm{W})=\|\bm{Z}^{\dagger}-\bm{QXWV}^\prime\|_F^2,
    \end{align*}
    where $\bm{Q}=\mathrm{diag}(q_1,q_2,\cdots,q_n),\;q_i=(t_i+1)/2-\pi(\bm{x}_i)$, and $\bm{Z}^\dagger=\bm{QX}\bm{W}^{(t)}\bm{V}^{(t)\prime} + 4\bm{Y}\odot \bm{\Phi}^{\dagger(t)}\in\mathbb{R}^{n\times m}$ using the $\bm{V}^{(t)}$ and $\bm{W}^{(t)}$ from the $t$th iteration.
    $\bm{\Phi}^{\dagger(t)}\in [0,1]^{n\times m}$ is a matrix, in which each element is equal to $\phi_{ij}^{\dagger(t)}$.
    \begin{align*}
        \phi_{ij}^{\dagger(t)}
        =
        \frac{\exp(-q_i\bm{v}_{(j)}^{(t)\prime}\bm{W}^{(t)\prime}\bm{x}_i)}{1+\exp(-q_i\bm{v}_{(j)}^{(t)\prime}\bm{W}^{(t)\prime}\bm{x}_i)}\;(i=1,2,\cdots,n;\;j=1,2,\cdots,m).
    \end{align*}
\end{prop}
\begin{proof}
$L_A(\bm{V},\bm{W},\bm{X})$ can be written as
\begin{align}
\nonumber
    L_A(\bm{V},\bm{W},\bm{X})
    &=\sum_{i=1}^n \sum_{j=1}^m y_{ij}\log \left\{1+\exp\left(-\{(t_i+1)/2-\pi(\bm{x}_i)\}\times \bm{v}_{(j)}^\prime\bm{W}^\prime\bm{x}_i\right)\right\}
    \\
    \label{MM_A}
    &=\sum_{i=1}^n \sum_{j=1}^m y_{ij}\log \left\{1+\exp\left(-q_i\bm{v}_{(j)}^\prime\bm{W}^\prime\bm{x}_i\right)\right\}
    =\sum_{i=1}^n \sum_{j=1}^m y_{ij}\log \left\{1+\exp\left(-\theta_{ij}^\dagger\right)\right\},
\end{align}
where $\theta^\dagger_{ij}=q_i\bm{v}_{(j)}^\prime\bm{W}^\prime\bm{x}_i$.
From the Eq. (\ref{MM_A}), we isolate the terms involving parameter $\theta^\dagger_{ij}$ and define them as $\mathcal{L}(\theta^\dagger_{ij})$, that is,
\begin{align*}
    \mathcal{L}(\theta_{ij}^\dagger)
    =y_{ij}\log \left\{1+\exp\left(-\theta_{ij}^\dagger\right)\right\}.
\end{align*}
To derive the majorization function of $\mathcal{L}(\theta_{ij})$, we calculate the first and second derivates of $\mathcal{L}(\theta_{ij})$ with respect to the parameter.
Because these are the same as those in the weighting approach, the following holds:
\begin{gather*}
    \xi(\theta^{\dagger}_{ij})\equiv \frac{\partial \mathcal{L}(\theta^{\dagger}_{ij})}{\partial \theta^{\dagger}_{ij}}
    =-y_{ij}\frac{\exp(-\theta^{\dagger}_{ij})}{1+\exp(-\theta^{\dagger}_{ij})}.
    \\
    \frac{\partial^2 \mathcal{L}(\theta^{\dagger}_{ij})}{\partial \theta^{\dagger2}_{ij}}
    =y_{ij}\frac{\exp(-\theta^{\dagger}_{ij})}{1+\exp(-\theta^{\dagger}_{ij})}\frac{1}{1+\exp(-\theta^{\dagger}_{ij})}\leq \frac{1}{4}.
\end{gather*}
Thus, we majorize this function using the least squares function in the same manner as in the proof of Proposition \ref{prop_w}:
\begin{align*}
    \mathcal{L}(\theta^{\dagger}_{ij})
    &\leq \mathcal{L}(\vartheta^\dagger_{ij})
    +\xi (\vartheta^\dagger_{ij})(\theta^{\dagger}_{ij}-\vartheta^\dagger_{ij})
    +\frac{1}{8}(\theta^{\dagger}_{ij}-\vartheta^\dagger_{ij})^2
    \\
    &=\frac{1}{8}(z_{ij}^\dagger-\theta^{\dagger}_{ij})^2+\mathrm{const.}\;,
\end{align*}
where $\vartheta^\dagger_{ij}=q_i\bm{v}^{(t)\prime}_{(j)}\bm{W}^{(t)\prime}\bm{x}_i$ in $t$th iteration, and $z_{ij}^\dagger=\vartheta^\dagger_{ij}-4\xi(\vartheta^\dagger_{ij})$.
Therefore, owing to the summation property, the majorization function of $L_A(\bm{V},\bm{W},\bm{X})$ is 
\begin{align*}
    \frac{1}{8}\sum_{i=1}^n\sum_{j=1}^m (z_{ij}^\dagger-\theta^\dagger_{ij})^2 + \mathrm{const.}\;
\end{align*}
In the matrix notation, we obtain
\begin{align*}
    \frac{1}{8}\sum_{i=1}^n \sum_{j=1}^m (z_{ij}^\dagger-\theta^\dagger_{ij})^2 + \mathrm{const.}
    =\frac{1}{8}\|\bm{Z}^\dagger-\bm{\Theta}^\dagger\| + \mathrm{const.}
\end{align*}
where $\bm{\Theta}^\dagger$ and $\bm{Z}^\dagger$ are the matrices with elements $\theta_{ij}^\dagger$ and $\bm{z}_{ij}^\dagger$, respectively.
Because each $\theta_{ij}^\dagger$ is expressed as $q_i\bm{v}_{(j)}^\prime\bm{W}^\prime\bm{x}_i$ and $\bm{\Theta}=\bm{QXWV}^\prime$, it follows that 
\begin{align*}
   \frac{1}{8}\|\bm{Z}^\dagger-\bm{QXWV}^\prime\|_F^2 + \mathrm{const.}
\end{align*}
Here, since $z_{ij}^\dagger=q_i\bm{v}^{(t)\prime}_{(j)}\bm{W}^{(t)\prime}\bm{x}_i +4y_{ij}\frac{\exp(-M_i\bm{v}^{(t)\prime}_{(j)}\bm{W}^{(t)\prime}\bm{x}_i)}{1+\exp(-M_i\bm{v}^{(t)\prime}_{(j)}\bm{W}^{(t)\prime}\bm{x}_i)}$ in the $t$th iteration, $\bm{Z}^\dagger=\bm{QXW}^{(t)}\bm{V}^{(t)\prime}+4\bm{Y}\odot\bm{\Phi}^{\dagger(t)}\in\mathbb{R}^{n\times m}$.
By excluding this constant factor, $\mathcal{M}_A(\bm{V},\bm{W})$ is derived.
\end{proof}
Therefore, in every iteration of the MM algorithm, the following minimization problem should be considered:
\begin{gather}
\label{obj2}
    \underset{\bm{V},\bm{W}}{\min}\;\mathcal{M}_A(\bm{V},\bm{W})\;\mathrm{subject\;to\;}\bm{V}^\prime\bm{V}=\bm{I}_r.
    \end{gather}
We present the following two propositions for estimating the parameters $\bm{V}$ and $\bm{W}$.
\begin{prop}
Given $\bm{W}$, $\bm{V}$ is updated to $\hat{\bm{V}}$ to solve the minimization problem Eq. (\ref{obj2}).
\begin{align*}
\hat{\bm{V}}=\bm{K}^\dagger\bm{L}^{\dagger\prime}
\end{align*}
where $2\bm{Z}^{\dagger\prime}\bm{QXW}=\bm{K}^\dagger\bm{\Lambda}^\dagger\bm{L}^{\dagger\prime}$ through singular value decomposition.
$\bm{K}^\dagger\in\mathbb{R}^{m\times r}$ and $\bm{L}^\dagger\in\mathbb{R}^{r\times r}$ are the left- and right-singular vectors, respectively.
$\bm{\Lambda}^\dagger=\mathrm{diag}(\lambda^\dagger_1,\lambda^\dagger_2,\cdots,\lambda^\dagger_r)$ is a square-diagonal matrix.
\end{prop}
\begin{proof}
$\mathcal{M}_A(\bm{V},\bm{W})$ in Eq. ~ (\ref{obj2}) can be written as 
\begin{align*}
    \|\bm{Z}^\dagger-\bm{QXWV}^\prime\|_F^2
    =\mathrm{tr}\left(\bm{Z}^{\dagger\prime}\bm{Z}^\dagger-2\bm{VW}^\prime\bm{X}^\prime\bm{Q}\bm{Z}^\dagger+\bm{W}^\prime\bm{X}^\prime\bm{Q}^2\bm{XW}\right),
\end{align*}
where $\bm{Q}^2$ is the square of matrix $\bm{Q}$.
Considering only the term related to parameter $\bm{V}$,
\begin{align*}
    \underset{\bm{V}}{\mathrm{argmax}}\;2\mathrm{tr}\left((\bm{Z}^{\dagger\prime}\bm{QXW})^\prime\bm{V}\right)\;\mathrm{subject\;to}\;\bm{V}^\prime\bm{V}=\bm{I}_r.
\end{align*}
Referring to Theorem A.4.2 in Adachi (2016) \cite{adachi2016matrix}, this maximization problem can be solved using singular value decomposition, thereby proving this proposition.
\end{proof}
\begin{prop}
    Suppose that $\bm{X}^\prime\bm{X}$ is a regular matrix.
    Given $\bm{V}$, $\bm{W}$ is updated to $\hat{\bm{W}}$ to solve the minimization problem in Eq. (\ref{obj2}).
\begin{align*}
\hat{\bm{W}}=\left(\bm{X}^\prime\bm{Q}^2\bm{X}\right)^{-1}\bm{X}^\prime\bm{Q}\bm{Z}^\dagger\bm{V}
\end{align*}
\end{prop}
\begin{proof}
    $\mathcal{M}_A(\bm{V},\bm{W})$ in Eqs. (\ref{obj}) can be written as 
\begin{align*}
   \|\bm{Z}^\dagger-\bm{QXWV}^\prime\|_F^2
    =\mathrm{tr}\left(\bm{Z}^{\dagger\prime}\bm{Z}^\dagger-2\bm{VW}^\prime\bm{X}^\prime\bm{Q}\bm{Z}^\dagger+\bm{W}^\prime\bm{X}^\prime\bm{Q}^2\bm{XW}\right).
\end{align*}
Considering only the terms related to parameter $\bm{W}$,
\begin{align*}
    \underset{\bm{W}}{\mathrm{argmin}}\;-2\mathrm{tr}\left(\bm{W}^\prime\bm{X}^\prime\bm{Q}\bm{Z}^\dagger\bm{V}\right)+\mathrm{tr}\left(\bm{W}^\prime\bm{X}^\prime\bm{Q}^2\bm{XW}\right)
\end{align*}
Differentiating with respect to $\bm{W}$ to determine $\hat{\bm{W}}$, it holds that
\begin{align*}
    -2\bm{X}^\prime\bm{Q}\bm{Z}^\dagger\bm{V}
    +2\bm{X}^\prime\bm{Q}^2\bm{XW}=\bm{O}_{p\times r}
    \Rightarrow \hat{\bm{W}}=\left(\bm{X}^\prime\bm{Q}^2\bm{X}\right)^{-1}\bm{X}^\prime\bm{Q}\bm{Z}^\dagger\bm{V}.
\end{align*}
\end{proof}
The pseudocode for the estimation algorithm is presented in Algorithm \ref{alg2}.
\begin{figure}[t]
\begin{algorithm}[H]
    \caption{Logistic reduced rank regression based on the A-learner}
    \label{alg2}
    \begin{algorithmic}[1]    
    \REQUIRE$\bm{X}$,$\bm{Y}$,$\bm{T}$,$\bm{Q}$
    \ENSURE $\bm{W},\bm{V}$
    \STATE $\varepsilon>0$
    \STATE $t\leftarrow 1$
    \STATE Set initial value for $\bm{W}^{(0)}$,$\bm{V}^{(0)}$ 
\WHILE{$\mathcal{M}_A(\bm{V}^{(t)},\bm{W}^{(t)})-\mathcal{M}_A(\bm{V}^{(t-1)},\bm{W}^{(t-1)})<\varepsilon$}
    \STATE $\bm{\Phi}^{\dagger(t)}=(\phi^{\dagger(t)}_{ij})\leftarrow \frac{\exp(-q_i\bm{v}_{(j)}^{(t-1)\prime}\bm{W}^{(t-1)\prime}\bm{x}_i)}{1+\exp(-q_i\bm{v}_{(j)}^{(t-1)\prime}\bm{W}^{(t-1)\prime}\bm{x}_i)}$
    \STATE $\bm{Z}^\dagger\leftarrow \bm{QX}\bm{W}^{(t-1)}\bm{V}^{(t-1)\prime}+4\bm{Y}\odot\bm{\Phi}^{\dagger(t)}$
    \STATE  $\bm{K}^\dagger\bm{\Lambda}^\dagger\bm{L}^{\dagger\prime}$ $\leftarrow$ singular value decomposition of $2\bm{Z}^{\dagger\prime}\bm{Q}\bm{XW}^{(t-1)}$
     \STATE $\bm{V}^{(t)}\leftarrow\bm{K}^\dagger\bm{L}^{\dagger\prime}$
     \STATE $\bm{W}^{(t)}\leftarrow (\bm{X}^\prime\bm{Q}^2\bm{X})^{-1}\bm{X}^\prime\bm{QZ}^\dagger\bm{V}^{(t)}$
     \STATE $t \leftarrow t+1$
    \ENDWHILE
      \STATE$\bm{W} \leftarrow \bm{W}^{(t)}$
      \STATE$\bm{V} \leftarrow \bm{V}^{(t)}$
    \RETURN $\bm{W},\bm{V}$
    \end{algorithmic}
\end{algorithm}
\end{figure}
\section{Numerical study}
In this section, we present the simulation designs, including data generation, comparison methods, evaluation indices, and results of the numerical simulations.
First, we describe the data-generation method.
We generate a matrix corresponding to binary outcome $\bm{Y}\in\{0,1\}^{n\times m}$ as $Y_{ij}=I(Y_{ij}^* >0)$, where $Y_{ij}^*$ is the $(i,j)$th element of $\bm{Y}^*=\bm{XD}\odot \bm{XD}+\bm{TXWV}^\prime+\bm{E}$. 
This generation method was performed in the same manner as that used by Tian et al. (2014) \cite{tian2014simple}. The first term represents the main effect and $\bm{XWV}^\prime$ in the second term corresponds to the treatment effect. The matrix of explanatory variables is $\bm{X}\in\mathbb{R}^{n\times p}\sim N(\bm{0}_p,(1-\rho_1)\bm{I}_p+\rho_1\bm{1}_p\bm{1}^\prime_p)$, the coefficient matrix corresponding to the main effect is $\bm{D}=(\bm{d}_1,\bm{d}_2,\cdots,\bm{d}_m)\in\mathbb{R}^{p\times m}$, and the matrix of the random error is $\bm{E}\in\mathbb{R}^{n\times m}\sim N(\bm{0}_m,(1-\rho_2)\bm{I}_m+\rho_2\bm{1}_m\bm{1}^\prime_m)$, where $\boldsymbol{1}_{p}$ is the $p$-dimensional vector whose elements are all one.
The treatment assignment $T_i$ for the $i$-th subject is generated from a Bernoulli distribution satisfying $P(T_i=1)=P(T_i=-1)=0.5$ in the case of RCTs, and from simple logistic model $P(T_i=1|\bm{X}_i=\bm{x}_i)=1/\{1+\exp (1-x_{i1})\}$, following Chen et al. (2017), in the case of observational studies.
The true values of $\bm{D}$ and $\bm{W}$ are generated as random variables from a normal distribution $N(0,1)$. The true values of $\bm{V}$ are obtained by performing QR decomposition on a matrix of random variables from a normal distribution $N(0,1)$. Here, we assume that we mistakenly estimate the nonlinear main effects using a linear function. We thus consider the following settings for sample size, number of variables, and hyperparameters:
\begin{enumerate}
    \item Sample size; $n=\{100,300,500\}$
    \item The number of explanatory variables; $p=\{10,50\}$
    \item The number of outcomes; $m=\{5,10\}$
    \item The rank of the coefficient matrix for estimating treatment effects; $r=\{3,5\}$
    \item Hyperparameters; $\rho_1=\{0,1/3,2/3\}$,\;$\rho_2=\{0,1/3,2/3\}$
\end{enumerate}
As the proposed method assumes a low-rank structure in the coefficient matrix, we set $m=10$ when $r=5$.
The number of iterations for each setting is 100.
\par
Next, we describe the comparison methods.
These values are determined following the study of Chen et al. (2017) \cite{chen2017general}.
Specifically, they consist of the following seven components:
\begin{enumerate}
\item Full model (represented as ``Full,'' e.g., \cite{chen2017general} and \cite{tian2014simple});
using multiple outcomes $\bm{Y}$ as response variables, $\bm{X}$ as the explanatory variable for the main effect, and $\bm{TX}$ as the explanatory variables for the interaction effect, we first estimate the parameters using standard logistic regression.
Subsequently, we construct the HTE based on these estimates.
    \item Multivariate logistic regression in the framework of the A-learner (represented as ``MA'');
    we apply the A-learner to each outcome.
    Therefore, we solve the following minimization problem:
    \begin{align*}
        \underset{\bm{\varGamma}}{\min}\;\sum_{i=1}^n \sum_{j=1}^m y_{ij}\log \left\{1+\exp\left(-\{(t_i+1)/2-\pi(\bm{x}_i)\}\times \bm{\gamma}_{j}^\prime\bm{x}_i\right)\right\}
    \end{align*}
    \item The bias-corrected MA (represented as ``MAmod'');
    we adjust for the bias identified by Eq. (\ref{bias_mod}) in Theorem \ref{thm_bias} when computing the treatment effect using the estimated $\bm{\varGamma}$ from MA.
    \item Multivariate logistic regression in the framework of the W-method (represented as ``MW'');
    we apply the W-method to each outcome.
    Therefore, we solve the following minimization problem:
    \begin{align*}
        \underset{\bm{\varGamma}}{\min}\;\sum_{i=1}^n \sum_{j=1}^m \frac{y_{ij}\log \left\{1+\exp\left(-t_i\bm{\gamma}_{j}^\prime\bm{x}_i\right)\right\}}{t_i\pi(\bm{x}_i)+(1-t_i)/2}
    \end{align*}
     \item Proposed method in the framework of the A-learner (represented as ``R3A'');
     we estimate $\bm{V}$ and $\bm{W}$ using Algorithm $\ref{alg2}$.
     \item The bias-corrected R3A (represented as ``R3Amod'');
     we adjust for the bias identified by Theorem \ref{bias_RRRA1} and \ref{bias_RRRA2} when computing the treatment effect using the estimated $\hat{\bm{W}}$ and $\hat{\bm{V}}$ from R3A as follows;
     \begin{align*}
         \log \frac{E[Y_{ij}|T_i=1,X_i=\bm{x}_i]}{E[Y_{ij}|T_i=-1,X_i=\bm{x}_i]}
    =\hat{\bm{v}}_{(j)}^\prime\hat{\bm{W}}^\prime\bm{x}_i+
    \log \frac{1+e^{-\{1-\pi(\bm{x}_i)\}\hat{\bm{v}}_{(j)}^\prime\hat{\bm{W}}^\prime\bm{x}_i}}{1+e^{\pi(\bm{x}_i)\hat{\bm{v}}_{(j)}^\prime\hat{\bm{W}}^\prime\bm{x}_i}}.
     \end{align*}
    \item Proposed method in the framework of the W-method (represented as ``R3W'');
    we estimate $\bm{V}$ and $\bm{W}$ using Algorithm $\ref{alg1}$.
\end{enumerate}
\par
Finally, we describe the evaluation indices.
The objective of this simulation is to accurately estimate the treatment effect $\bm{XWV}^\prime$ for binary outcomes and identify subgroups in which the treatment is effective.
Therefore, we utilize the mean squared error (MSEs) and area under the curve (AUC) as evaluation indices.
MSEs are defined as follows:
\begin{align*}
    \frac{1}{nm}\|h(\hat{\bm{\varGamma}})-\bm{XWV}^\prime\|_F^2,
\end{align*}
where $\hat{\bm{\varGamma}}=(\hat{\bm{\gamma}}_1,\hat{\bm{\gamma}}_2,\cdots,\hat{\bm{\gamma}}_m)\in\mathbb{R}^{p\times m}$ is the estimator of $\bm{\varGamma}$ and $h(\bm{\varGamma}):\mathbb{R}^{p\times m}\mapsto \mathbb{R}^{n\times m}$ represents the function of estimated treatment effects.
Next, we define the false-positive rate (FPR) and false-negative rate (FNR) to evaluate the AUC:
\begin{gather*}
    \mathrm{FPR} =\mathrm{card}\left\{i\middle|\sum_{j=1}^m \hat{\bm{\gamma}}_j^\prime\bm{x}_i > 0\; \mathrm{and}\; \sum_{j=1}^m \hat{\bm{\gamma}}_j^\prime\bm{x}_i\leq 0\right\}/\mathrm{card}\left\{i\middle|\sum_{j=1}^m \hat{\bm{\gamma}}_j^\prime\bm{x}_i\leq 0\right\},
    \\
    \mathrm{FNR} =\mathrm{card}\left\{i\middle|\sum_{j=1}^m \hat{\bm{\gamma}}_j^\prime\bm{x}_i \leq  0\; \mathrm{and}\; \sum_{j=1}^m \hat{\bm{\gamma}}_j^\prime\bm{x}_i > 0\right\}/\mathrm{card}\left\{i\middle|\sum_{j=1}^m \hat{\bm{\gamma}}_j^\prime\bm{x}_i> 0\right\},
\end{gather*}
where $\mathrm{card}(\cdot)$ denotes the number of elements in the set.
When evaluating the AUC, to ensure a clear interpretation of the plots, we focus on the ROC curves for Full, MAmod, MW, R3Amod, and R3W.
\subsection{Results}
In this subsection, we discuss the results of the numerical simulations assuming RCTs from the perspectives of MSEs and the AUC, respectively.
The results of the numerical simulations, assuming observational studies, are presented in the Supplementary Materials.
First, we discuss the results of MSEs from the perspective of the estimation accuracy of treatment effects.
The results are presented for scenarios with a large number of explanatory variables ($p=50$) and various combinations of correlations between the explanatory variables and the outcome, specifically, $(\rho_1,\rho_2)=(2/3,0),\;(0,2/3)$ and $(2/3,2/3)$.
The figures for all scenarios, including these scenarios, are presented in the Supplementary Material.
Throughout the numerical study of the MSEs, in every scenario, R3W consistently exhibits a higher estimation accuracy than R3Amod.
Comparing R3A with R3Amod and MA with MAmod clearly indicates that the modifications allow for correction of the bias inherent in the original A-learner, as demonstrated in Theorems \ref{thm_bias}, \ref{bias_RRRA1}, and \ref{bias_RRRA2}, thereby accurately estimating the treatment effects.
\par
When the sample size is large ($n=500$), the proposed methods, R3Amod and R3W, outperform the other comparison methods in terms of estimating the treatment effects, regardless of the number of explanatory variables.
Furthermore, Full performs poorly in terms of MSEs overall because it misspecifies the nonlinear main effects as linear functions.
Referencing Figures \ref{set6} and \ref{set14}, it is evident that in scenarios with high correlation among explanatory variables $(\rho_1=2/3,\;\rho_2=0)$ and in scenarios with high correlation among outcomes $(\rho_1=0,\rho_2=2/3)$, the proposed method, particularly R3W, achieves lower MSEs values compared with other methods when the sample size is large.
Based on Figure \ref{set18}, in scenarios where there is high correlation both among explanatory variables and among outcomes $(\rho_1=2/3,\;\rho_2=2/3)$, the proposed method outperforms others when $n$ is relatively large $(n=300,\;n=500)$.
In addition, by comparing the first and second rows in Figures \ref{set6}, \ref{set14}, and \ref{set18}, it is apparent that, when the number of outcomes is large $(m=10)$ and the rank of the treatment effects is small $(r=3)$, the proposed method exhibits superior accuracy in terms of MSEs.
This indicates that the proposed method correctly identifies the rank of treatment effects and considers the correlation structure of the treatment effects.
\par
Subsequently, we discuss the results of the AUC from the perspective of accuracy in identifying the subgroups in which the treatment is effective.
The results are presented for scenarios with a large number of explanatory variables ($p=50$) and various combinations of correlations between the explanatory variables and the outcome; specifically, $(\rho_1,\rho_2)=(1/3,0),\;(2/3,0),\;(0,2/3),\;(1/3,2/3)$, and $(2/3,2/3)$.
In all the figures, the blue, green, pink, red, and black lines represent R3W, R3Amod, MRW, MRAmod, and Full, respectively.
The figures for all scenarios, including these scenarios, are presented in the Supplementary Material.
Overall, the proposed method R3W demonstrates higher accuracy in terms of AUC than the other methods. In contrast, R3Amod shows comparable accuracy to MRAmod, and is more accurate than MRW. When the sample size is moderately large, such as $n=300$ or $n=500$, the proposed R3W method consistently achieves the highest AUC values in all scenarios.
Figure \ref{roc_set4} show that, in scenarios with moderate correlation among explanatory variables $(\rho_1=1/3,\;\rho_2=0)$, the proposed method, particularly when the sample size is large $(n=500)$, demonstrates significantly higher AUC values compared to other methods.
However, as shown in Figures \ref{roc_set6}, in scenarios with high correlation among explanatory variables $(\rho_1=2/3,\;\rho_2=0)$, the proposed method does not exhibit substantial differences in accuracy compared to other methods.
Figures \ref{roc_set14} and \ref{roc_set16} confirm that the proposed method robustly identifies subgroups even in cases with high correlation among outcomes $(\rho_1=0,1/3,\;\rho_2=2/3)$.
In addition, as shown in Fig \ref{roc_set18}, in scenarios where there is high correlation among both outcomes and explanatory variables $(\rho_1=2/3,\;\rho_2=2/3)$, the proposed method achieves superior AUC values when the number of outcomes is relatively small.
These results confirm that the proposed method, R3W, can robustly discriminate between subgroups in the presence of correlations between the predictors and outcomes.
\begin{figure}[htbp]
  \centering
  \hspace*{-15mm}
\includegraphics[width=1.2\textwidth]{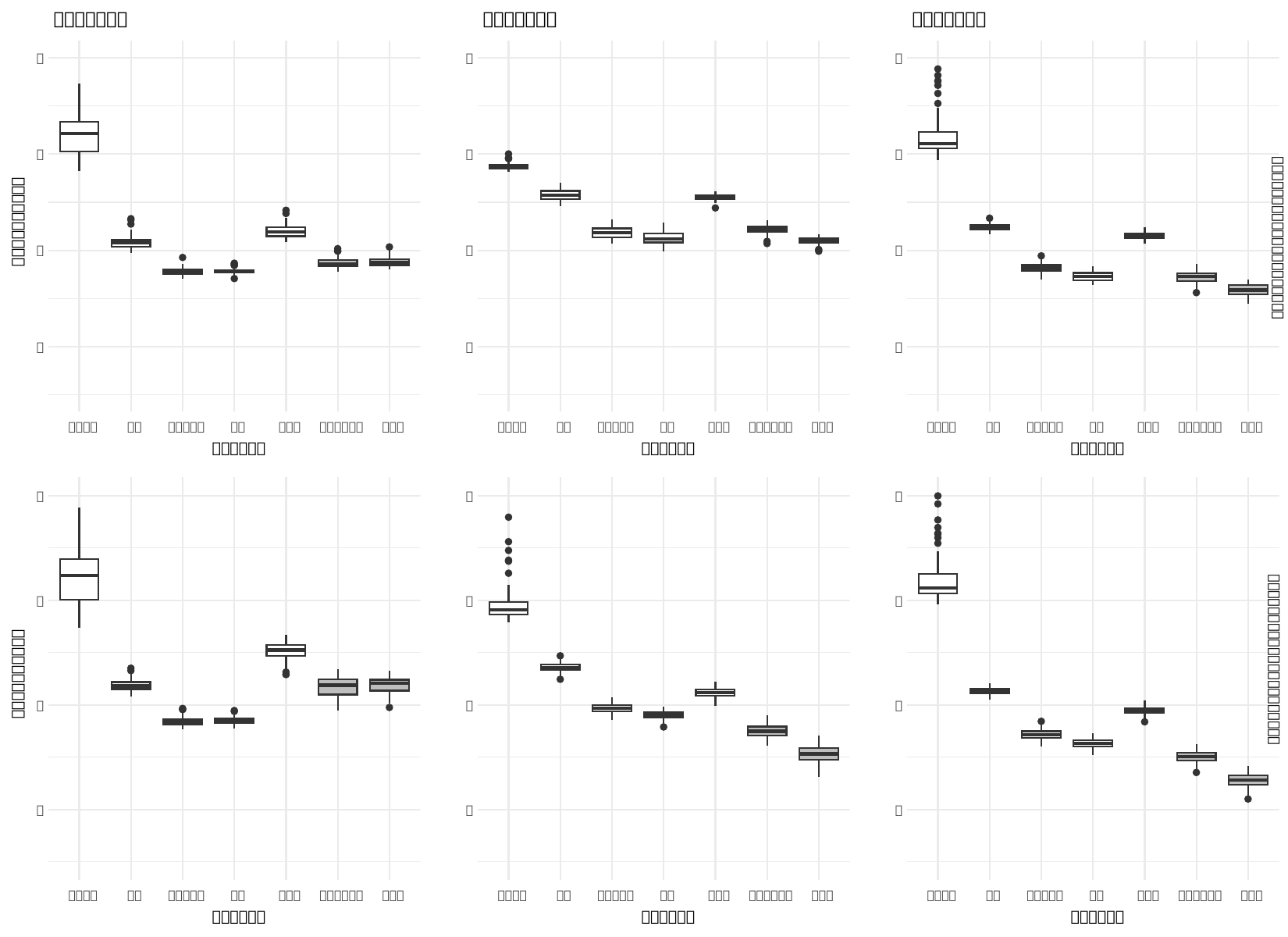}
  \caption{Boxplot of MSE calculated for scenarios with 50 explanatory variables that are highly correlated and errors that are uncorrelated, i.e., $\rho_1=2/3$ and $\rho_2=0$}
  \label{set6}
\end{figure}
\begin{figure}[htbp]
  \centering
  \hspace*{-15mm}
\includegraphics[width=1.2\textwidth]{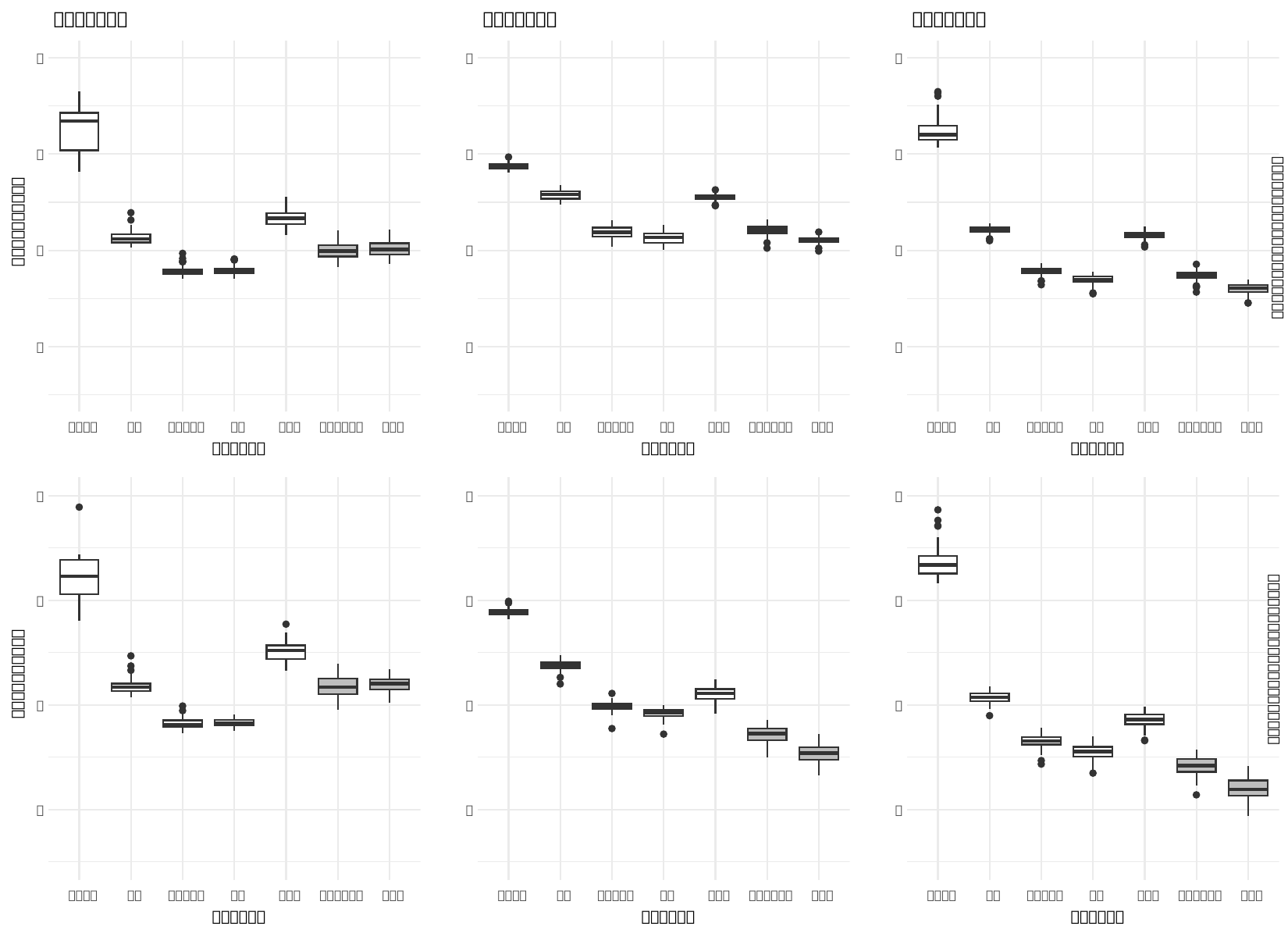}
  \caption{Boxplot of MSE calculated for scenarios with 50 explanatory variables that are uncorrelated and errors that are highly correlated, i.e., $\rho_1=0$ and $\rho_2=2/3$}
  \label{set14}
\end{figure}
\begin{figure}[htbp]
  \centering
  \hspace*{-15mm}
\includegraphics[width=1.2\textwidth]{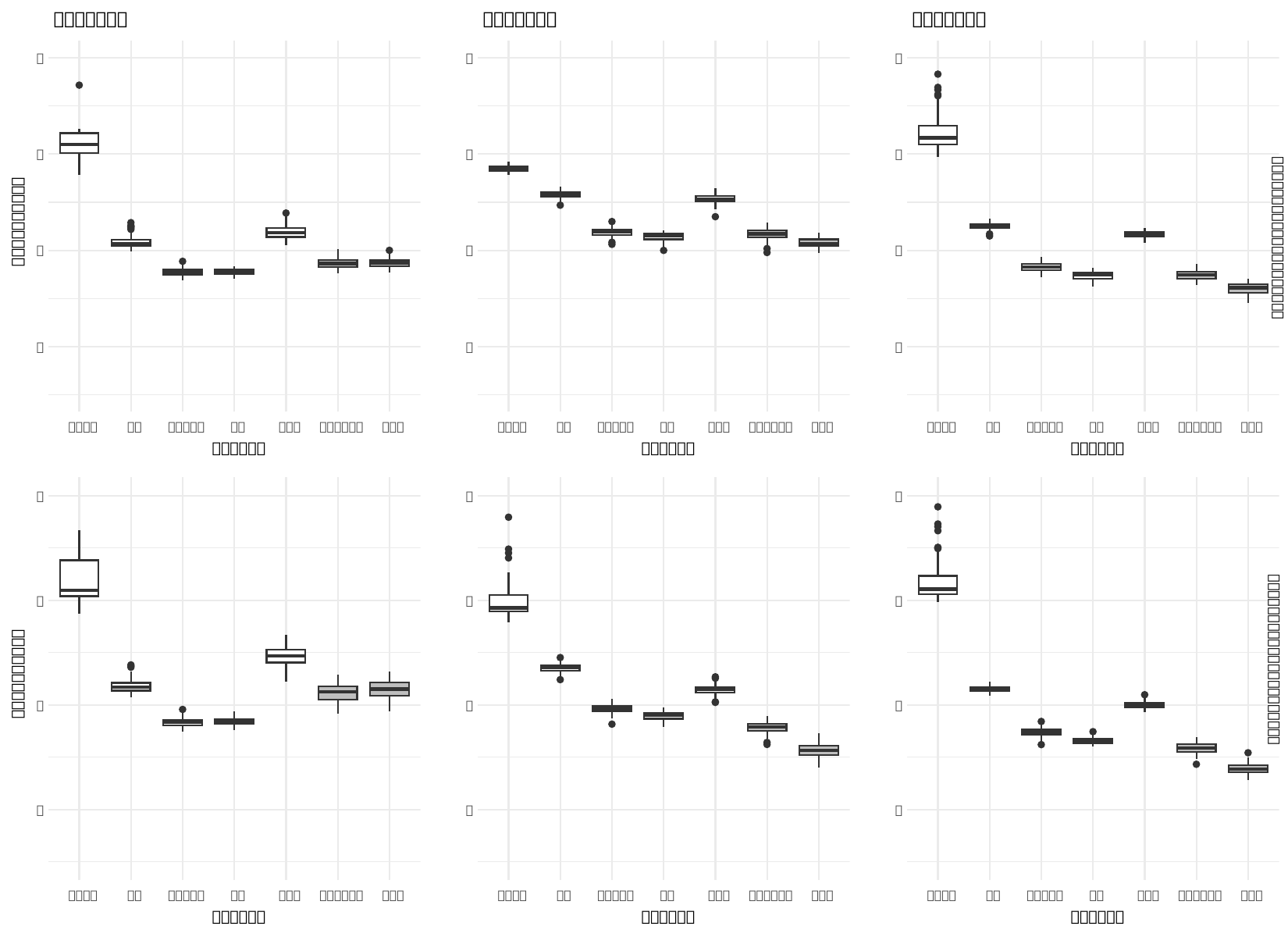}
  \caption{Boxplot of MSE calculated for scenarios with 50 explanatory variables and errors, both of which are highly correlated, i.e., $\rho_1=2/3$ and $\rho_2=2/3$}
  \label{set18}
\end{figure}
\begin{figure}[htbp]
  \centering
  \hspace*{-15mm}
\includegraphics[width=1.2\textwidth]{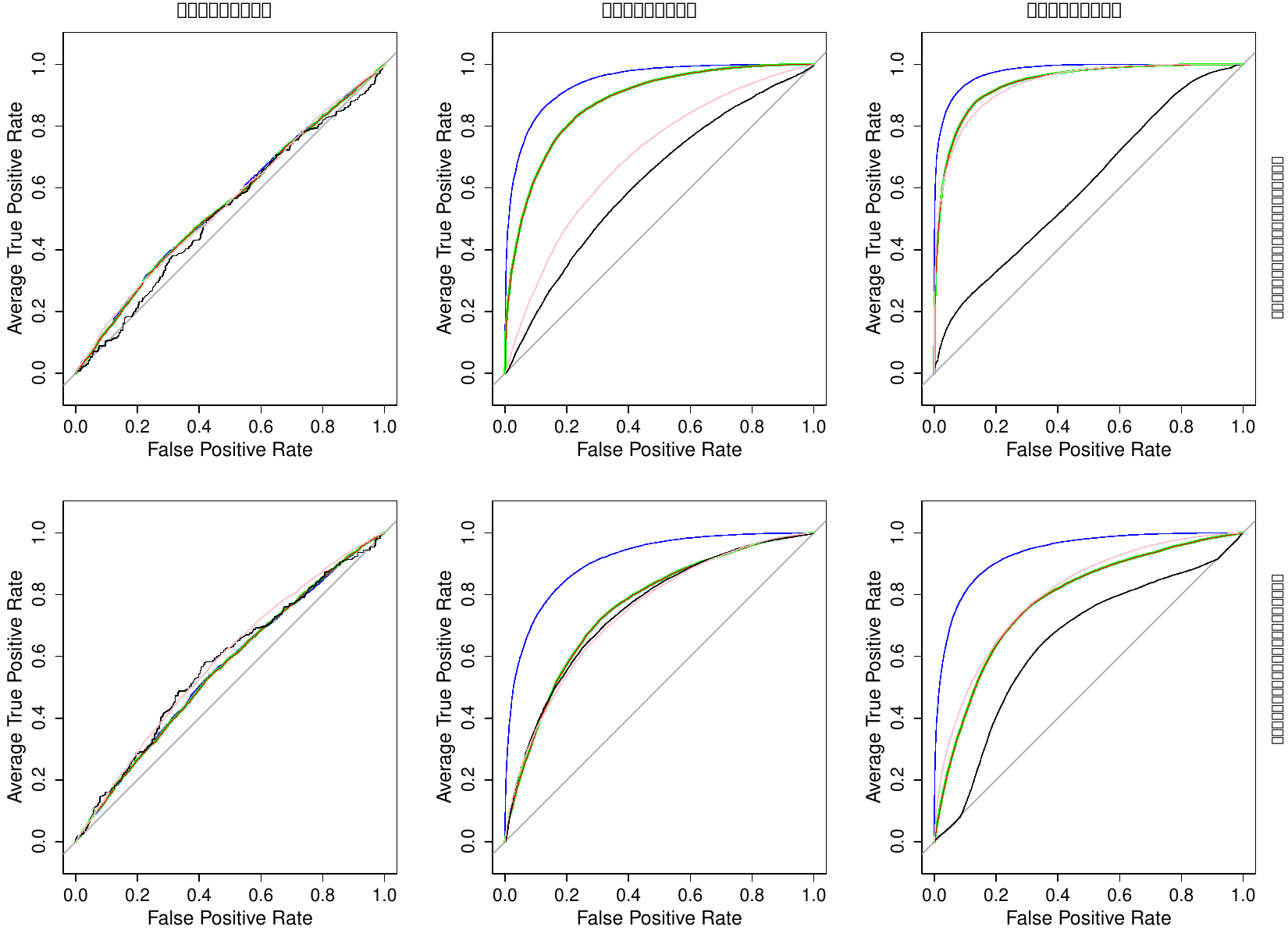}
  \caption{ROC curves for scenarios with 50 explanatory variables that are moderately correlated and errors that are uncorrelated, i.e., $\rho_1=1/3$ and $\rho_2=0$.
  The blue, green, pink, red, and black lines represent R3W, R3Amod, MRW, MRAmod, and Full, respectively.}
  \label{roc_set4}
\end{figure}
\begin{figure}[htbp]
  \centering
  \hspace*{-15mm}
\includegraphics[width=1.2\textwidth]{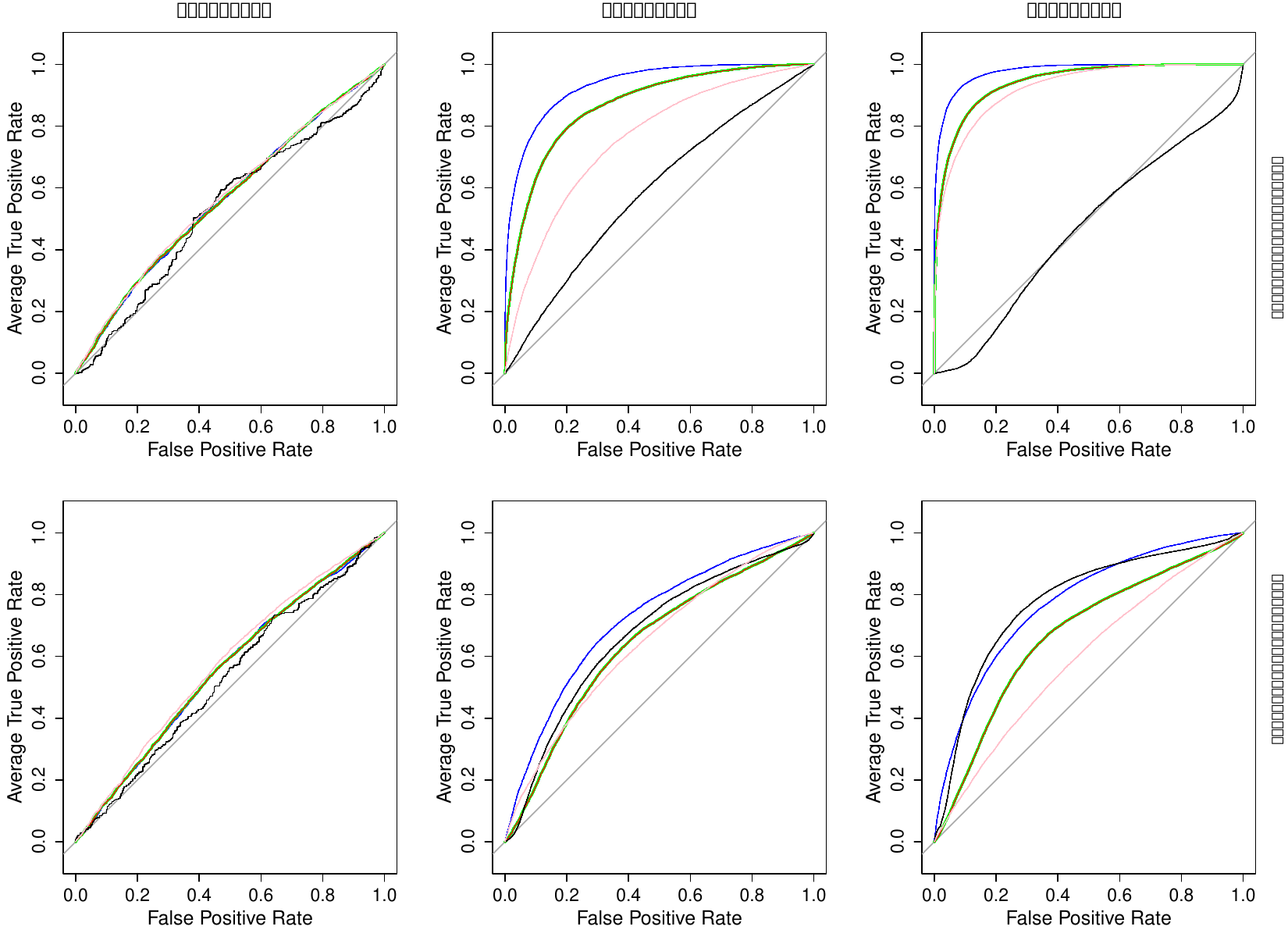}
  \caption{ROC curves for scenarios with 50 explanatory variables that are highly correlated and errors that are uncorrelated, i.e., $\rho_1=2/3$ and $\rho_2=0$.
  The blue, green, pink, red, and black lines represent R3W, R3Amod, MRW, MRAmod, and Full, respectively.}
  \label{roc_set6}
\end{figure}
\begin{figure}[htbp]
  \centering
  \hspace*{-15mm}
\includegraphics[width=1.2\textwidth]{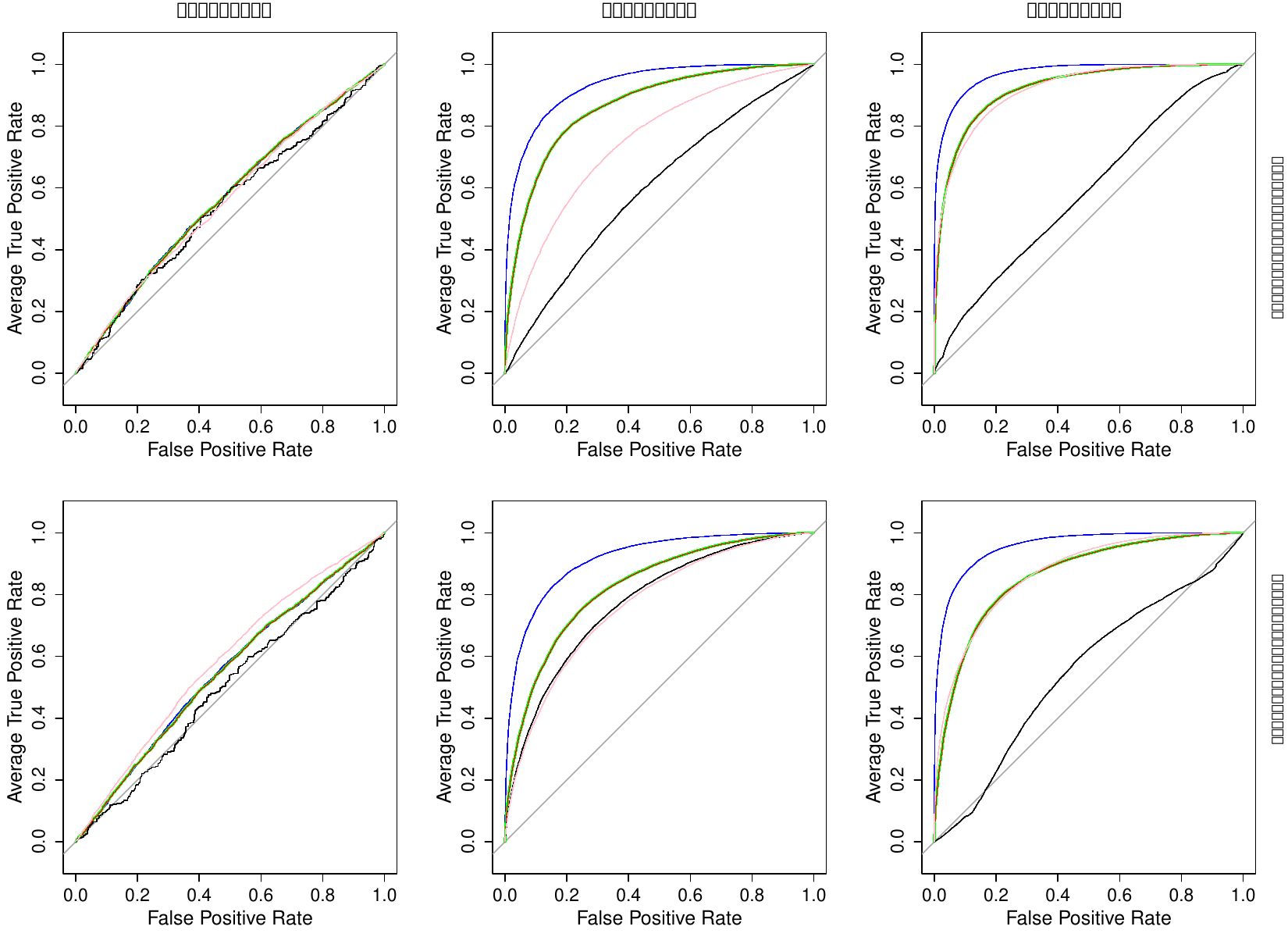}
  \caption{ROC curves for scenarios with 50 explanatory variables that are uncorrelated and errors that are highly correlated, i.e., $\rho_1=0$ and $\rho_2=2/3$.
  The blue, green, pink, red, and black lines represent R3W, R3Amod, MRW, MRAmod, and Full, respectively.}
  \label{roc_set14}
\end{figure}
\begin{figure}[htbp]
  \centering
  \hspace*{-15mm}
\includegraphics[width=1.2\textwidth]{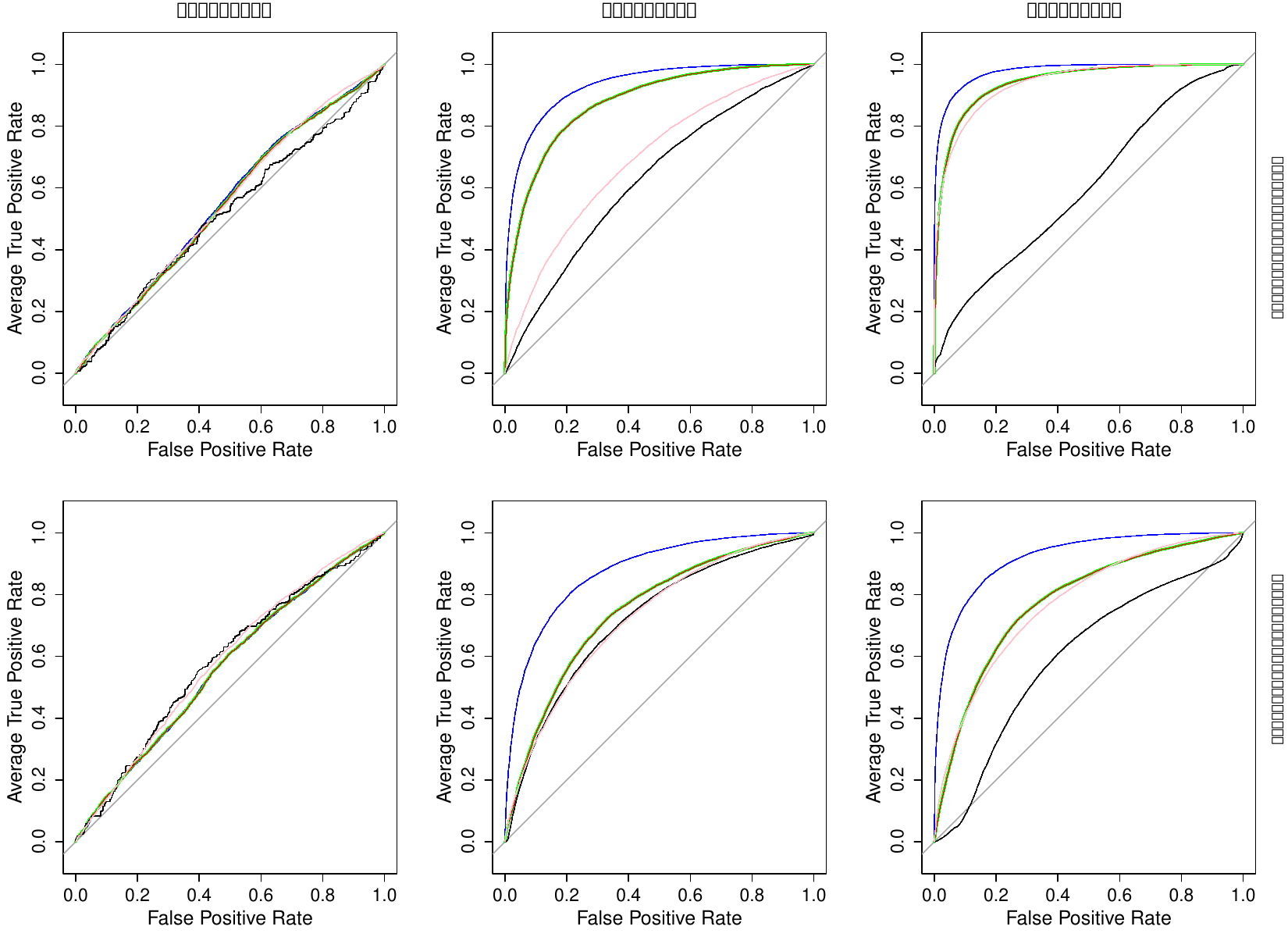}
  \caption{ROC curves for scenarios with 50 explanatory variables that are moderately correlated and errors that are highly correlated, i.e., $\rho_1=1/3$ and $\rho_2=2/3$.
  The blue, green, pink, red, and black lines represent R3W, R3Amod, MRW, MRAmod, and Full, respectively.}
  \label{roc_set16}
\end{figure}
\begin{figure}[htbp]
  \centering
  \hspace*{-15mm}
\includegraphics[width=1.2\textwidth]{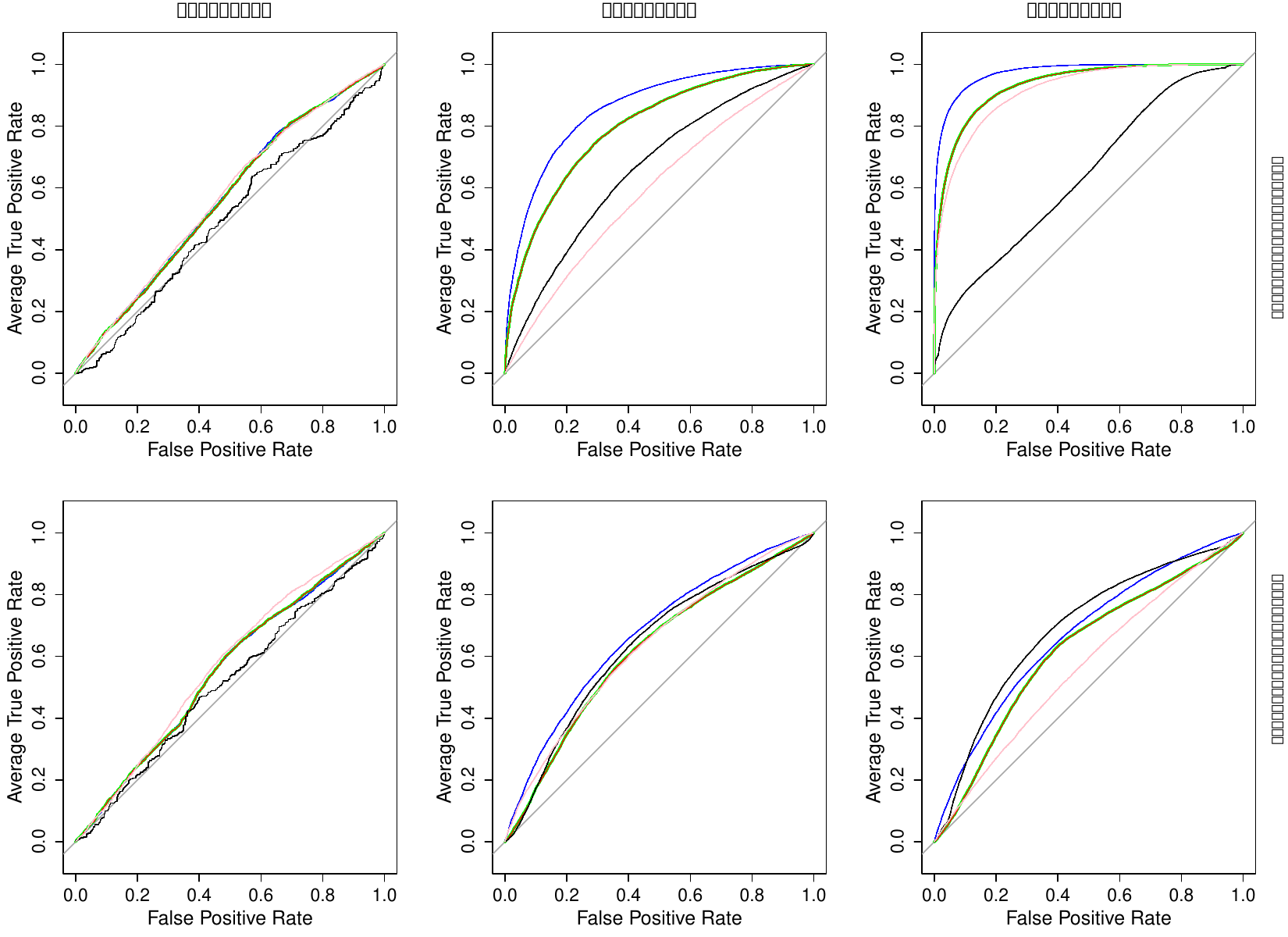}
  \caption{ROC curves for scenarios with 50 explanatory variables and errors, both of which are highly correlated, i.e., $\rho_1=2/3$ and $\rho_2=2/3$.
  The blue, green, pink, red, and black lines represent R3W, R3Amod, MRW, MRAmod, and Full, respectively.}
  \label{roc_set18}
\end{figure}
\section{Application}
In this section, we describe the application of the proposed method to real data to demonstrate its effectiveness. We identify subgroups and interpret the results based on the estimated loadings.
\subsection{Data description}
In this subsection, we explain the data acquired from the AIDS Clinical Trial Group Study 175 (ACTG175).
This dataset is available in the \ textbf{speff2trial}package of R \cite{juraska2022package}.
This study, known as ACTG175, was a randomized controlled trial designed to evaluate the effectiveness of different treatments in HIV-positive patients.
In this trial, the patients were randomized into four treatment groups: 1. zidovudine alone, 2. zidovudine plus didanosine, 3. zidovudine plus zalcitabine, and 4. didanosine alone.
We define the test therapy as a combination therapy with zidovudine plus didanosine ($n = 332$) and the control therapy as monotherapy with zidovudine ($n = 318$).
We set three binary outcomes: CD4 T cell count at 20 $\pm$ 5 weeks (represented as ``cd420''), CD4 T cell count at 96 $\pm$ 5 weeks (``cd496''), and CD8 T cell count at 20 $\pm$ 5 weeks (``cd820'') are each dichotomized using the median values as the cutoff.
We set fifteen covariates: age in years at baseline (``age''); 
weight in kg at baseline (``wtkg''); 
hemophilia: 0=no, 1=yes (``hemo''); homosexual activity:
0=no, 1=yes (``homo''); 
Karnofsky score(on a scale of 0-100) (``karn''); 
zidovudine use in the 30 days prior to treatment initiation: 0=no, 1=yes (``z30''); 
Race: 0=white, 1=non-white(``race'') 
history of intravenous drug use: 0=no, 1=yes (``drugs'');
gender: 0=female, 1=male (``gend''); 
antiretroviral history: 0=naive, 1=experienced (``str2'');
symptomatic indicator: 0=asymptomatic, 1=symptomatic (``symp'');
non-zidovudine antiretroviral therapy prior to initiation of study treatment: 0=no, 1=yes (``opri''); CD4 T cell count at baseline (``cd40''); 
CD8 T cell count at baseline (``cd80'').
The purpose of this application was to identify and interpret subgroups in which zidovudine plus didanosine treatment showed higher efficacy than zidovudine alone by applying the proposed method to this dataset.
Based on the results of the numerical simulations, the proposed method, R3W, demonstrated superior accuracy in estimating treatment effects and distinguishing subgroups compared to the comparison methods in terms of both MSEs and AUC.
Therefore, in this subsection, we apply the R3W to the ACTG175 dataset to interpret the results.
Here, we set the rank $r$ of the coefficient matrix to estimate treatment effects at $2$.
\subsection{Results}
This subsection discusses the results of the application and interpretation of the subgroups.
Figure \ref{path} depicts the estimated $\bm{V}$ and $\bm{W}$ as path diagrams.
In Figure $\ref{path}$, the first axis of the low-rank structure is labeled as ``Dim1,'' and the second axis as ``Dim2.''
Furthermore, we removed the coefficients of $\bm{W}$ with absolute values less than $0.1$ from Figure $\ref{path}$.
\begin{figure}[htbp]
  \centering
\includegraphics[width=1.05\textwidth]{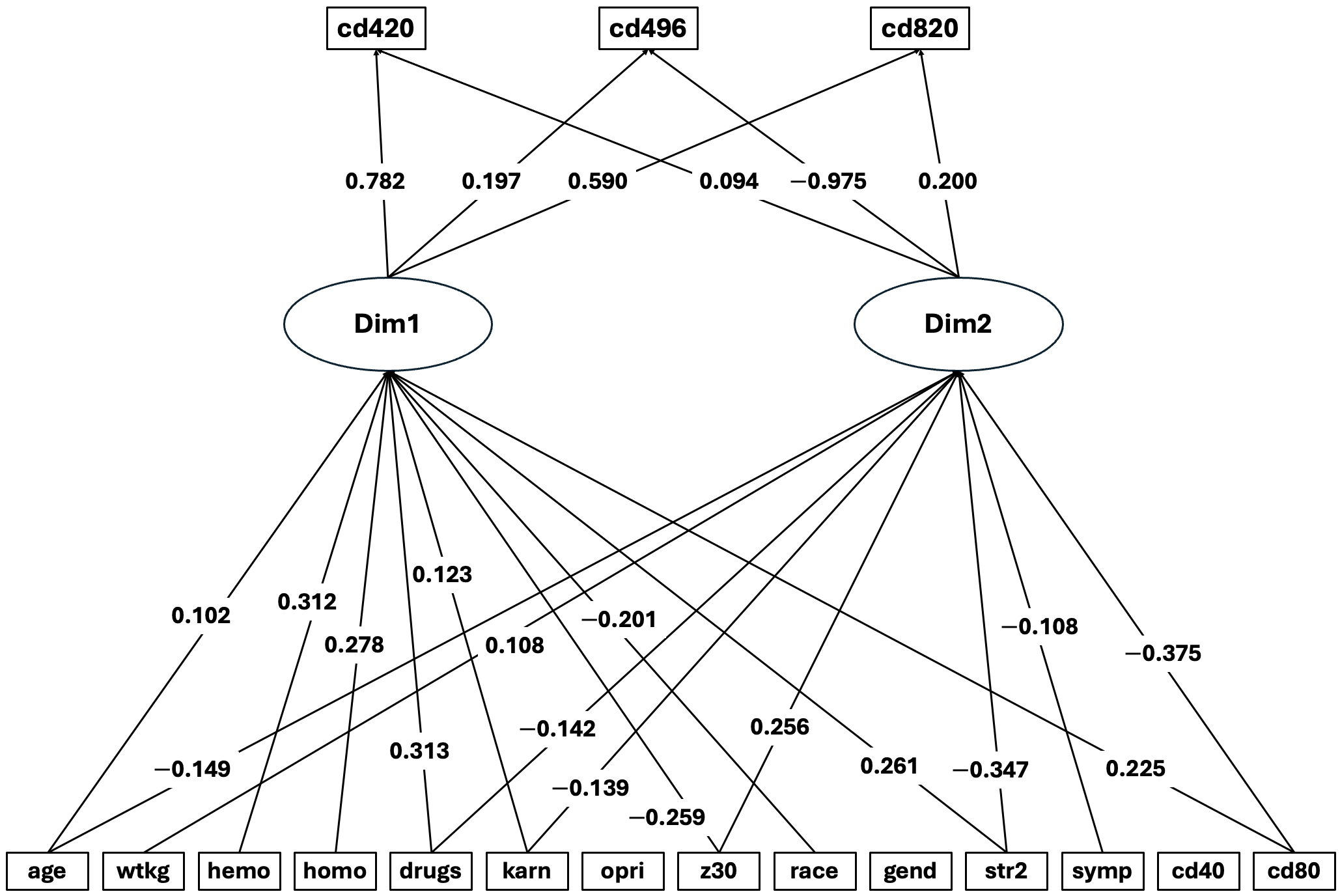}
  \caption{The path diagram using estimated $\bm{V}$ (top) and $\bm{W}$ (bottom)}
  \label{path}
\end{figure}
From this path diagram, ``Dim1'' is identified as a factor where an increase in its value corresponds to higher treatment effects on CD4 and CD8 cells.
Conversely, ``Dim2'' is associated with a decrease in the treatment effect on CD4 cells as its value increases.
From these findings, we interpret that the subgroup in which zidovudine plus didanosine treatment shows higher efficacy than zidovudine alone can be described as follows:
\\
    $\bullet$ Older patients\quad 
    $\bullet$ Patients with a history of hemophilia
    \\
    $\bullet$ Patients with history of homosexual activity
    \\
    $\bullet$ Patients with a history of intravenous drug use
    \\
    $\bullet$ Patients with high Karnofsky scores.
    \\
    $\bullet$ Patients who did not receive zidovudine within the past 30 days\quad
    $\bullet$ White people.
    \\
    $\bullet$ Patients with experienced antiretroviral history\quad 
    $\bullet$ Symptomatic patients.
    \\
    $\bullet$ Patients with high baseline of CD8cells
\section{Conclusion}
In this study, we propose a method that allows the estimation of HTE for multiple binary outcomes in both randomized controlled trials and observational studies using interpretable linear functions.
These methods extend the W-method and A-learner within the framework of a reduced-rank regression, thereby capturing the correlation structure among multiple outcomes.
In addition, by assuming a low-rank structure, the latent variables can be interpreted as subgroups.
We also demonstrate that the conventional A-learner includes a bias dependent on the propensity score. 
In our numerical study, the proposed methods outperforms the comparison methods in terms of MSEs and AUC, except in certain specific scenarios.
Overall, the proposed method based on the W-method demonstrates higher accuracy in terms of both MSEs and AUC compared with the approach based on the A-learner, and this tendency is also observed for existing methods.
Furthermore, the bias in the A-learner is accurately detected, and after correction, the MSEs values decrease significantly.
In a real data application using the ACTG 175 dataset, we identify factors for which the combination of zidovudine and didanosine show greater treatment effects on CD4 and CD8 cell counts than zidovudine alone.
\par
Three issues need to be addressed in future research.
First, in the numerical simulations and real data applications, we assume scenarios in which the numbers of explanatory variables and binary outcomes are relatively small, specifically $p,m < n$.
However, in studies focusing on genes, the number of explanatory variables and binary outcomes can be high-dimensional (e.g., \cite{kobak2021sparse} and \cite{park2024low}).
In such situations, the interpretation of subgroups can become complex and the model may overfit the data. 
Therefore, in future, we plan to incorporate sparse estimations into the loading matrices.
This approach enables interpretable and versatile modeling, even in high-dimensional settings.
Second, this study considers only cases in which the true rank of the treatment effect is correctly identified. 
However, in many cases, the true rank is unknown, and misidentifying it can introduce bias into the estimators \cite{anderson2002specification}.
Therefore, future work should address the extent of bias introduced by misidentifying the rank of the treatment effect in the proposed methods and discuss strategies for selecting an appropriate rank (e.g., \cite{bunea2011optimal} and \cite{chen2013reduced}).
Finally, our numerical simulations are performed under the assumption of randomized controlled trials. However, the proposed method is capable of estimating the treatment effects in both randomized controlled trials and observational studies.
\bibliography{Manu}

\begin{thebibliography}{10}

\bibitem{rothwell2005subgroup}
Peter~M Rothwell.
\newblock Subgroup analysis in randomised controlled trials: importance, indications, and interpretation.
\newblock {\em The Lancet}, Vol. 365, No. 9454, pp. 176--186, 2005.

\bibitem{strobl2009introduction}
Carolin Strobl, James Malley, and Gerhard Tutz.
\newblock An introduction to recursive partitioning: rationale, application, and characteristics of classification and regression trees, bagging, and random forests.
\newblock {\em Psychological methods}, Vol.~14, No.~4, p. 323, 2009.

\bibitem{foster2011subgroup}
Jared~C Foster, Jeremy~MG Taylor, and Stephen~J Ruberg.
\newblock Subgroup identification from randomized clinical trial data.
\newblock {\em Statistics in medicine}, Vol.~30, No.~24, pp. 2867--2880, 2011.

\bibitem{shen2015inference}
Juan Shen and Xuming He.
\newblock Inference for subgroup analysis with a structured logistic-normal mixture model.
\newblock {\em Journal of the American Statistical Association}, Vol. 110, No. 509, pp. 303--312, 2015.

\bibitem{kosorok2015adaptive}
Michael~R Kosorok and Erica~EM Moodie.
\newblock {\em Adaptive treatment strategies in practice: planning trials and analyzing data for personalized medicine}.
\newblock SIAM, 2015.

\bibitem{siriwardhana2020personalized}
Chathura Siriwardhana and Karunarathna~B Kulasekera.
\newblock Personalized treatment plans with multivariate outcomes.
\newblock {\em Biometrical Journal}, Vol.~62, No.~8, pp. 1973--1985, 2020.

\bibitem{kulasekera2022quantiles}
Karunarathna~B Kulasekera and Chathura Siriwardhana.
\newblock Quantiles based personalized treatment selection for multivariate outcomes and multiple treatments.
\newblock {\em Statistics in medicine}, Vol.~41, No.~15, pp. 2695--2710, 2022.

\bibitem{williams1996design}
Paige Williams and Louise Ryan.
\newblock Design of multiple binary outcome studies with intentionally missing data.
\newblock {\em Biometrics}, pp. 1498--1514, 1996.

\bibitem{inan2017joint}
G{\"u}l Inan and R~Yucel.
\newblock Joint gees for multivariate correlated data with incomplete binary outcomes.
\newblock {\em Journal of Applied Statistics}, Vol.~44, No.~11, pp. 1920--1937, 2017.

\bibitem{dunson2000bayesian}
David~B Dunson.
\newblock Bayesian latent variable models for clustered mixed outcomes.
\newblock {\em Journal of the Royal Statistical Society: Series B (Statistical Methodology)}, Vol.~62, No.~2, pp. 355--366, 2000.

\bibitem{yuki2023estimation}
Shintaro Yuki, Kensuke Tanioka, and Hiroshi Yadohisa.
\newblock Estimation and visualization of heterogeneous treatment effects for multiple outcomes.
\newblock {\em Statistics in Medicine}, Vol.~42, No.~5, pp. 693--715, 2023.

\bibitem{massy1965principal}
William~F Massy.
\newblock Principal components regression in exploratory statistical research.
\newblock {\em Journal of the American Statistical Association}, Vol.~60, No. 309, pp. 234--256, 1965.

\bibitem{izenman1975reduced}
Alan~Julian Izenman.
\newblock Reduced-rank regression for the multivariate linear model.
\newblock {\em Journal of multivariate analysis}, Vol.~5, No.~2, pp. 248--264, 1975.

\bibitem{yee2003reduced}
Thomas~W Yee and Trevor~J Hastie.
\newblock Reduced-rank vector generalized linear models.
\newblock {\em Statistical modelling}, Vol.~3, No.~1, pp. 15--41, 2003.

\bibitem{rosenbaum1983central}
Paul~R Rosenbaum and Donald~B Rubin.
\newblock The central role of the propensity score in observational studies for causal effects.
\newblock {\em Biometrika}, Vol.~70, No.~1, pp. 41--55, 1983.

\bibitem{reinsel2022multivariate}
Gregory~C Reinsel, Raja~P Velu, and Kun Chen.
\newblock {\em Multivariate reduced-rank regression: theory, methods and applications}, Vol. 225.
\newblock Springer Nature, 2022.

\bibitem{hunter2004tutorial}
David~R Hunter and Kenneth Lange.
\newblock A tutorial on mm algorithms.
\newblock {\em The American Statistician}, Vol.~58, No.~1, pp. 30--37, 2004.

\bibitem{de2024new}
Mark de~Rooij.
\newblock A new algorithm and a discussion about visualization for logistic reduced rank regression.
\newblock {\em Behaviormetrika}, Vol.~51, No.~1, pp. 389--410, 2024.

\bibitem{adachi2016matrix}
Kohei Adachi.
\newblock {\em Matrix-based introduction to multivariate data analysis}.
\newblock Springer, 2016.

\bibitem{tian2014simple}
Lu~Tian, Ash~A Alizadeh, Andrew~J Gentles, and Robert Tibshirani.
\newblock A simple method for estimating interactions between a treatment and a large number of covariates.
\newblock {\em Journal of the American Statistical Association}, Vol. 109, No. 508, pp. 1517--1532, 2014.

\bibitem{chen2017general}
Shuai Chen, Lu~Tian, Tianxi Cai, and Menggang Yu.
\newblock A general statistical framework for subgroup identification and comparative treatment scoring.
\newblock {\em Biometrics}, Vol.~73, No.~4, pp. 1199--1209, 2017.

\bibitem{juraska2022package}
Michal Juraska and Maintainer~Michal Juraska.
\newblock Package ‘speff2trial’.
\newblock 2022.

\bibitem{kobak2021sparse}
Dmitry Kobak, Yves Bernaerts, Marissa~A Weis, Federico Scala, Andreas~S Tolias, and Philipp Berens.
\newblock Sparse reduced-rank regression for exploratory visualisation of paired multivariate data.
\newblock {\em Journal of the Royal Statistical Society Series C: Applied Statistics}, Vol.~70, No.~4, pp. 980--1000, 2021.

\bibitem{park2024low}
Seyoung Park, Eun~Ryung Lee, and Hongyu Zhao.
\newblock Low-rank regression models for multiple binary responses and their applications to cancer cell-line encyclopedia data.
\newblock {\em Journal of the American Statistical Association}, Vol. 119, No. 545, pp. 202--216, 2024.

\bibitem{anderson2002specification}
TW~Anderson.
\newblock Specification and misspecification in reduced rank regression.
\newblock {\em Sankhy{\=a}: The Indian Journal of Statistics, Series A}, pp. 193--205, 2002.

\bibitem{bunea2011optimal}
Florentina Bunea, Yiyuan She, and Marten~H Wegkamp.
\newblock Optimal selection of reduced rank estimators of high-dimensional matrices.
\newblock 2011.

\bibitem{chen2013reduced}
Kun Chen, Hongbo Dong, and Kung-Sik Chan.
\newblock Reduced rank regression via adaptive nuclear norm penalization.
\newblock {\em Biometrika}, Vol. 100, No.~4, pp. 901--920, 2013.

\end{thebibliography}
\bibliographystyle{junsrt}

\newpage
\section{Appendix}
Here, we present figures illustrating the simulation results conducted in Section 4 of the main paper.
The figures include the results of numerical simulations assuming randomized controlled trials (RCTs) and observational studies, specifically boxplots related to mean squared errors (MSEs) and ROC curves related to the area under the cover (AUC).
For details on the simulation design, please refer to the main paper.
Additionally, figures for all scenarios, including those presented in the main paper, are provided here.
\begin{figure}[ht]
    \centering
    \begin{minipage}[b]{0.8\textwidth}
        \centering
        \begin{subfigure}[b]{\textwidth}
            \centering            \includegraphics[width=0.9\textwidth]{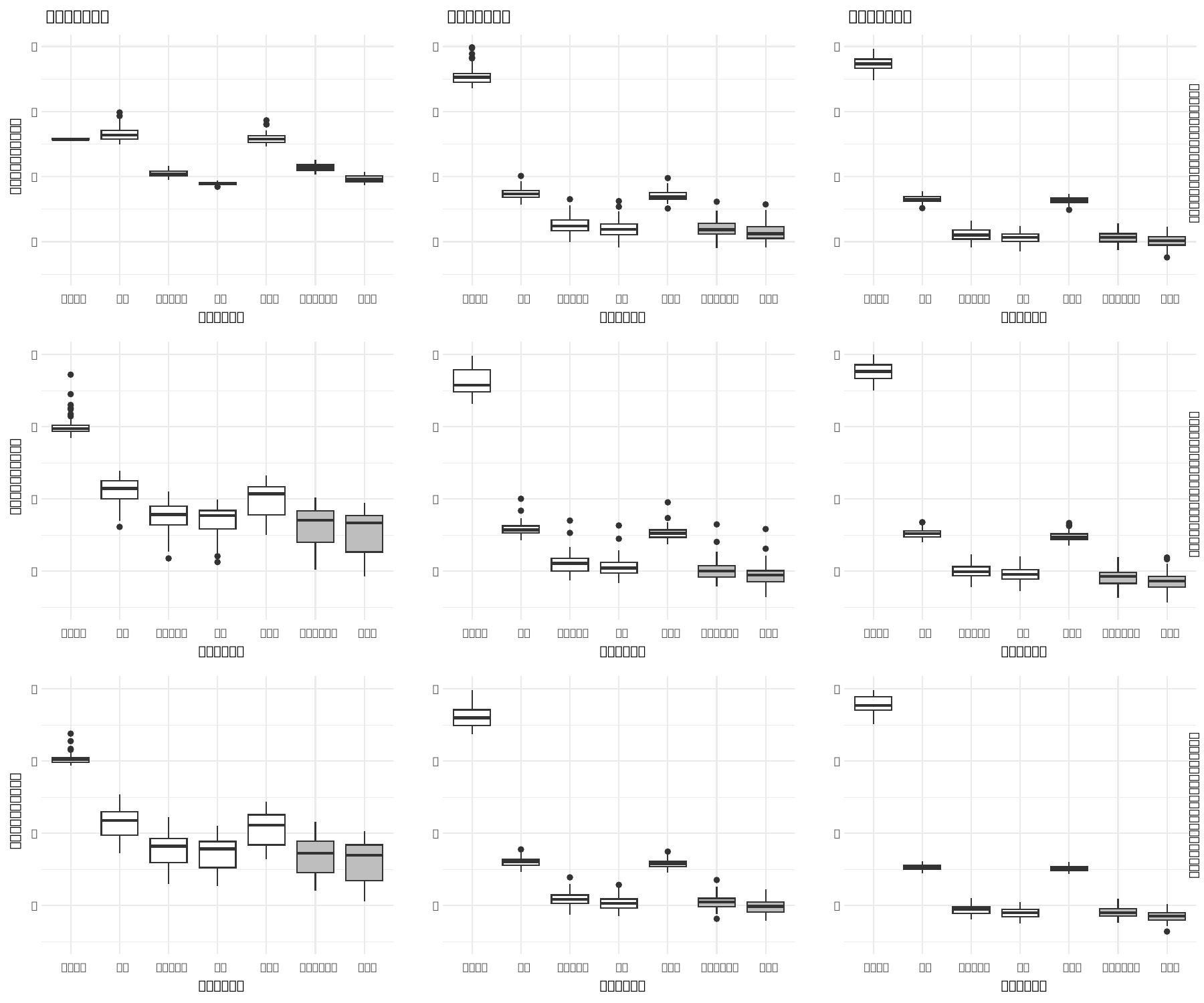}
            \caption{RCTs: Boxplot of MSE calculated for scenarios with 10 explanatory variables and errors that are uncorrelated, i.e., $\rho_1=0$ and $\rho_2=0$}
        \end{subfigure}
    \end{minipage}
    \vspace{0.5cm} 
    \begin{minipage}[b]{0.8\textwidth}
        \centering
        \begin{subfigure}[b]{\textwidth}
            \centering
            \includegraphics[width=0.9\textwidth]{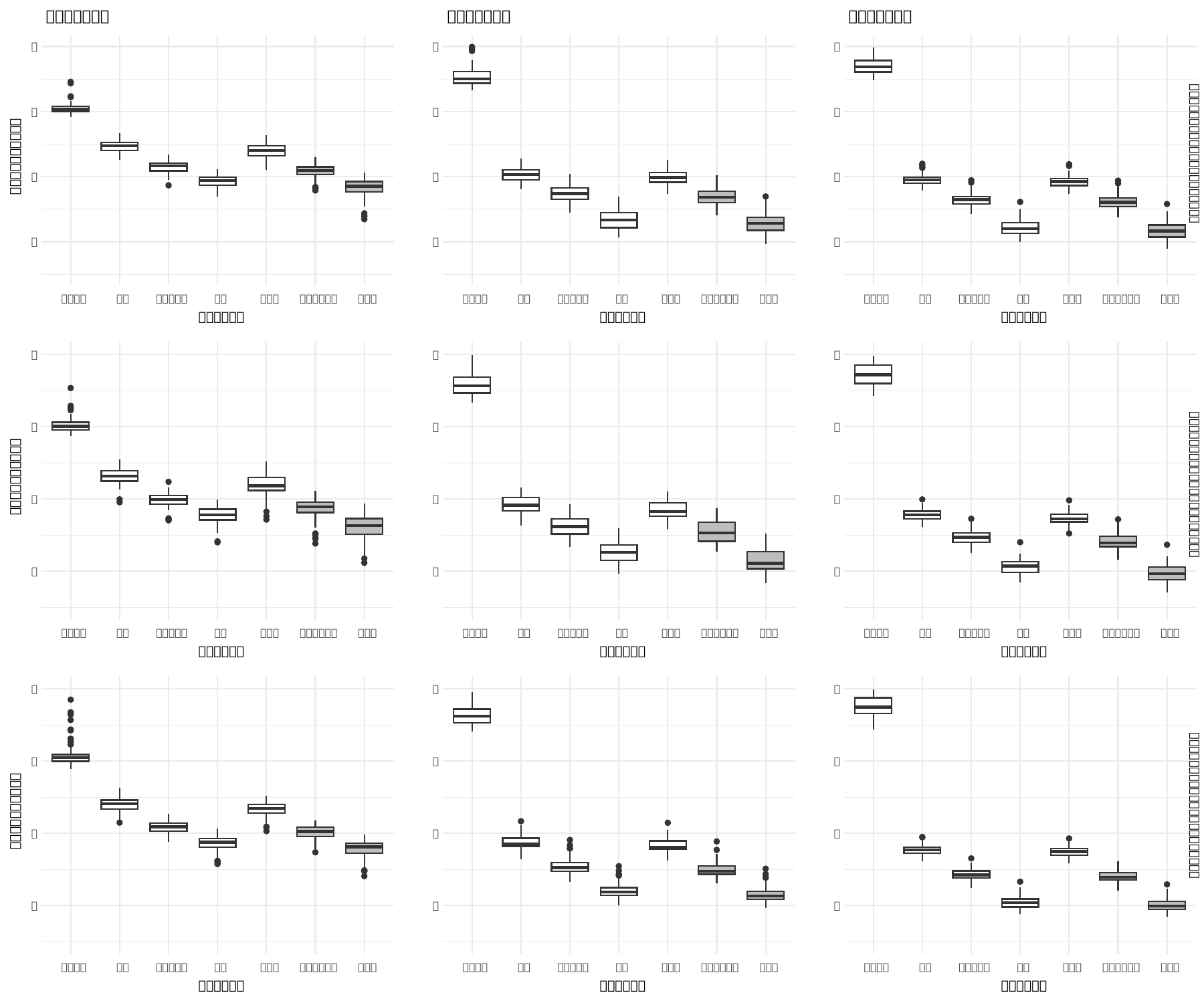}
            \caption{Observational Studies: Boxplot of MSE calculated for scenarios with 10 explanatory variables and errors that are uncorrelated, i.e., $\rho_1=0$ and $\rho_2=0$}
        \end{subfigure}
    \end{minipage}
\end{figure}
\begin{figure}[ht]
    \centering
    \begin{minipage}[b]{0.8\textwidth}
        \centering
        \begin{subfigure}[b]{\textwidth}
            \centering
            \includegraphics[width=0.9\textwidth]{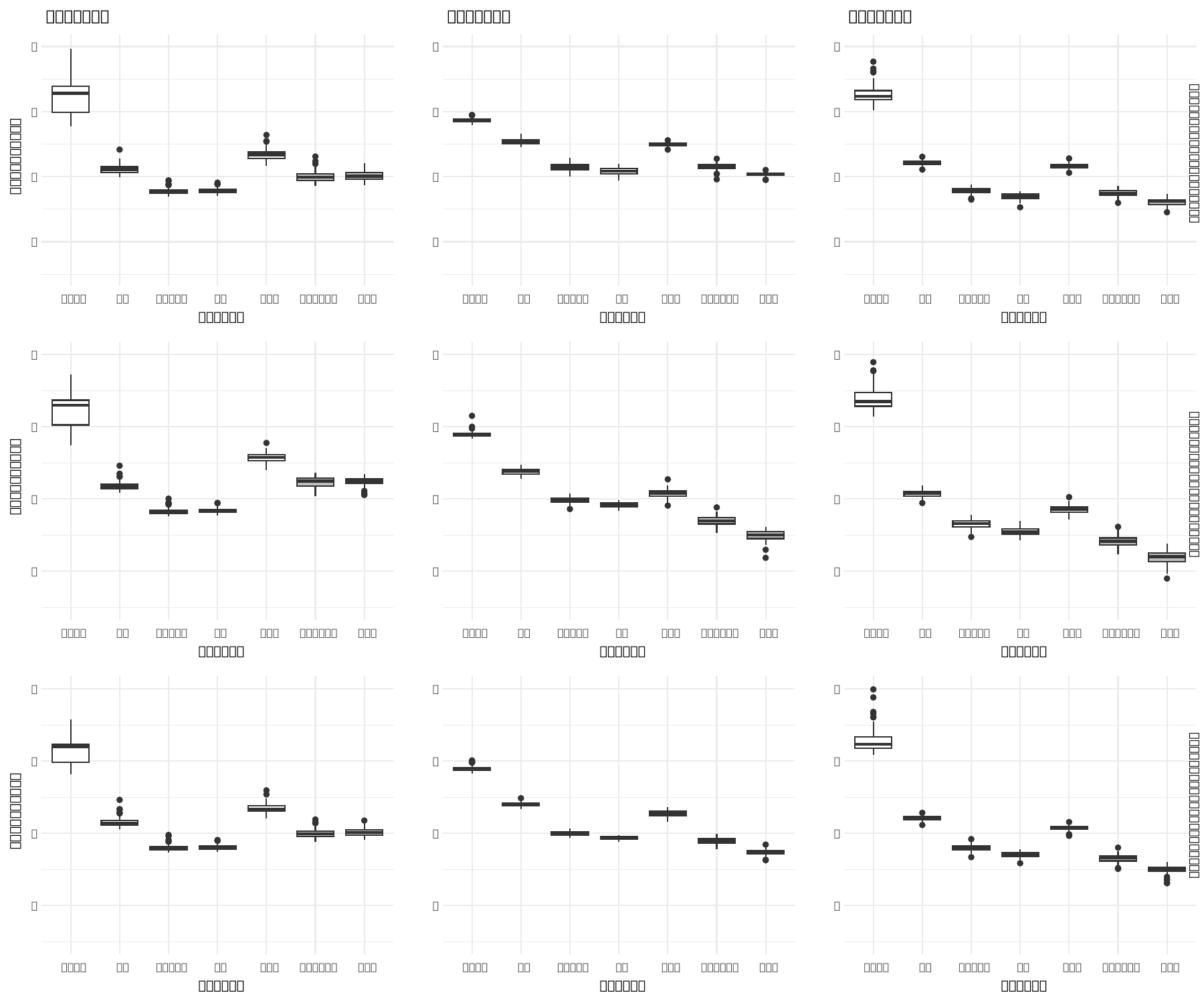}
            \caption{RCTs: Boxplot of MSE calculated for scenarios with 50 explanatory variables and errors that are uncorrelated, i.e., $\rho_1=0$ and $\rho_2=0$}
        \end{subfigure}
    \end{minipage}
    \vspace{0.5cm} 
    \begin{minipage}[b]{0.8\textwidth}
        \centering
        \begin{subfigure}[b]{\textwidth}
            \centering
            \includegraphics[width=0.9\textwidth]{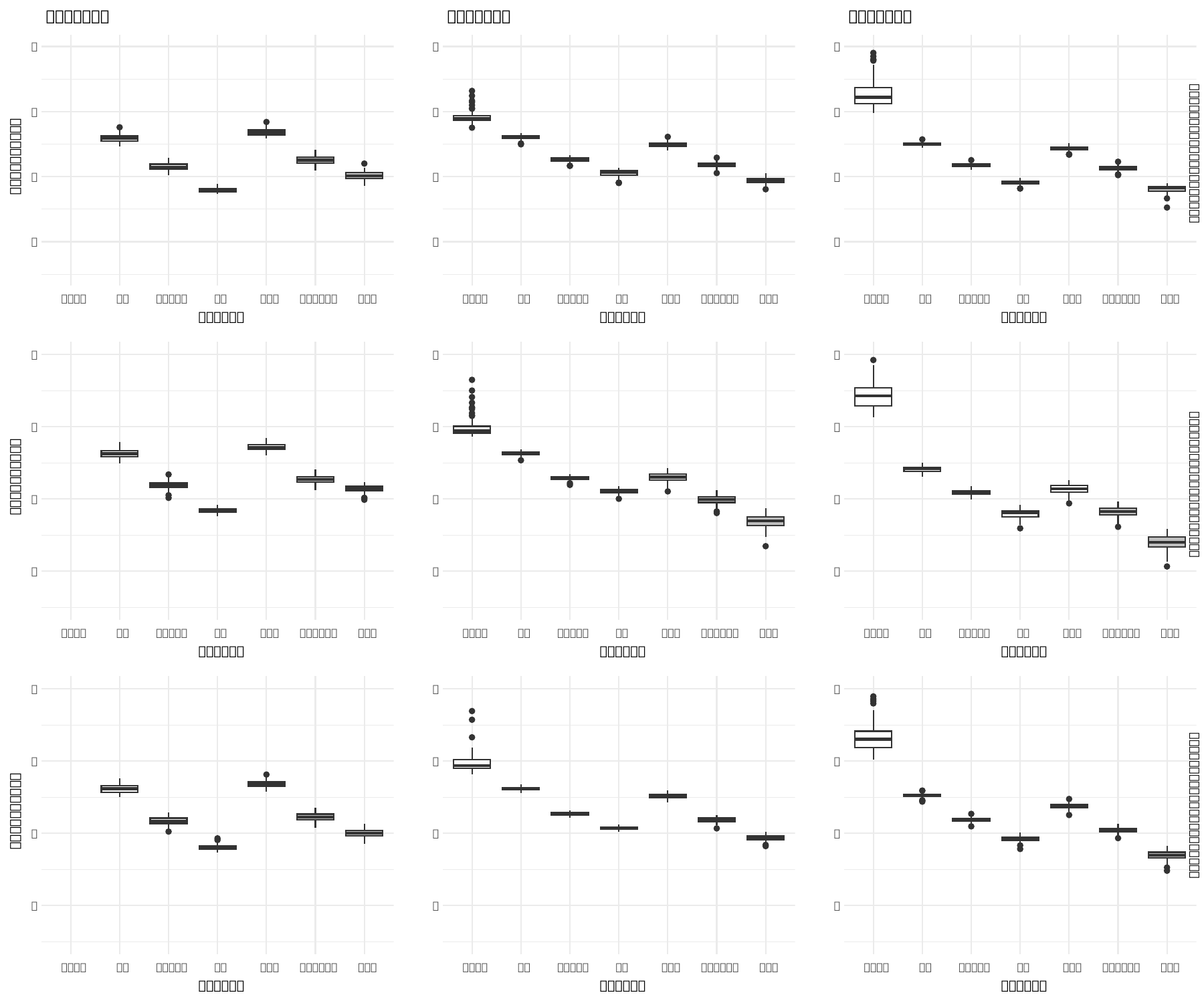}
            \caption{Observational Studies: Boxplot of MSE calculated for scenarios with 50 explanatory variables and errors that are uncorrelated, i.e., $\rho_1=0$ and $\rho_2=0$}
        \end{subfigure}
    \end{minipage}
\end{figure}
\begin{figure}[ht]
    \centering
    \begin{minipage}[b]{0.8\textwidth}
        \centering
        \begin{subfigure}[b]{\textwidth}
            \centering
            \includegraphics[width=0.9\textwidth]{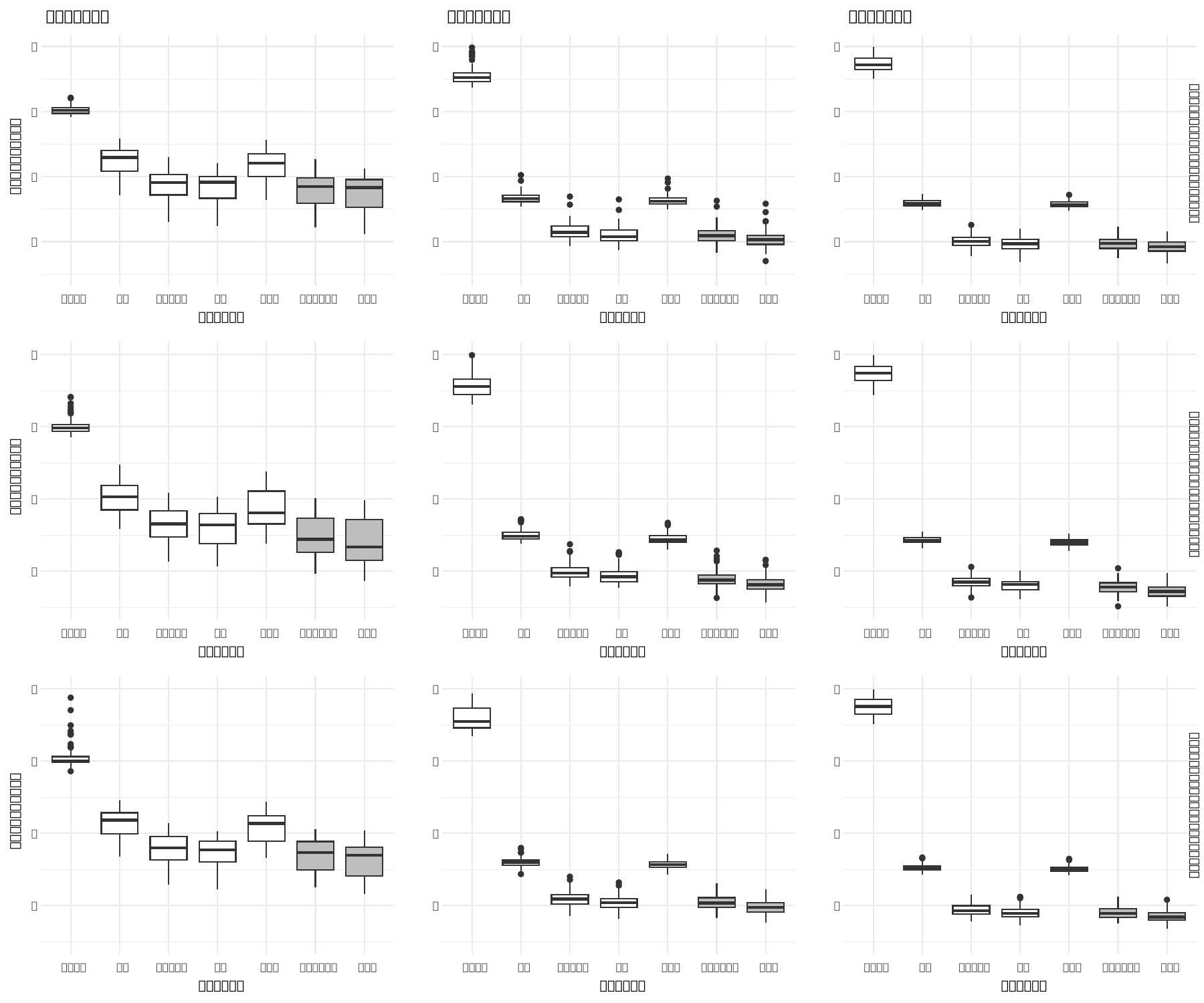}
            \caption{RCTs: Boxplot of MSE calculated for scenarios with 10 explanatory variables that are moderately correlated and errors that are uncorrelated, i.e., $\rho_1=1/3$ and $\rho_2=0$}
        \end{subfigure}
    \end{minipage}
    \vspace{0.5cm} 
    \begin{minipage}[b]{0.8\textwidth}
        \centering
        \begin{subfigure}[b]{\textwidth}
            \centering
            \includegraphics[width=0.9\textwidth]{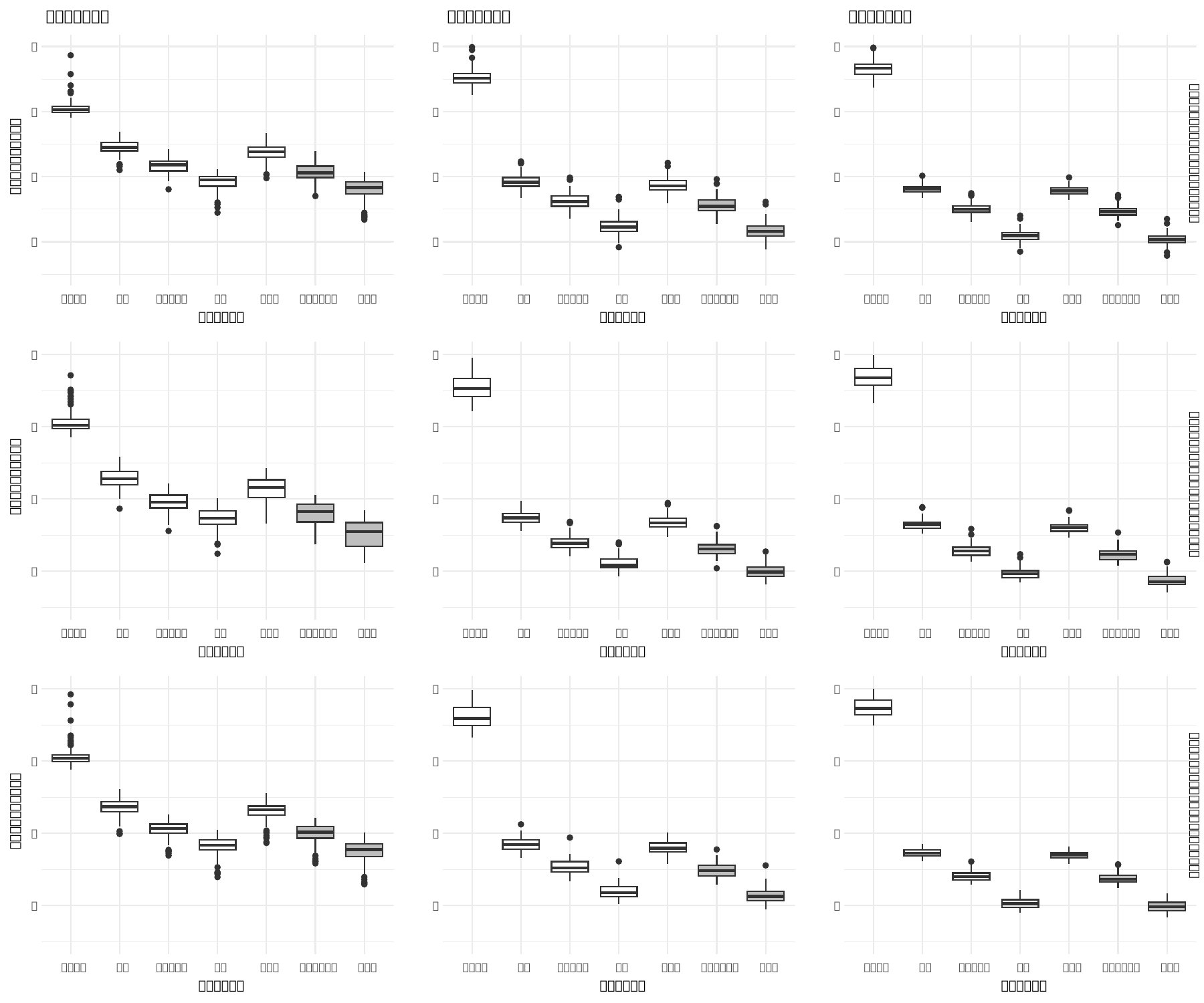}
            \caption{Observational Studies: Boxplot of MSE calculated for scenarios with 10 explanatory variables that are moderately correlated and errors that are uncorrelated, i.e., $\rho_1=1/3$ and $\rho_2=0$}
        \end{subfigure}
    \end{minipage}
\end{figure}
\begin{figure}[ht]
    \centering
    \begin{minipage}[b]{0.8\textwidth}
        \centering
        \begin{subfigure}[b]{\textwidth}
            \centering
            \includegraphics[width=0.9\textwidth]{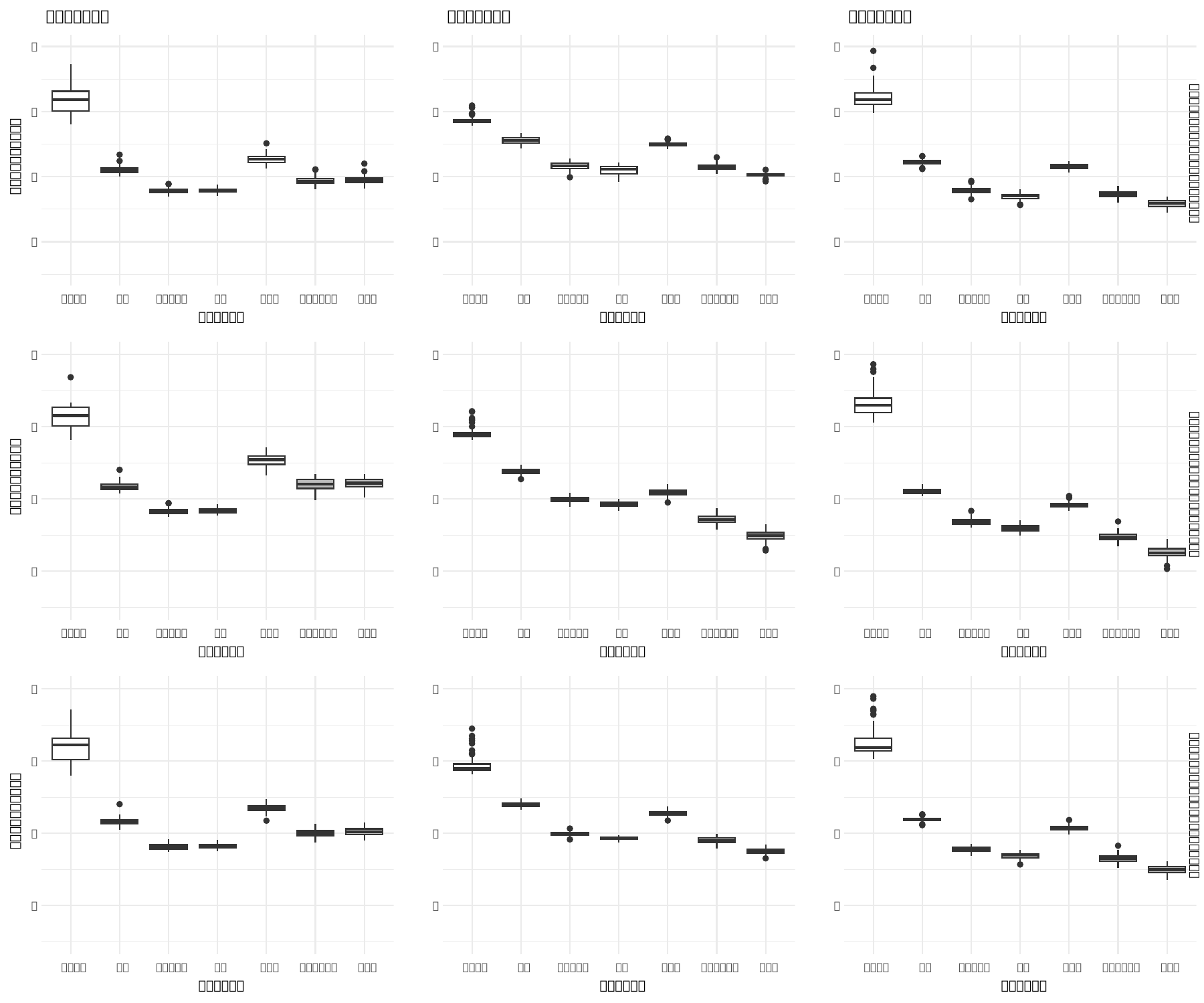}
            \caption{RCTs: Boxplot of MSE calculated for scenarios with 50 explanatory variables that are moderately correlated and errors that are uncorrelated, i.e., $\rho_1=1/3$ and $\rho_2=0$}
        \end{subfigure}
    \end{minipage}
    \vspace{0.5cm} 
    \begin{minipage}[b]{0.8\textwidth}
        \centering
        \begin{subfigure}[b]{\textwidth}
            \centering
            \includegraphics[width=0.9\textwidth]{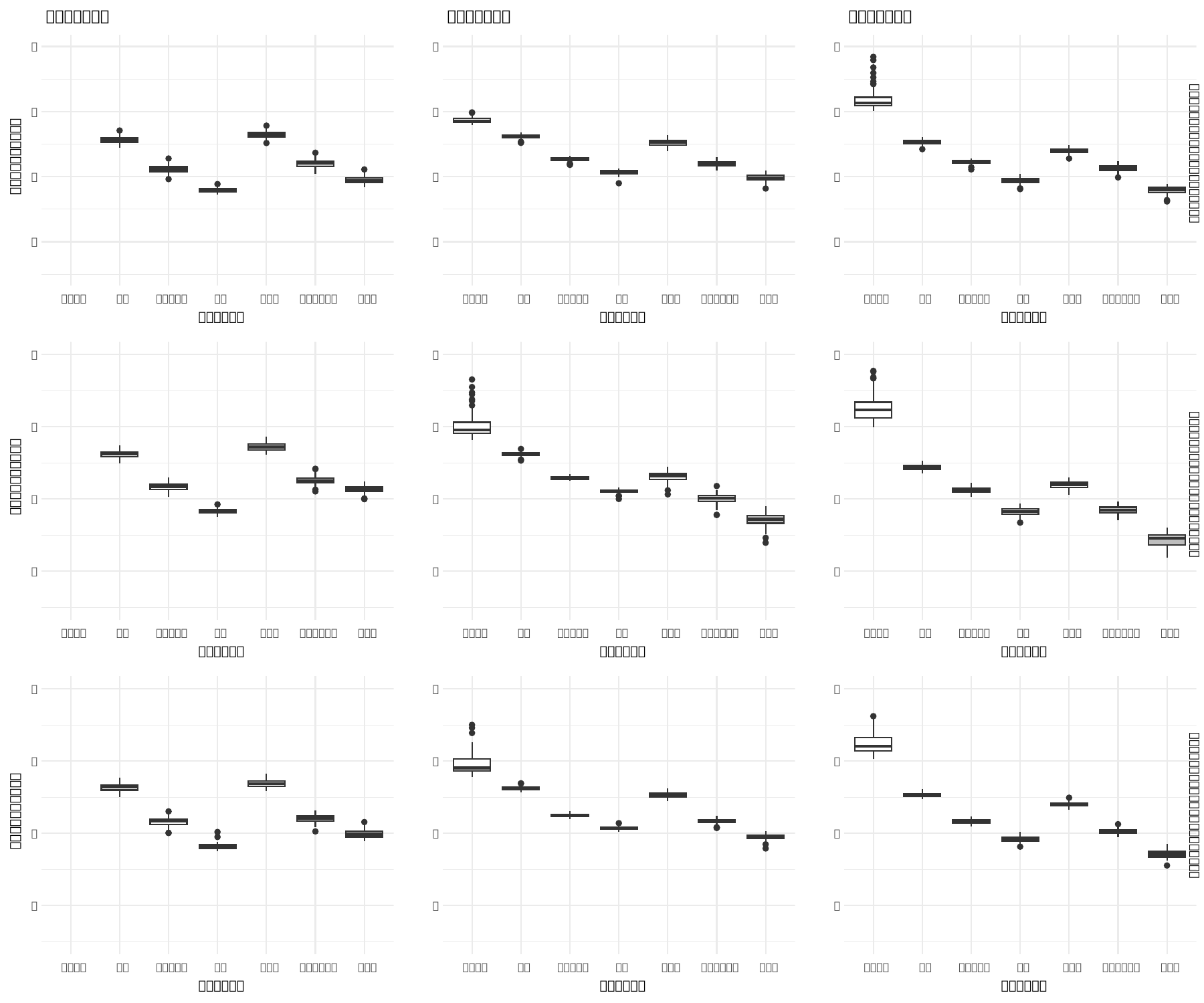}
            \caption{Observational Studies: Boxplot of MSE calculated for scenarios with 50 explanatory variables that are moderately correlated and errors that are uncorrelated, i.e., $\rho_1=1/3$ and $\rho_2=0$}
        \end{subfigure}
    \end{minipage}
\end{figure}
\begin{figure}[ht]
    \centering
    \begin{minipage}[b]{0.8\textwidth}
        \centering
        \begin{subfigure}[b]{\textwidth}
            \centering
            \includegraphics[width=0.9\textwidth]{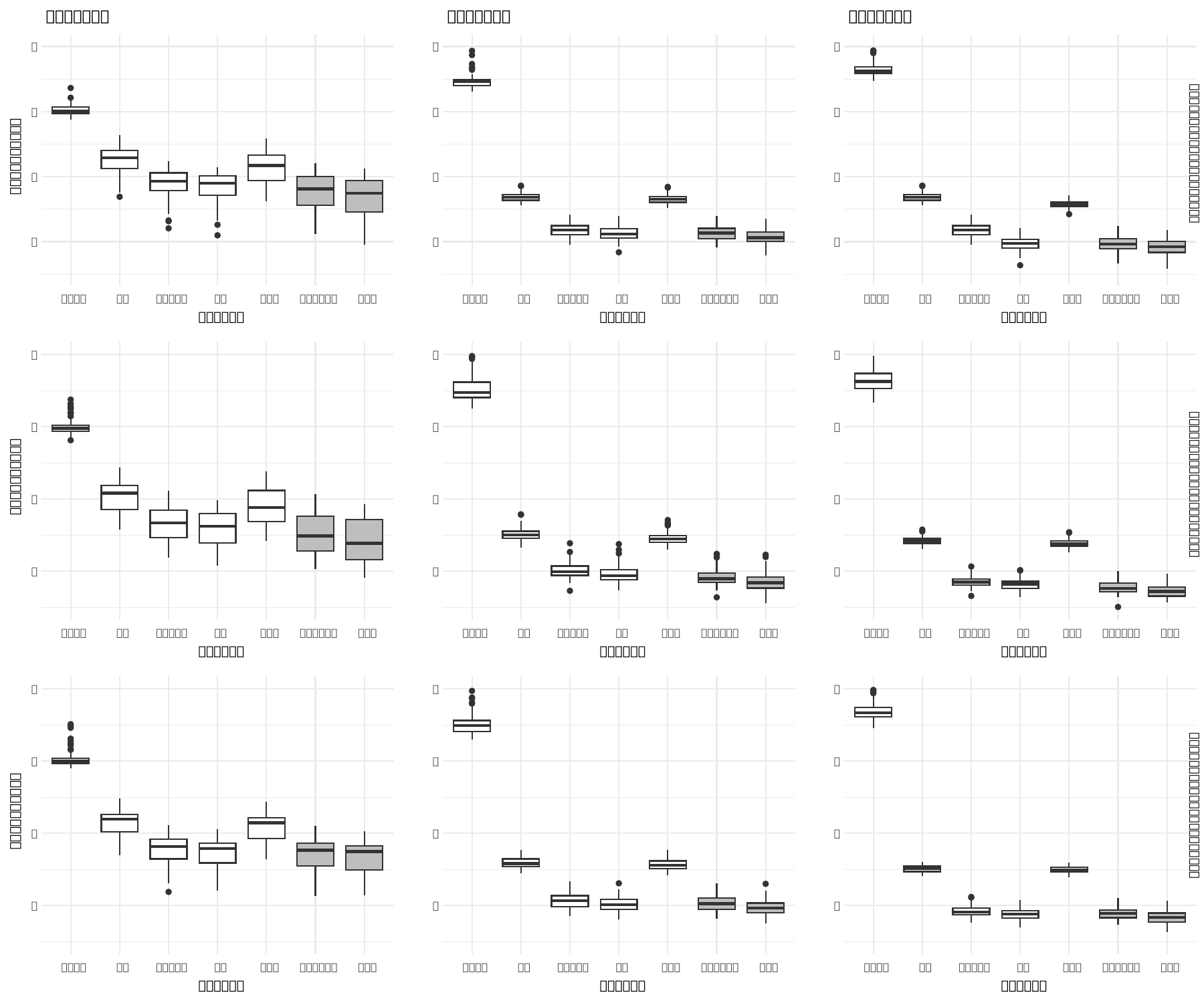}
            \caption{RCTs: Boxplot of MSE calculated for scenarios with 10 explanatory variables that are highly correlated and errors that are uncorrelated, i.e., $\rho_1=2/3$ and $\rho_2=0$}
        \end{subfigure}
    \end{minipage}
    \vspace{0.5cm} 
    \begin{minipage}[b]{0.8\textwidth}
        \centering
        \begin{subfigure}[b]{\textwidth}
            \centering
            \includegraphics[width=0.9\textwidth]{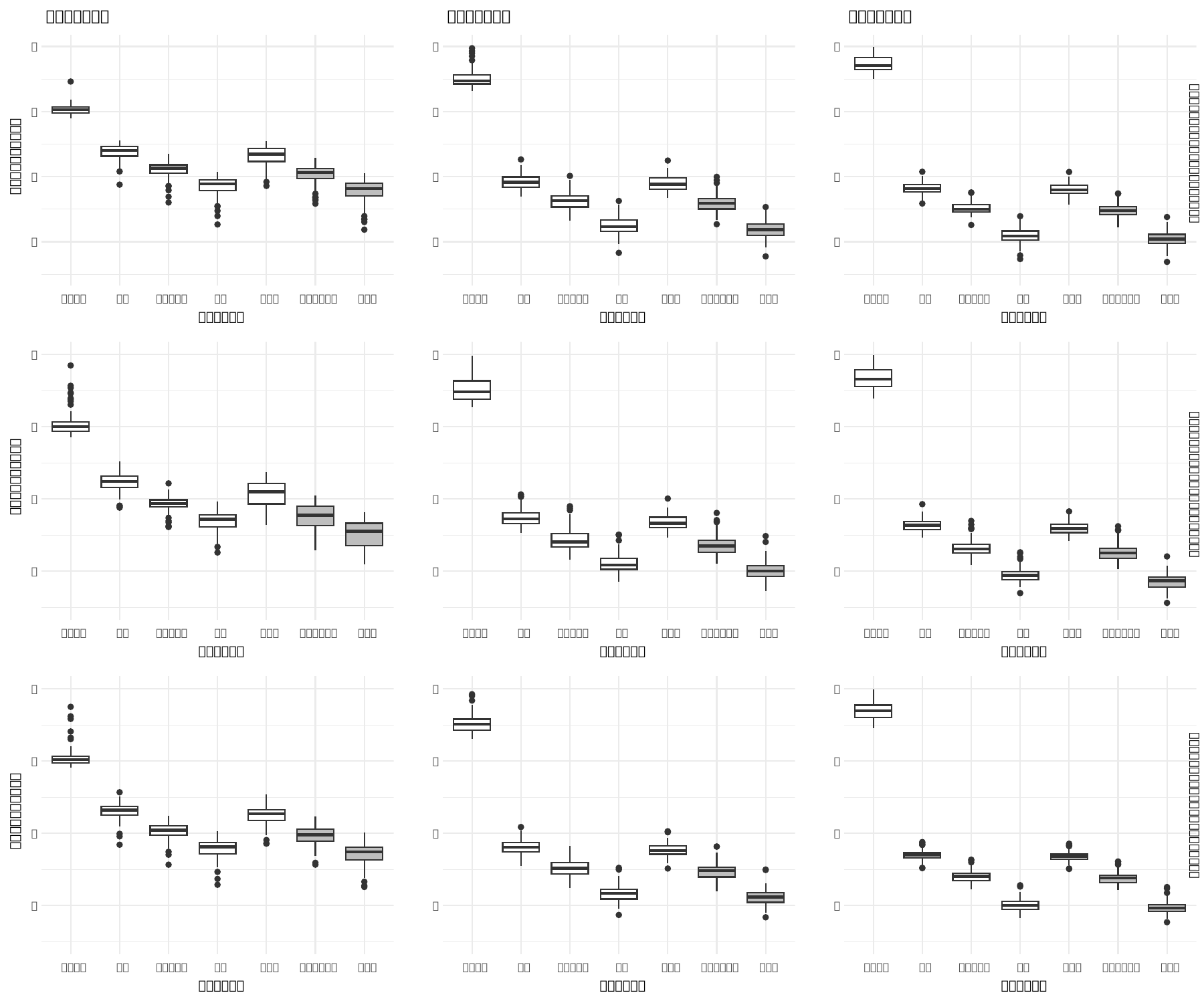}
            \caption{Observational Studies: Boxplot of MSE calculated for scenarios with 10 explanatory variables that are highly correlated and errors that are uncorrelated, i.e., $\rho_1=2/3$ and $\rho_2=0$}
        \end{subfigure}
    \end{minipage}
\end{figure}
\begin{figure}[ht]
    \centering
    \begin{minipage}[b]{0.8\textwidth}
        \centering
        \begin{subfigure}[b]{\textwidth}
            \centering
            \includegraphics[width=0.9\textwidth]{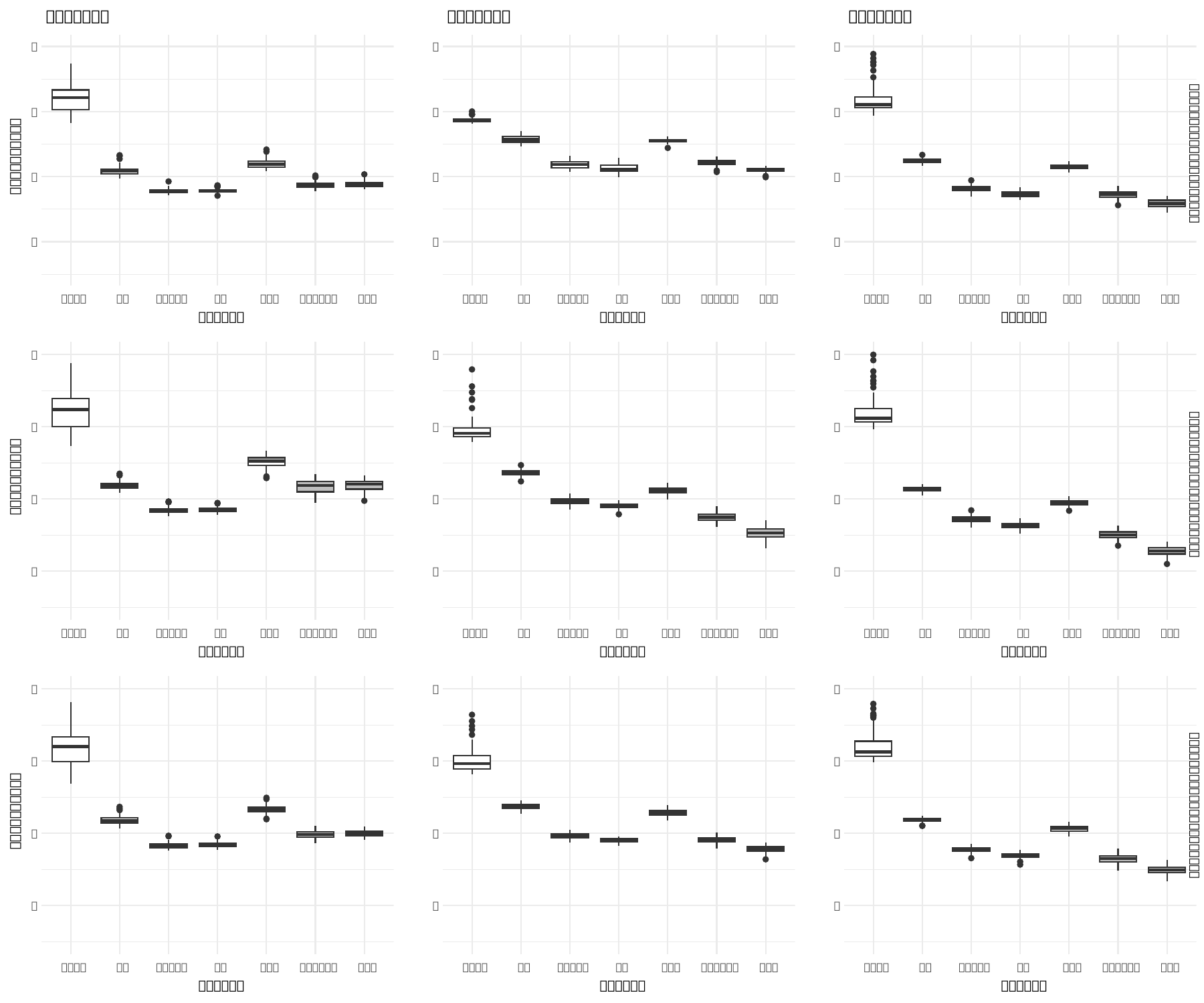}
            \caption{RCTs: Boxplot of MSE calculated for scenarios with 50 explanatory variables that are highly correlated and errors that are uncorrelated, i.e., $\rho_1=2/3$ and $\rho_2=0$}
        \end{subfigure}
    \end{minipage}
    \vspace{0.5cm} 
    \begin{minipage}[b]{0.8\textwidth}
        \centering
        \begin{subfigure}[b]{\textwidth}
            \centering
            \includegraphics[width=0.9\textwidth]{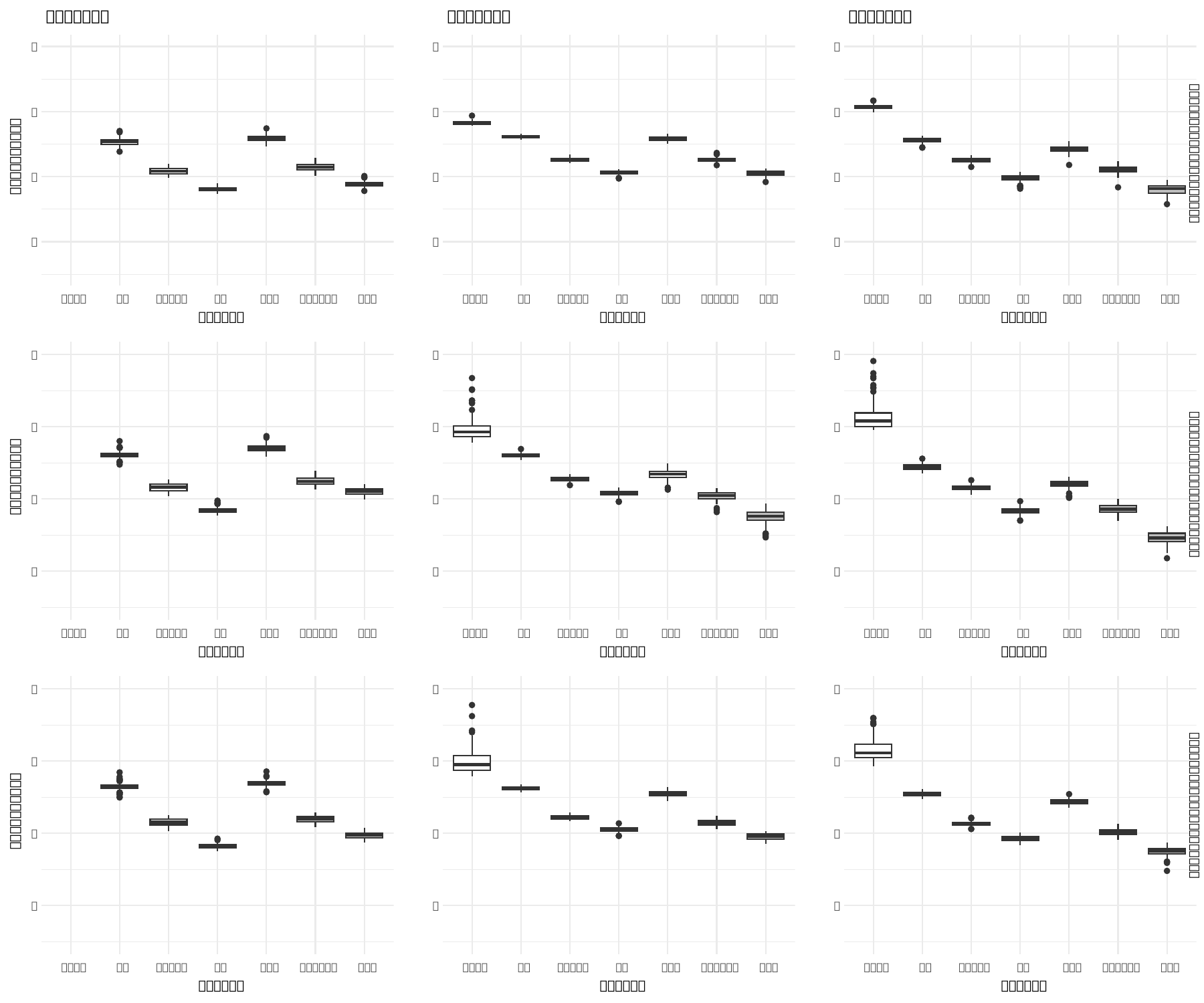}
            \caption{Observational Studies: Boxplot of MSE calculated for scenarios with 50 explanatory variables that are highly correlated and errors that are uncorrelated, i.e., $\rho_1=2/3$ and $\rho_2=0$}
        \end{subfigure}
    \end{minipage}
    
\end{figure}
\begin{figure}[ht]
    \centering
    \begin{minipage}[b]{0.8\textwidth}
        \centering
        \begin{subfigure}[b]{\textwidth}
            \centering
            \includegraphics[width=0.9\textwidth]{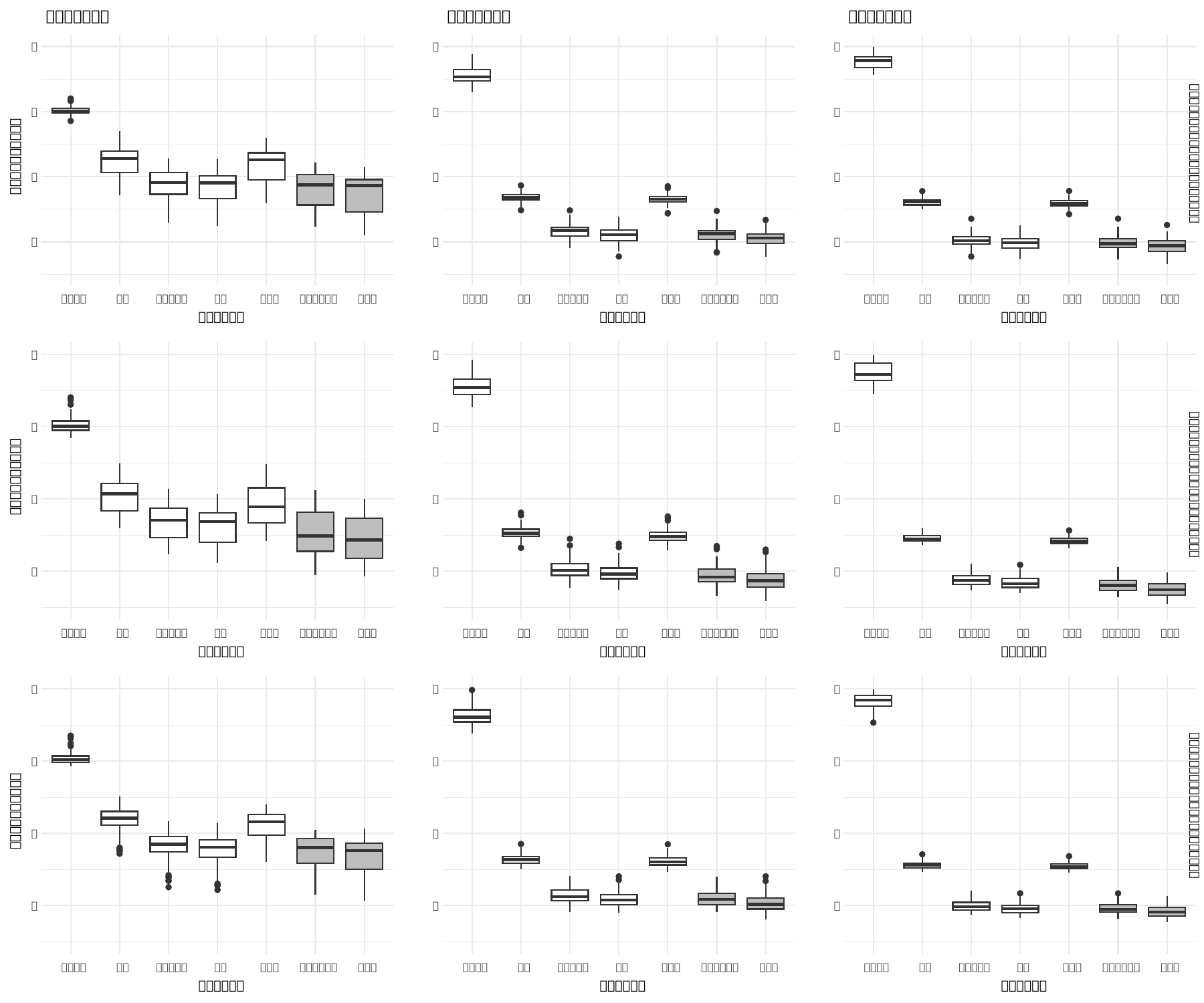}
            \caption{RCTs: Boxplot of MSE calculated for scenarios with 10 explanatory variables that are uncorrelated and errors that are moderately correlated, i.e., $\rho_1=0$ and $\rho_2=1/3$}
        \end{subfigure}
    \end{minipage}
    \vspace{0.5cm} 
    \begin{minipage}[b]{0.8\textwidth}
        \centering
        \begin{subfigure}[b]{\textwidth}
            \centering
            \includegraphics[width=0.9\textwidth]{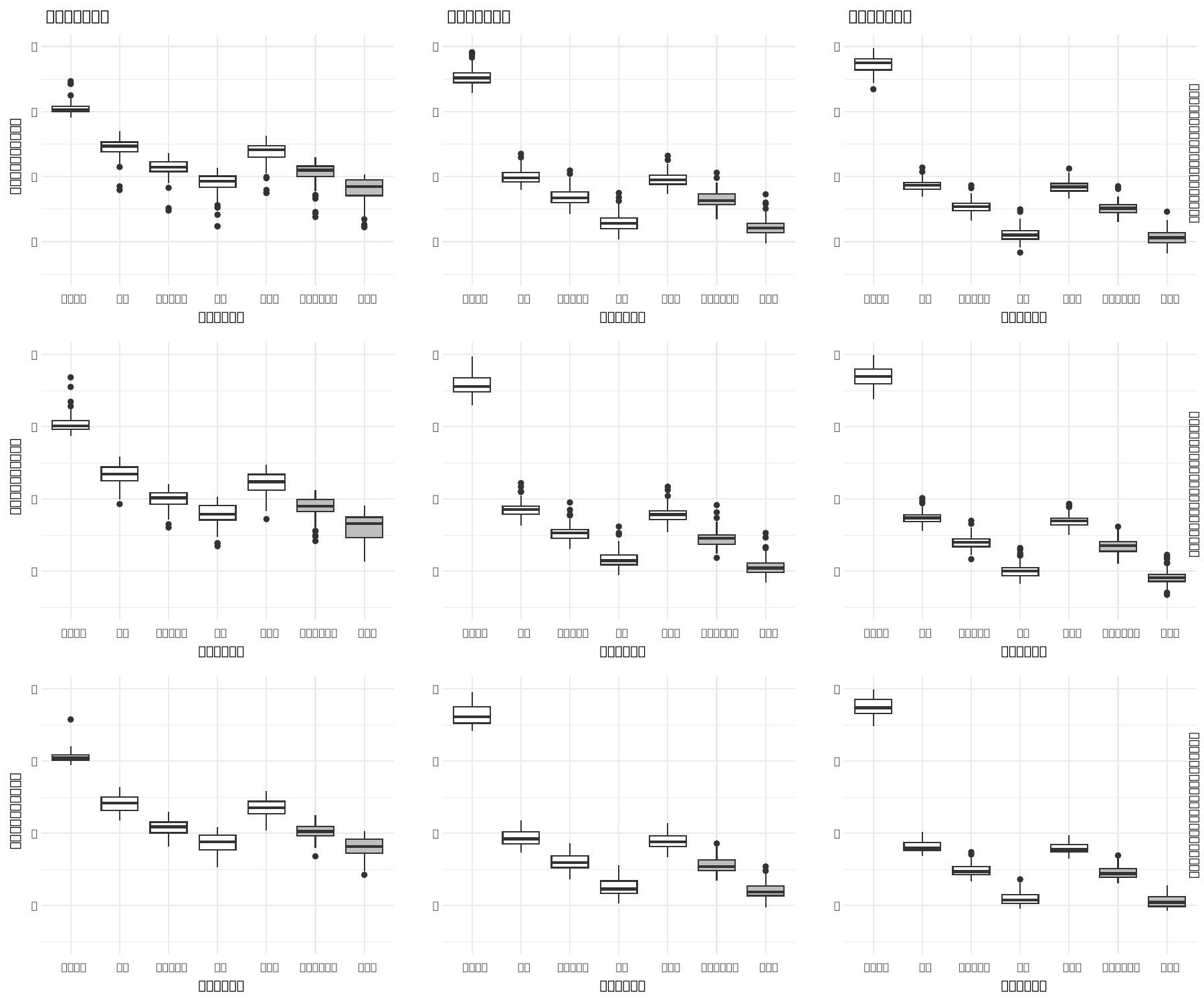}
            \caption{Observational Studies: Boxplot of MSE calculated for scenarios with 10 explanatory variables that are uncorrelated and errors that are moderately correlated, i.e., $\rho_1=0$ and $\rho_2=1/3$}
        \end{subfigure}
    \end{minipage}
\end{figure}
\begin{figure}[ht]
    \centering
    \begin{minipage}[b]{0.8\textwidth}
        \centering
        \begin{subfigure}[b]{\textwidth}
            \centering
            \includegraphics[width=0.9\textwidth]{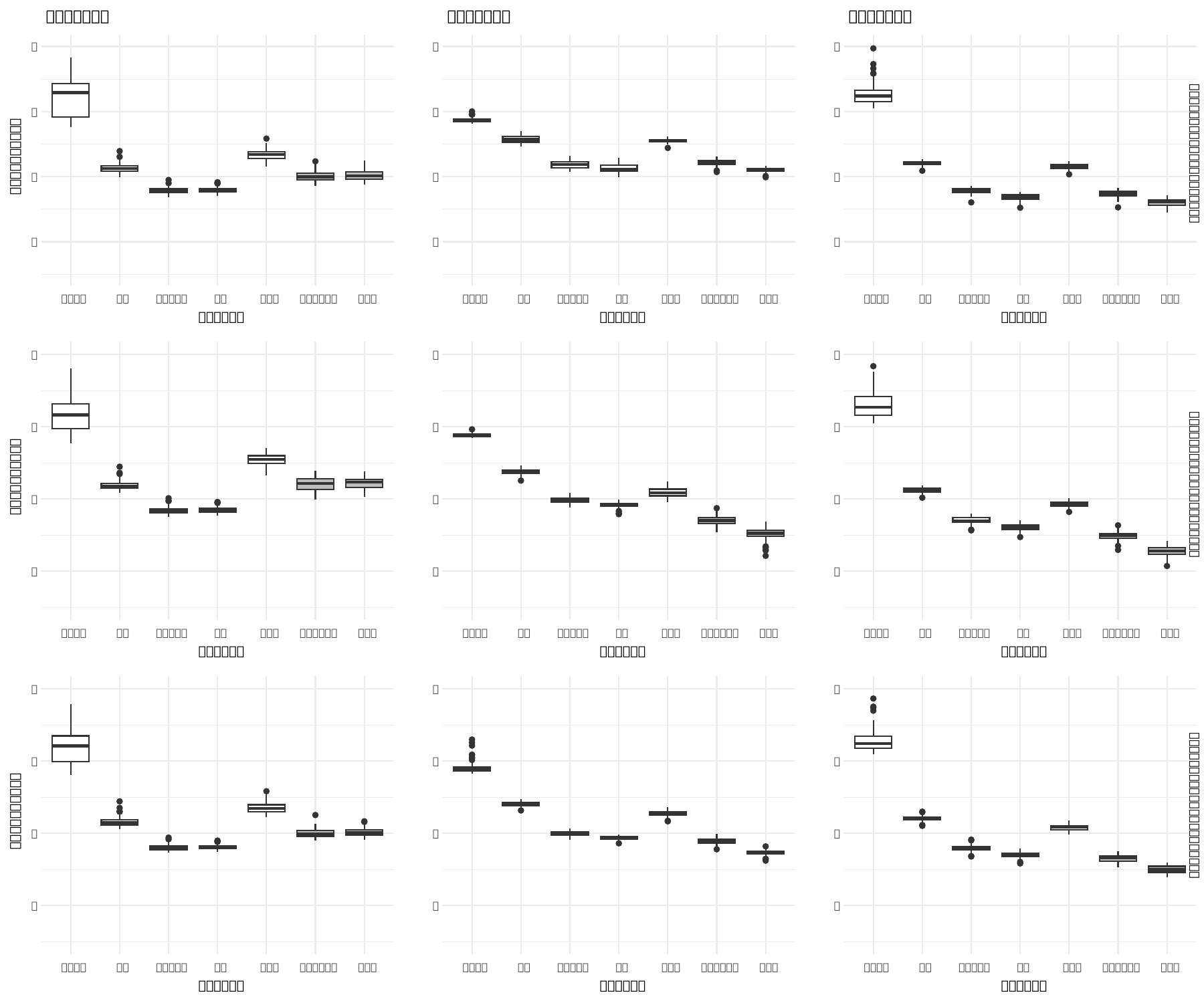}
            \caption{RCTs: Boxplot of MSE calculated for scenarios with 50 explanatory variables that are uncorrelated and errors that are moderately correlated, i.e., $\rho_1=0$ and $\rho_2=1/3$}
        \end{subfigure}
    \end{minipage}
    \vspace{0.5cm} 
    \begin{minipage}[b]{0.8\textwidth}
        \centering
        \begin{subfigure}[b]{\textwidth}
            \centering
            \includegraphics[width=0.9\textwidth]{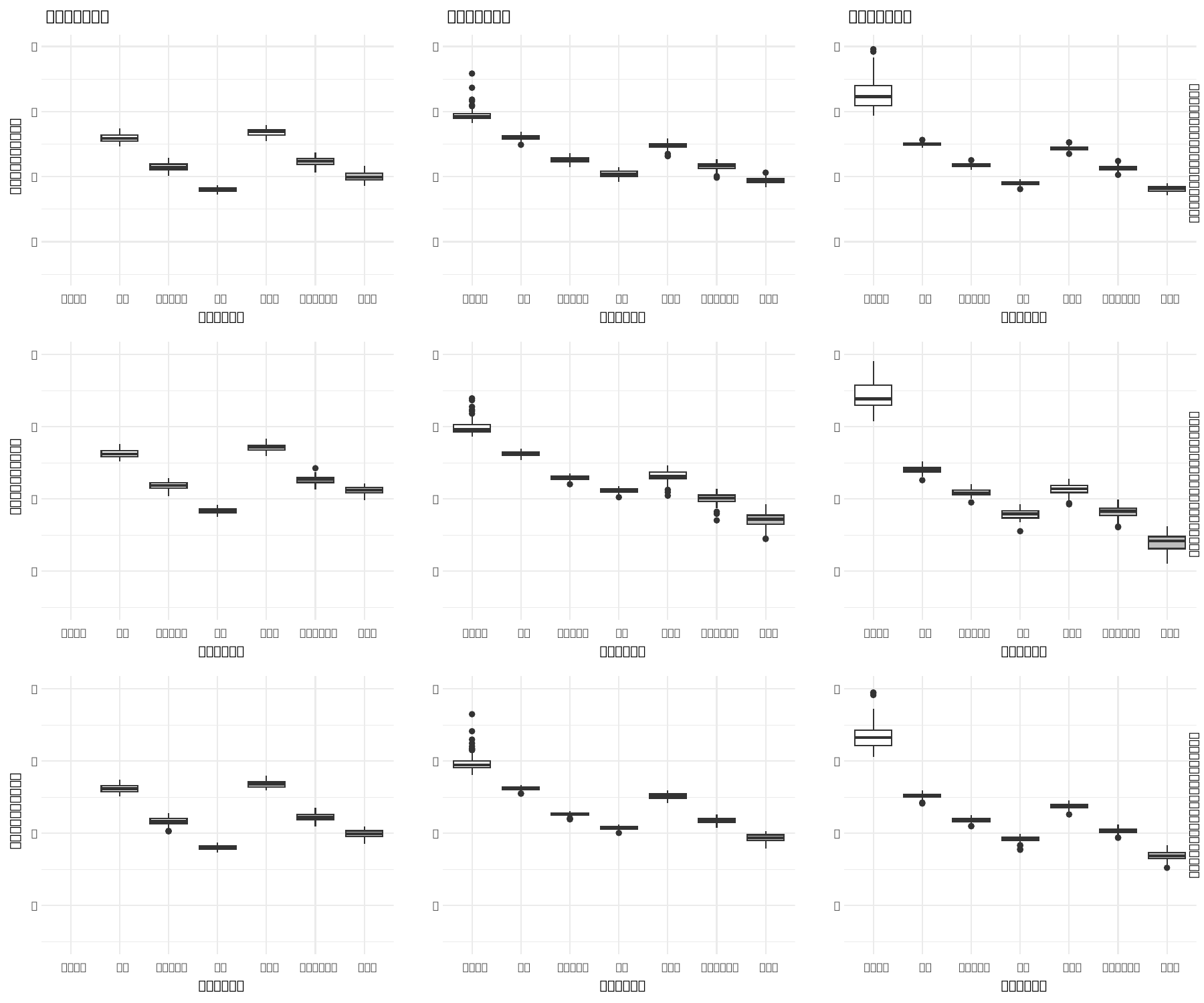}
            \caption{Observational Studies: Boxplot of MSE calculated for scenarios with 50 explanatory variables that are uncorrelated and errors that are moderately correlated, i.e., $\rho_1=0$ and $\rho_2=1/3$}
        \end{subfigure}
    \end{minipage}
    \end{figure}
\begin{figure}[ht]
    \centering
    \begin{minipage}[b]{0.8\textwidth}
        \centering
        \begin{subfigure}[b]{\textwidth}
            \centering
            \includegraphics[width=0.9\textwidth]{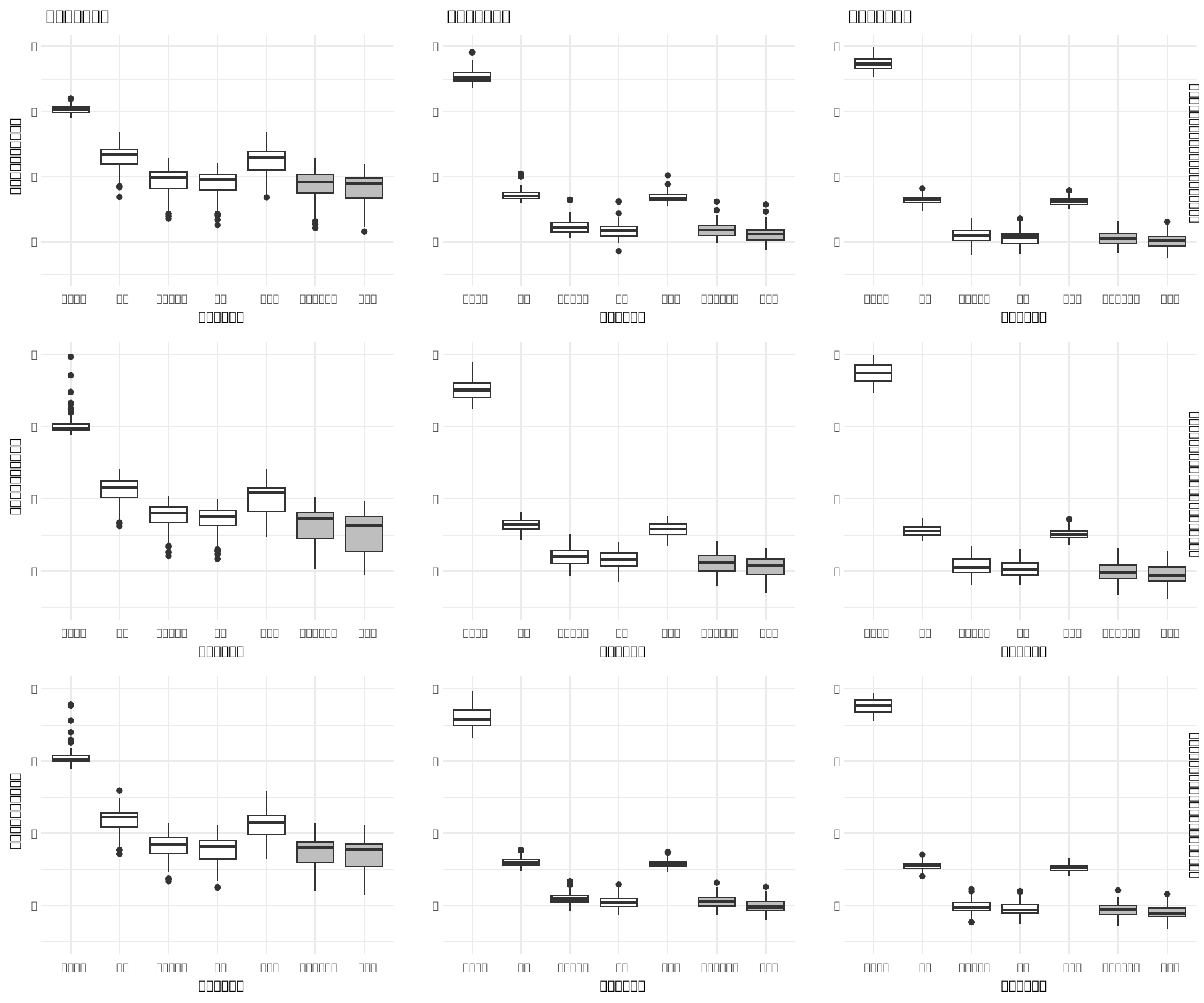}
            \caption{RCTs: Boxplot of MSE calculated for scenarios with 10 explanatory variables and errors that are moderately correlated, i.e., $\rho_1=1/3$ and $\rho_2=1/3$}
        \end{subfigure}
    \end{minipage}
    \vspace{0.5cm} 
    \begin{minipage}[b]{0.8\textwidth}
        \centering
        \begin{subfigure}[b]{\textwidth}
            \centering
            \includegraphics[width=0.9\textwidth]{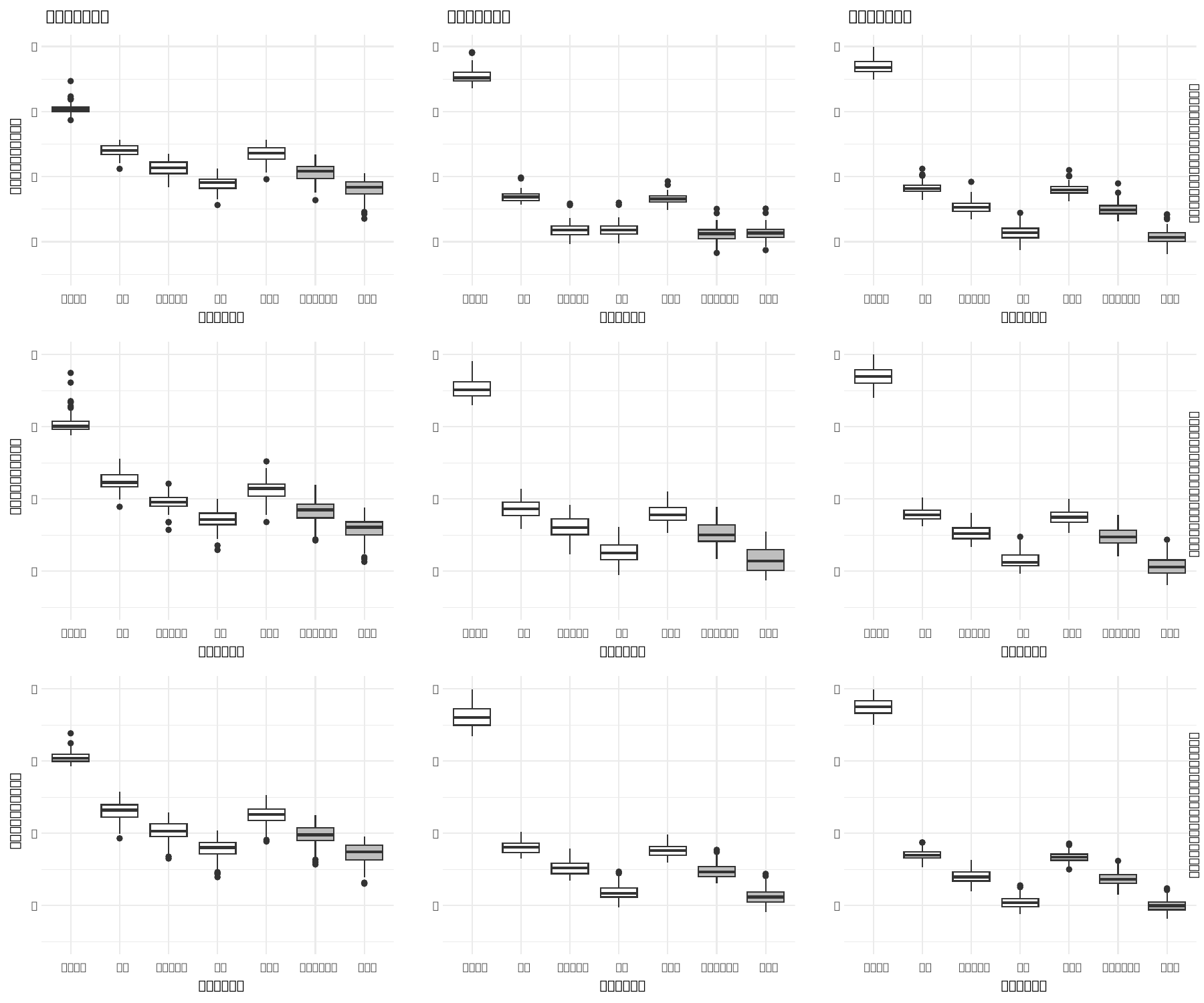}
            \caption{Observational Studies: Boxplot of MSE calculated for scenarios with 10 explanatory variables and errors that are moderately correlated, i.e., $\rho_1=1/3$ and $\rho_2=1/3$}
        \end{subfigure}
    \end{minipage}
\end{figure}
\begin{figure}[ht]
    \centering
    \begin{minipage}[b]{0.8\textwidth}
        \centering
        \begin{subfigure}[b]{\textwidth}
            \centering
            \includegraphics[width=0.9\textwidth]{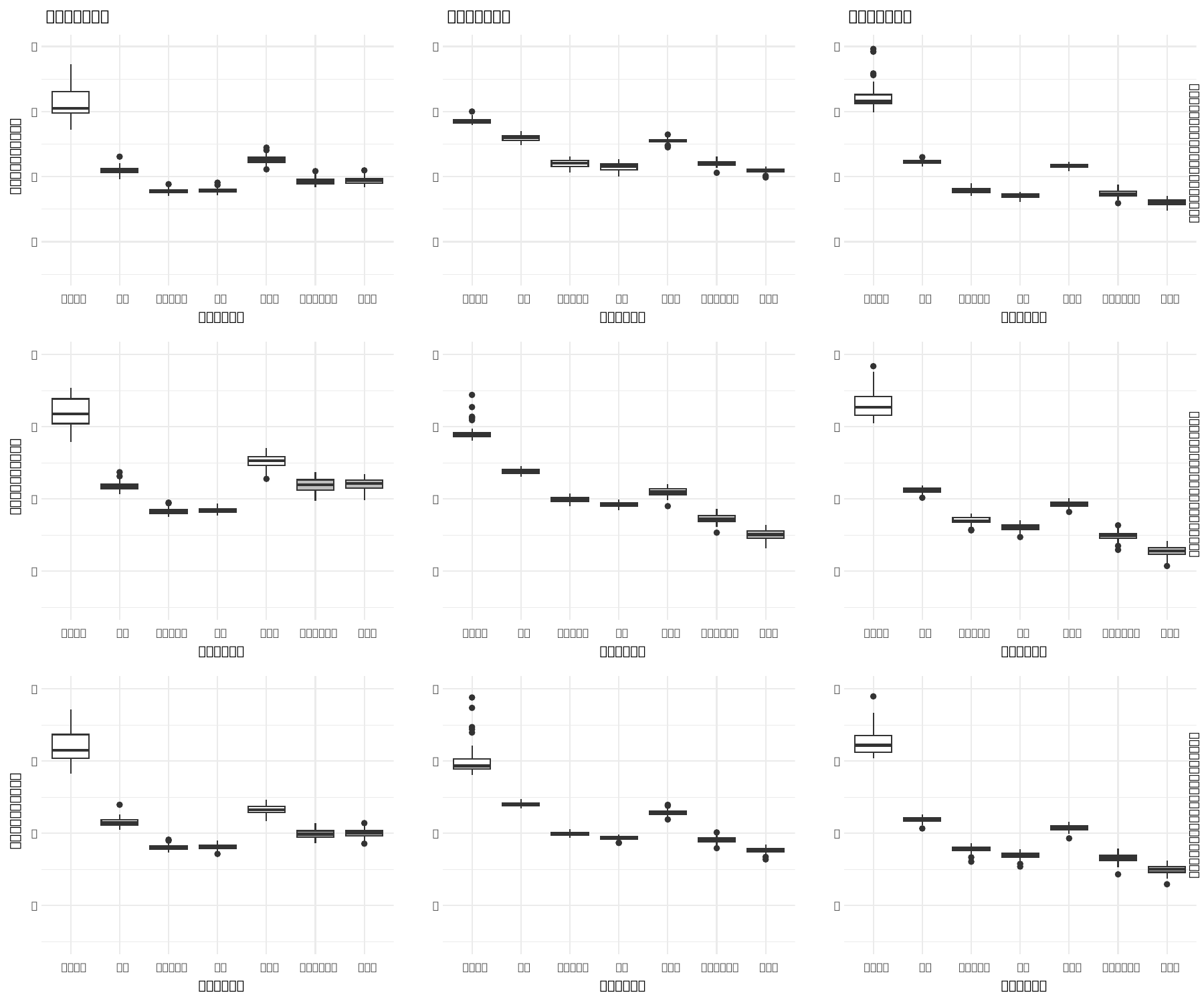}
            \caption{RCTs: Boxplot of MSE calculated for scenarios with 50 explanatory variables and errors that are moderately correlated, i.e., $\rho_1=1/3$ and $\rho_2=1/3$}
        \end{subfigure}
    \end{minipage}
    \vspace{0.5cm} 
    \begin{minipage}[b]{0.8\textwidth}
        \centering
        \begin{subfigure}[b]{\textwidth}
            \centering
            \includegraphics[width=0.9\textwidth]{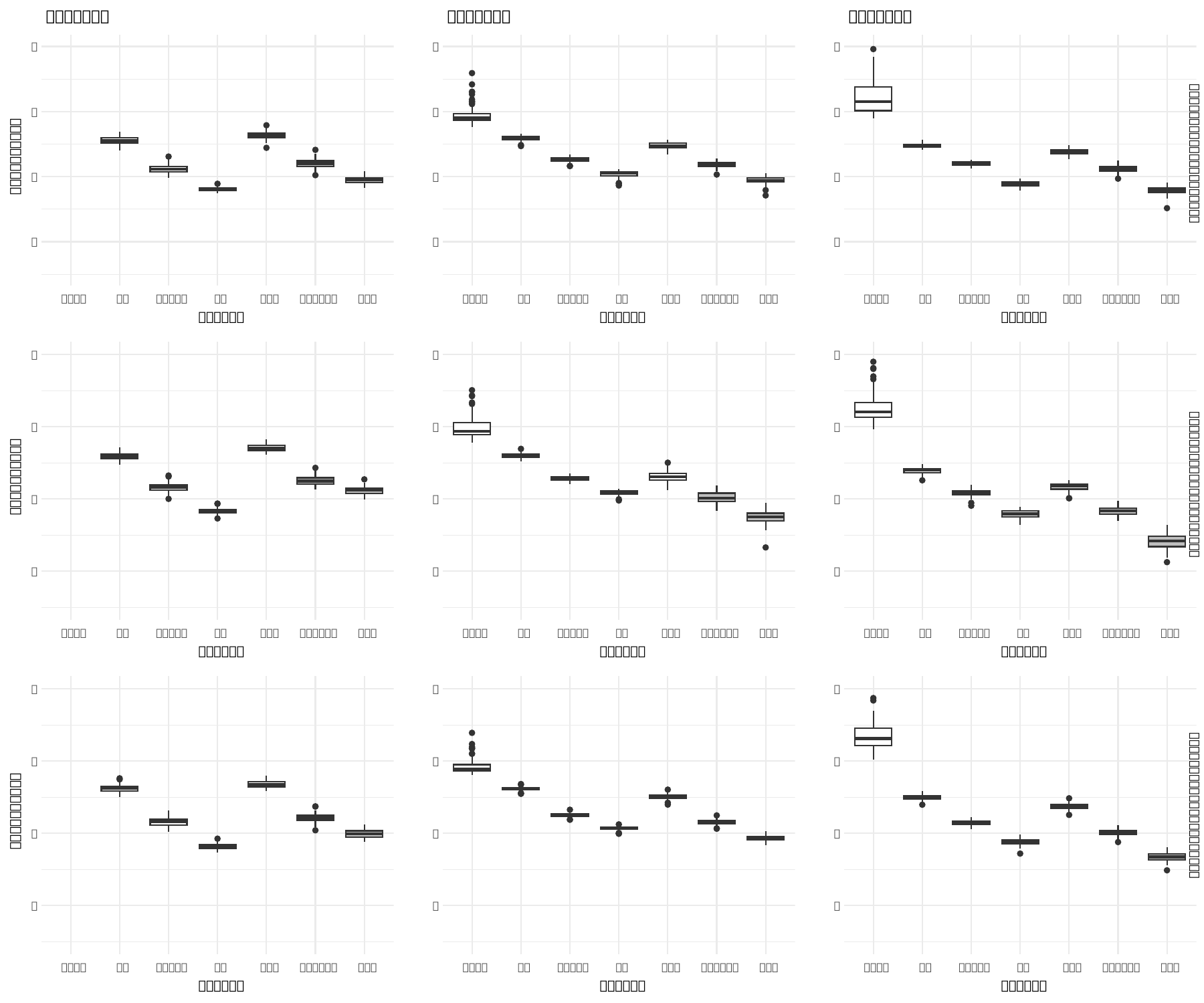}
            \caption{Observational Studies: Boxplot of MSE calculated for scenarios with 50 explanatory variables and errors that are moderately correlated, i.e., $\rho_1=1/3$ and $\rho_2=1/3$}
        \end{subfigure}
    \end{minipage}
    \end{figure}
\begin{figure}[ht]
    \centering
    \begin{minipage}[b]{0.8\textwidth}
        \centering
        \begin{subfigure}[b]{\textwidth}
            \centering
            \includegraphics[width=0.9\textwidth]{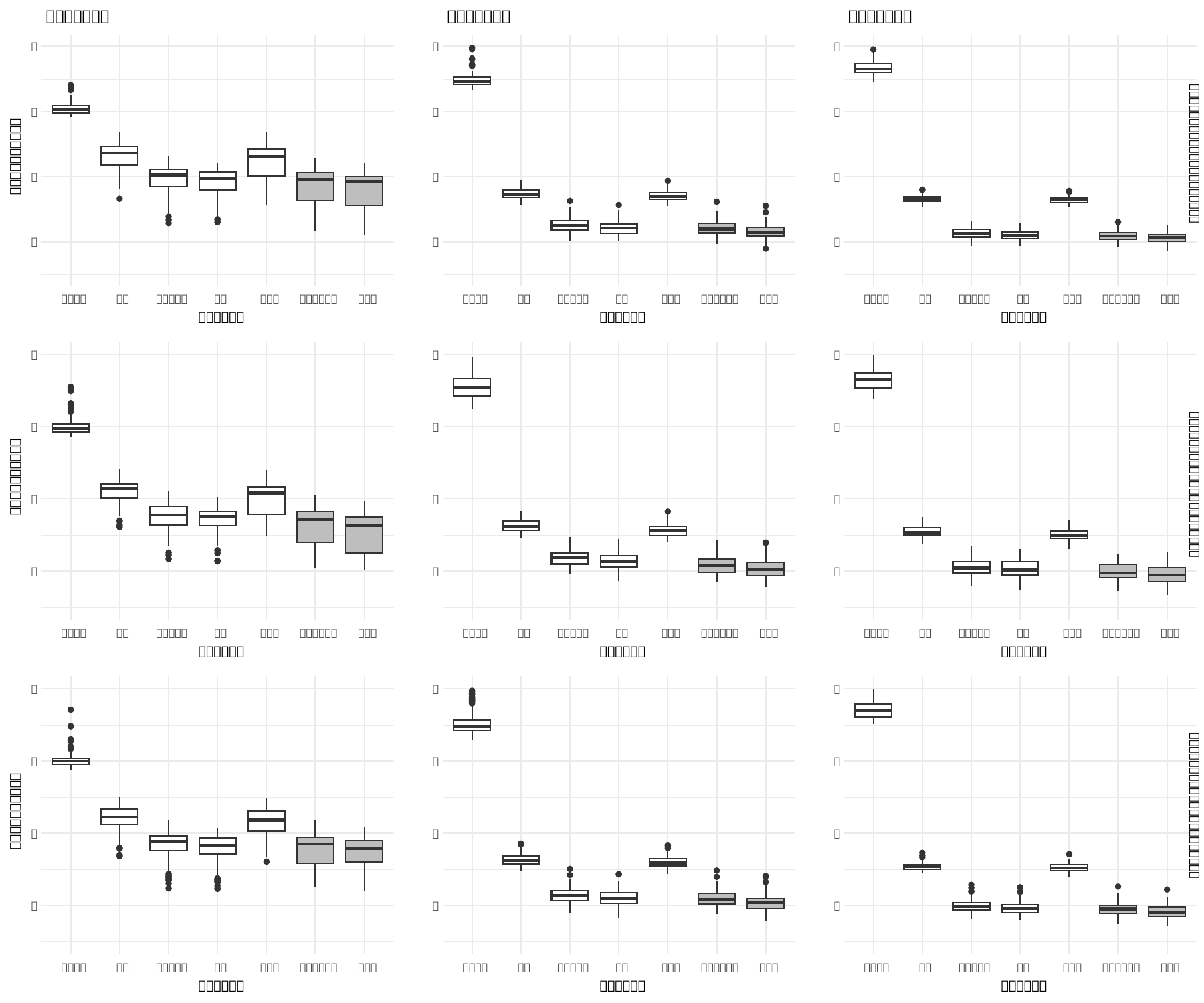}
            \caption{RCTs: Boxplot of MSE calculated for scenarios with 10 explanatory variables that are highly correlated and errors that are moderately correlated, i.e., $\rho_1=2/3$ and $\rho_2=1/3$}
        \end{subfigure}
    \end{minipage}
    \vspace{0.5cm} 
    \begin{minipage}[b]{0.8\textwidth}
        \centering
        \begin{subfigure}[b]{\textwidth}
            \centering
            \includegraphics[width=0.9\textwidth]{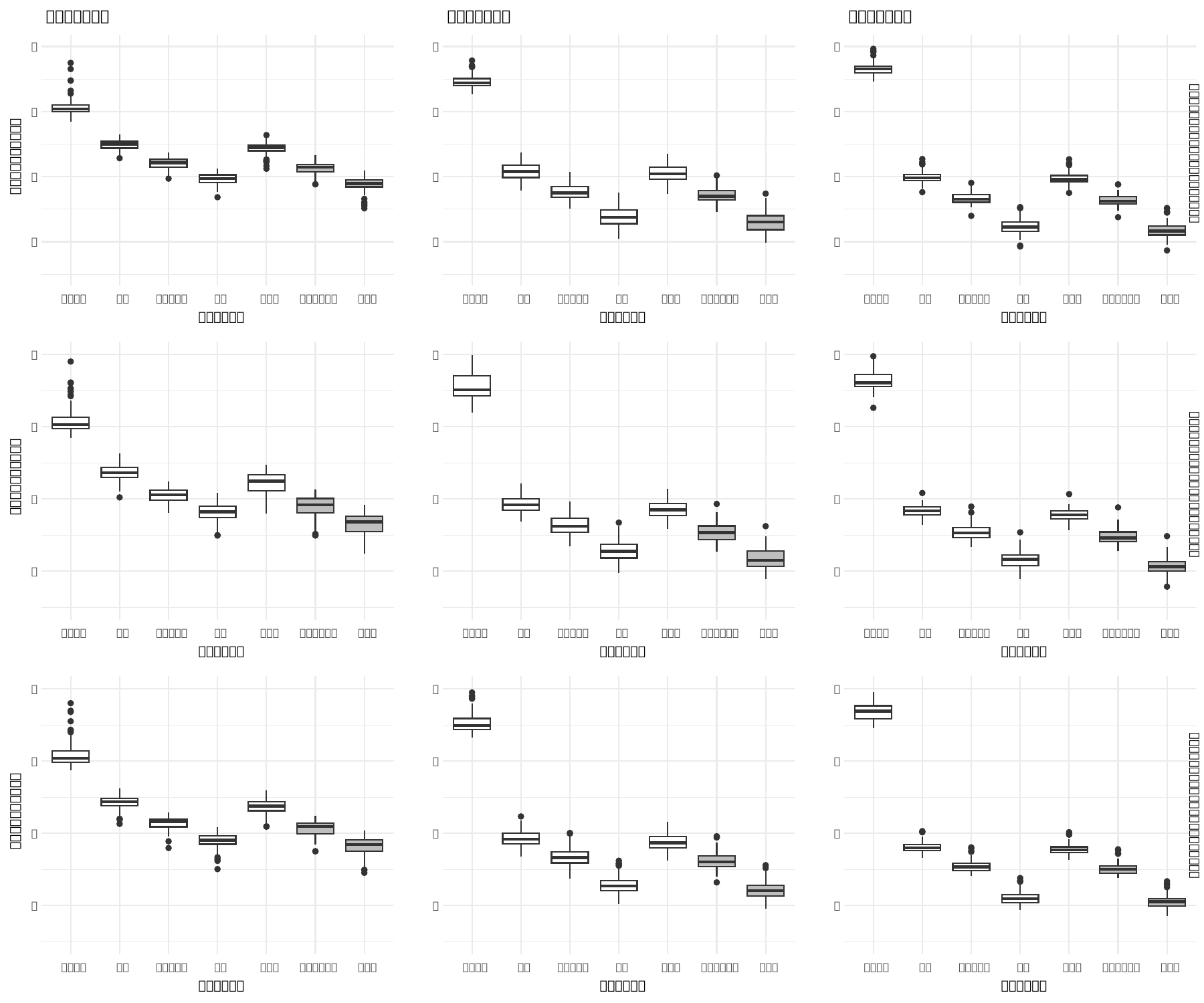}
            \caption{Observational Studies: Boxplot of MSE calculated for scenarios with 10 explanatory variables that are highly correlated and errors that are moderately correlated, i.e., $\rho_1=2/3$ and $\rho_2=1/3$}
        \end{subfigure}
    \end{minipage}
\end{figure}
\begin{figure}[ht]
    \centering
    \begin{minipage}[b]{0.8\textwidth}
        \centering
        \begin{subfigure}[b]{\textwidth}
            \centering
            \includegraphics[width=0.9\textwidth]{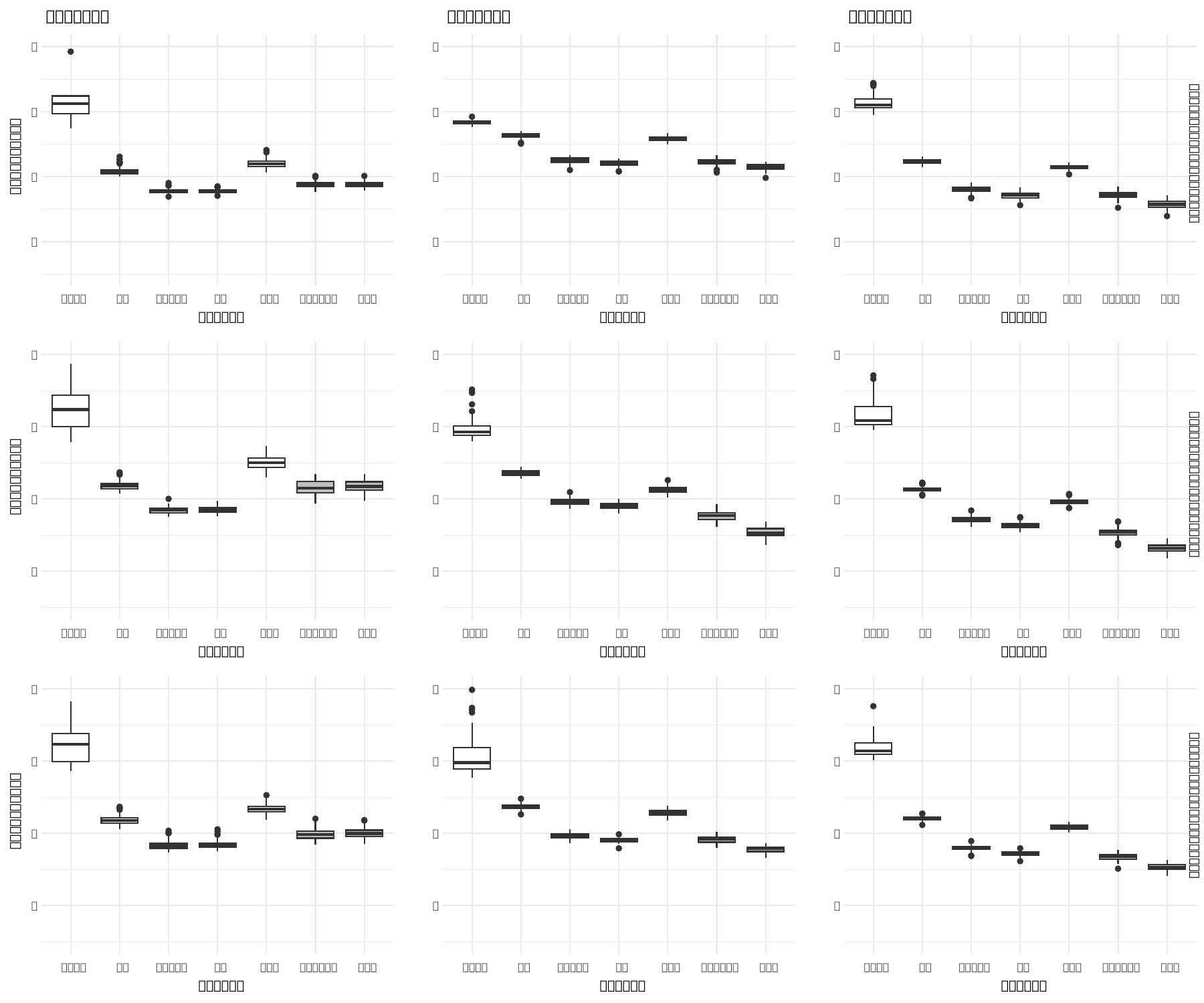}
            \caption{RCTs: Boxplot of MSE calculated for scenarios with 50 explanatory variables that are highly correlated and errors that are moderately correlated, i.e., $\rho_1=2/3$ and $\rho_2=1/3$}
        \end{subfigure}
    \end{minipage}
    \vspace{0.5cm} 
    \begin{minipage}[b]{0.8\textwidth}
        \centering
        \begin{subfigure}[b]{\textwidth}
            \centering
            \includegraphics[width=0.9\textwidth]{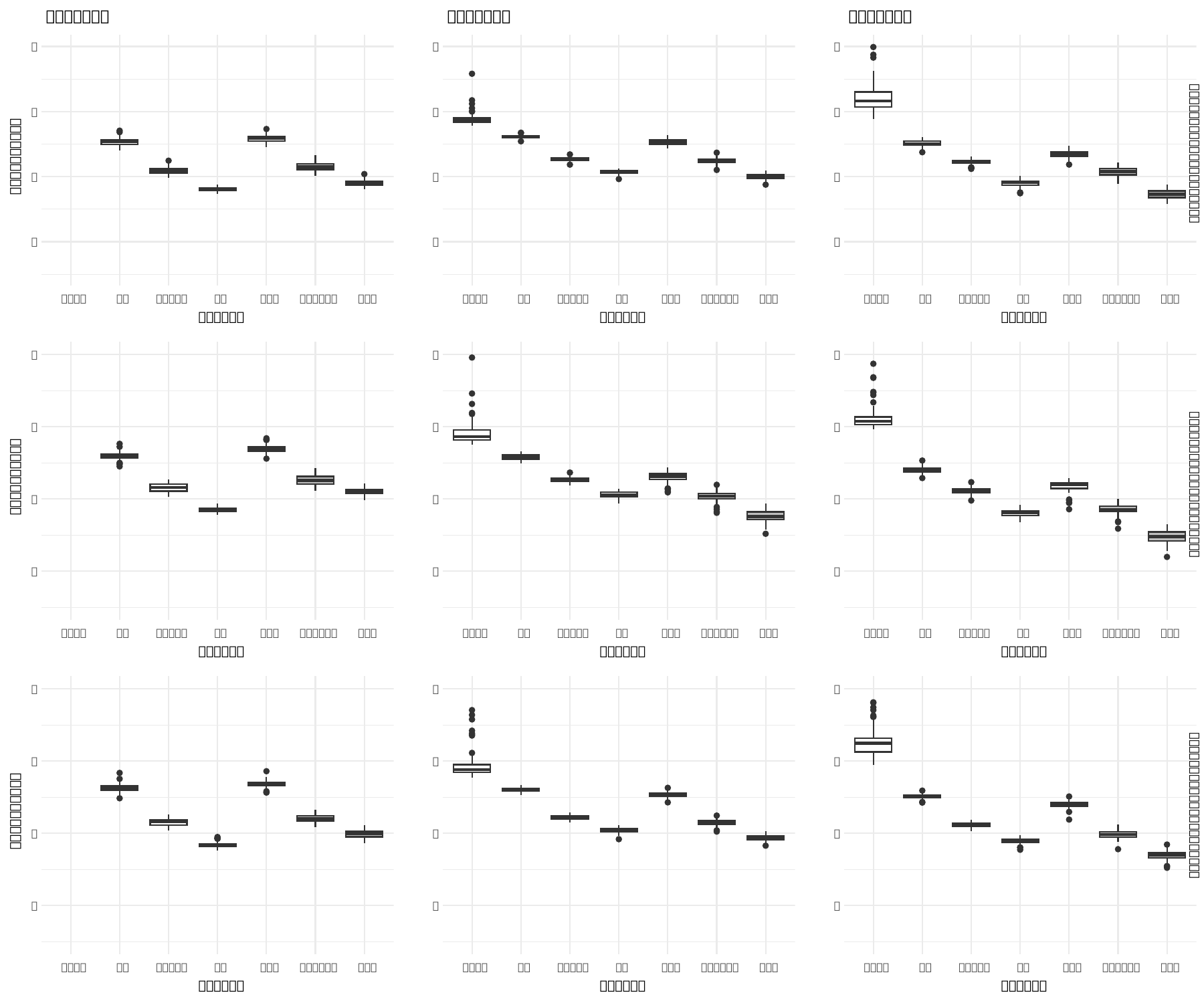}
            \caption{Observational Studies: Boxplot of MSE calculated for scenarios with 50 explanatory variables that are highly correlated and errors that are moderately correlated, i.e., $\rho_1=2/3$ and $\rho_2=1/3$}
        \end{subfigure}
    \end{minipage}
    \end{figure}
\begin{figure}[ht]
    \centering
    \begin{minipage}[b]{0.8\textwidth}
        \centering
        \begin{subfigure}[b]{\textwidth}
            \centering
            \includegraphics[width=0.9\textwidth]{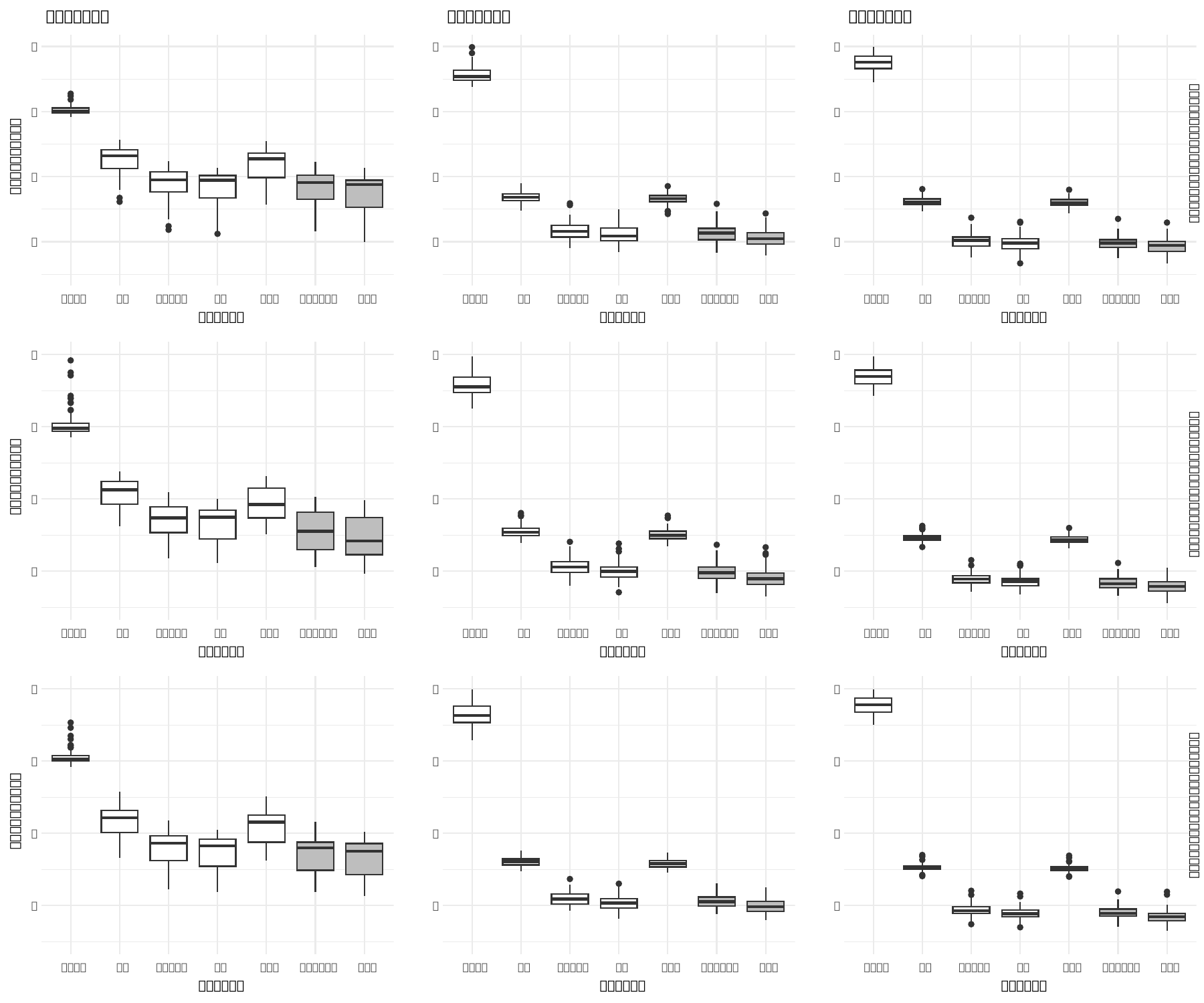}
            \caption{RCTs: Boxplot of MSE calculated for scenarios with 10 explanatory variables that are uncorrelated and errors that are highly correlated, i.e., $\rho_1=0$ and $\rho_2=2/3$}
        \end{subfigure}
    \end{minipage}
    \vspace{0.5cm} 
    \begin{minipage}[b]{0.8\textwidth}
        \centering
        \begin{subfigure}[b]{\textwidth}
            \centering
            \includegraphics[width=0.9\textwidth]{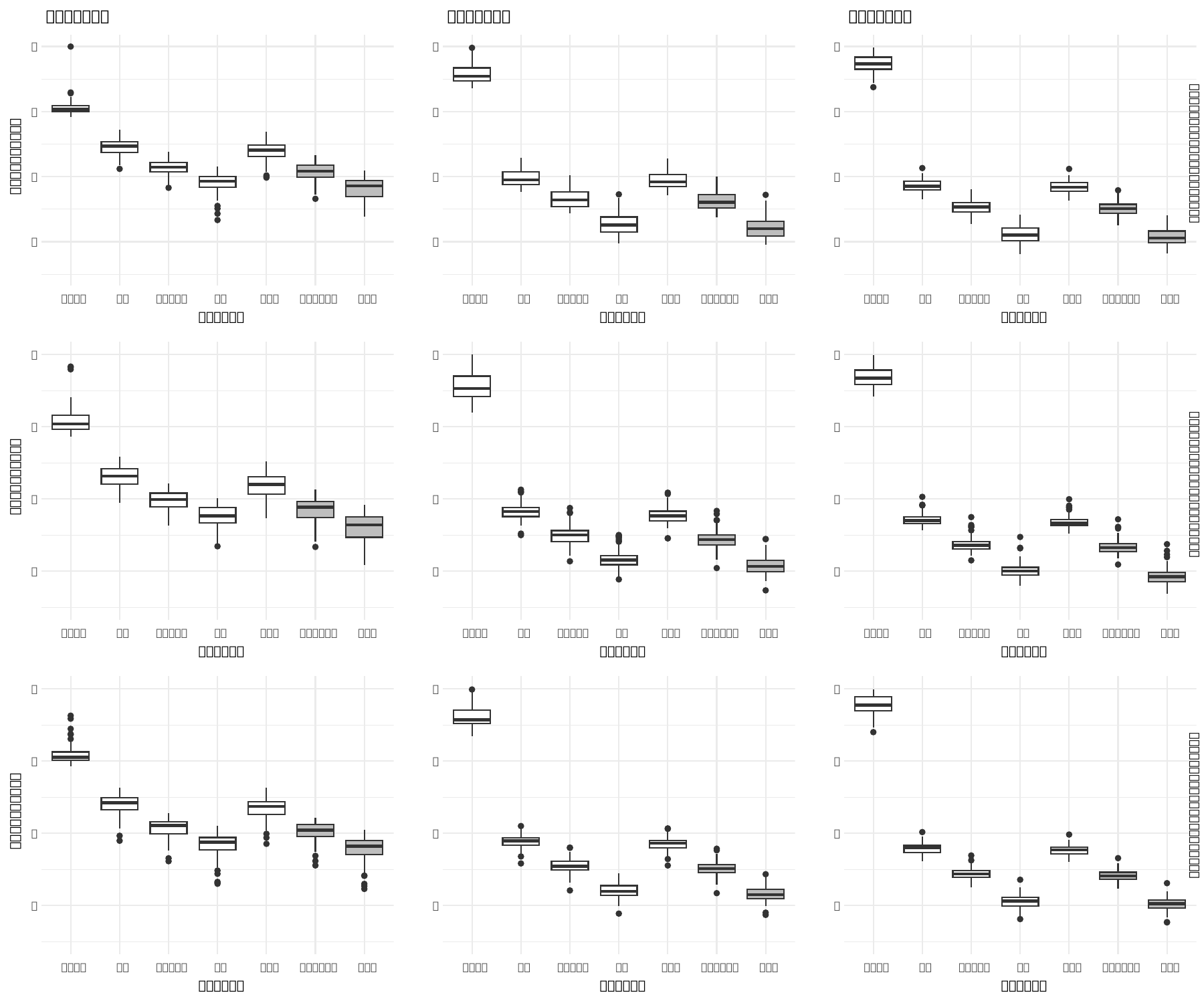}
            \caption{Observational Studies: Boxplot of MSE calculated for scenarios with 10 explanatory variables that are uncorrelated and errors that are highly correlated, i.e., $\rho_1=0$ and $\rho_2=2/3$}
        \end{subfigure}
    \end{minipage}
\end{figure}
\begin{figure}[ht]
    \centering
    \begin{minipage}[b]{0.8\textwidth}
        \centering
        \begin{subfigure}[b]{\textwidth}
            \centering
            \includegraphics[width=0.9\textwidth]{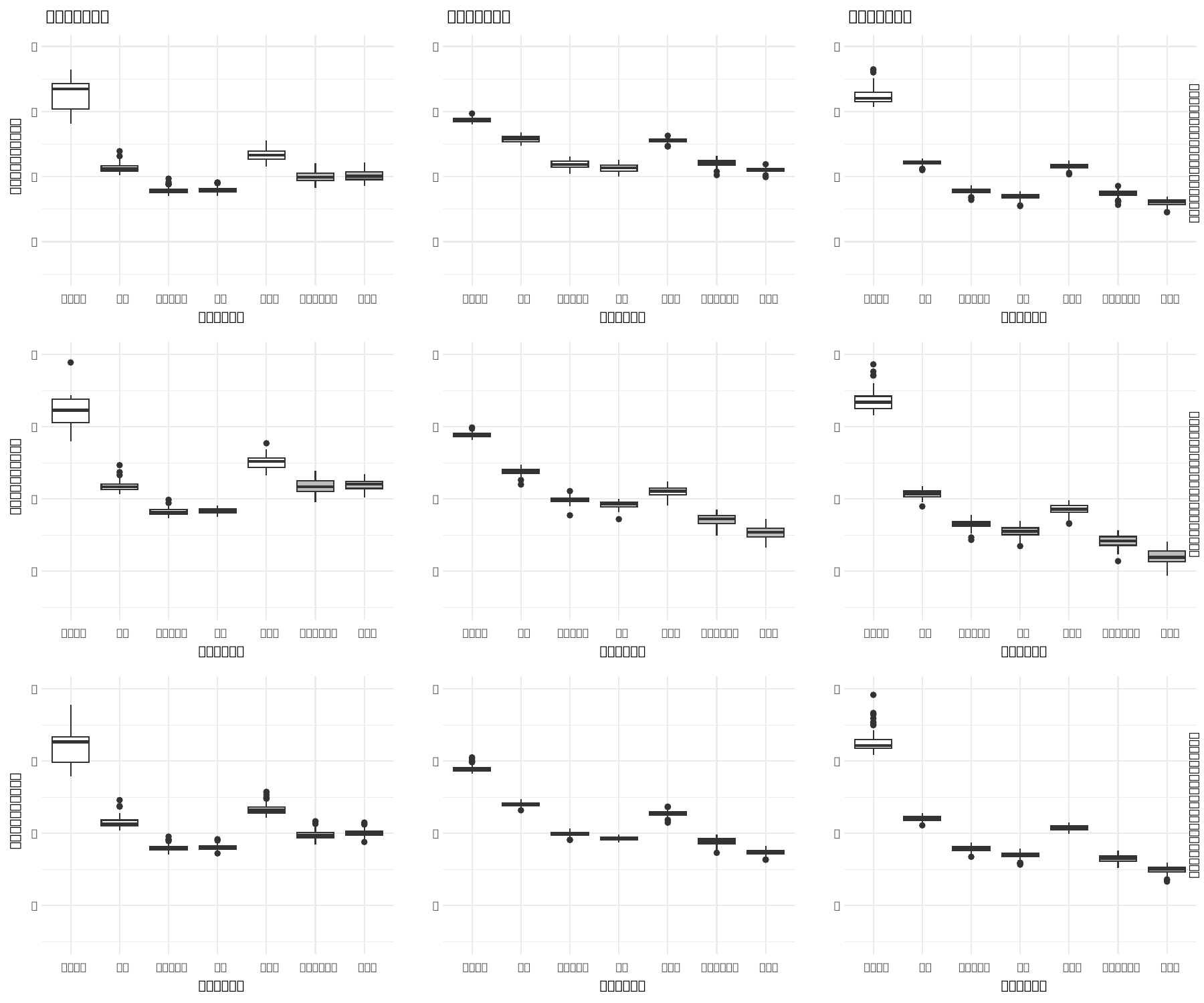}
            \caption{RCTs: Boxplot of MSE calculated for scenarios with 50 explanatory variables that are uncorrelated and errors that are highly correlated, i.e., $\rho_1=0$ and $\rho_2=2/3$}
        \end{subfigure}
    \end{minipage}
    \vspace{0.5cm} 
    \begin{minipage}[b]{0.8\textwidth}
        \centering
        \begin{subfigure}[b]{\textwidth}
            \centering
            \includegraphics[width=0.9\textwidth]{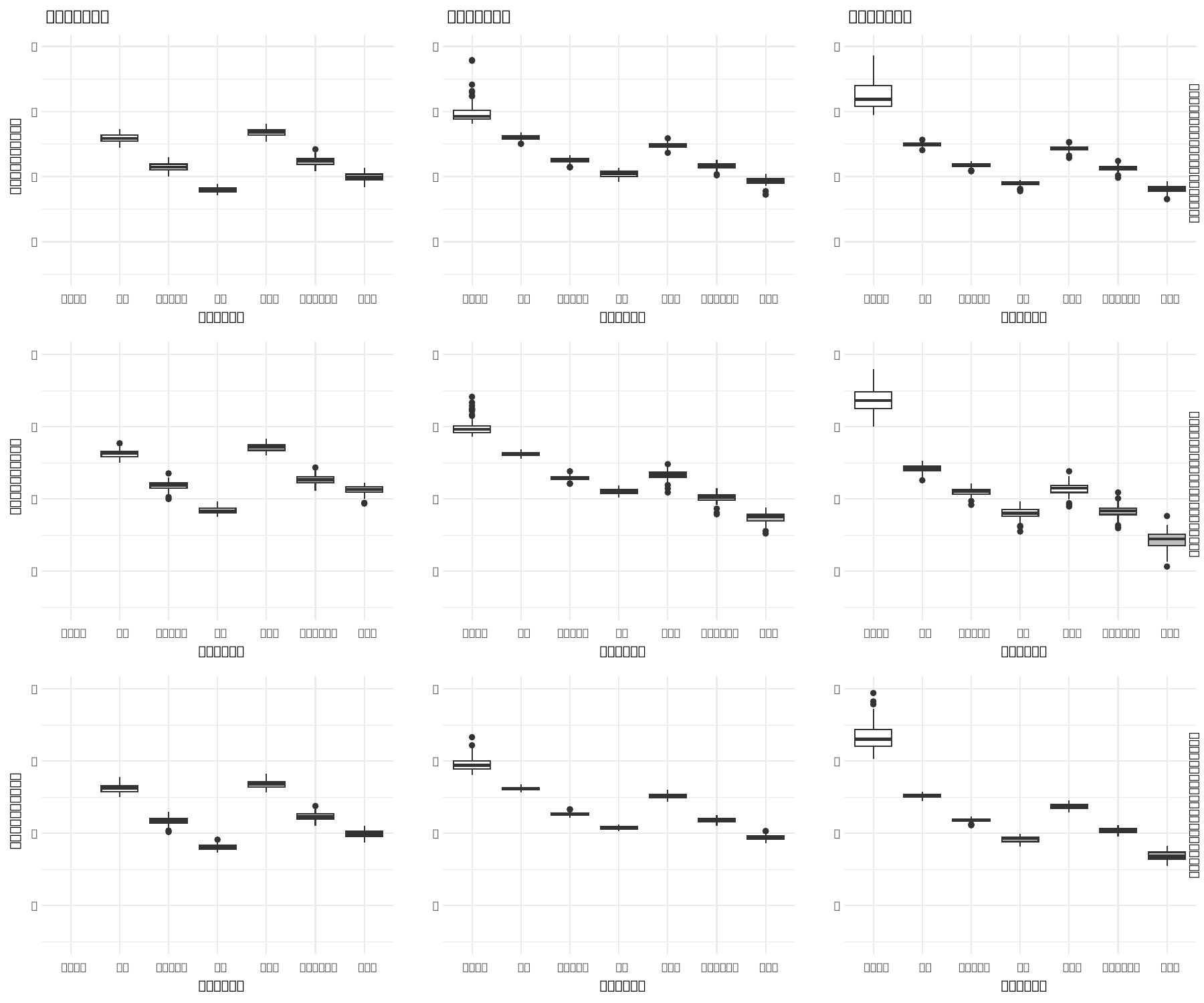}
            \caption{Observational Studies: Boxplot of MSE calculated for scenarios with 50 explanatory variables that are uncorrelated and errors that are highly correlated, i.e., $\rho_1=0$ and $\rho_2=2/3$}
        \end{subfigure}
    \end{minipage}
    \end{figure}
\begin{figure}[ht]
    \centering
    \begin{minipage}[b]{0.8\textwidth}
        \centering
        \begin{subfigure}[b]{\textwidth}
            \centering
            \includegraphics[width=0.9\textwidth]{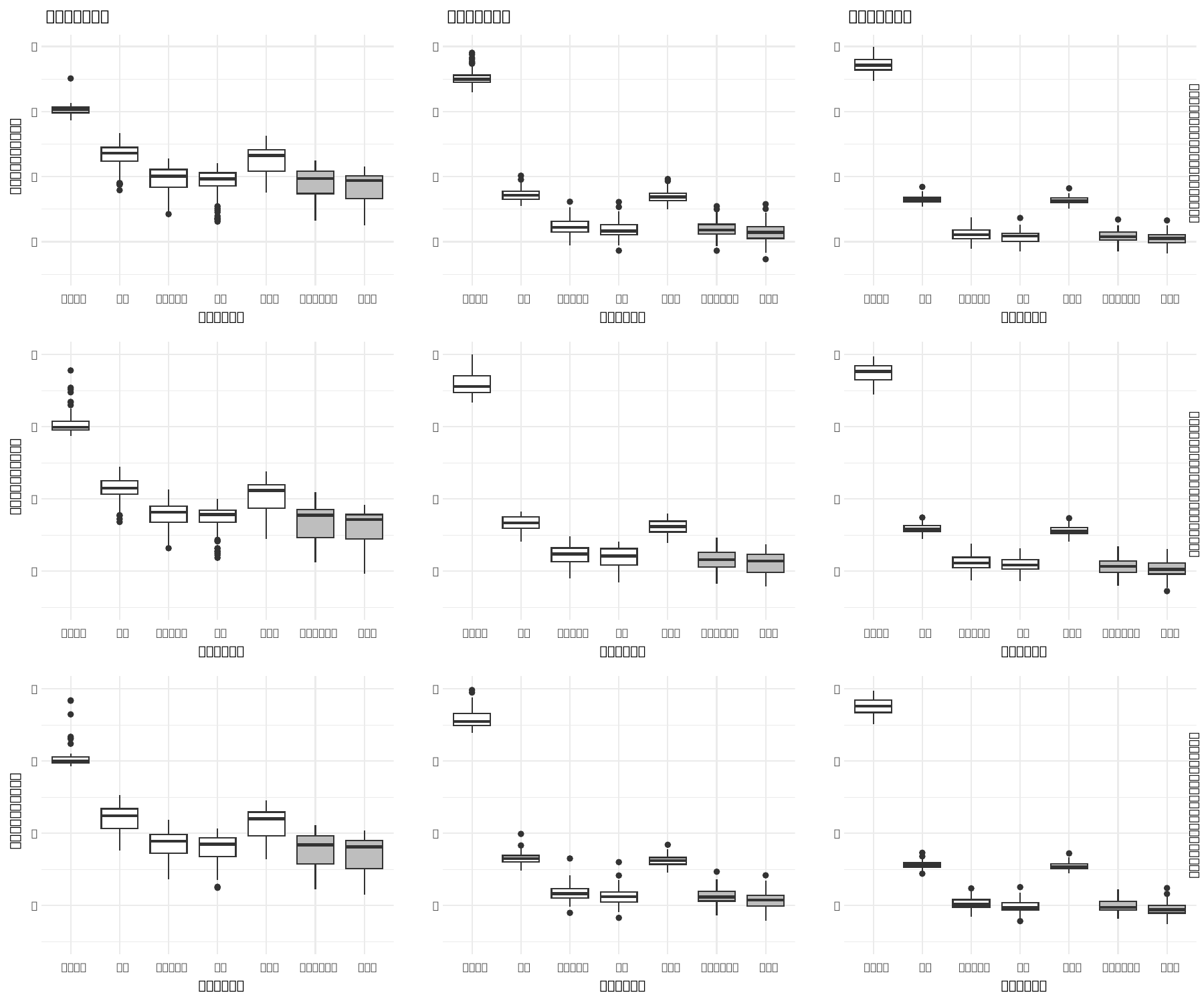}
            \caption{RCTs: Boxplot of MSE calculated for scenarios with 10 explanatory variables that are moderately correlated and errors that are highly correlated, i.e., $\rho_1=1/3$ and $\rho_2=2/3$}
        \end{subfigure}
    \end{minipage}
    \vspace{0.5cm} 
    \begin{minipage}[b]{0.8\textwidth}
        \centering
        \begin{subfigure}[b]{\textwidth}
            \centering
            \includegraphics[width=0.9\textwidth]{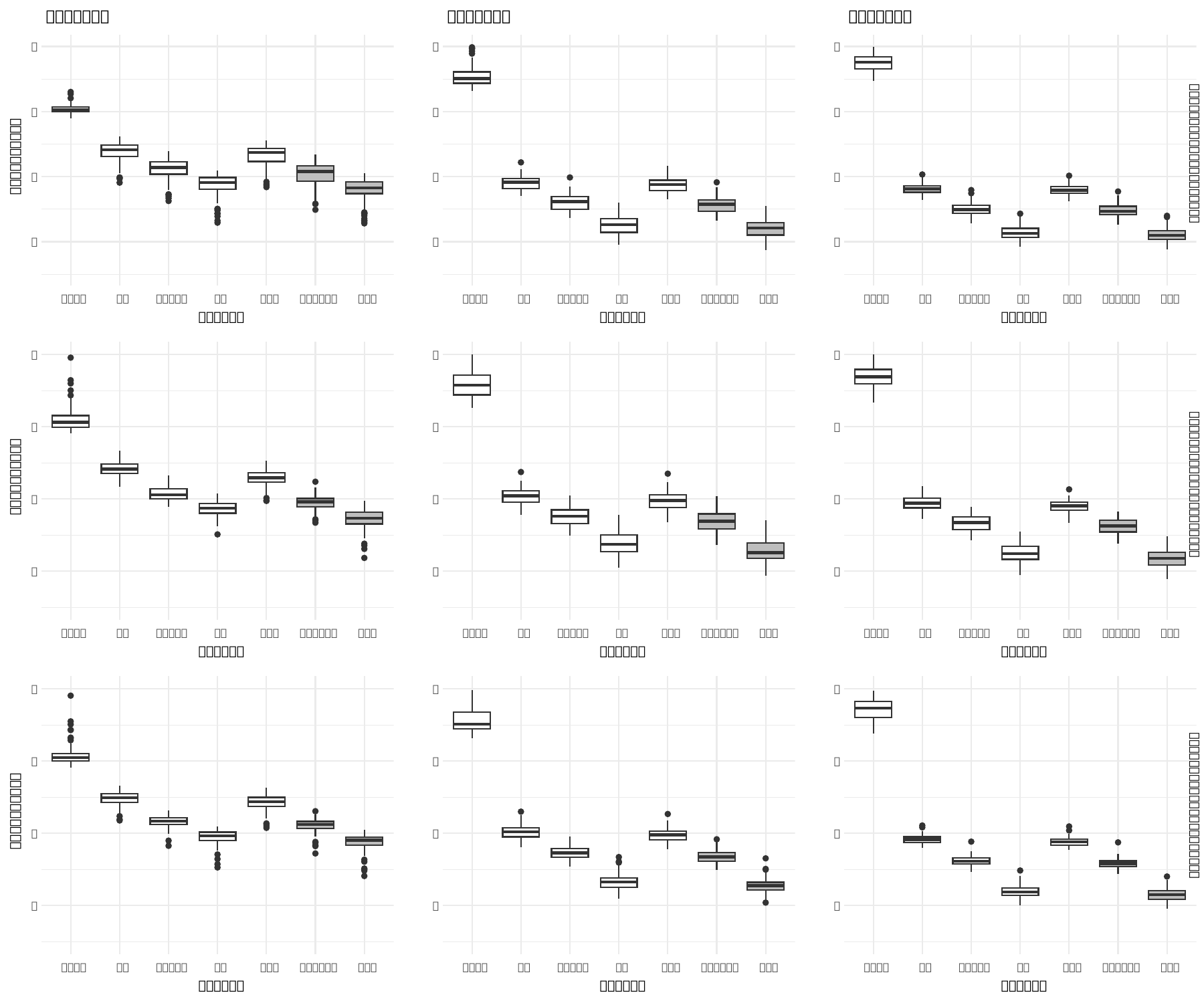}
            \caption{Observational Studies: Boxplot of MSE calculated for scenarios with 10 explanatory variables that are moderately correlated and errors that are highly correlated, i.e., $\rho_1=1/3$ and $\rho_2=2/3$}
        \end{subfigure}
    \end{minipage}
\end{figure}
\begin{figure}[ht]
    \centering
    \begin{minipage}[b]{0.8\textwidth}
        \centering
        \begin{subfigure}[b]{\textwidth}
            \centering
            \includegraphics[width=0.9\textwidth]{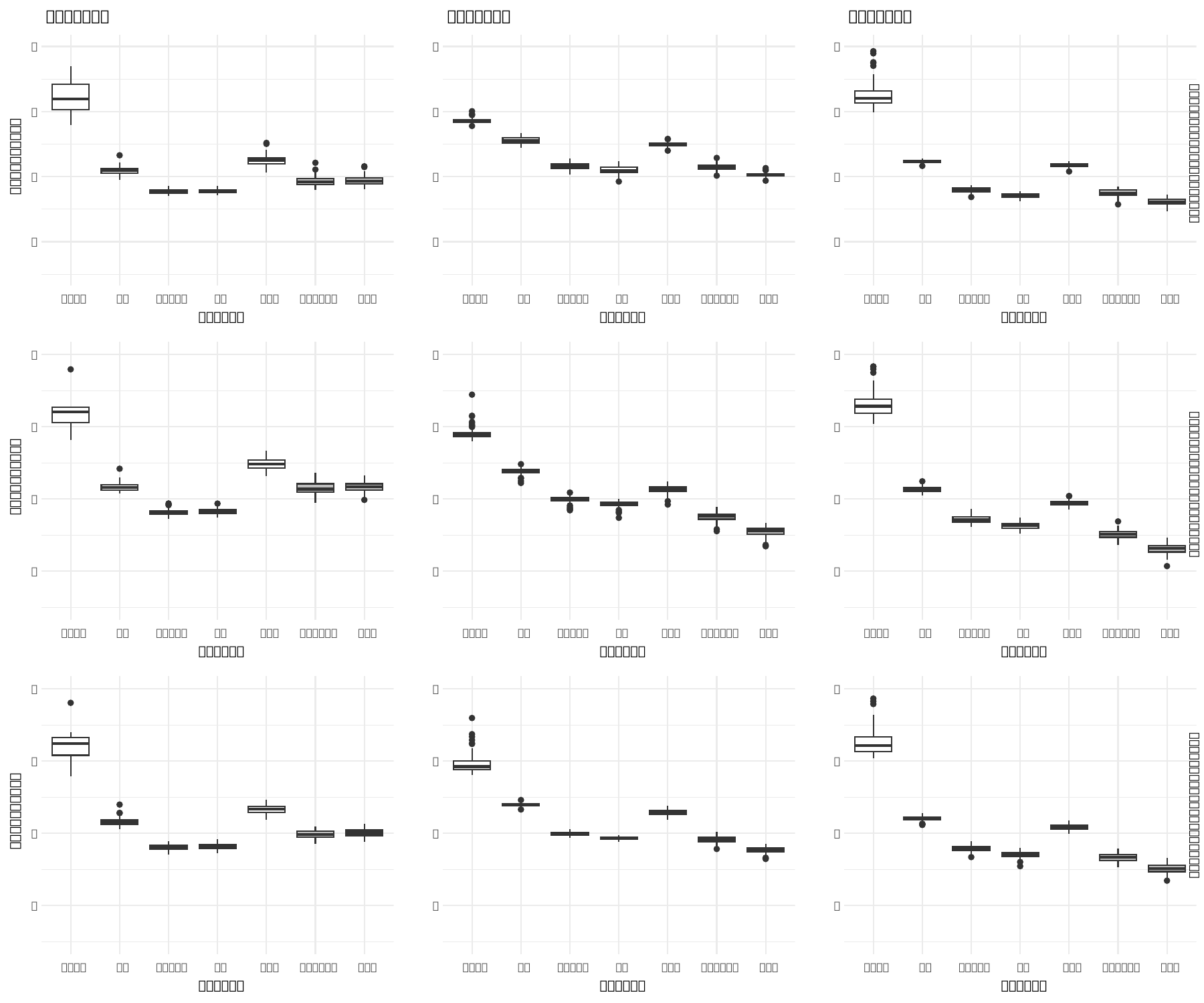}
            \caption{RCTs: Boxplot of MSE calculated for scenarios with 50 explanatory variables that are moderately correlated and errors that are highly correlated, i.e., $\rho_1=1/3$ and $\rho_2=2/3$}
        \end{subfigure}
    \end{minipage}
    \vspace{0.5cm} 
    \begin{minipage}[b]{0.8\textwidth}
        \centering
        \begin{subfigure}[b]{\textwidth}
            \centering
            \includegraphics[width=0.9\textwidth]{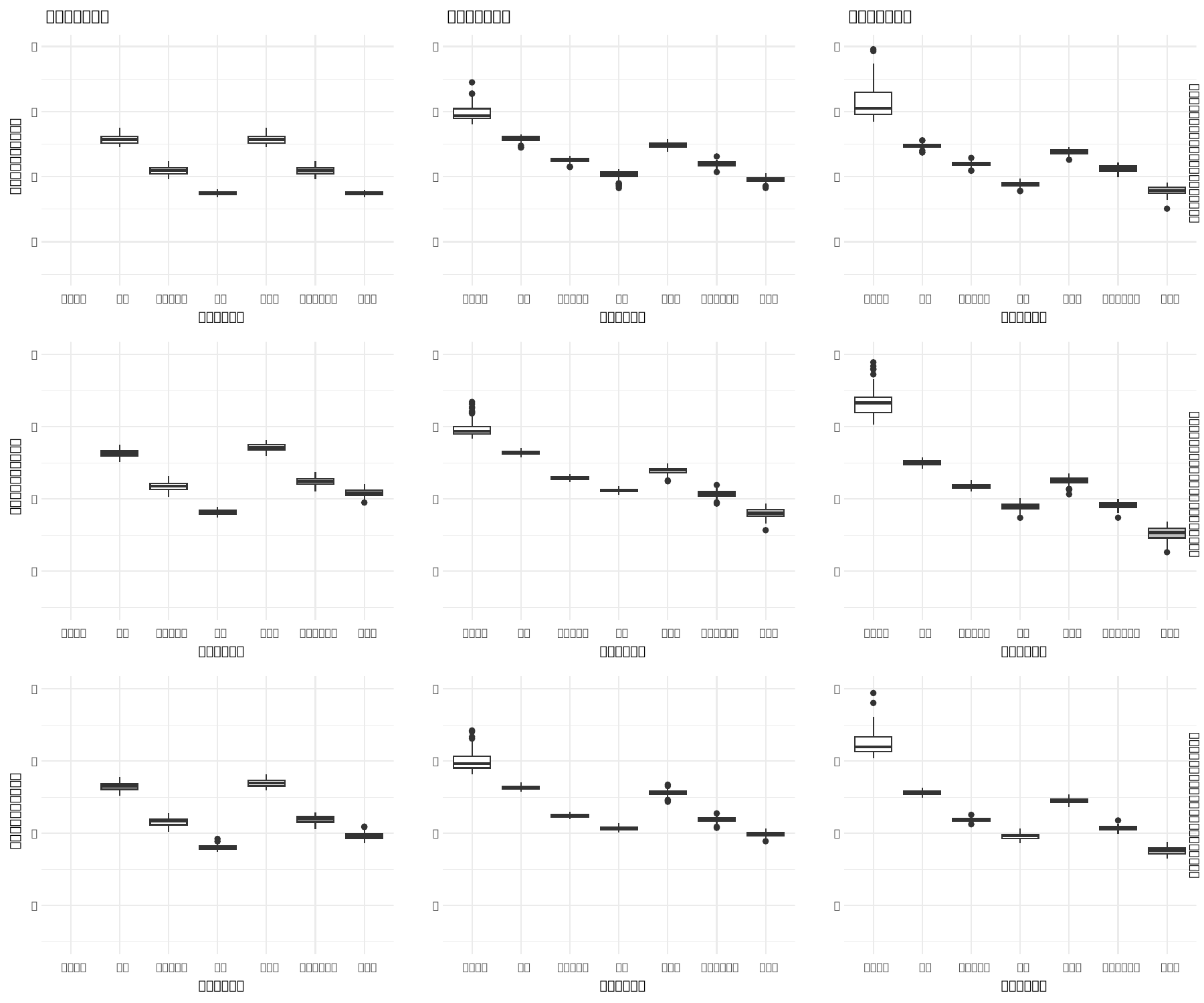}
            \caption{Observational Studies: Boxplot of MSE calculated for scenarios with 50 explanatory variables that are moderately correlated and errors that are highly correlated, i.e., $\rho_1=1/3$ and $\rho_2=2/3$}
        \end{subfigure}
    \end{minipage}
    \end{figure}
\begin{figure}[ht]
    \centering
    \begin{minipage}[b]{0.8\textwidth}
        \centering
        \begin{subfigure}[b]{\textwidth}
            \centering
            \includegraphics[width=0.9\textwidth]{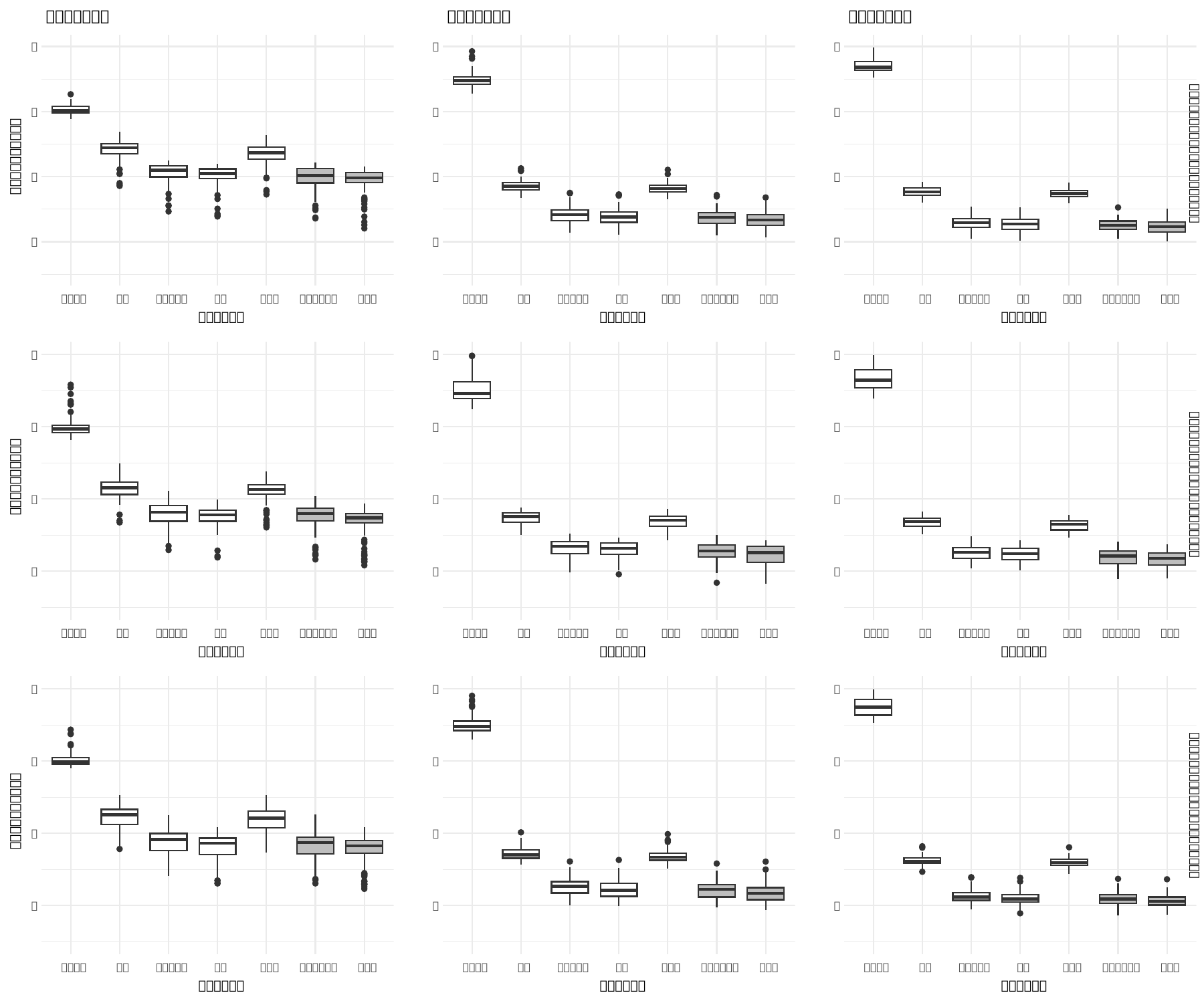}
            \caption{RCTs: Boxplot of MSE calculated for scenarios with 10 explanatory variables and errors that are highly correlated, i.e., $\rho_1=2/3$ and $\rho_2=2/3$}
        \end{subfigure}
    \end{minipage}
    \vspace{0.5cm} 
    \begin{minipage}[b]{0.8\textwidth}
        \centering
        \begin{subfigure}[b]{\textwidth}
            \centering
            \includegraphics[width=0.9\textwidth]{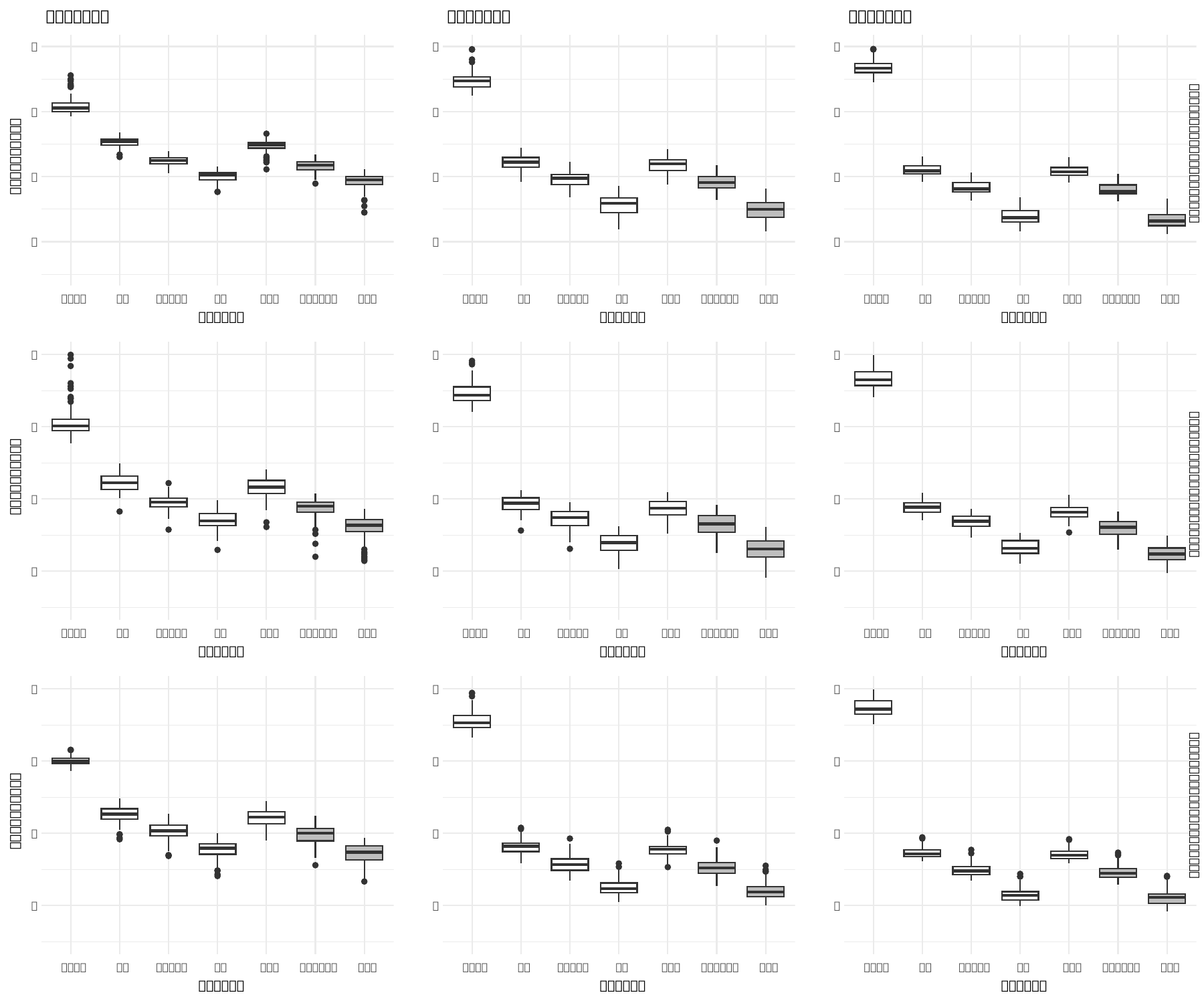}
            \caption{Observational Studies: Boxplot of MSE calculated for scenarios with 10 explanatory variables and errors that are highly correlated, i.e., $\rho_1=2/3$ and $\rho_2=2/3$}
        \end{subfigure}
    \end{minipage}
\end{figure}
\begin{figure}[ht]
    \centering
    \begin{minipage}[b]{0.8\textwidth}
        \centering
        \begin{subfigure}[b]{\textwidth}
            \centering
            \includegraphics[width=0.9\textwidth]{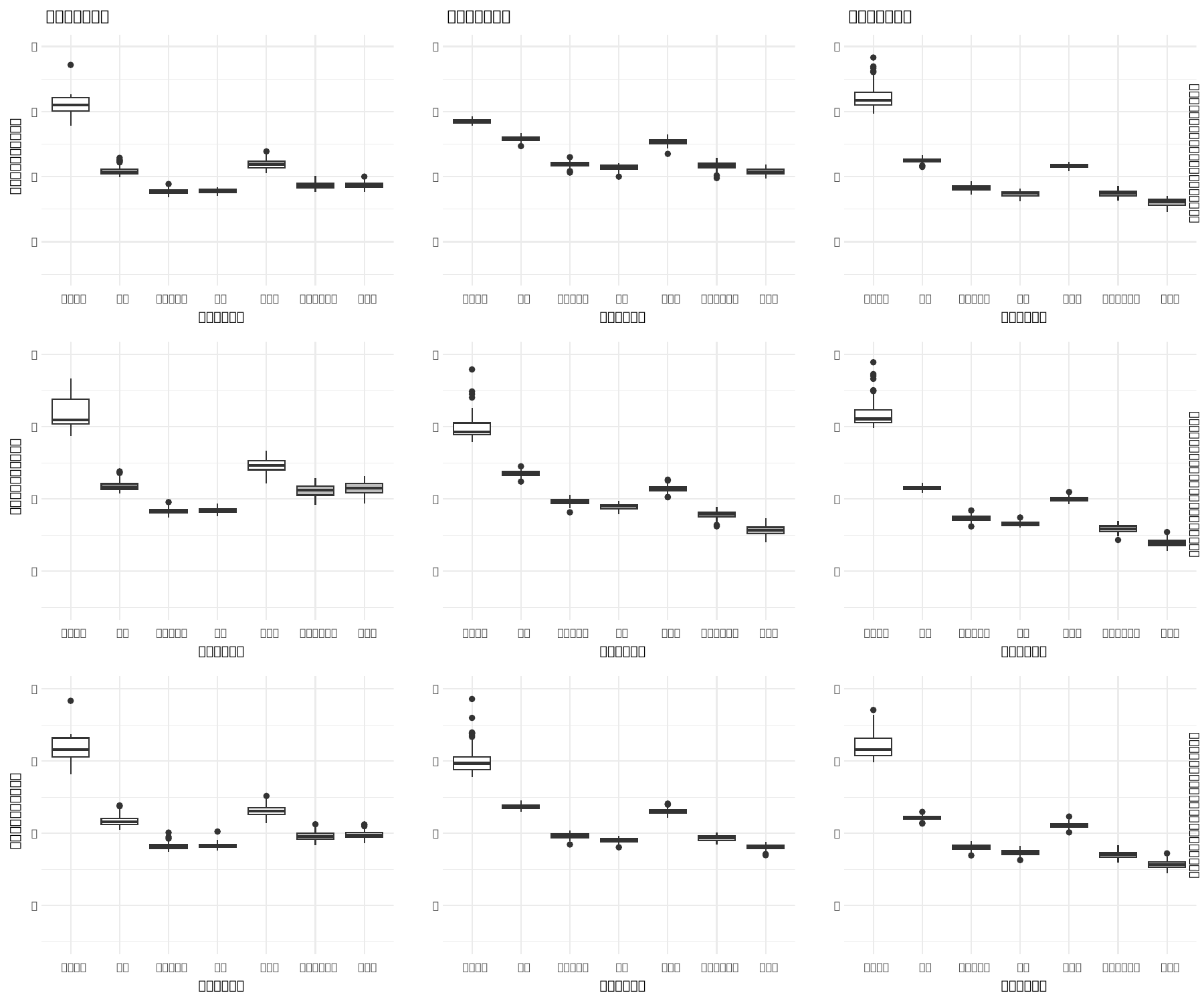}
            \caption{RCTs: Boxplot of MSE calculated for scenarios with 50 explanatory variables and errors that are highly correlated, i.e., $\rho_1=2/3$ and $\rho_2=2/3$}
        \end{subfigure}
    \end{minipage}
    \vspace{0.5cm} 
    \begin{minipage}[b]{0.8\textwidth}
        \centering
        \begin{subfigure}[b]{\textwidth}
            \centering
            \includegraphics[width=0.9\textwidth]{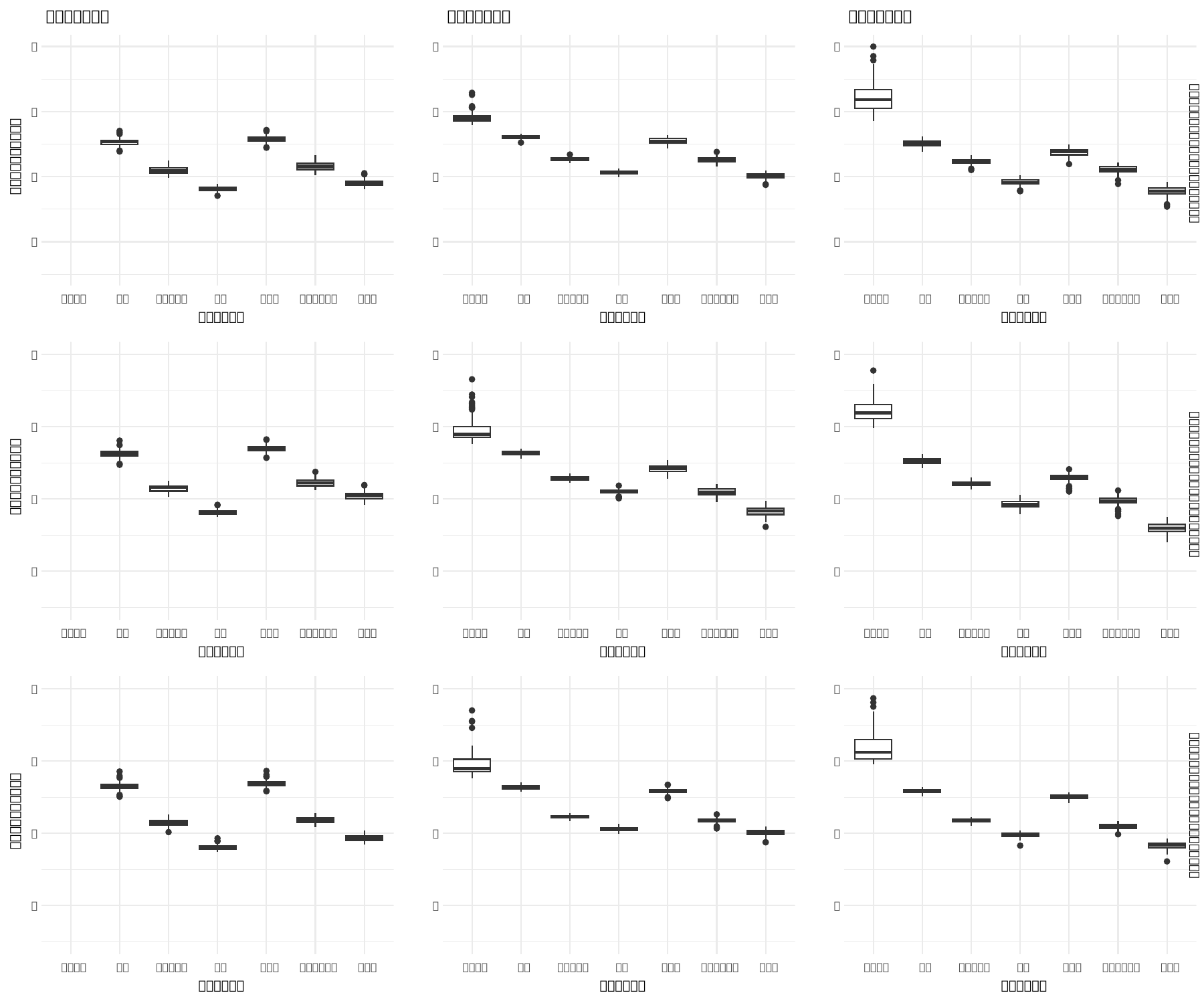}
            \caption{Observational Studies: Boxplot of MSE calculated for scenarios with 50 explanatory variables and errors that are highly correlated, i.e., $\rho_1=2/3$ and $\rho_2=2/3$}
        \end{subfigure}
    \end{minipage}
    \end{figure}
\begin{figure}[ht]
    \centering
    \begin{minipage}[b]{0.8\textwidth}
        \centering
        \begin{subfigure}[b]{\textwidth}
            \centering            \includegraphics[width=0.9\textwidth]{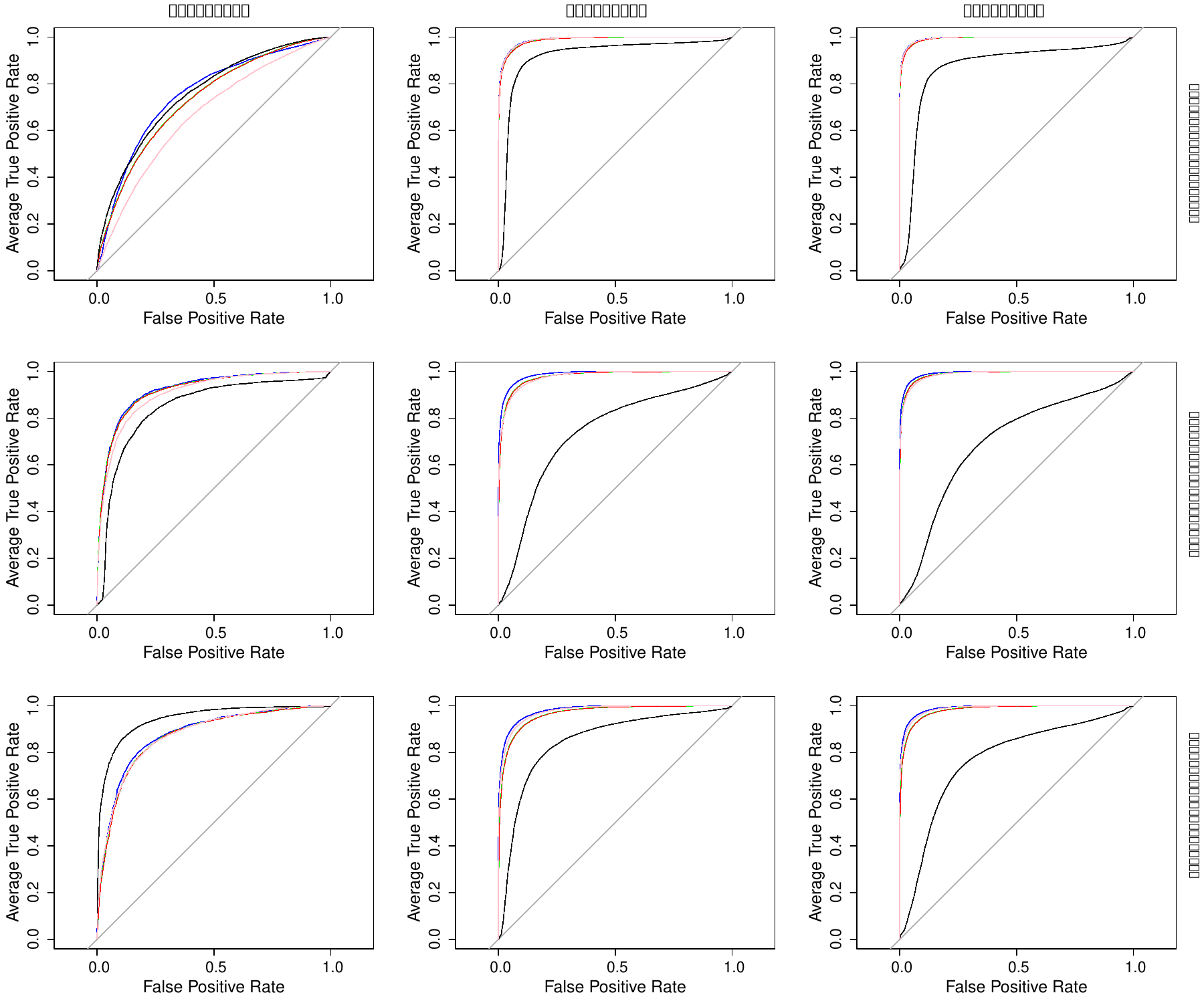}
            \caption{RCTs: ROC curves for scenarios with 10 explanatory variables and errors that are uncorrelated, i.e., $\rho_1=0$ and $\rho_2=0$}
        \end{subfigure}
    \end{minipage}
    \vspace{0.5cm} 
    \begin{minipage}[b]{0.8\textwidth}
        \centering
        \begin{subfigure}[b]{\textwidth}
            \centering
            \includegraphics[width=0.9\textwidth]{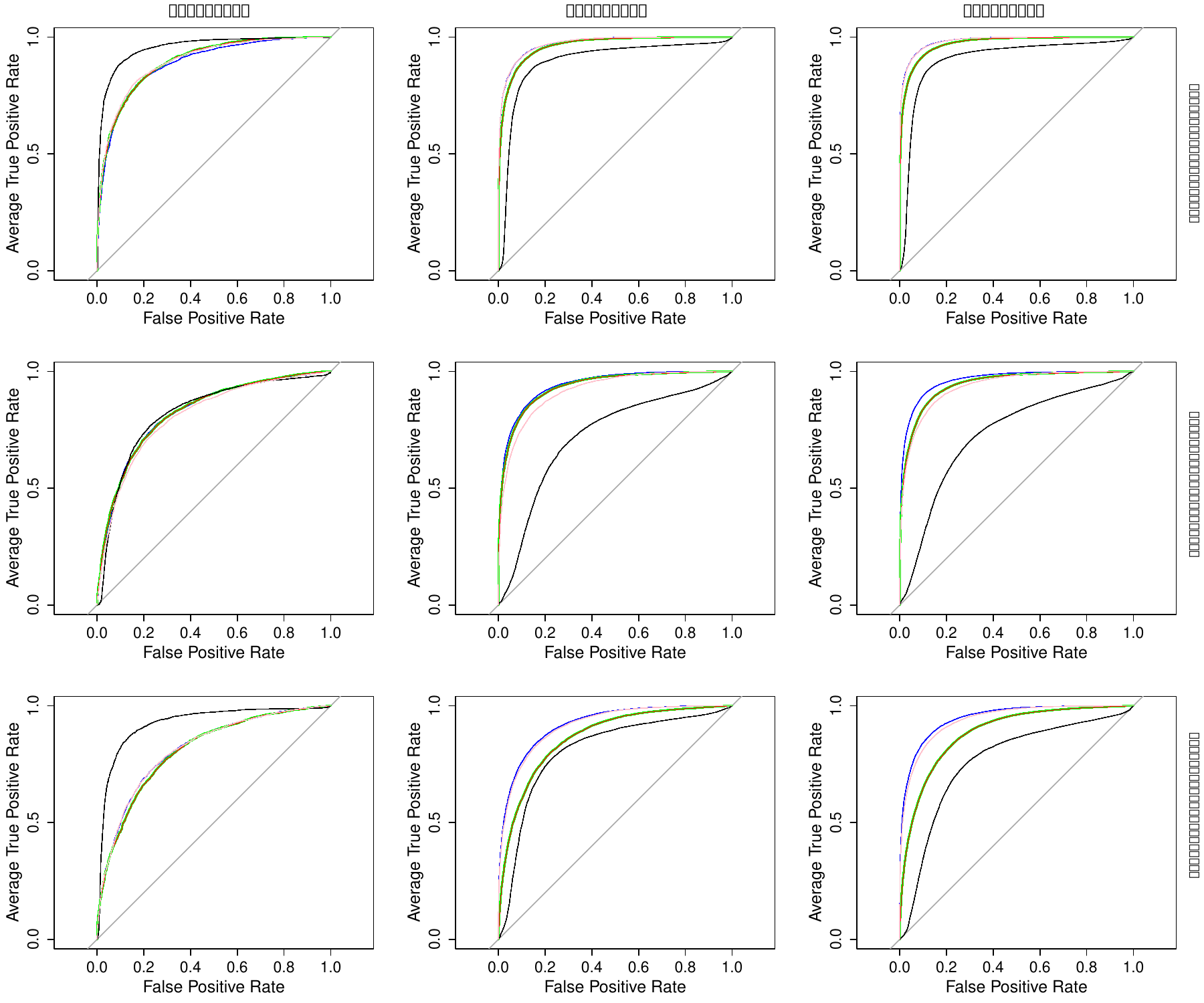}
            \caption{Observational Studies: ROC curves for scenarios with 10 explanatory variables and errors that are uncorrelated, i.e., $\rho_1=0$ and $\rho_2=0$}
        \end{subfigure}
    \end{minipage}
\end{figure}
\begin{figure}[ht]
    \centering
    \begin{minipage}[b]{0.8\textwidth}
        \centering
        \begin{subfigure}[b]{\textwidth}
            \centering
            \includegraphics[width=0.9\textwidth]{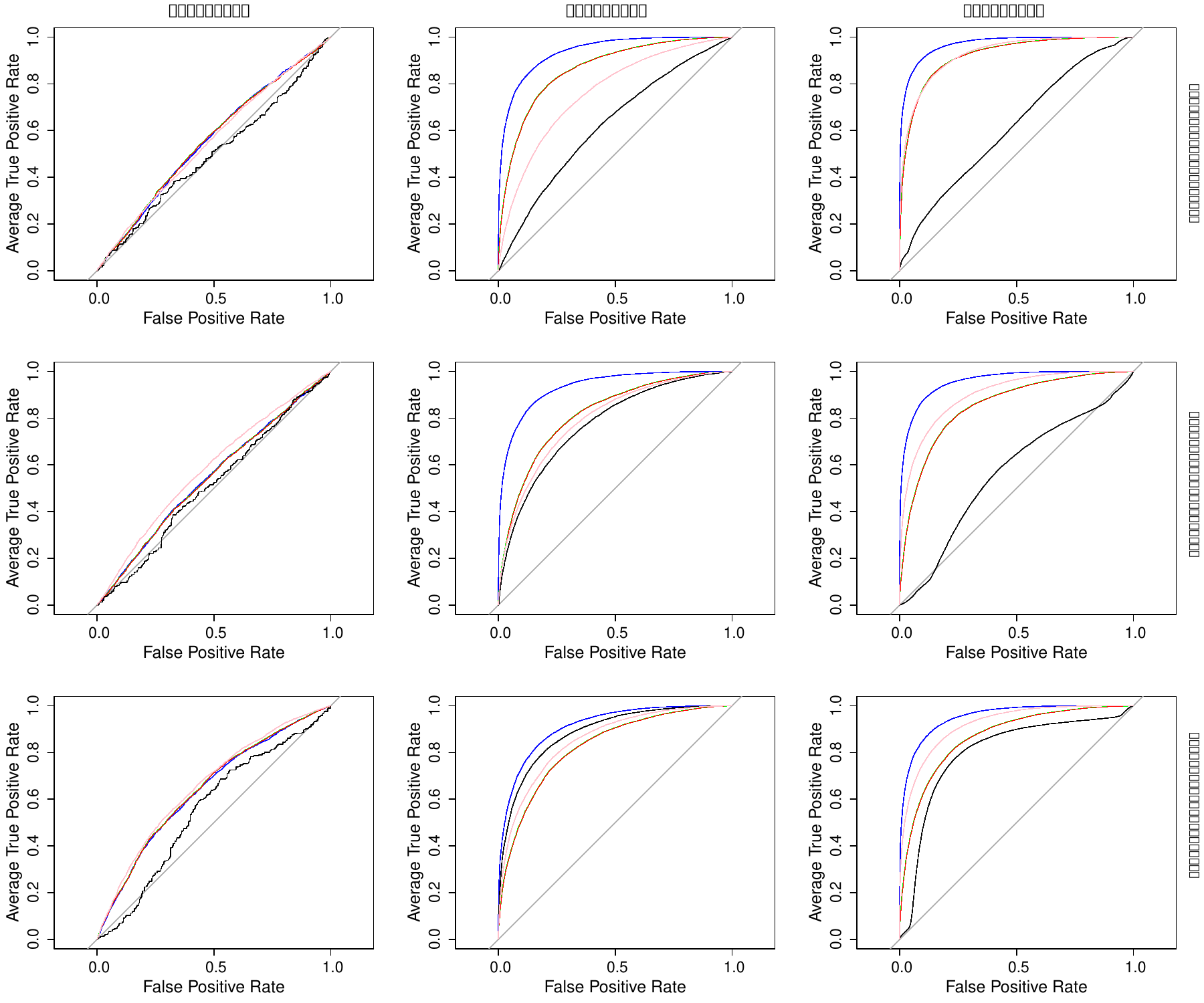}
            \caption{RCTs: ROC curves for scenarios with 50 explanatory variables and errors that are uncorrelated, i.e., $\rho_1=0$ and $\rho_2=0$}
        \end{subfigure}
    \end{minipage}
    \vspace{0.5cm} 
    \begin{minipage}[b]{0.8\textwidth}
        \centering
        \begin{subfigure}[b]{\textwidth}
            \centering
            \includegraphics[width=0.9\textwidth]{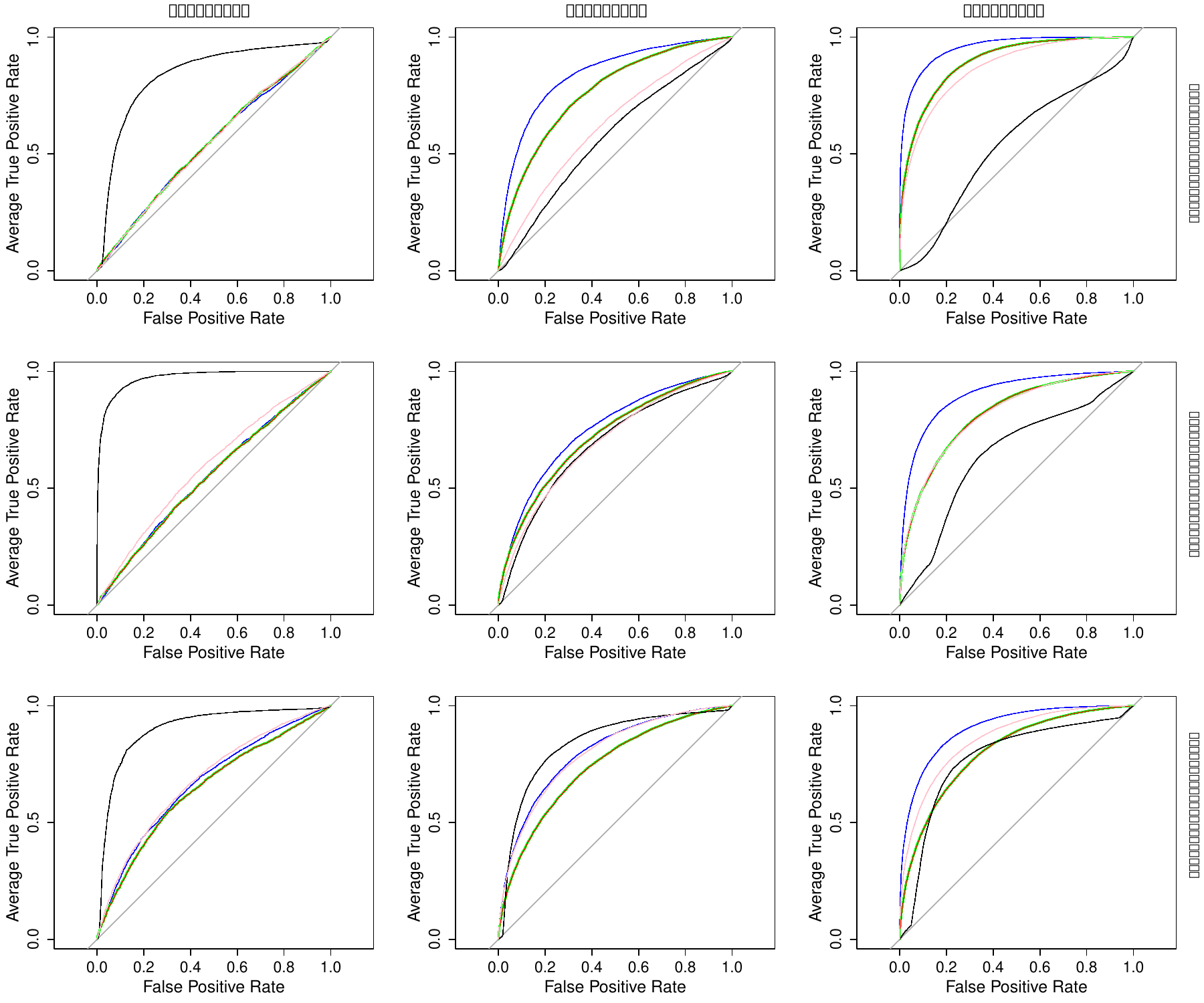}
            \caption{Observational Studies: ROC curves for scenarios with 50 explanatory variables and errors that are uncorrelated, i.e., $\rho_1=0$ and $\rho_2=0$}        \end{subfigure}
    \end{minipage}
\end{figure}
\begin{figure}[ht]
    \centering
    \begin{minipage}[b]{0.8\textwidth}
        \centering
        \begin{subfigure}[b]{\textwidth}
            \centering
            \includegraphics[width=0.9\textwidth]{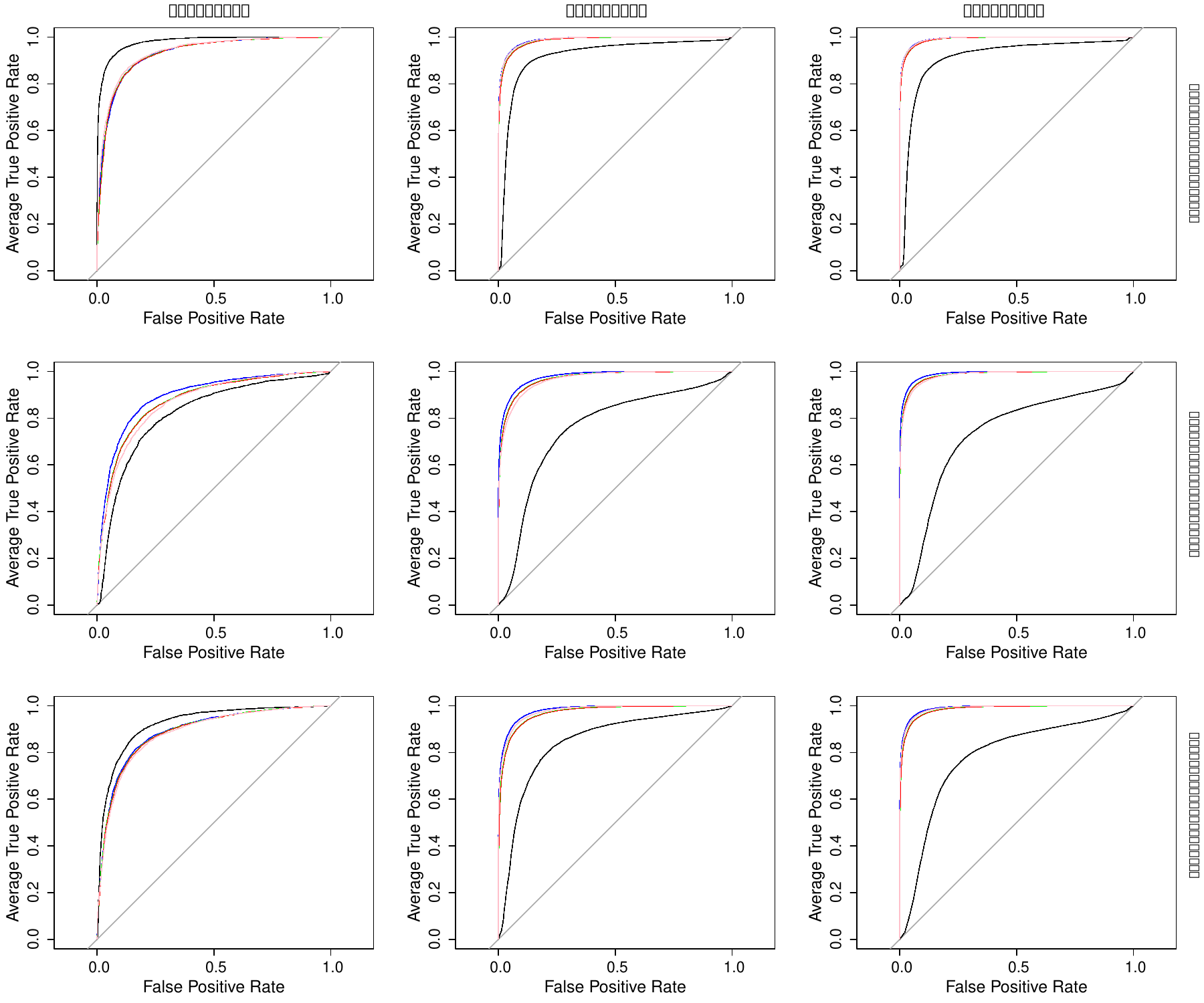}
            \caption{RCTs: ROC curves for scenarios with 10 explanatory variables that are moderately correlated and errors that are uncorrelated, i.e., $\rho_1=1/3$ and $\rho_2=0$}
        \end{subfigure}
    \end{minipage}
    \vspace{0.5cm} 
    \begin{minipage}[b]{0.8\textwidth}
        \centering
        \begin{subfigure}[b]{\textwidth}
            \centering
            \includegraphics[width=0.9\textwidth]{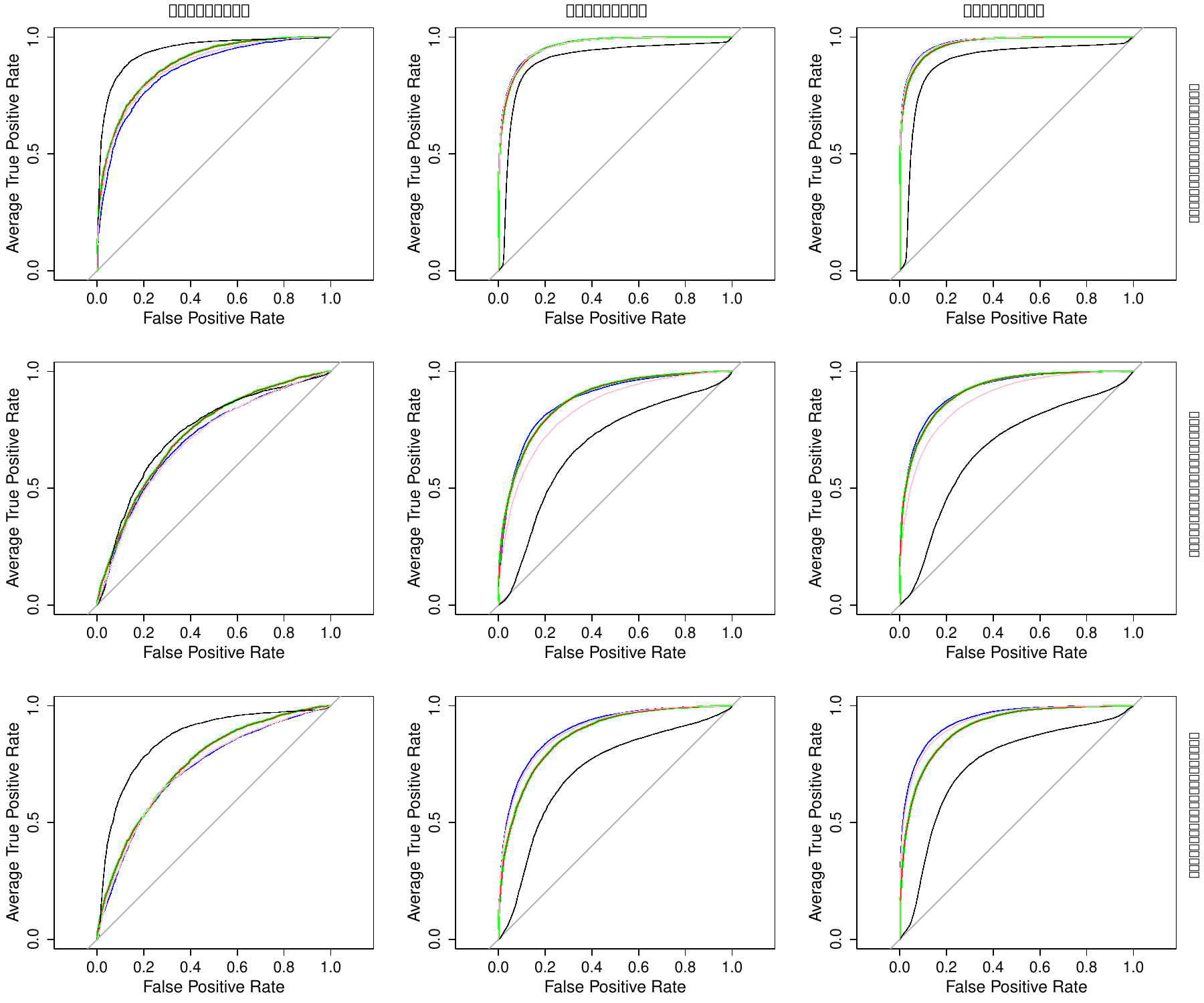}
            \caption{Observational Studies: ROC curves for scenarios with 10 explanatory variables that are moderately correlated and errors that are uncorrelated, i.e., $\rho_1=1/3$ and $\rho_2=0$}
        \end{subfigure}
    \end{minipage}
\end{figure}
\begin{figure}[ht]
    \centering
    \begin{minipage}[b]{0.8\textwidth}
        \centering
        \begin{subfigure}[b]{\textwidth}
            \centering
            \includegraphics[width=0.9\textwidth]{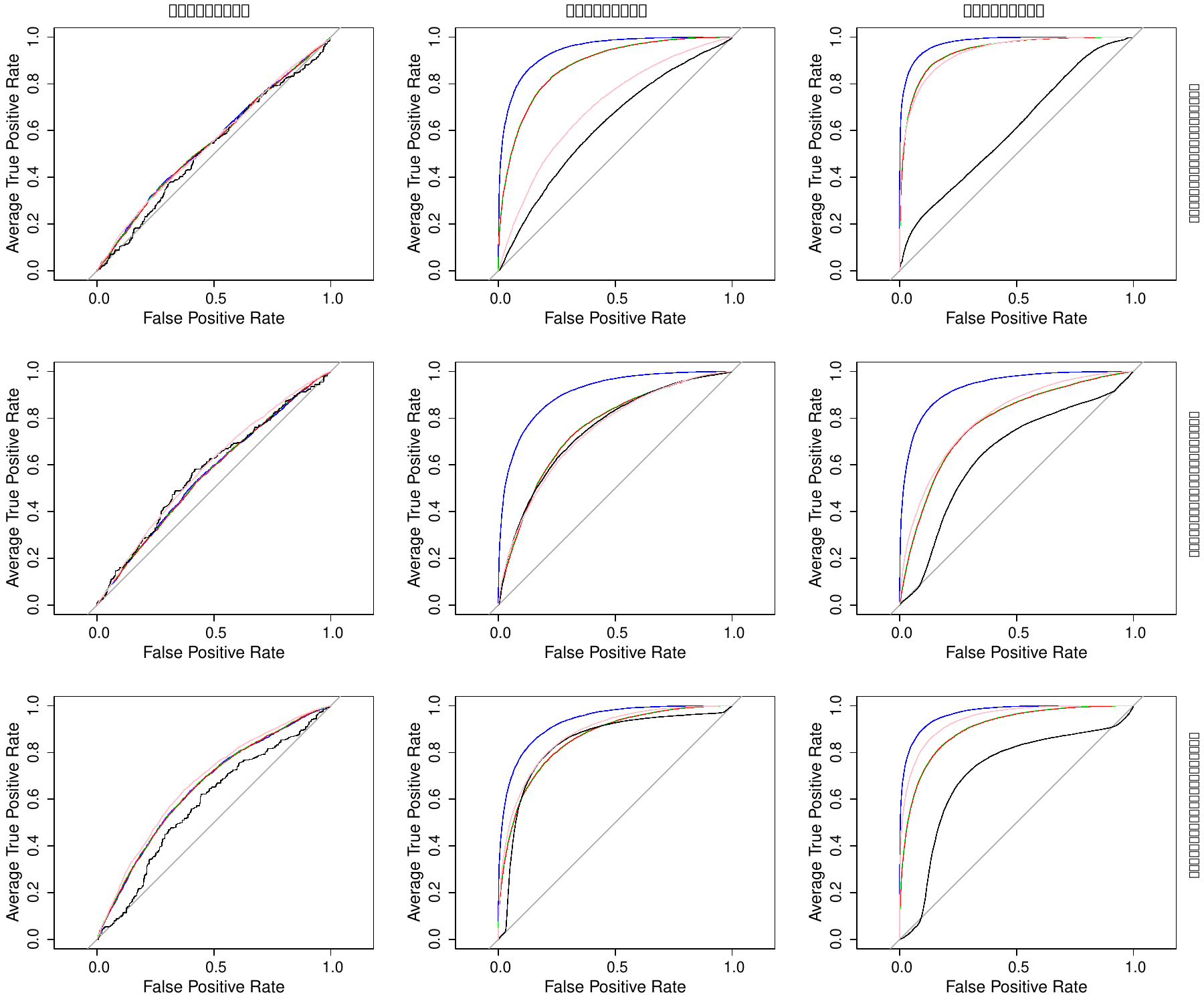}
            \caption{RCTs: ROC curves for scenarios with 50 explanatory variables that are moderately correlated and errors that are uncorrelated, i.e., $\rho_1=1/3$ and $\rho_2=0$}
        \end{subfigure}
    \end{minipage}
    \vspace{0.5cm} 
    \begin{minipage}[b]{0.8\textwidth}
        \centering
        \begin{subfigure}[b]{\textwidth}
            \centering
            \includegraphics[width=0.9\textwidth]{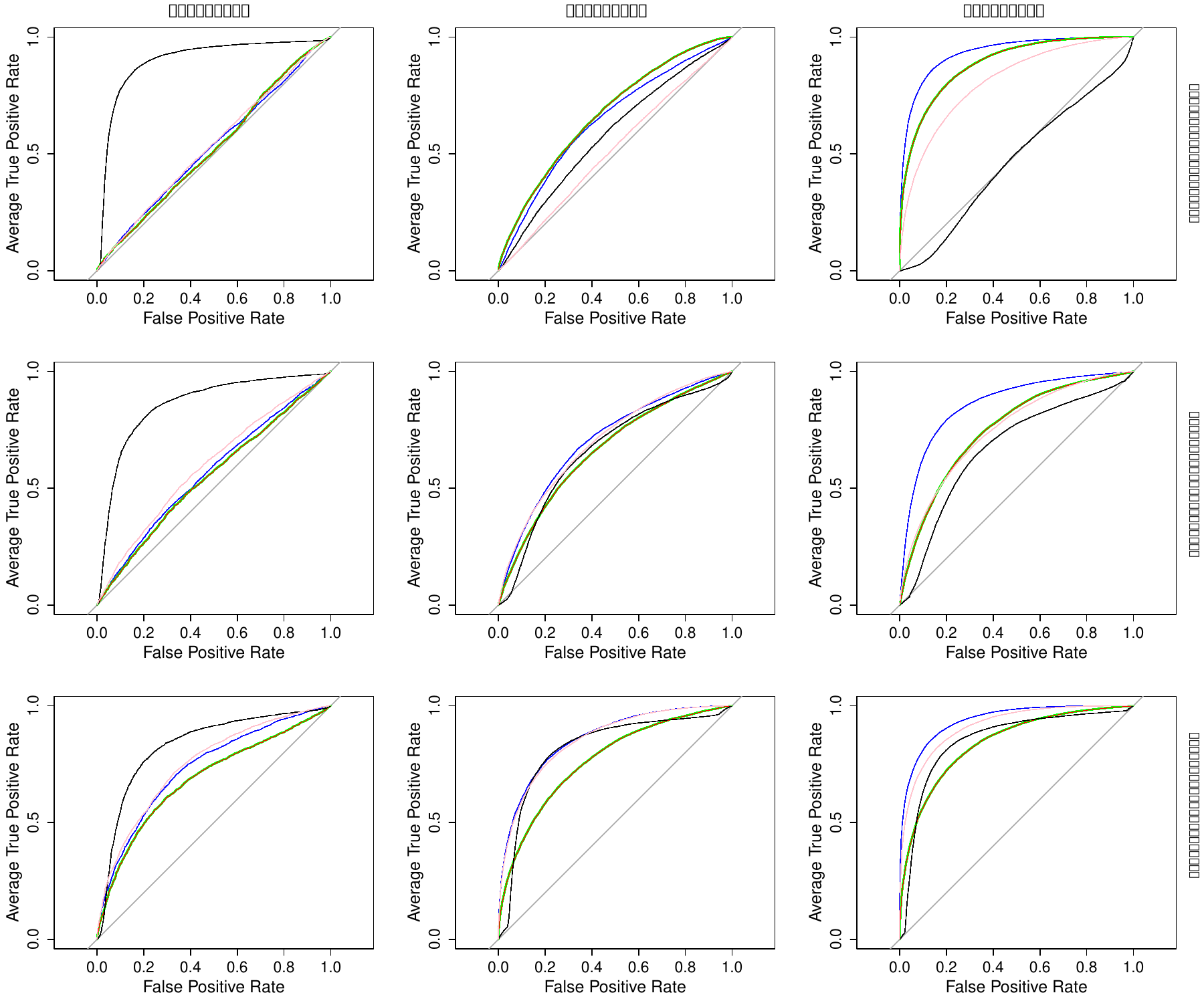}
            \caption{Observational Studies: ROC curves for scenarios with 50 explanatory variables that are moderately correlated and errors that are uncorrelated, i.e., $\rho_1=1/3$ and $\rho_2=0$}
        \end{subfigure}
    \end{minipage}
\end{figure}
\begin{figure}[ht]
    \centering
    \begin{minipage}[b]{0.8\textwidth}
        \centering
        \begin{subfigure}[b]{\textwidth}
            \centering
            \includegraphics[width=0.9\textwidth]{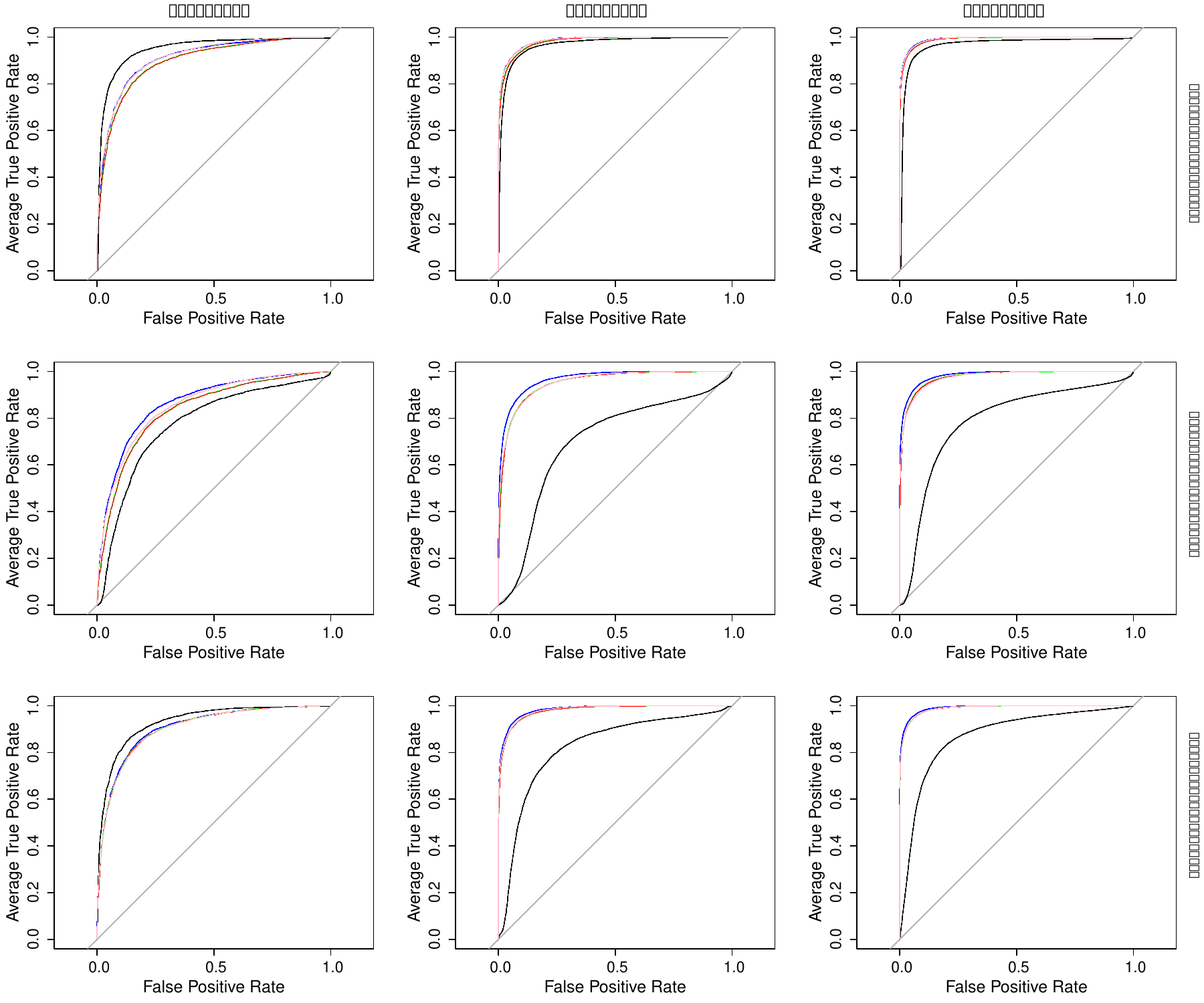}
            \caption{RCTs: ROC curves for scenarios with 10 explanatory variables that are highly correlated and errors that are uncorrelated, i.e., $\rho_1=2/3$ and $\rho_2=0$}
        \end{subfigure}
    \end{minipage}
    \vspace{0.5cm} 
    \begin{minipage}[b]{0.8\textwidth}
        \centering
        \begin{subfigure}[b]{\textwidth}
            \centering
            \includegraphics[width=0.9\textwidth]{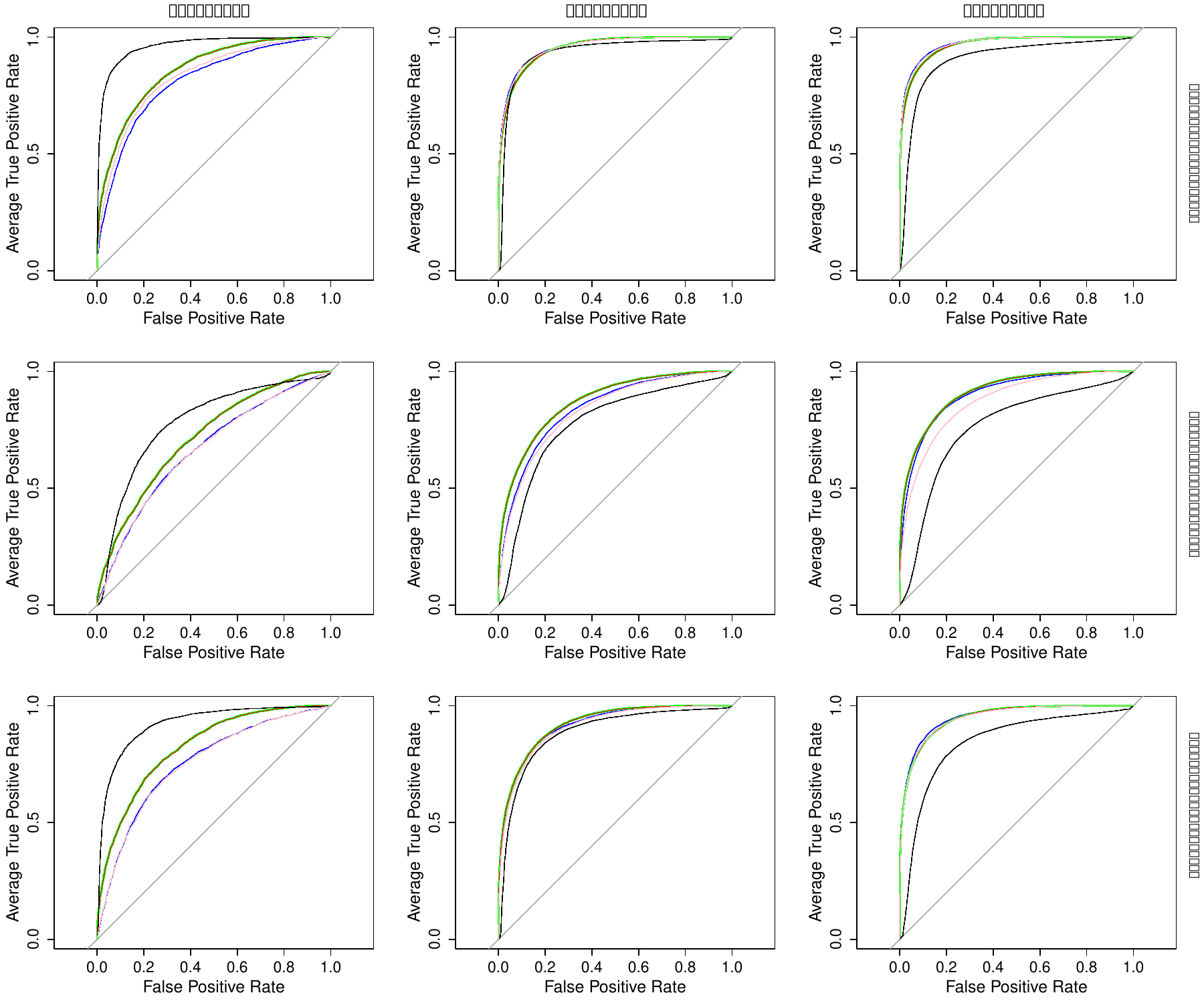}
            \caption{Observational Studies: ROC curves for scenarios with 10 explanatory variables that are highly correlated and errors that are uncorrelated, i.e., $\rho_1=2/3$ and $\rho_2=0$}
        \end{subfigure}
    \end{minipage}
\end{figure}
\begin{figure}[ht]
    \centering
    \begin{minipage}[b]{0.8\textwidth}
        \centering
        \begin{subfigure}[b]{\textwidth}
            \centering
            \includegraphics[width=0.9\textwidth]{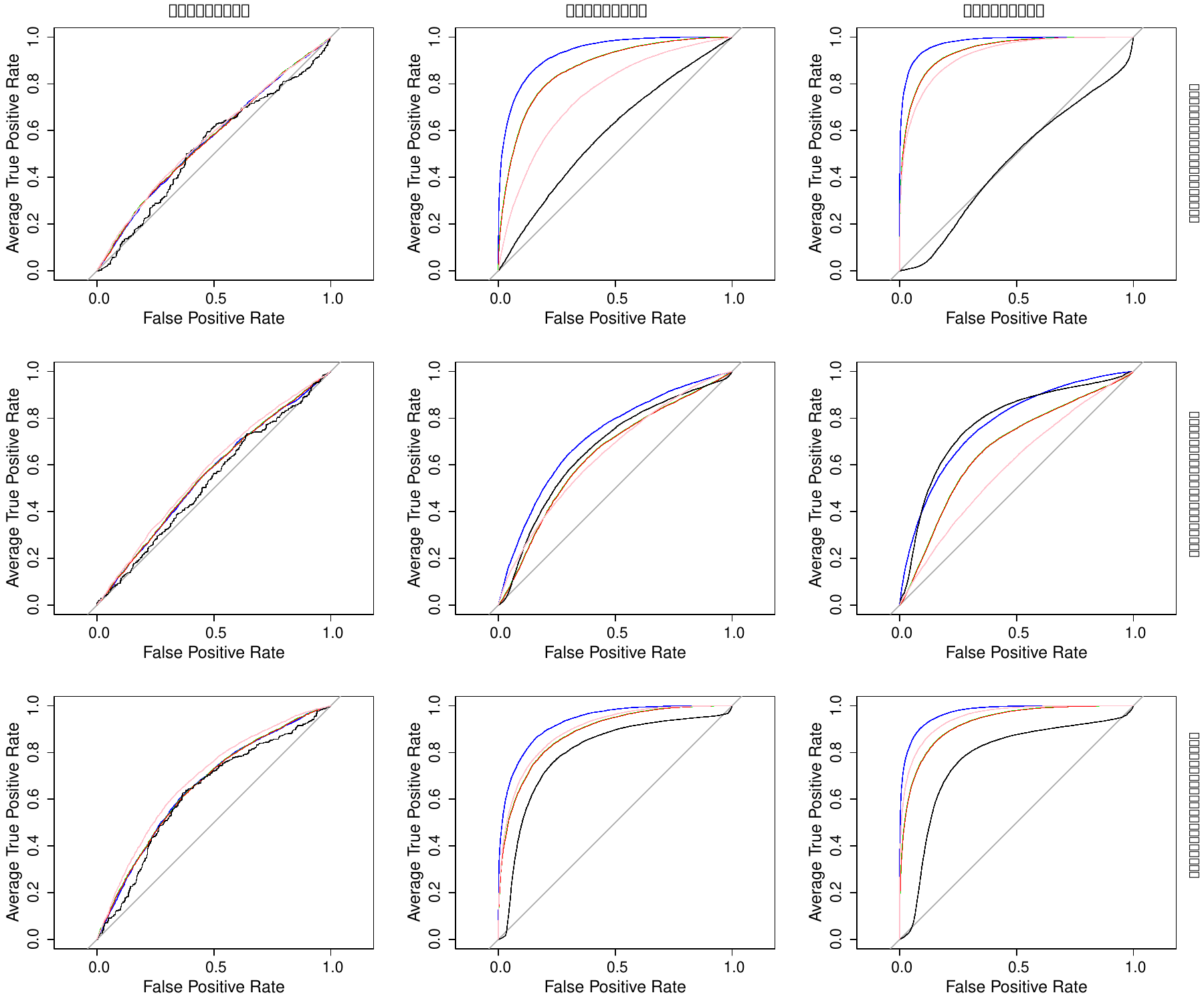}
            \caption{RCTs: ROC curves for scenarios with 50 explanatory variables that are highly correlated and errors that are uncorrelated, i.e., $\rho_1=2/3$ and $\rho_2=0$}
        \end{subfigure}
    \end{minipage}
    \vspace{0.5cm} 
    \begin{minipage}[b]{0.8\textwidth}
        \centering
        \begin{subfigure}[b]{\textwidth}
            \centering
            \includegraphics[width=0.9\textwidth]{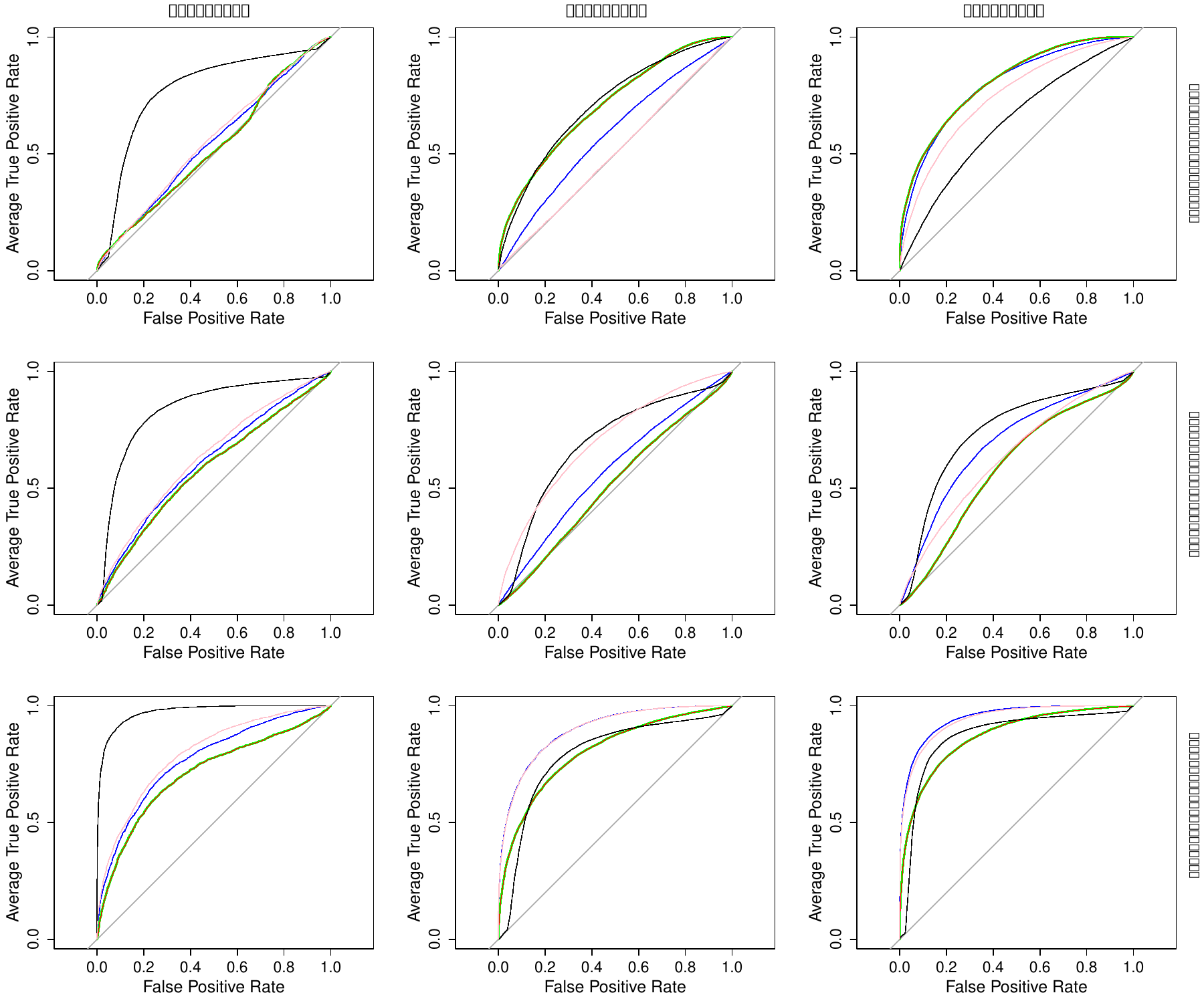}
            \caption{Observational Studies: ROC curves for scenarios with 50 explanatory variables that are highly correlated and errors that are uncorrelated, i.e., $\rho_1=2/3$ and $\rho_2=0$}        \end{subfigure}
    \end{minipage}
    
\end{figure}
\begin{figure}[ht]
    \centering
    \begin{minipage}[b]{0.8\textwidth}
        \centering
        \begin{subfigure}[b]{\textwidth}
            \centering
            \includegraphics[width=0.9\textwidth]{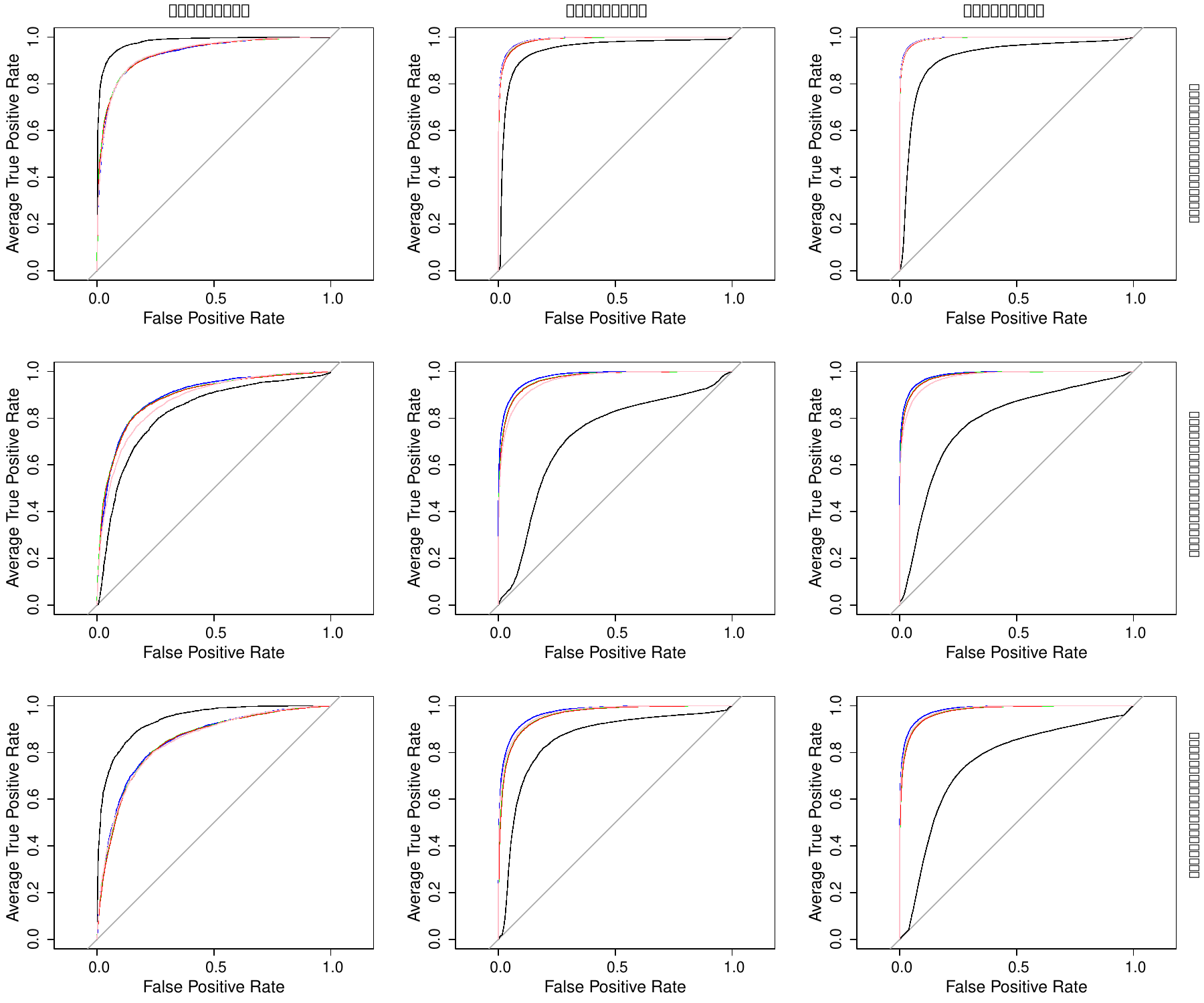}
            \caption{RCTs: ROC curves for scenarios with 10 explanatory variables that are uncorrelated and errors that are moderately correlated, i.e., $\rho_1=0$ and $\rho_2=1/3$}
        \end{subfigure}
    \end{minipage}
    \vspace{0.5cm} 
    \begin{minipage}[b]{0.8\textwidth}
        \centering
        \begin{subfigure}[b]{\textwidth}
            \centering
            \includegraphics[width=0.9\textwidth]{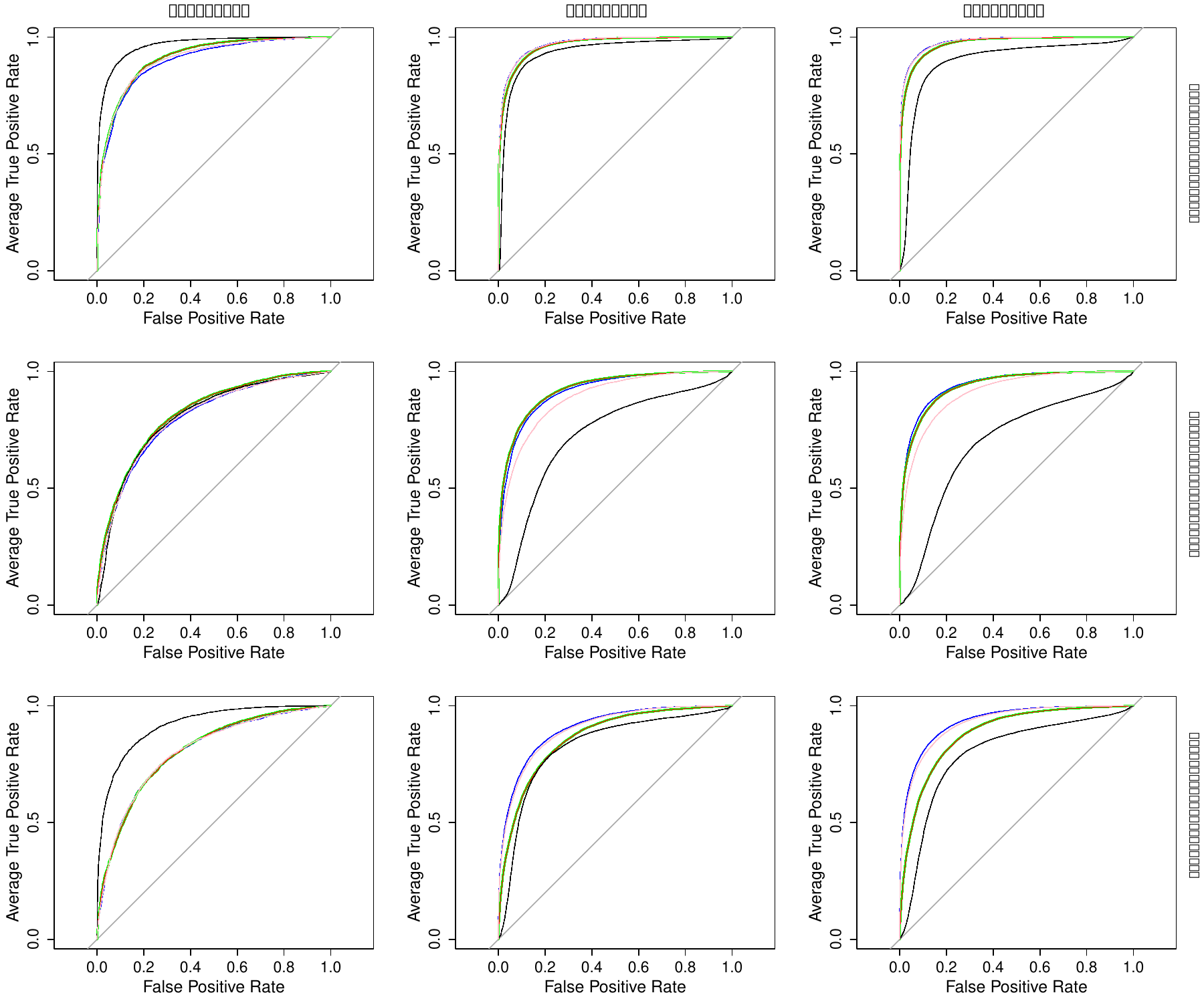}
            \caption{Observational Studies: ROC curves for scenarios with 10 explanatory variables that are uncorrelated and errors that are moderately correlated, i.e., $\rho_1=0$ and $\rho_2=1/3$}
        \end{subfigure}
    \end{minipage}
\end{figure}
\begin{figure}[ht]
    \centering
    \begin{minipage}[b]{0.8\textwidth}
        \centering
        \begin{subfigure}[b]{\textwidth}
            \centering
            \includegraphics[width=0.9\textwidth]{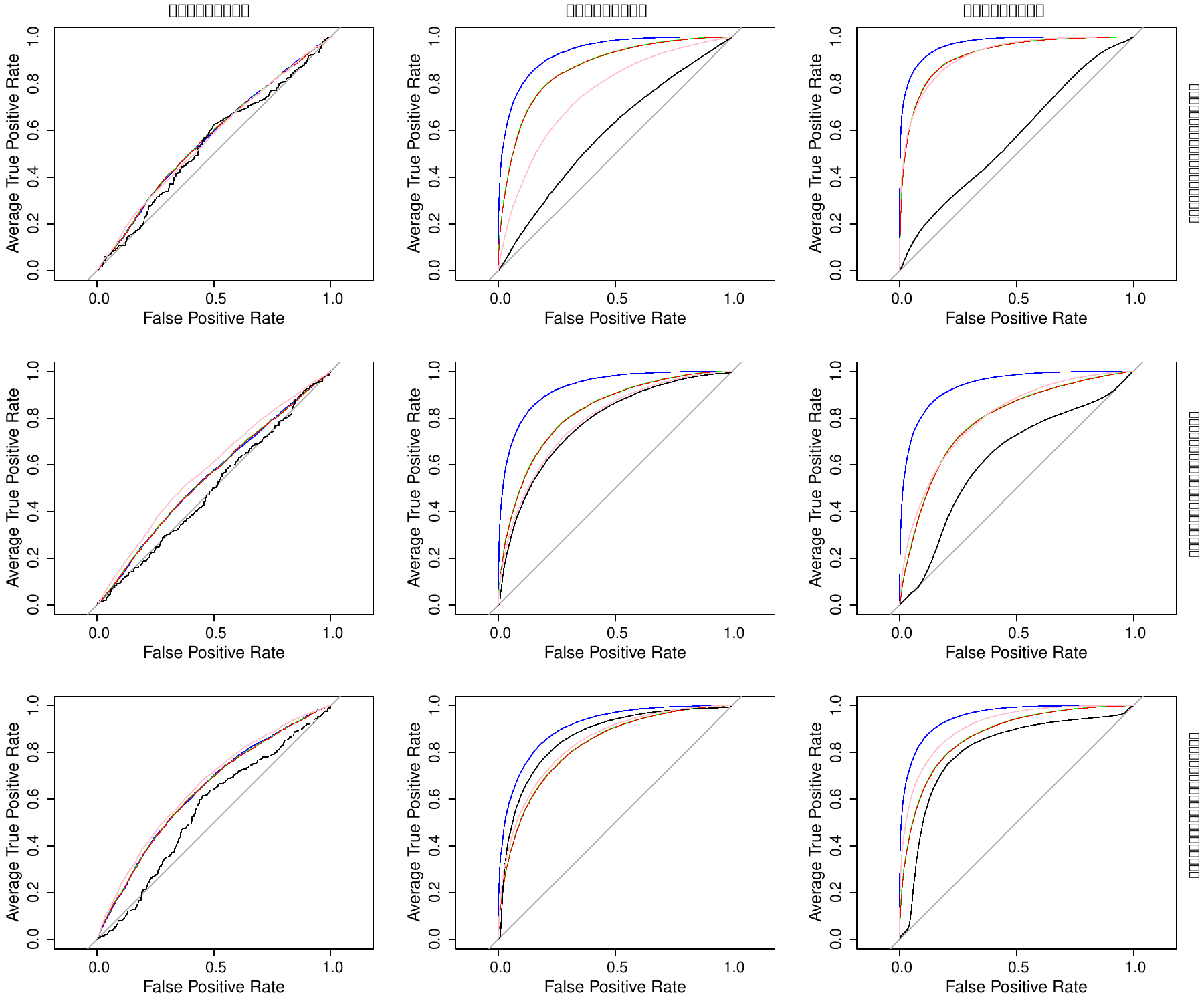}
            \caption{RCTs: ROC curves for scenarios with 50 explanatory variables that are uncorrelated and errors that are moderately correlated, i.e., $\rho_1=0$ and $\rho_2=1/3$}
        \end{subfigure}
    \end{minipage}
    \vspace{0.5cm} 
    \begin{minipage}[b]{0.8\textwidth}
        \centering
        \begin{subfigure}[b]{\textwidth}
            \centering
            \includegraphics[width=0.9\textwidth]{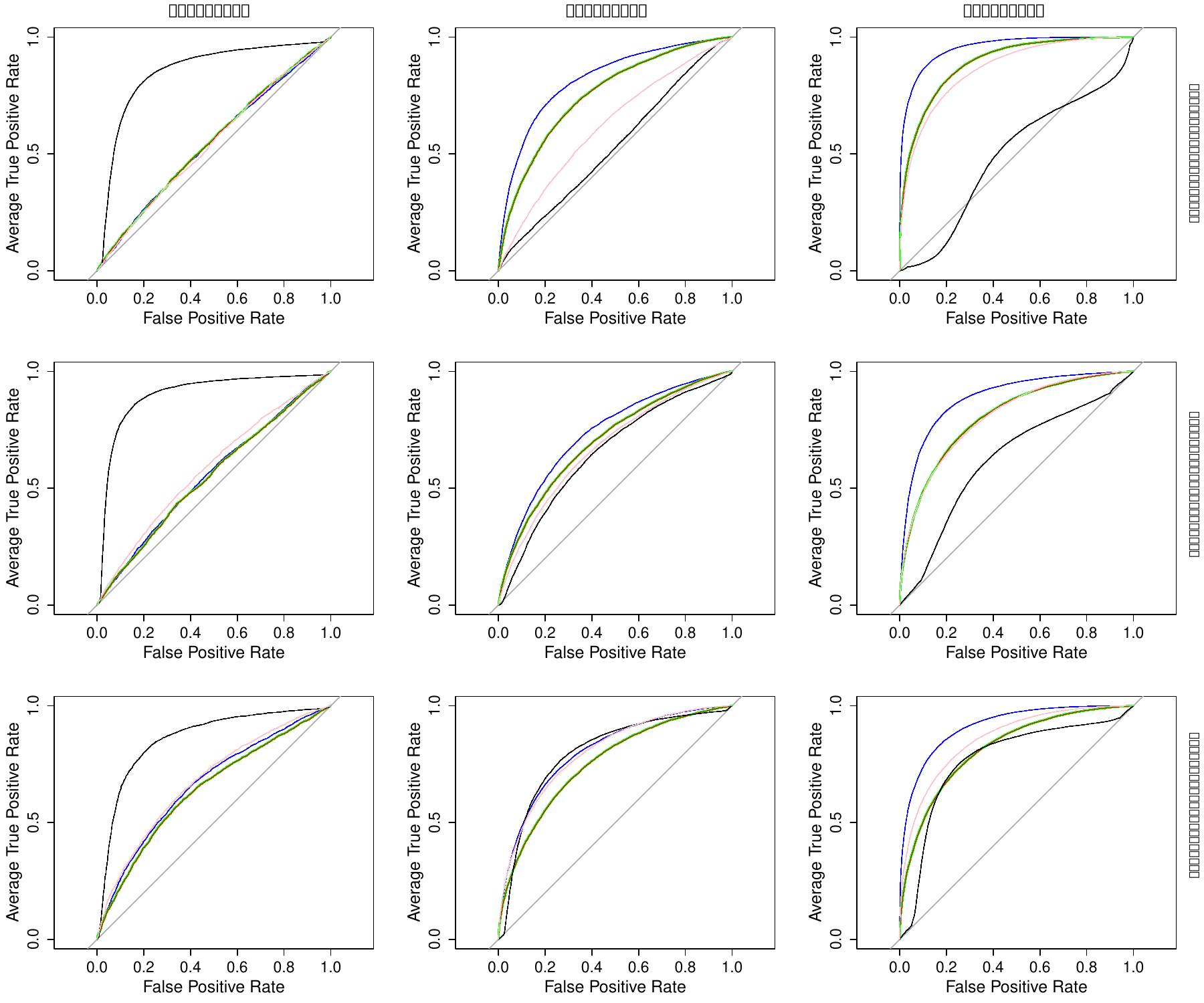}
            \caption{Observational Studies: ROC curves for scenarios with 50 explanatory variables that are uncorrelated and errors that are moderately correlated, i.e., $\rho_1=0$ and $\rho_2=1/3$}
        \end{subfigure}
    \end{minipage}
    \end{figure}
\begin{figure}[ht]
    \centering
    \begin{minipage}[b]{0.8\textwidth}
        \centering
        \begin{subfigure}[b]{\textwidth}
            \centering
            \includegraphics[width=0.9\textwidth]{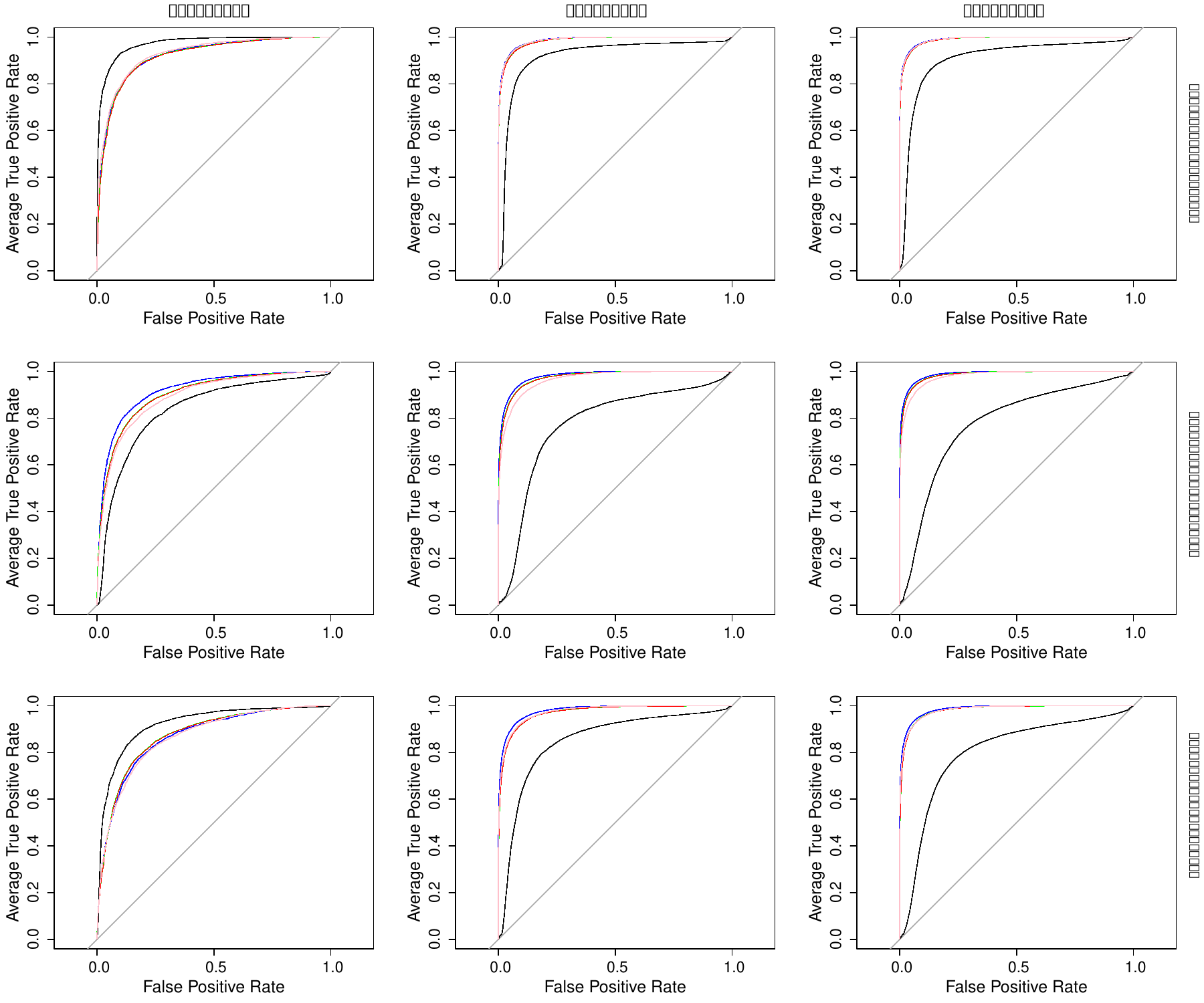}
            \caption{RCTs: ROC curves for scenarios with 10 explanatory variables and errors that are moderately correlated, i.e., $\rho_1=1/3$ and $\rho_2=1/3$}
        \end{subfigure}
    \end{minipage}
    \vspace{0.5cm} 
    \begin{minipage}[b]{0.8\textwidth}
        \centering
        \begin{subfigure}[b]{\textwidth}
            \centering
            \includegraphics[width=0.9\textwidth]{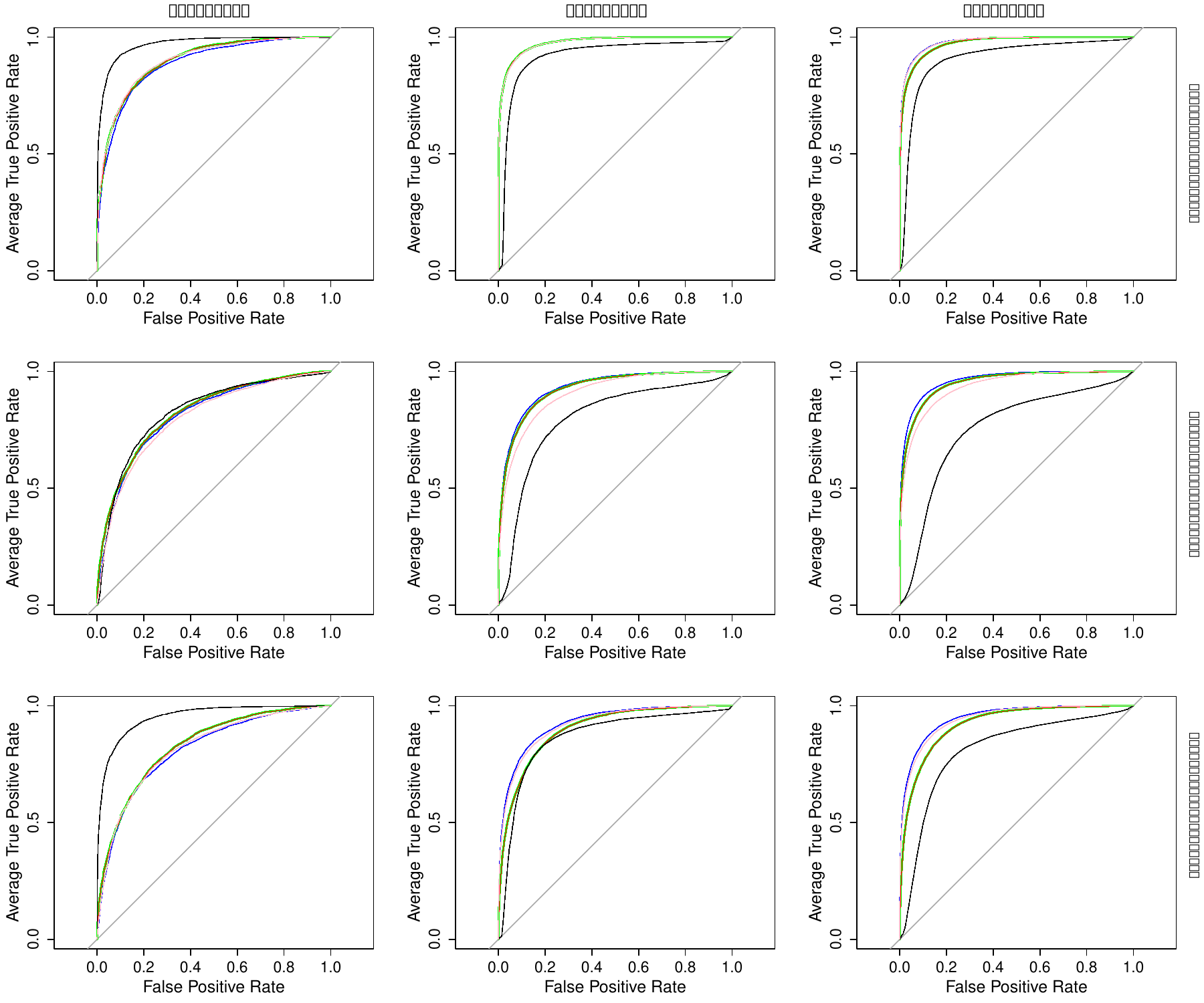}
            \caption{Observational Studies: ROC curves for scenarios with 10 explanatory variables and errors that are moderately correlated, i.e., $\rho_1=1/3$ and $\rho_2=1/3$}
        \end{subfigure}
    \end{minipage}
\end{figure}
\begin{figure}[ht]
    \centering
    \begin{minipage}[b]{0.8\textwidth}
        \centering
        \begin{subfigure}[b]{\textwidth}
            \centering
            \includegraphics[width=0.9\textwidth]{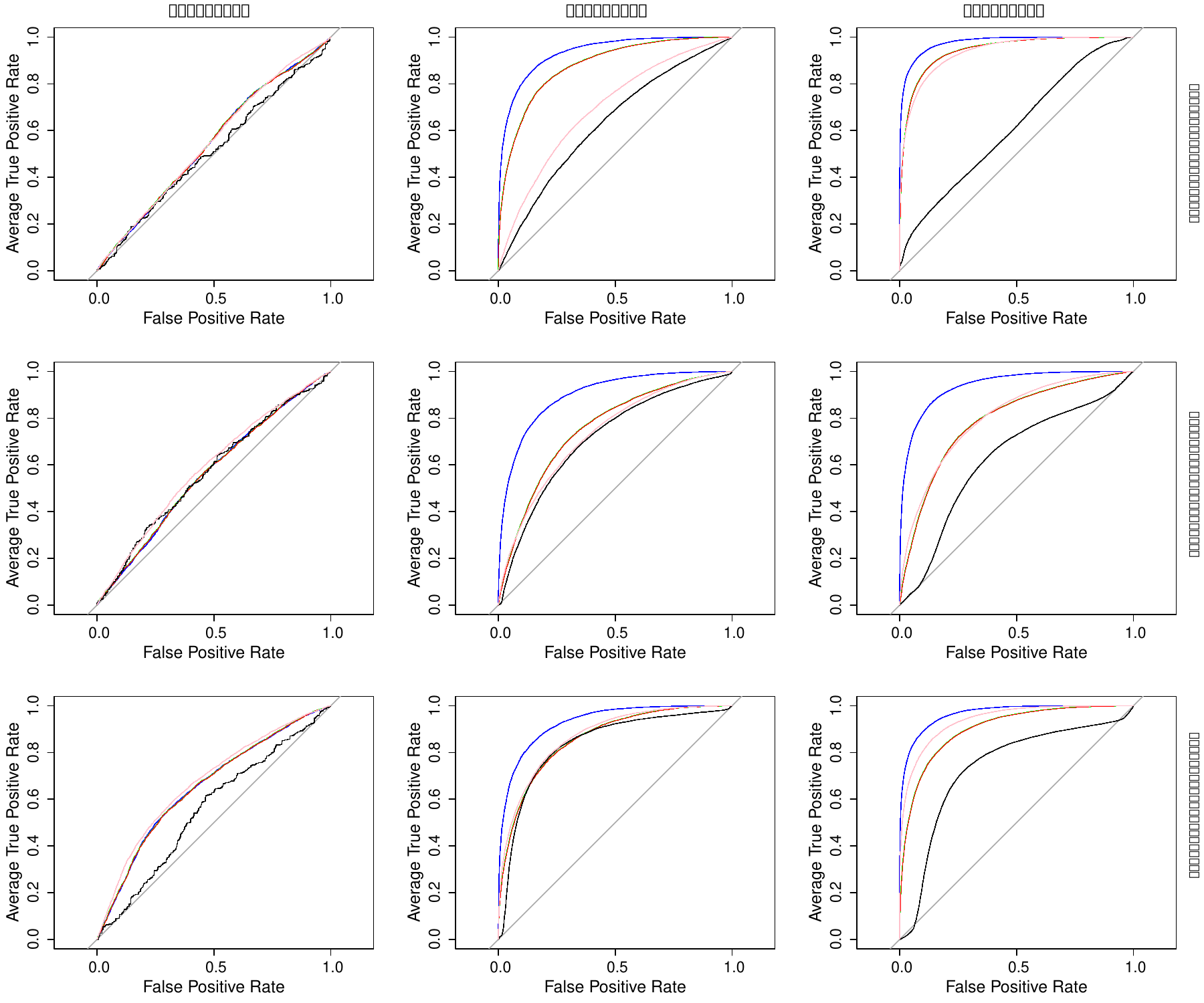}
            \caption{RCTs: ROC curves for scenarios with 50 explanatory variables and errors that are moderately correlated, i.e., $\rho_1=1/3$ and $\rho_2=1/3$}
        \end{subfigure}
    \end{minipage}
    \vspace{0.5cm} 
    \begin{minipage}[b]{0.8\textwidth}
        \centering
        \begin{subfigure}[b]{\textwidth}
            \centering
\includegraphics[width=0.9\textwidth]{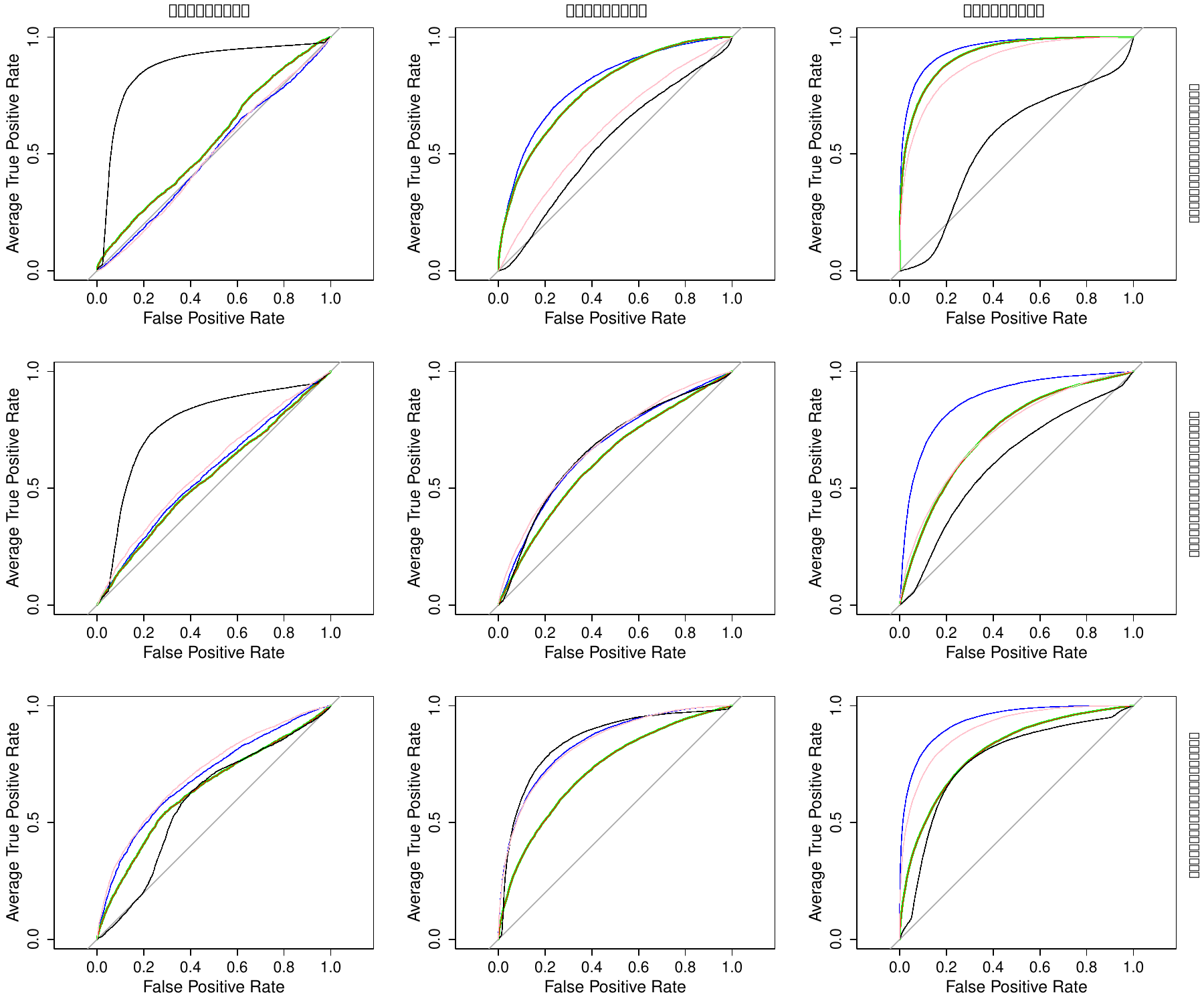}
            \caption{Observational Studies: ROC curves for scenarios with 50 explanatory variables and errors that are moderately correlated, i.e., $\rho_1=1/3$ and $\rho_2=1/3$}
        \end{subfigure}
    \end{minipage}
    \end{figure}
\begin{figure}[ht]
    \centering
    \begin{minipage}[b]{0.8\textwidth}
        \centering
        \begin{subfigure}[b]{\textwidth}
            \centering
            \includegraphics[width=0.9\textwidth]{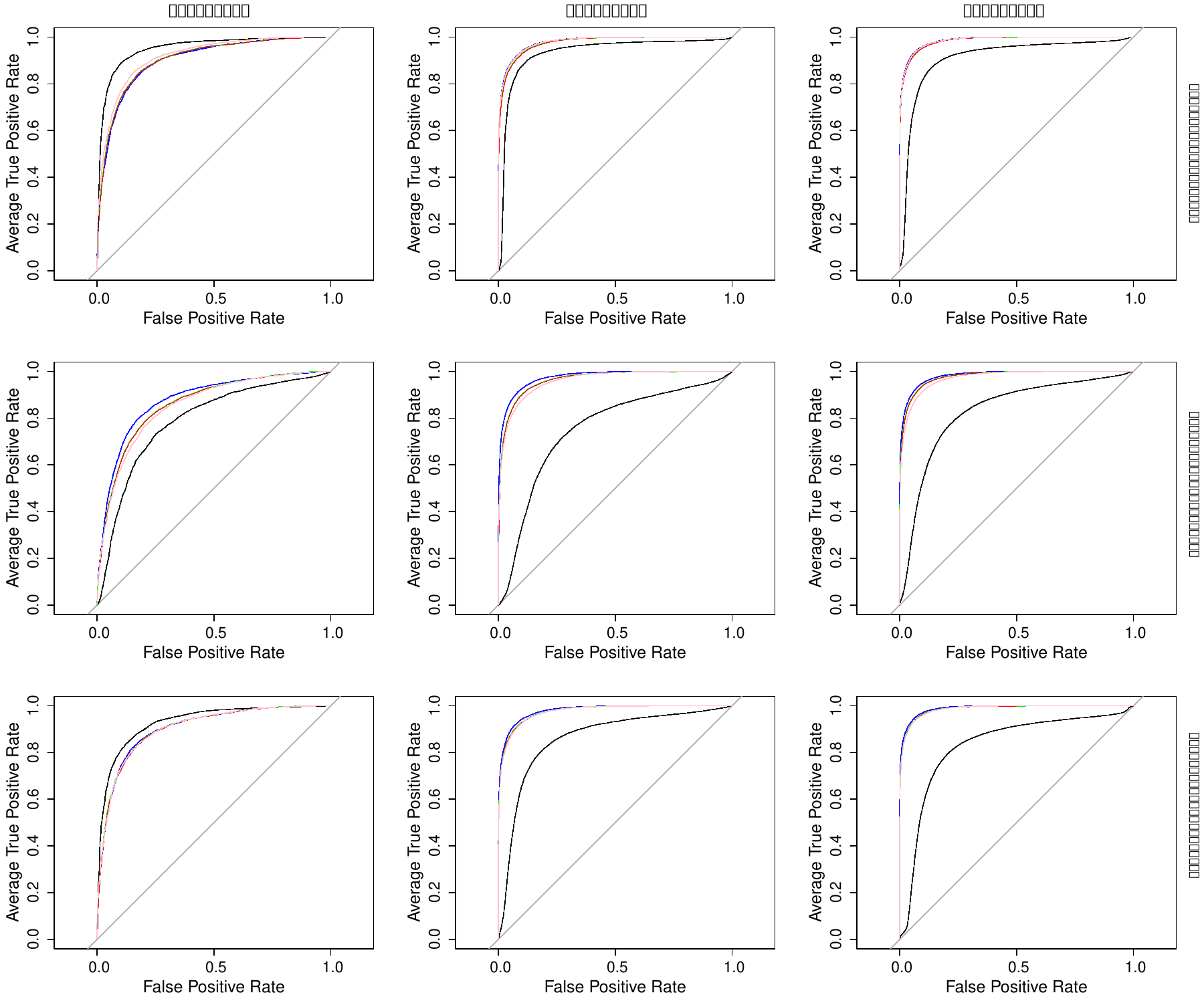}
            \caption{RCTs: ROC curves for scenarios with 10 explanatory variables that are highly correlated and errors that are moderately correlated, i.e., $\rho_1=2/3$ and $\rho_2=1/3$}
        \end{subfigure}
    \end{minipage}
    \vspace{0.5cm} 
    \begin{minipage}[b]{0.8\textwidth}
        \centering
        \begin{subfigure}[b]{\textwidth}
            \centering
            \includegraphics[width=0.9\textwidth]{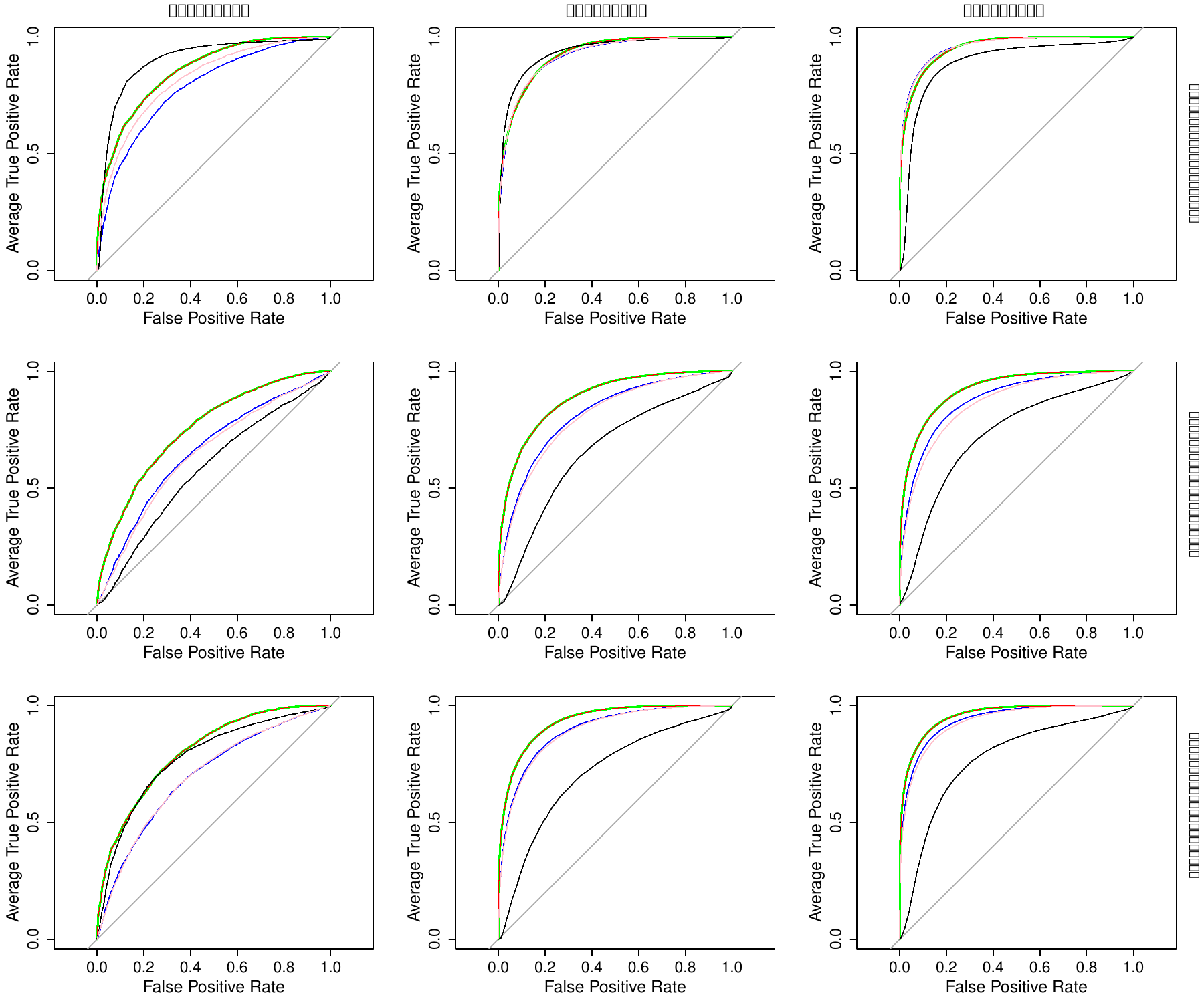}
            \caption{Observational Studies: ROC curves for scenarios with 10 explanatory variables that are highly correlated and errors that are moderately correlated, i.e., $\rho_1=2/3$ and $\rho_2=1/3$}
        \end{subfigure}
    \end{minipage}
\end{figure}
\begin{figure}[ht]
    \centering
    \begin{minipage}[b]{0.8\textwidth}
        \centering
        \begin{subfigure}[b]{\textwidth}
            \centering
            \includegraphics[width=0.9\textwidth]{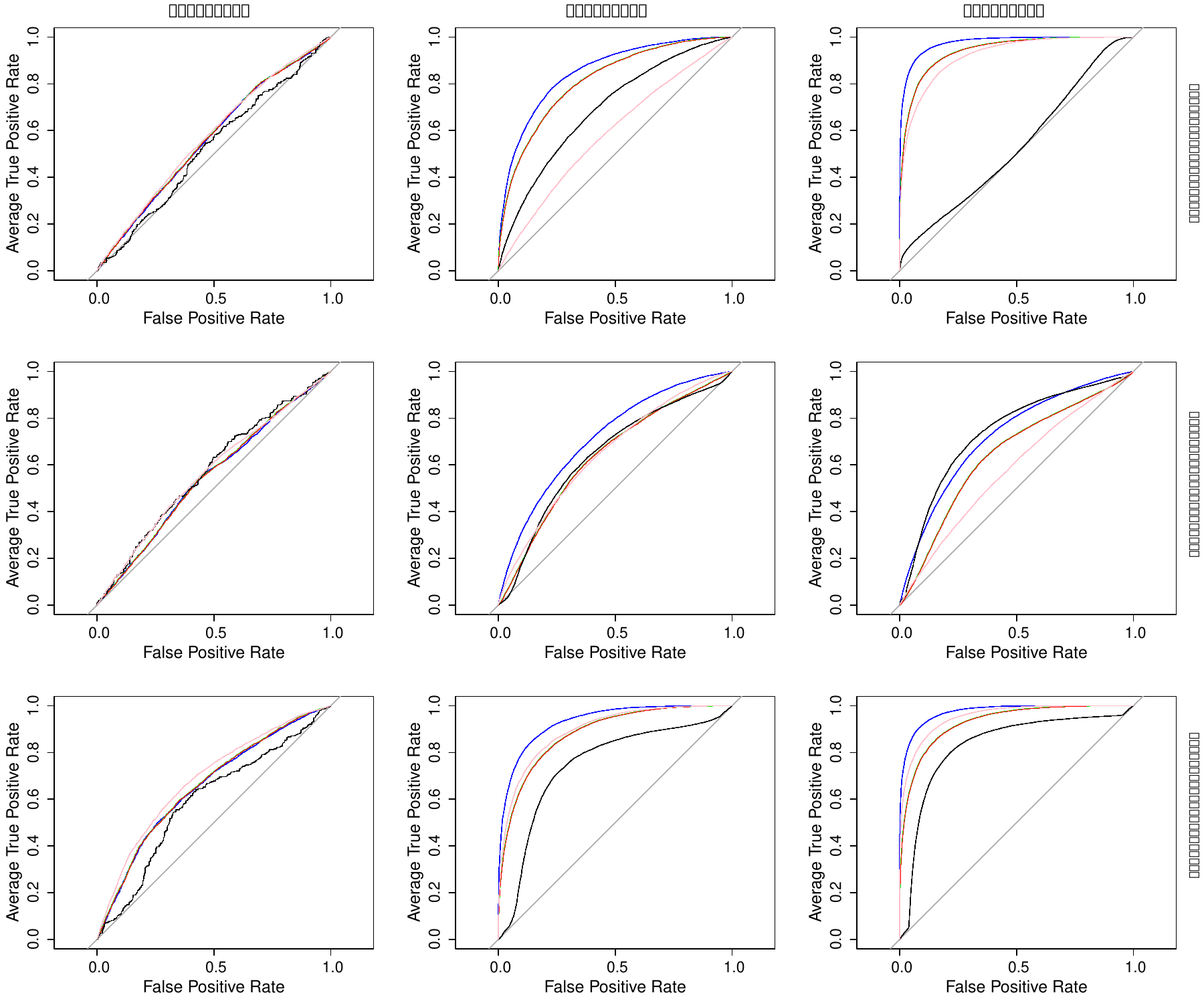}
            \caption{RCTs: ROC curves for scenarios with 50 explanatory variables that are highly correlated and errors that are moderately correlated, i.e., $\rho_1=2/3$ and $\rho_2=1/3$}
        \end{subfigure}
    \end{minipage}
    \vspace{0.5cm} 
    \begin{minipage}[b]{0.8\textwidth}
        \centering
        \begin{subfigure}[b]{\textwidth}
            \centering
            \includegraphics[width=0.9\textwidth]{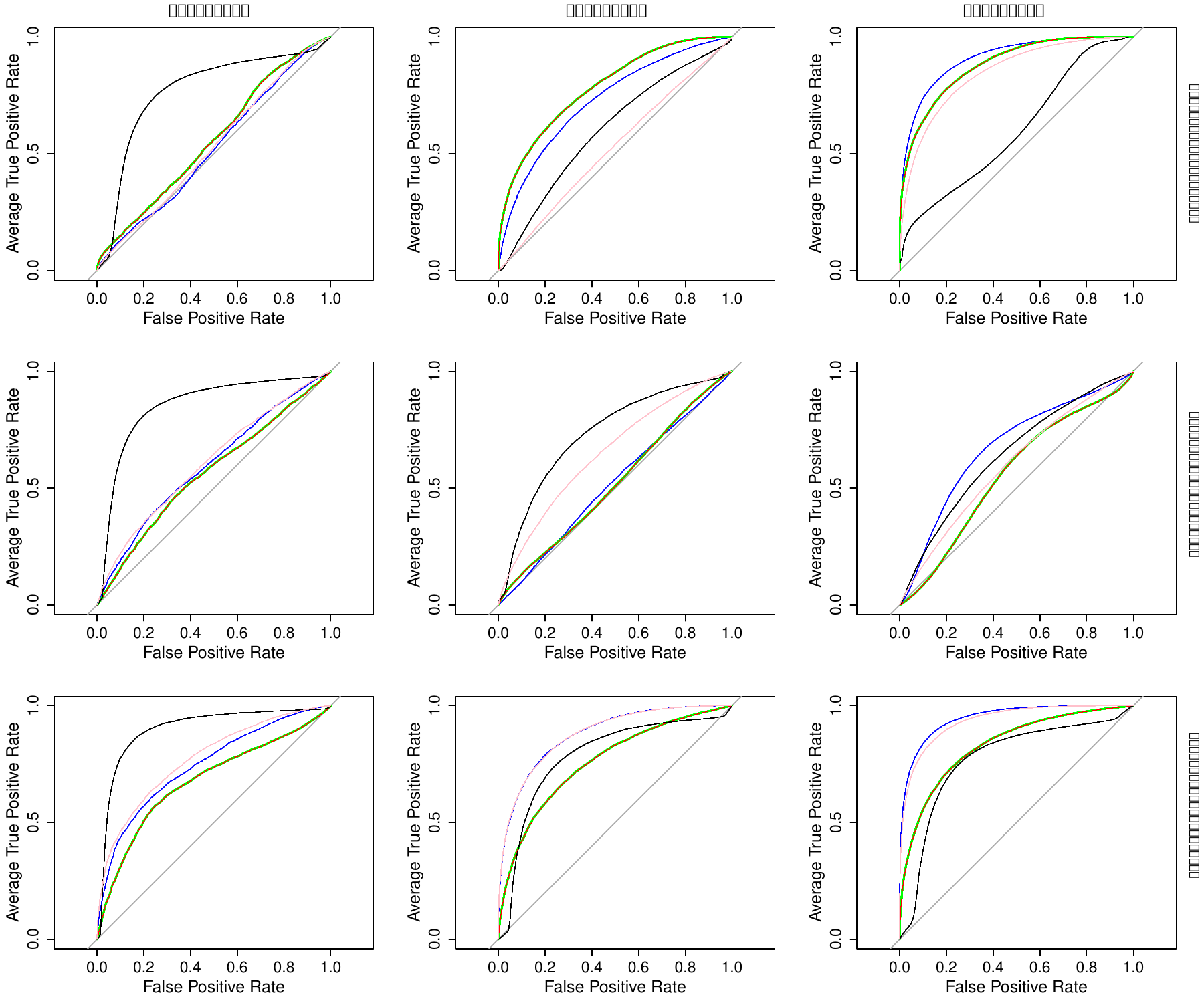}
            \caption{Observational Studies: ROC curves for scenarios with 50 explanatory variables that are highly correlated and errors that are moderately correlated, i.e., $\rho_1=2/3$ and $\rho_2=1/3$}
        \end{subfigure}
    \end{minipage}
    \end{figure}
\begin{figure}[ht]
    \centering
    \begin{minipage}[b]{0.8\textwidth}
        \centering
        \begin{subfigure}[b]{\textwidth}
            \centering
            \includegraphics[width=0.9\textwidth]{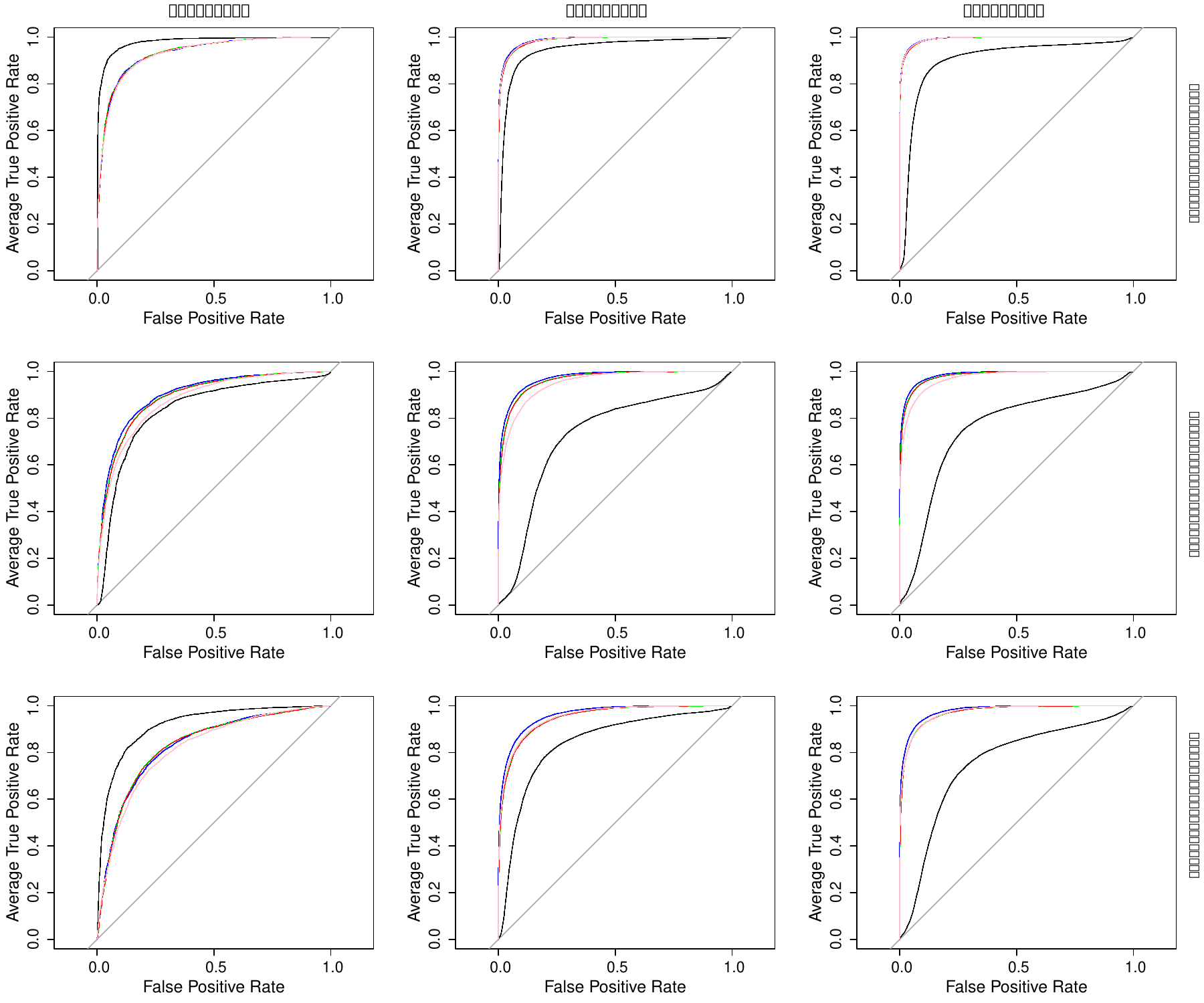}
            \caption{RCTs: ROC curves for scenarios with 10 explanatory variables that are uncorrelated and errors that are highly correlated, i.e., $\rho_1=0$ and $\rho_2=2/3$}
        \end{subfigure}
    \end{minipage}
    \vspace{0.5cm} 
    \begin{minipage}[b]{0.8\textwidth}
        \centering
        \begin{subfigure}[b]{\textwidth}
            \centering
            \includegraphics[width=0.9\textwidth]{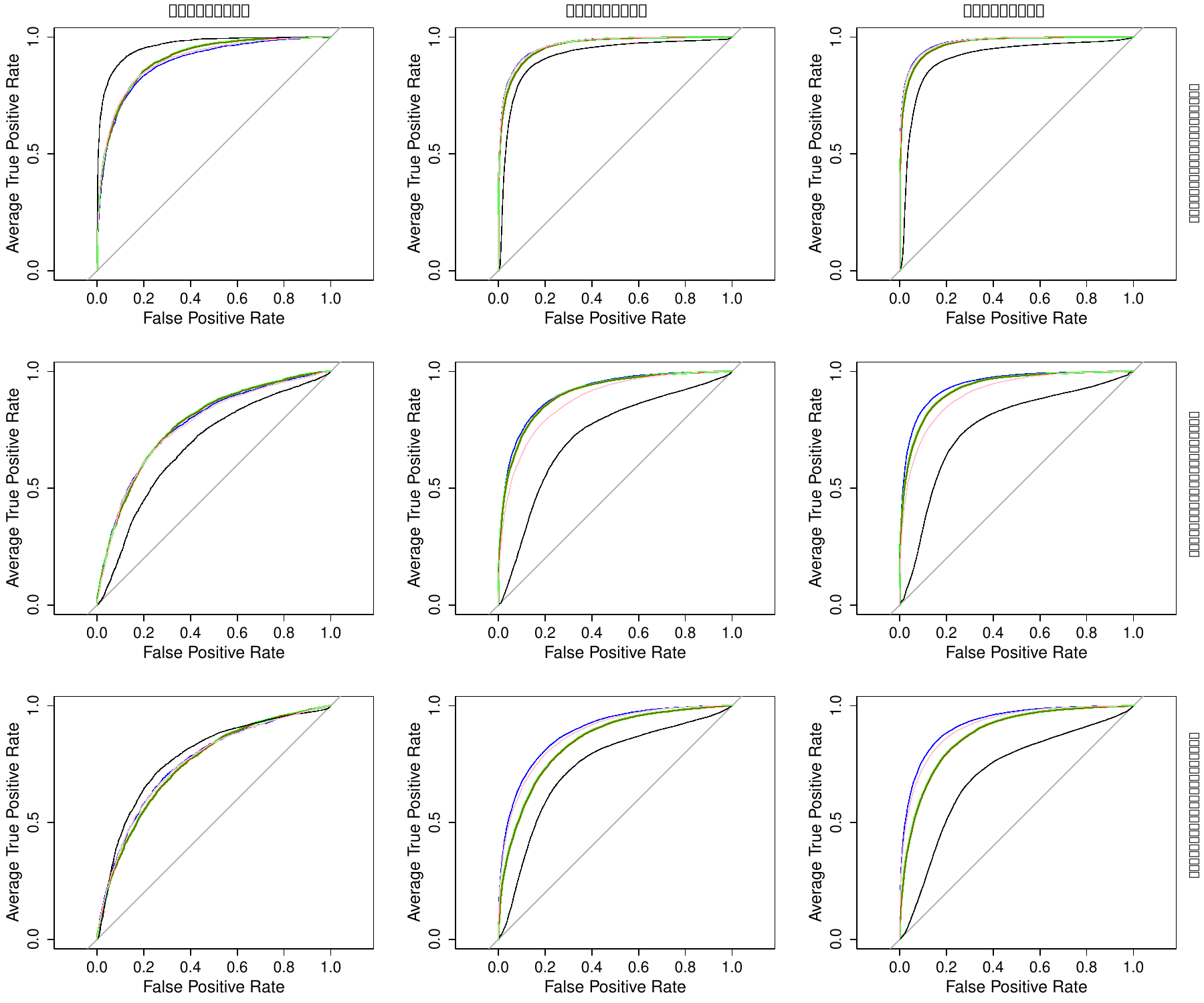}
            \caption{Observational Studies: ROC curves for scenarios with 10 explanatory variables that are uncorrelated and errors that are highly correlated, i.e., $\rho_1=0$ and $\rho_2=2/3$}
        \end{subfigure}
    \end{minipage}
\end{figure}
\begin{figure}[ht]
    \centering
    \begin{minipage}[b]{0.8\textwidth}
        \centering
        \begin{subfigure}[b]{\textwidth}
            \centering
            \includegraphics[width=0.9\textwidth]{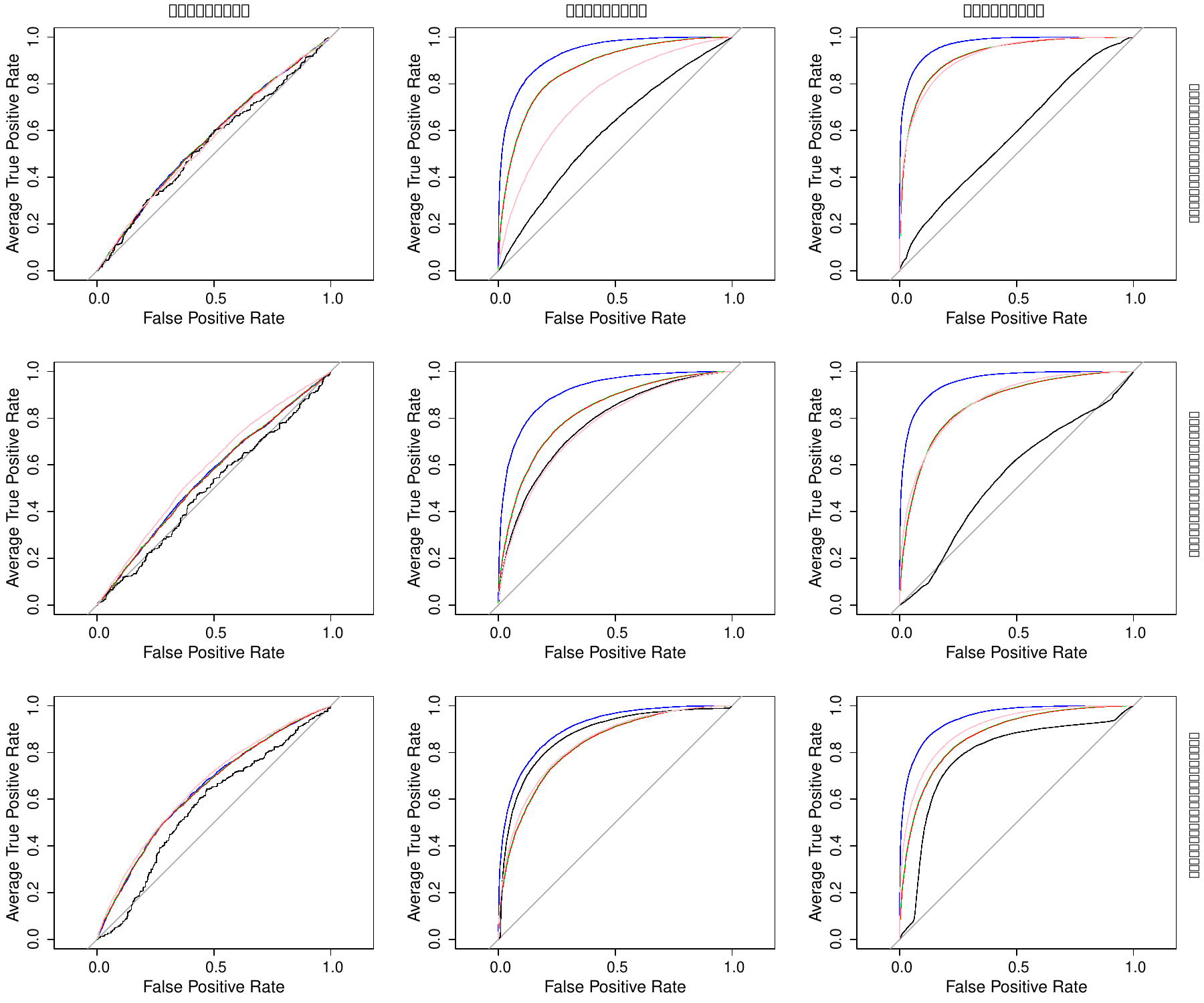}
            \caption{RCTs: ROC curves for scenarios with 50 explanatory variables that are uncorrelated and errors that are highly correlated, i.e., $\rho_1=0$ and $\rho_2=2/3$}
        \end{subfigure}
    \end{minipage}
    \vspace{0.5cm} 
    \begin{minipage}[b]{0.8\textwidth}
        \centering
        \begin{subfigure}[b]{\textwidth}
            \centering
            \includegraphics[width=0.9\textwidth]{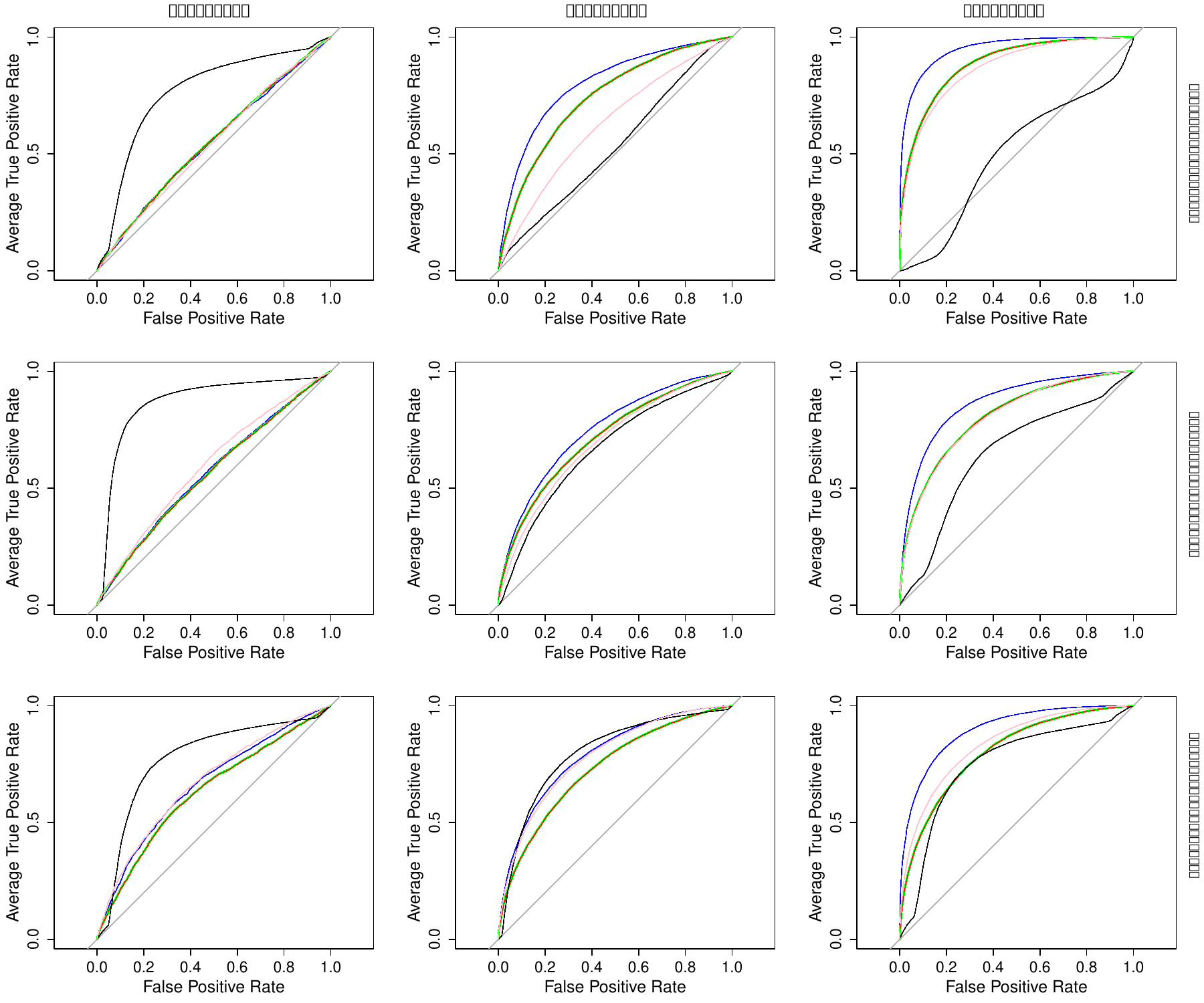}
            \caption{Observational Studies: ROC curves for scenarios with 50 explanatory variables that are uncorrelated and errors that are highly correlated, i.e., $\rho_1=0$ and $\rho_2=2/3$}
        \end{subfigure}
    \end{minipage}
    \end{figure}
\begin{figure}[ht]
    \centering
    \begin{minipage}[b]{0.8\textwidth}
        \centering
        \begin{subfigure}[b]{\textwidth}
            \centering
            \includegraphics[width=0.9\textwidth]{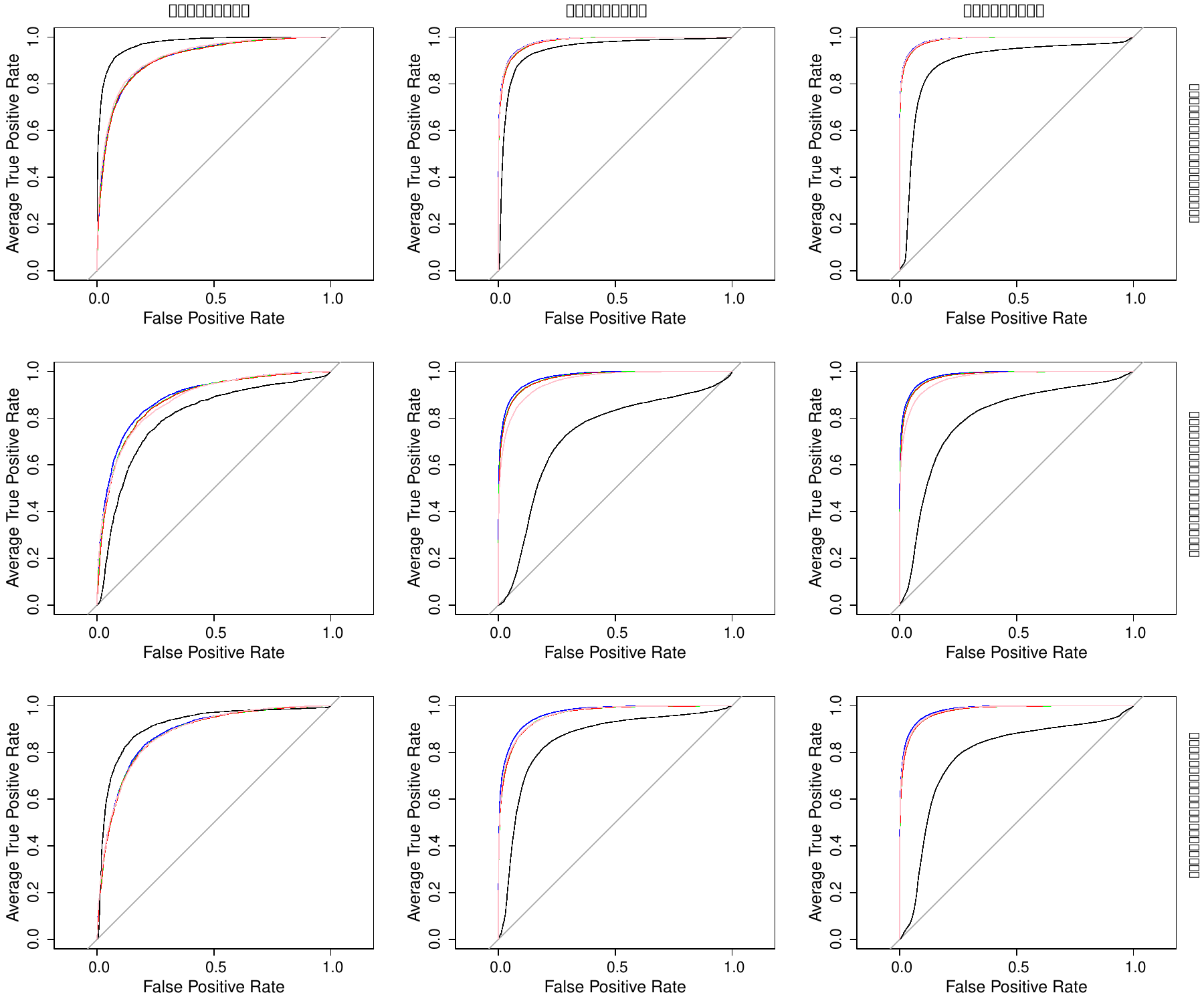}
            \caption{RCTs: ROC curves for scenarios with 10 explanatory variables that are moderately correlated and errors that are highly correlated, i.e., $\rho_1=1/3$ and $\rho_2=2/3$}
        \end{subfigure}
    \end{minipage}
    \vspace{0.5cm} 
    \begin{minipage}[b]{0.8\textwidth}
        \centering
        \begin{subfigure}[b]{\textwidth}
            \centering
            \includegraphics[width=0.9\textwidth]{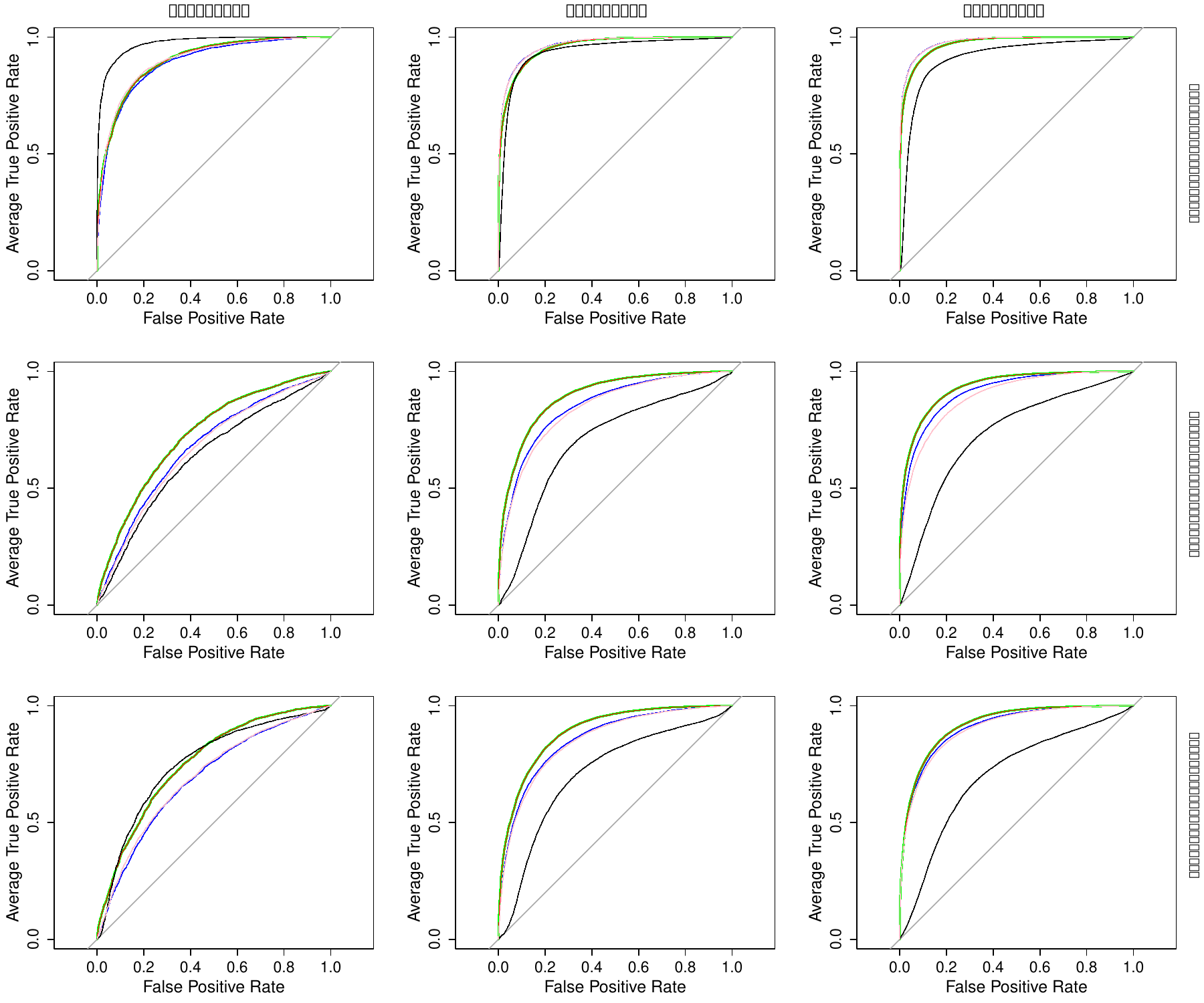}
            \caption{Observational Studies: ROC curves for scenarios with 10 explanatory variables that are moderately correlated and errors that are highly correlated, i.e., $\rho_1=1/3$ and $\rho_2=2/3$}
        \end{subfigure}
    \end{minipage}
\end{figure}
\begin{figure}[ht]
    \centering
    \begin{minipage}[b]{0.8\textwidth}
        \centering
        \begin{subfigure}[b]{\textwidth}
            \centering
            \includegraphics[width=0.9\textwidth]{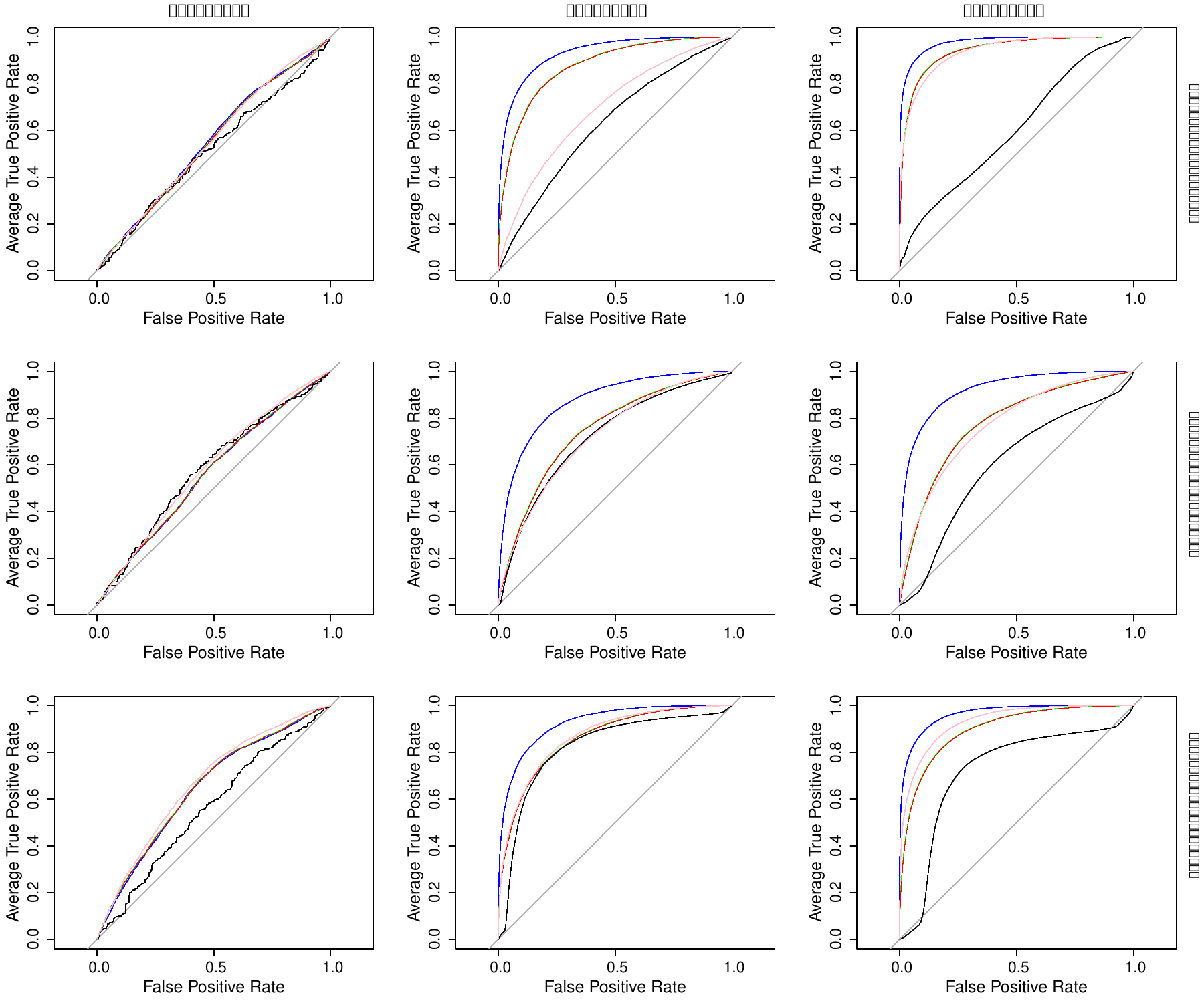}
            \caption{RCTs: ROC curves for scenarios with 50 explanatory variables that are moderately correlated and errors that are highly correlated, i.e., $\rho_1=1/3$ and $\rho_2=2/3$}
        \end{subfigure}
    \end{minipage}
    \vspace{0.5cm} 
    \begin{minipage}[b]{0.8\textwidth}
        \centering
        \begin{subfigure}[b]{\textwidth}
            \centering
            \includegraphics[width=0.9\textwidth]{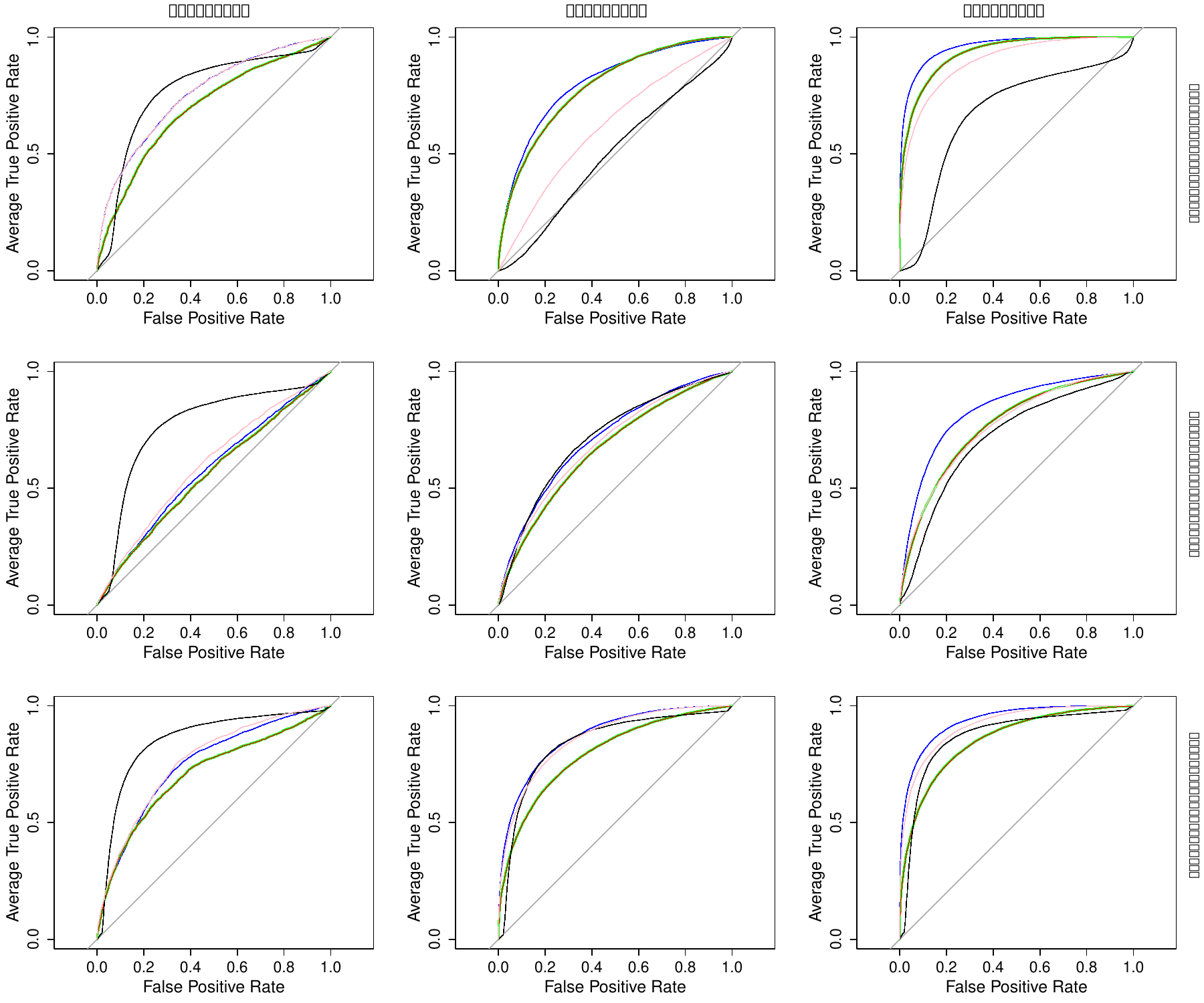}
            \caption{Observational Studies: ROC curves for scenarios with 50 explanatory variables that are moderately correlated and errors that are highly correlated, i.e., $\rho_1=1/3$ and $\rho_2=2/3$}
        \end{subfigure}
    \end{minipage}
    \end{figure}
\begin{figure}[ht]
    \centering
    \begin{minipage}[b]{0.8\textwidth}
        \centering
        \begin{subfigure}[b]{\textwidth}
            \centering
            \includegraphics[width=0.9\textwidth]{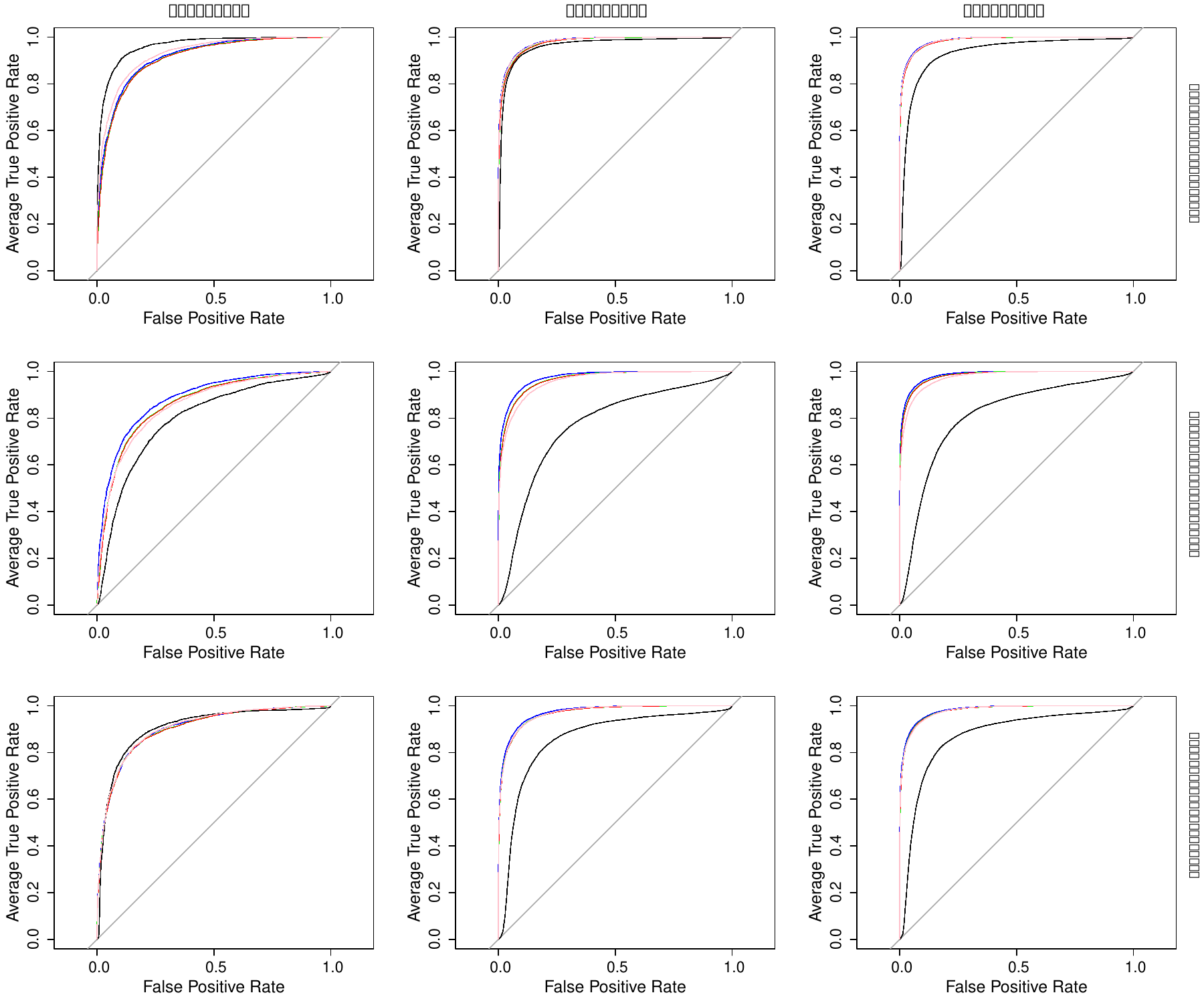}
            \caption{RCTs: ROC curves for scenarios with 10 explanatory variables and errors that are highly correlated, i.e., $\rho_1=2/3$ and $\rho_2=2/3$}
        \end{subfigure}
    \end{minipage}
    \vspace{0.5cm} 
    \begin{minipage}[b]{0.8\textwidth}
        \centering
        \begin{subfigure}[b]{\textwidth}
            \centering
            \includegraphics[width=0.9\textwidth]{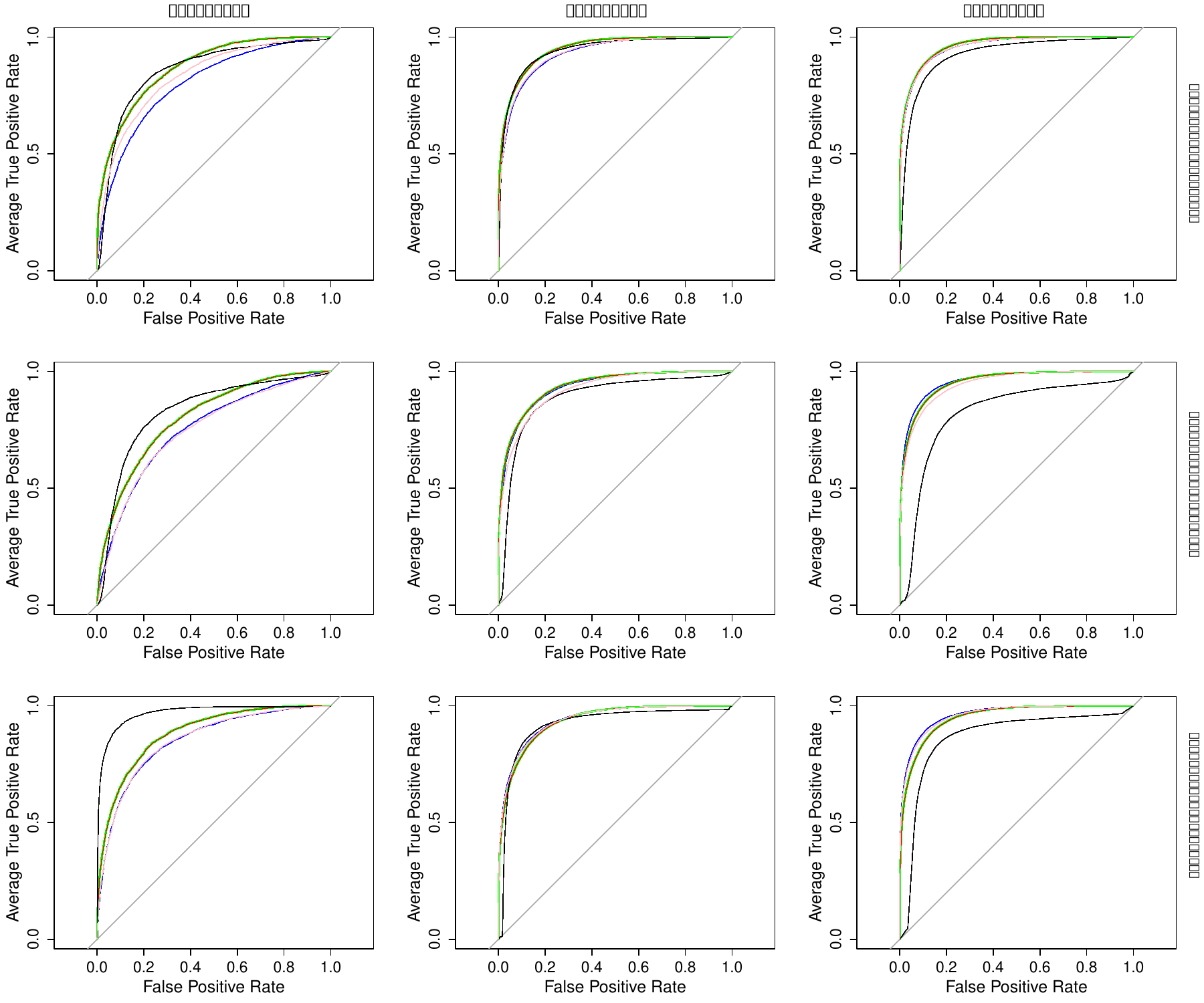}
            \caption{Observational Studies: ROC curves for scenarios with 10 explanatory variables and errors that are highly correlated, i.e., $\rho_1=2/3$ and $\rho_2=2/3$}
        \end{subfigure}
    \end{minipage}
\end{figure}
\begin{figure}[ht]
    \centering
    \begin{minipage}[b]{0.8\textwidth}
        \centering
        \begin{subfigure}[b]{\textwidth}
            \centering
            \includegraphics[width=0.9\textwidth]{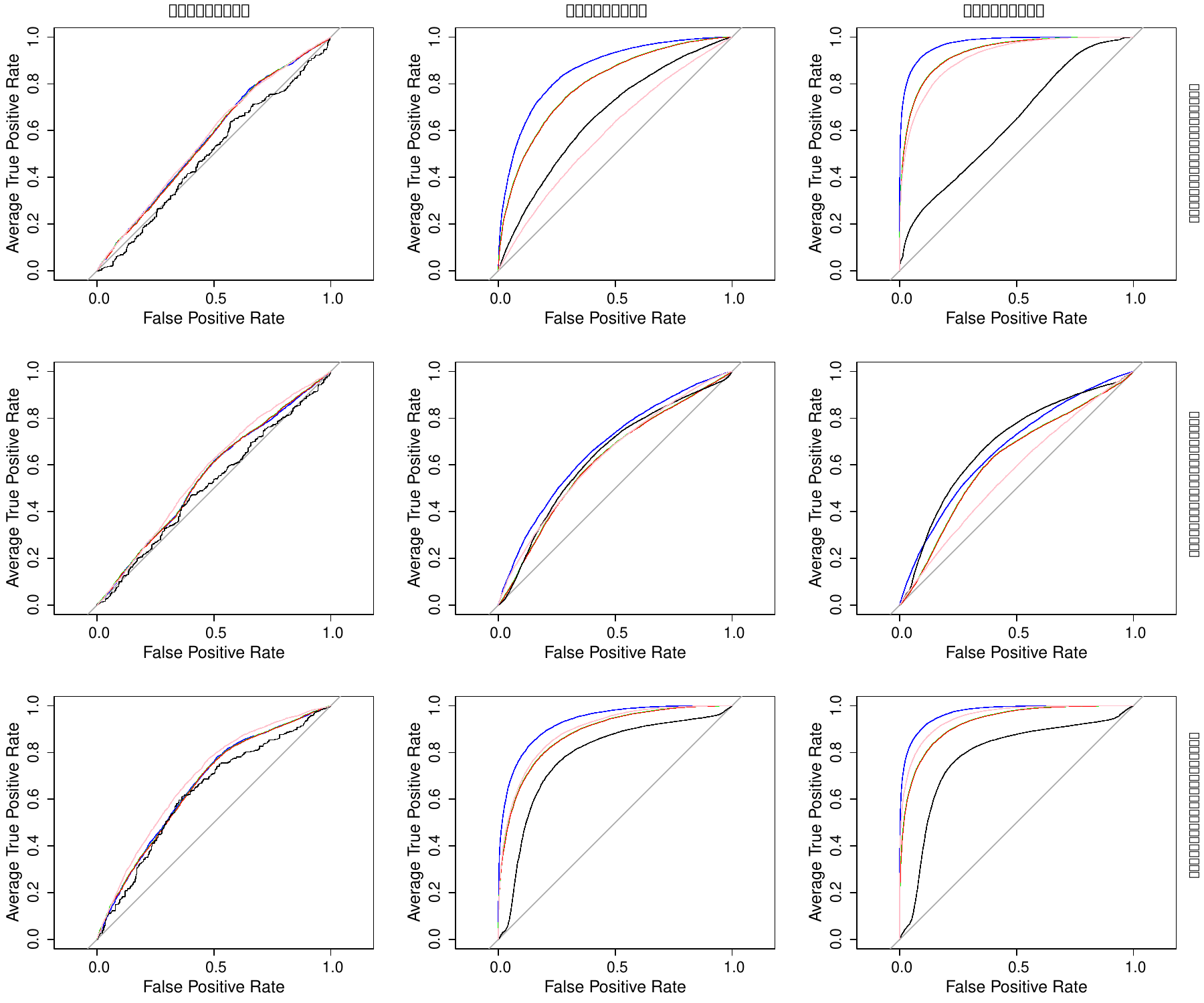}
            \caption{RCTs: ROC curves for scenarios with 50 explanatory variables and errors that are highly correlated, i.e., $\rho_1=2/3$ and $\rho_2=2/3$}
        \end{subfigure}
    \end{minipage}
    \vspace{0.5cm} 
    \begin{minipage}[b]{0.8\textwidth}
        \centering
        \begin{subfigure}[b]{\textwidth}
            \centering
            \includegraphics[width=0.9\textwidth]{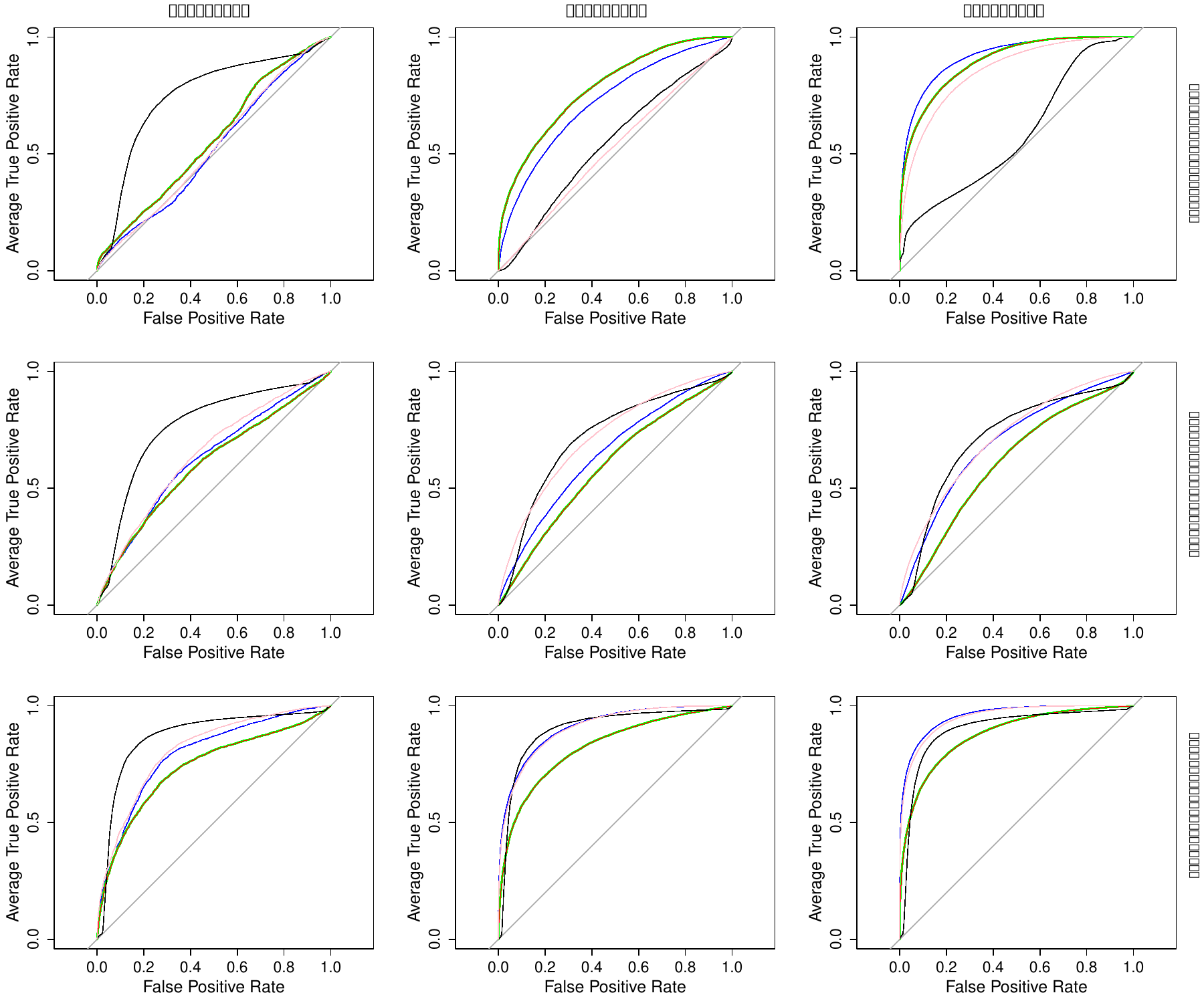}
            \caption{Observational Studies: ROC curves for scenarios with 50 explanatory variables and errors that are highly correlated, i.e., $\rho_1=2/3$ and $\rho_2=2/3$}
        \end{subfigure}
    \end{minipage}
    \end{figure}

\end{document}